\documentclass[final,leqno,onefignum]{siamltex}

\usepackage[utf8]{inputenc}
\usepackage[dvips]{epsfig}
\usepackage{amsmath,amssymb,amsfonts,mathrsfs,latexsym}
\usepackage{enumerate,enumitem} 
\usepackage{graphicx,color}

\newtheorem{example}{Example}[section]
\newtheorem{remark}{Remark}[section]
\newtheorem{assumption}{Assumption}

\newcommand{\nn}{\nonumber}
\newcommand{\beq}{\begin{equation}}
\newcommand{\eeq}{\end{equation}}
\newcommand{\beqa}{\begin{eqnarray}}
\newcommand{\eeqa}{\end{eqnarray}}
\newcommand{\lara}[1]{\langle #1 \rangle}
\newcommand{\Lara}[1]{\Big\langle #1 \Big\rangle}
\newcommand{\arr}[2]{\begin{array}{#1}#2\end{array}}
\newcommand{\norm}[2]{{\| #1 \|}_{#2}}
\newcommand{\sbs}{\subseteq}
\newcommand{\sps}{\supseteq}
\newcommand{\ovl}[1]{\overline{#1}}
\newcommand{\wtil}[1]{\widetilde{#1}}

\renewcommand{\d}{\ensuremath{\partial}}
\renewcommand{\div}{\mathrm{div}\:}
\newcommand{\surf}{\wtil}
\newcommand{\sd}{\surf{\d}}

\newcommand{\sdiv}{\surf{\div}}
\newcommand{\snabla}{\surf{\nabla}}
\newcommand{\dint}{{d}}

\newcommand{\dvol}{{d}V}
\newcommand{\dsurf}{{d}S}
\newcommand{\dline}{{d}l}

\newcommand{\Id}{\mathrm{Id}}
\newcommand{\tr}{\mathrm{tr}}

\newcommand{\loc}{\ensuremath{\mathrm{loc}}}
\newcommand{\cpt}{\ensuremath{\mathrm{comp}}}
\newcommand{\op}{\ensuremath{\mathrm{op}}}
\newcommand{\sym}{\ensuremath{ \mathrm{sym}}}

\newcommand{\cC}{\ensuremath{\mathcal{C}}}
\newcommand{\cD}{\ensuremath{{\cal D}}}
\newcommand{\cE}{\ensuremath{{\cal E}}}
\newcommand{\cS}{\ensuremath{{\cal S}}}
\newcommand{\Lip}{\mathrm{Lip}}

\newcommand{\mb}[1]{\ensuremath{\mathbb{#1}}}
\newcommand{\NN}{\mb{N}}

\newcommand{\RR}{\mb{R}}
\newcommand{\CC}{\mb{C}}

\newcommand{\al}{\alpha}

\newcommand{\ga}{\gamma}
\newcommand{\Ga}{\Gamma}
\newcommand{\de}{\delta}

\newcommand{\eps}{\epsilon}
\newcommand{\veps}{\varepsilon}

\newcommand{\vphi}{\varphi}

\newcommand{\la}{\lambda}
\newcommand{\La}{\Lambda}

\newcommand{\Om}{\Omega}

\newcommand{\Sig}{\Sigma}

\newcommand{\earth}{{\ensuremath{B}}}
\newcommand{\earthi}{{\ensuremath{\scriptstyle{B}}}}
\newcommand{\earthfs}{\ensuremath{\textstyle{\earth^{\fluid\solid}}}}

\newcommand{\earthfsc}{\ensuremath{\textstyle{\earth^{\fluid\solid\cpl}}}}
\newcommand{\earthfsci}{\ensuremath{\scriptstyle{\earthi^{\fluid\solid\cpl}}}}
\newcommand{\func}{\mathcal{J}}
\newcommand{\action}{\mathscr A}

\newcommand{\tim}{t}
\newcommand{\timd}{{t'}}
\newcommand{\timdd}{{t''}}
\newcommand{\timO}{{t_0}}
\newcommand{\kin}{\mathrm{\scriptstyle{kin}}}
\newcommand{\pot}{\mathrm{\scriptstyle{pot}}}
\newcommand{\solid}{\mathrm{\scriptscriptstyle{S}}}
\newcommand{\fluid}{\mathrm{\scriptscriptstyle{F}}}
\newcommand{\cpl}{\mathrm{\scriptscriptstyle{C}}}
\newcommand{\fxf}{\mathrm{\scriptscriptstyle{FF}}}
\newcommand{\fxs}{\mathrm{\scriptscriptstyle{FS}}}
\newcommand{\sxs}{\mathrm{\scriptscriptstyle{SS}}}
\newcommand{\PK}{\mathrm{\scriptstyle{PK}}}

\newcommand{\elast}{\mathrm{\scriptstyle{elast}}}

\newcommand{\grav}{\mathrm{\scriptstyle{gravity}}}
\newcommand{\ext}{\mathrm{\scriptstyle{force}}}
\newcommand{\motion}{\mathrm{\scriptstyle{motion}}}
\newcommand{\density}{\mathrm{\scriptstyle{density}}}

\newcommand{\Reg}{\mathrm{Reg}}
\newcommand{\Am}{\mathcal{A}}
\newcommand{\scJ}{\mathscr{J}}
\newcommand{\std}{\mathrm{ss}}

\title{Variational formulation of the earth's elastic-gravitational
  deformations under low regularity conditions}

\author{Katharina Brazda~\thanks{Fakult\"{a}t f\"{u}r Mathematik,
    Universit\"{a}t Wien, 1090 Vienna, Austria ({\tt
      katharina.brazda@univie.ac.at})}
\and Maarten V. de Hoop~\thanks{Simons Chair in Computational and
  Applied Mathematics and Earth Science, Rice University, Houston TX
  77005, USA ({\tt mdehoop@rice.edu})}
\and G\"unther H\"ormann~\thanks{Fakult\"{a}t f\"{u}r Mathematik,
  Universit\"{a}t Wien, 1090 Vienna, Austria ({\tt
    guenther.hoermann@univie.ac.at})}
}

\begin{document}
\maketitle

\pagestyle{myheadings}
\thispagestyle{plain}
\markboth{BRAZDA, DE HOOP, and H\"{O}RMANN}{Earth's
  elastic-gravitational deformations}

\begin{abstract}
We present a construction of the action, in the framework of the
calculus of variations and Sobolev spaces, describing deformations and
the oscillations of a uniformly rotating, elastic and self-gravitating
earth. We establish the Fr\'{e}chet differentiability of the action
under minimal regularity assumptions, which constrain the possible
composition of an earth model. Thus we obtain well-defined
Euler-Lagrange equations, weakly and strongly, that is, the system of
elastic-gravitational equations.
\end{abstract}


\section{Introduction} \label{sec:Introduction}

We present a construction of the action, in the framework of the
calculus of variations and Sobolev spaces, describing the oscillations
of a rotating earth, or any terrestrial planet, and establish its
Fr\'{e}chet differentiability. The key result pertains to obtaining
minimal regularity conditions representative of a ``real'' earth in a
consistent variational formulation with well-defined
elastic-gravitational, Euler-Lagrange equations. In the process, we
highlight the underlying conservation laws.

Under significantly stronger regularity conditions, the action has
been well known and, although without a rigorous analysis, described
in the literature (e.g., \cite{PeJa:58, WoDa:78, Gilbert:80,
  Valette:86, LoRo:90, DaTr:98, WoDe:07}). Minimal regularity
conditions are reflected in the allowable roughness of material
parameters of the interior and of the topography of the (major)
interior boundaries, namely, discontinuities such as the Moho, the
core-mantle boundary and the inner-core boundary. In a second paper
\cite{dHHP:15}, following a variational formulation, we establish the
existence and uniqueness of solutions, that is, a minimizer of the
action.

The regularity conditions constrain the composition of an earth model
on different levels. As a whole, the earth is viewed as a continuous
body, in fact, as a composite domain with interior boundaries. These
accommodate, for example, the outer core, the crust, the oceans, but
the interior boundaries can also coincide with general phase
transitions and separate distinct regions. The interior boundaries
correspond to a segmentation of the earth into subbodies that are
modeled as Lipschitz domains. On each part, the (particle) motions
minimally have to be positively oriented Lipschitz regular. This
follows from the principles of continuum mechanics
\cite{MaHu:83rep}. We can allow the material parameters relevant to
gravity and elasticity, the initial density of mass and the elements
of the stiffness tensor, to be $L^{\infty}$.

Even though we start from a general, nonlinear physics perspective, we
then assume small displacements for the oscillations, that is, small
perturbations of the equilibrium position, and invoke a natural
linearization on which, in fact, normal mode seismology is based. The
reference, here, is an unperturbed rotating earth with an associated
reference gravitational potential and prestress. These and the
associated perturbations can be shown to lie in subsets of appropriate
Sobolev spaces, hence the calculus of variations applies.

Perhaps the most technical part of our analysis appears in gluing
fluid and solid subdomains, in a composite domain, together. The
fluid-solid interior boundaries give explicit contributions to the
action. The mathematical construction from first principles (in fact,
Newton's third law) of appropriate fluid-solid interior boundary
conditions, with frictionless tangential slip, requires special
care. These have profound implications on the analyses of
well-posedness and of the spectrum. The form of the equations that is
directly obtained from the variational principle does not have a
spatial component operator that is coercive. Thus the existence of
energy estimates could be unclear. The remedy for this is discussed in
the mentioned second paper. Moreover, the presence of a fluid outer
core in a solid mantle implies the existence of an essential
spectrum. We discuss this, and a general characterization of the
spectrum, in a third paper [De Hoop, Holman, Jimbo \& Nakamura, in
  preparation]. Finally, we note that the fluid-solid interior
boundary term in the action can be straightforwardedly modified to
represent a rupture, in the linearized framework, and provides a
coupling with nonlinear friction laws in rupture dynamics.

The outline of the paper is as follows. In Section~\ref{sec:2} we
introduce the geometric, kinematic and physical field components and
their regularity. In particular, we define admissible motions. In
Section~\ref{sec:3} we construct and analyze the action, that is,
volume and surface (fluid-solid boundary) Lagrangian densities
starting from nonlinear physics; the regularity conditions introduced
in Section~\ref{sec:2} guarantee that it is well defined. We briefly
prove the consistency with underlying conservation laws. In
Section~\ref{sec:4}, we develop the linearization. The resulting
action is the subject of the further analysis based on quadratic
volume and surface Lagrangian densities. In Section~\ref{sec:5}, we
justify the application of Hamilton's principle by proving the
existence of the Fr\'{e}chet derivative of the action. We conclude
with writing the Euler-Lagrange equations in the spatially weak
form. As we mentioned before, their intuitive derivation has been well
known under significantly stronger regularity conditions. In
Section~\ref{sec:ruptures} we add shear faulting or shear ruptures,
incorporating a friction law, to the spatially weak formulation of the
Euler-Lagrange equations.
 
\subsection{Basic notation} \label{ssec:basic}

Inner products are denoted by brackets $\lara{\cdot|\cdot}$, whereas
$\lara{\cdot,\cdot}$ stands for a duality. We frequently write $a
\cdot b := \lara{a|b}_{\RR^n}$ for the Euclidean inner product (dot
product) of $a,b\in\RR^n$. Matrix multiplication (composition of
linear operators) also is denoted by a dot: $A \cdot B$. This is
consistent, as we identify row vectors with column vectors
($\RR^{1\times n} = \RR^{n\times 1} = \RR^n$). The matrix inner
product of $A,B\in\RR^{m\times n}$ is based on the Frobenius norm:
$A:B = \lara{A|B}_{\RR^{m\times n}} := \tr(A^T\cdot B)$. Components
are defined with respect to Cartesian coordinates: $a\cdot b=a_ib_i$
and $A:B = A_{ij} B_{ij}$. Here, and occasionally also the main text,
we employ summation convention. The derivative $Df$ of $f\colon
\RR^n\to \RR^m$ is a linear operator acting from $\RR^n$ to $\RR^m$
and $Df(a)$ is the derivative of $f$ in direction of $a\in\RR^n$ (see
Definition \ref{def:FrechetGateaux}). By abuse of terminology, $Df$ is
also referred to as the gradient of $f$, even though the true gradient
is given by the column vector $\mathrm{grad}\:f=(Df)^T$ if $m=1$, see
also \cite[p.\ xii]{MaHu:83rep}). We generally replace $D$ by $\nabla$
if only spatial coordinates are involved and denote the time
derivative by $\d_t$ or an over dot. The matrix representation of $Df$
follows from the identification $Df \cdot a := Df(a)$. Thus the
components of the matrix $Df$ read $(Df)_{ij}=\partial_j f_i$, for
$1\leq i\leq m$ and $1\leq j\leq n$.  We further use the row-wise
definition $\div f=(\d_jf_{ij})_{i=1}^m$ of the divergence of
$f\colon\RR^n\to\RR^{m\times n}$, that is, the derivative operator is
always contracted with the last index of $f$. Note that this
convention is transposed to the one in \cite{DaTr:98}, which also
explains the difference of their surface derivative operator
\cite[p.\ 827, (A.73)]{DaTr:98} and of (\ref{eq:surfdiff}), which is
the same as e.g.\ \cite{WoDe:07}, \cite[p.\ 355]{AmFuPa:00}, or
\cite[p.\ 95]{Gurtin:00}. The notation for surface operations
$\snabla$ and $\sdiv$ is summarized in Appendix \ref{app:surf}.

\subsection{Hamilton's principle of stationary action} 

Classical and relativistic mechanics and field theory can be
formulated in terms of the Hamilton's principle of stationary action
\cite{Synge:60, SeWh:68, FrSc:78, Goldstein:80, Salmon:88, MaRa:94}.
In classical point mechanics, when $q\colon [t_0,t_1]\to\RR^m$
describes the trajectory of a particle in $\RR^m$ between the times
$t_0<t_1\in\RR$, the action typically is given by an integral
$\action(q)=\int_I\mathscr{L}(t,q(t),\dot q(t))\:\dint t$ where
$\mathscr{L}\colon I\times\RR^m\times\RR^m\to\RR$ is called the
Lagrangian and $I:=[t_0,t_1]$. In field theory the state of the system
is characterized by a function $y$ of space and time, that is $y\colon
B\times I\to\RR^m$ for an interval $I\subseteq\RR$ and
$B\subseteq\RR^n$ open and bounded.  The {\bf Lagrangian}
$\mathscr{L}$ for a field theory is a functional of $y$ which may be
expressed as a volume and possibly also a surface integral. Hence, the
action is of the form
\begin{equation}\label{eq:Ham_action} 
\action(y)=\int_I\mathscr{L}(y)\:\dint t
=\int_I \left(\int_BL\:\dint V+\int_{S}{L_S}\:\dint S+\int_{\d B}{L_{\d B}}\:\dint S\right)\dint t.
\end{equation}
The integrand $L \colon B \times I \times \RR^m \times \RR^{m \times
  n} \times \RR^m \to \RR$ is called the volume Lagrangian
density. Its argument is $(x,t,y(x,t),\nabla y(x,t),\dot y(x,t))$. For
conservative systems $L$ is given by the difference of kinetic energy
density $E_{\kin}$ and potential energy density $E_{\pot}$ (see
\cite[p.\ 108]{Synge:60}, \cite[p.\ 22]{Goldstein:80})
\begin{equation}
L=E_{\kin}-E_{\pot}.
\end{equation}
The surface Lagrangian density $L_S\colon S\times I\times\RR^m \times
\RR^{m \times n} \times \RR^m \to \RR$ takes into account the
interaction energy of different regions within $B$ separated by the
hypersurface $S\sbs B$. The exterior Lagrangian $L_{\d B}\colon \d
B\times I\times\RR^m \times \RR^{m \times n} \times \RR^m \to\RR$
corresponds to conditions prescribed on the exterior boundary $\d
B$. Both surface Lagrangians depend on $(x,t,y(x,t),\snabla
y(x,t),\dot y(x,t))$. Here $\snabla$ denotes the surface gradient
operator introduced in (\ref{eq:surfdiff}). The Lagrangians may also
depend on higher derivatives or even in a nonlocal manner on the state
variables $y$.

Under suitable regularity conditions (see Appendix \ref{app:var})
stationarity of $\action$ at $y$ implies that $y$ is a solution of the
\textbf{Euler-Lagrange equations (EL)}
\begin{equation}\label{eq:Ham_EL}
\d_t(\d_{\dot y}L)+\nabla\cdot(\d_{\nabla y}L)-\d_yL=0
\quad\textrm{in}\quad B\times I
\end{equation}
and, if $L_{\d B}=0$, satisfies the \textbf{natural boundary
  conditions (NBC)}
\begin{equation}\label{eq:Ham_NBC} 
(\d_{\nabla y}L)\cdot\nu=0
\quad\textrm{on}\quad\d B\times I,
\end{equation}
where $\nu\colon\d B\to\RR^n$ denotes the exterior unit normal vector
to $\d B$. On sufficiently smooth orientable interior surfaces
$\Sig\sbs B$ with unit normal $\nu\colon\Sig\to\RR^n$, $y$ satisfies
the natural jump relations, also referred to as \textbf{natural
  interior boundary conditions (NIBC)},
\begin{equation}\label{eq:Ham_IBC} 
[\d_{\nabla y}L]_-^{+}\cdot\nu=0
\quad\textrm{on}\quad(\Sig\setminus S)\times I
\end{equation}
and
\begin{equation}\label{eq:Ham_IBCS} 
\d_{y}{L_S}-\d_t(\d_{\dot{y}}{L_S})-\snabla\cdot(\d_{\snabla y}{L_S})-[\d_{\nabla y}L]_-^{+}\cdot\nu=0
\quad\textrm{on}\quad S\times I
\end{equation}
with $[\:.\:]_-^{+}$ denoting the jump of the enclosed quantity across
the surface, see (\ref{eq:jump}).  The EL (\ref{eq:Ham_EL}) coincide
with Newton's equation of motion and the NBC (\ref{eq:Ham_NBC}) and
NIBC (\ref{eq:Ham_IBC}, \ref{eq:Ham_IBCS}) are the associated
dynamical boundary or jump conditions. Thus, Hamilton's principle
incorporates Newton's balance of momentum equation, boundary
conditions, and even additional constraints in one single
functional. Moreover, by Noether's theorem, the symmetry properties of
the action yield the conserved quantities of the system.

An integral formulation of the equations of motion typically requires
less regularity from the fields involved than the classical
differential form.  In particular, the weak EL, that is, the integral
formulation of the stationarity of the action $\de\action=0$,
coincide with the ``principle of virtual work'' (also, though
incorrectly, known as the ``principle of virtual power''), stated for
$L_{\d B}=0$:
\begin{eqnarray}\label{eq:Ham_virtual} 
\de\action(y,h)&=&\int_I\int_B(\d_yL\cdot h+\d_{\dot y}L\cdot\dot h+\d_{\nabla y}L:\nabla h)\:\dint V\dint t\nn\\
&&+\int_I\int_{S}(\d_yL_S\cdot h+\d_{\dot y}L_S\cdot\dot h+\d_{\snabla y}L_S:\snabla h)\:\dint S\dint t=0
\end{eqnarray}
for all ``virtual displacements'' $h\colon B\times I\to\RR^m$ (where
for simplicity we have omitted the argument of the Lagrangian
densities).  Here $\de\action(y,h)$ denotes the first variation of
$\action$ at $y$ in direction of $h$. Formally, the differential EL
are obtained from the weak form (\ref{eq:Ham_virtual}) by the
divergence theorem (\ref{eq:IBC_calc}).  Moreover, the principle of
virtual work can be shown to be equivalent to the integral formulation
of Newton's second law without requiring classical regularity
\cite{AnOs:79}.

\section{An earth model of low regularity}
\label{sec:2}

\subsection{Continuous bodies and composite domains}
\label{ssec:continuum}

\subsubsection{Continuous bodies}

A body is an abstraction of heuristic ideas of aggregations of matter capable of
deformation and motion.   We model continuous bodies by suitable subsets $B \subseteq \RR^3$ taken from a specified family $\mathcal{B}$ in $\RR^3$, which in addition to volume possess mass and other physical properties and can support forces \cite{Antman:05}. The elements of a body are called particles. 
A reasonable requirement is measurability and boundedness of the sets $B\in\mathcal{B}$ that can be assumed by continuous bodies. However this class still is too general for most purposes of continuum mechanics, whereas the class of compact subsets with smooth boundary is often too restrictive in order to model the behavior of general, physically realistic, bodies in mathematical terms.  Since fundamental concepts and key arguments of continuum mechanics rely on the validity of a version of the divergence theorem, the set of bodies may be restricted to the following:
\begin{eqnarray}\label{eq:bodies}
&& \mathcal{B}:=\{B\sbs\RR^3:\:B\:\text{open, bounded, connected, and with boundary $\d B$ such}\qquad\\
&&\qquad \text{ that the divergence theorem holds for $\cC^1$-vector fields}\}.\nn
\end{eqnarray}
The natural setting satisfying these requirements is based on geometric measure theory \cite[p.\ 56]{Silhavy:97}: Bodies $B\in\mathcal{B}$ are sets of finite perimeter and motions are functions of bounded variation (see also \cite{NoVi:88,DeOw:93,DelPiero:03, DelPiero:07}). Describing slip conditions associated to fracture and fluid-solid boundary conditions lead to the consideration of finite disjoint unions of such suitable regions. 
For our purposes of modeling parts of a composite fluid-solid earth, it suffices to consider continuous bodies described by connected open sets $B$ in $\RR^3$ together with any parts of their boundaries (in accordance with \cite[p.\ 432]{Antman:05}) and further assume that the boundaries $\d B$ are Lipschitz ($\Lip$) with $B$ located locally on one side of $\d B$.  

\begin{definition}[{{\bf $\Lip$-domains, boundaries, and surfaces}}]
\label{def:Lipdom}
\begin{enumerate}[label=(\roman*)]
\item A bounded open connected set $V\sbs\RR^n$ is called a
  $\Lip$-domain if its boundary $\d V$ is locally (in a finite cover
  by open neighborhoods) given as the graph of a Lipschitz continuous
  function and such that $V$ in this local representation is located
  only on one side of the graph describing the boundary
  (cf. \ \cite[Def.\ 1.2.1.1, p.\ 5]{Grisvard:85} or
  \cite[p.\ 64]{Ziemer:89} or \cite[p.\ 32]{Ciarlet:88}).
\item The closure of a $\Lip$-domain $V$ is an $n$-dimensional
  $\Lip$-submanifold in $\RR^n$ with boundary, see
  \cite[Def.\ 1.2.1.2, p.\ 6 and p.\ 7]{Grisvard:85}. In this case the
  boundary $\d V$ is referred to as a $\Lip$-boundary.  Any connected
  $(n-1)$-dimensional submanifold $S\subseteq\d V$ is a
  $\Lip$-submanifold (possibly with boundary) and is called a
  $\Lip$-(hyper-)surface. Note that a $\Lip$-boundary $\d V$ is
  orientable and we choose
$$
\nu\colon\d V\to\RR^n
$$ 
  as the exterior unit normal vector.  It satisfies $\nu\in
  L^{\infty}(\d V)^n$ as on a $\Lip$-surfaces $\nu$ exists almost
  everywhere (see \cite[p.\ 83]{KiOd:88}, \cite[p.\ 35]{Ciarlet:88}).
\end{enumerate}
\end{definition}
 
\begin{remark}[{{\bf Characterization of domains}}]
\begin{enumerate}[label=(\roman*)]
\item $\cC^{k,\al}$-domains, boundaries, and surfaces are defined
analogously by replacing $\Lip(=\cC^{0,1})$ in Definition \ref{def:Lipdom} by $\cC^{k,\al}$ for $k\in\NN_0$ and $\al\in(0,1]$. However, the case $\al=1$, that is, $\cC^{k,1}$, appears naturally in models of the earth.

\item We note that modeling of fractal boundaries would require uniform H\"older regularity $\cC^{0,\al}$ with $0 < \al < 1$ (\cite{Jaffard:97a,Jaffard:97b}), which is below Lipschitz regularity, thus does not allow for well-defined $L^\infty$ normal vector fields and applications of  divergence theorems close to the classical form. 
\end{enumerate}
\end{remark}

We note that the boundaries of $\Lip$-domains may have corners or
edges. We recall the classical divergence theorem for matrix-valued
fields on $\Lip$-domains, but state it already in the form of its
natural extension to $H^1$ regular fields:

\begin{lemma}[{{\bf Divergence theorem}}]\label{lem:divthm} 
If $V\sbs\RR^n$ is a $\Lip$-domain and $f\in H^1(V)^{m\times n}$,  then 
\begin{equation}\label{eq:divthm}
\int_V\div f\:\dvol=\int_{\d V}f\cdot\nu\:\dsurf.
\end{equation}
\end{lemma}

In \cite[Thm.\ 2.35]{ChToZi:09}, the divergence theorem is proven for
$V$ having finite perimeter, which is true for $\Lip$-domains,
cf.\ \cite[p.\ 248]{Ziemer:89}. Note that in the Sobolev space case
the surface integral in (\ref{eq:divthm}) has to be understood as
Sobolev duality of the traces of the component functions of $f$ in
$H^{\frac{1}{2}}(\d V)$ with corresponding components of $\nu$ in
$L^{\infty}(\d V)\sbs L^2(\d V)\sbs H^{-\frac{1}{2}}(\d V)$. Instead
of $H^1(V)^{m\times n}$ it suffices to assume that (see
\cite[p.\ 93]{KiOd:88}, \cite[p.\ 511]{DL:V5})
\begin{equation}\label{eq:Hdiv}
f\in H_{\div}(V)^{m\times n}:=\{f\in L^2(V)^{m\times n}:\:\div f\in  L^2(V)^m\}.
\end{equation}  
Thus, by Definition \ref{def:Lipdom} and Lemma \ref{lem:divthm}, a
$\Lip$-domain (or a $\cC^{k,\al}$-domain with $k \geq 1$) in
particular satisfies the properties of an element of $\mathcal{B}$
given in (\ref{eq:bodies}). Our basic setting will be the Lipschitz
framework and we will require higher regularity, that is $\cC^{k,1}$
with $k\geq 1$, only where necessary. Consequently, for our purposes
we may redefine the class of bodies as
\begin{equation}\label{eq:bodiesLip}
\mathcal{B}:=\{B\subseteq \RR^3:\:\text{the open interior } B^\circ \text{ is a $\Lip$-domain}\}.
\end{equation}

\subsubsection{Interior boundaries of composite domains}

We introduce an abstract notion of  so-called interior boundary for a general subset of $\RR^n$, whose intended meaning can be understood from the following very simple example: Let $V$ be the open unit ball centered at $0$ in $\RR^3$ with an arbitrary plane through $0$ removed. Then the unit sphere will be considered the exterior boundary of $V$, while the disc inside the removed plane, but without its surrounding unit circle, will be called  the interior boundary  $\Sig$ of $V$. Observe that in this example we obtain $\Sigma$ topologically by removing the boundary of the closed unit ball $\ovl{V}$, that is, the unit sphere, from the boundary of the two open disjoint half balls comprising $V$.

\begin{definition}[{{\bf Interior boundaries}}]\label{def:intbd}
The interior boundary $\Sig$ of a set $V\sbs\RR^n$ is defined by
\begin{equation}\label{eq:intbd}
\Sig:=\d V\setminus \d\ovl V.
\end{equation}
\end{definition}

According to the definition, interior boundary points of $V$ are
elements of the boundary $\d V$ which are not boundary points of the
closure $\ovl{V}$. We recall that one always has $\d\ovl{V}\sbs\d V$,
because $\d \ovl{V} = \ovl{\ovl{V}} \setminus \ovl{V}^\circ = \ovl{V}
\setminus \ovl{V}^\circ \subseteq \ovl{V} \setminus V^\circ = \d
V$. By construction, we have the disjoint union $\d
V=\Sig\cup\d\ovl{V}$ and combined with $\ovl{V}=V\cup\d V$, we obtain
the representation
\begin{equation}\label{eq:intbd_decomp_ovl}
\ovl{V}=V\cup\Sig\cup\d\ovl{V}
\end{equation}  
which gives a decomposition into disjoint sets, if $V$ is open. Moreover, noting that $\ovl{V}^{\:\circ}=\ovl{V}\setminus\d\ovl{V}$ we may then also deduce the identity
\begin{equation}\label{eq:intbd_decomp}
\ovl{V}^{\:\circ}=V\cup\Sig.
\end{equation}
Examples for interior boundaries are cuts inside $V$ and  hypersurfaces $S\sbs \ovl{V}^\circ \setminus V^\circ$ such that $V$ is located on both sides of $S$. In particular, we emphasize that, due to Definition \ref{def:Lipdom}, 
$$ 
  \Sigma = \emptyset \text{ for all $\Lip$-domains}. 
$$

However, we may model a set with $\Lip$-surfaces as interior
boundaries by considering the following composition of $\Lip$-domains:

\begin{definition}[{{\bf $\Lip$-composite domains}}]\label{def:composite}
A subset $V$ of $\RR^n$ is called a $\Lip$-composite domain if it can
be written as a finite union of pairwise disjoint $\Lip$-domains
$V_k\subseteq \RR^n$ ($k = 1, \ldots, k_0)$, that is,
\begin{equation}\label{eq:composite}
V=\bigcup_{k=1}^{k_0} V_k,\qquad \qquad V_k\cap V_{k'}=\emptyset \quad (k\neq k'),
\end{equation}
with the additional property that for every subset $M
\sbs\{1,\ldots,k_0\}$ the set $V^M := (\bigcup_{k\in M}
\ovl{V_k})^{\:\circ}$ is a finite disjoint union of $\Lip$-domains.
\end{definition}

Note that the condition on the sets $V^M$ is required to rule out sets
with cusps of the exterior or interior boundary. Furthermore, we
observe that if $V$ is a $\Lip$-composite domain with $\ovl{V}$
connected, then $\ovl{V}^{\:\circ}$ is a $\Lip$-domain. Indeed, as is
clear from Definition \ref{def:intbd}, taking the closure of a set
removes all its interior boundaries. (Taking the interior removes all
exterior boundaries which are defined as the interior boundaries of
the complement.)
 
If $V$ is a $\Lip$-composite domain with disjoint union
$V=\bigcup_{k=1}^{k_0} V_k$ as in the definition, then the following two
relations may be derived, where the first equality makes use of the
fact that $V$ is a disjoint union of open sets and the second equality
employs the property $\ovl{V_k}^\circ = V_k$ (i.e, empty interior
boundary) of the Lip-domains $V_k$:
\begin{equation}
\d V=\bigcup_{k=1}^{k_0} \d V_k=\bigcup_{1\leq k<k'\leq k_0}\Sig_{kk'}\:\cup\:\d\ovl V
\end{equation}
with
\begin{equation}
\Sig_{kk'}:=\d V_k\cap \d V_{k'}\quad\textrm{for}\quad k\neq k'. 
\end{equation}
Calling on \eqref{eq:intbd}, Definition \ref{def:intbd}, the interior
boundary $\Sig$ of $V$ is thus given by
\begin{equation}\label{eq:composite_intbd}
\Sig=\bigcup_{1\leq k<k'\leq k_0}\Sig_{kk'}\:\setminus\:\d\ovl V.
\end{equation}
We note that the sets $\Sig_{kk'}$ are $(n-1)$-dimensional
$\Lip$-surfaces in the sense of Definition \ref{def:Lipdom} and their
boundaries $\d\Sig_{kk'}$ (here, boundary in the sense of a
submanifold of $\d V$ with boundary) thus are $(n-2)$-dimensional
manifolds.  The boundary $\d\Sig$ of the interior boundary $\Sig$
consists of those parts of $\d\Sig_{kk'}$ that lie on the exterior
boundary $\d\ovl V$ of $V$, that is,
\begin{equation}\label{eq:composite_bdintbd}
\d\Sig
=\bigcup_{1\leq k<k'\leq k_0}\d\Sig_{kk'}\:\cap\:\d\ovl{V}.
\end{equation}
In particular, in $\RR^3$, $\d\Sig$ is a finite union of curves on the exterior boundary.

\subsubsection{Divergence theorem for composite domains and surfaces}

We present a variant (Lemma \ref{lem:divthmComposite}) of the
divergence theorem for $\Lip$-composite domains.  Compared to the
classical formulation (Lemma \ref{lem:divthm}) for $\Lip$-domains it
will contain an additional surface integral over the interior boundary
$\Sig$. Since by (\ref{eq:composite_intbd}) this is not a single
$\Lip$-surface but a finite union of those, the surface integral has
to be interpreted as the corresponding sum of surface integrals. More
precisely, if $g \colon \Sig \to \RR$ with restrictions
$g\!\mid_{\Sig_{kk'}} \in L^1(\Sig_{kk'})$ for all $1\leq k<k'\leq k_0$,
then $g\in L^1(\Sig)$ with
\begin{equation}
\int_{\Sig}g\:\dsurf\::=\sum_{1\leq k<k'\leq k_0}\int_{\Sig_{kk'}}g\:\dsurf.
\end{equation} 
Let $V=\bigcup_{k=1}^{k_0} V_k$ be a $\Lip$-composite domain as in \eqref{eq:composite}. By an $\RR^{m\times n}$-valued,  bounded piecewise $\cC^1$-function $f$ on $\ovl{V}$, with discontinuities being at most jumps contained in the interior boundaries $\Sig$, we mean
\begin{equation}
f\in\cC^1(\textstyle{\bigcup_{k=1}^{k_0} V_k})^{m\times n}\cap L^\infty(\ovl V)^{m\times n}
\end{equation}
such that every restriction $f|_{V_k}$ ($1 \leq k \leq k_0$) possesses
a $\cC^1$ extension to $\ovl{V_k}$, denoted by $f^{\overline{V}_k}$.
The classical partial derivative is continuous on $\bigcup_{k=1}^{k_0}
V_k$ and may be expressed as a sum involving cut-offs by
characteristic functions $\chi_{V_k}$ (having value $1$ on $V_k$ and
else $0$)
\begin{equation}
\d_j f \cdot \chi_V =\sum_{k=1}^{k_0}(\d_jf^{\ovl{V_k}})\chi_{V_k}\in L^\infty(\ovl V)^{m\times n}.
\end{equation}
Recalling that $\d V_k\setminus\Sig\sbs\d\ovl V$ we set
\begin{equation}\label{eq:partintpart}
\ovl{f}:=\sum_{k=1}^{k_0} f^{\ovl{V_k}}\chi_{_{\d V_k\setminus\Sig}}\in L^\infty(\d\ovl V)^{m\times n}
\end{equation}
and note that there is no contribution of $f^{\ovl{V_k}}$ to $\ovl{f}$
if $\d V_k\sbs\Sig$, that is, if $V_k$ is a completely interior
region. By construction, $\ovl f$ is $\cC^1$ on the finite union
$\d\ovl{V}\setminus\d\Sig$ of surfaces on the exterior boundary. To
simplify the notation we omit the bar introduced in
(\ref{eq:partintpart}) in the following, that is, we write $f$ instead
of the trace $\ovl{f}$ in surface integrals.

We define the {\bf jump} $[f]_-^{+}$ of $f$ when passing through
$\Sig$.  Strictly speaking, the notation $[.]_-^{+}$ for jumps across
the interior boundary is only applicable to a composite domain
\eqref{eq:composite} consisting of two parts, one labeled by $+$, the
other one labeled by $-$. Nevertheless it can also be extended to
composite domains which are such that every point has a neighborhood
that contains elements of only an even number of different interior
regions. In particular, triple-junctions or points where an odd number
of different interior regions meet, are not allowed. However, in the
later application, we will need to label the interior regions by two
different flags (fluid or solid). In addition, if two regions of the
same kind share a common boundary, they are glued together by taking
the closure of their union, see \eqref{eq:earth_fs}. This allows us to
consistently use the $[.]_-^{+}$-notation as well as the choice of the
normal vector.

After these preparatory observations we are ready to define the jump operator $[\:.\:]_-^{+}$ as the jump of the enclosed quantity across the surface $\Sig=\d V\setminus\d\ovl{V}$,
$[f]_-^{+}(x):=f^{\ovl{V_k}}(x)-f^{\ovl{V_{k'}}}(x)$ for $x\:\in\Sig_{kk'}$,
where $V_k$ corresponds to the region on the $+$- and $V_{k'}$ to the region on the $-$-side of $\Sig$.  
We observe that $[f]_-^{+}\in\cC^1(\Sig)^{m\times n}$. 
To simplify the notation we just write 
\begin{equation}\label{eq:jump}
[f]_-^+:=f^+-f^-
\end{equation} 
for the jump of $f$ from the $-$- to the $+$-side of a surface, with $f^\pm$ denoting the limit of $f$ when approaching the surface coming from the $\pm$-side. As \cite[Figure A.1, p.\ 826]{DaTr:98} we agree on the convention that the unit normal vector points in the direction of the jump, that is, from the $-$- to the $+$ side:
\begin{equation}\label{eq:normal}
\nu:=\nu^-=-\nu^+,
\end{equation}
with $\nu^\pm$ the normal vector of the region on the $\pm$-side of
the surface. Note that this is the reason for the negative sign in
front of the integrals over interior boundaries below.  By abuse of
notation we will occasionally write $[f\cdot\nu]_-^+$ instead of
$[f]_-^+\cdot\nu=f^+\cdot\nu-f^-\cdot\nu$. An $\RR^{n\times n}$-valued
field $f$ is said to satisfy the {\bf normality condition} on $\Sig$
\cite[eq.\ (3.82)]{DaTr:98} (or the $-$-side of $\Sig$) if
\begin{equation}\label{eq:normality}
f\cdot\nu=(\nu\cdot f\cdot\nu)\nu
\end{equation}
(or the same with $f$ replaced by $f^-$). See Appendix \ref{app:surf}
for further notation and identities related to surfaces in $\RR^n$.
The definitions above can be extended from the piecewise $\cC^1$-case
to piecewise $\Lip$- or $H^1$-functions $f$, if $\ovl{f}$ and
$[f]_-^+\in H^{\frac{1}{2}}(\Sig)^{m\times n}$ are interpreted almost
everywhere via the trace. Of particular importance will be the space
of piecewise $H^1$-vector fields with continuous normal component
across the interior boundary $\Sig$ of a $\Lip$-composite domain
$V=\bigcup_{k=1}^{k_0} V_k$:
\begin{equation}\label{eq:zeronormal}
H^1_{\Sig}(V)^{n}:=\{h\in H^1(V)^{n}:\:[h]_-^+\cdot\nu=0\:\text{ on }\:\Sig\}.
\end{equation}

We are now in the position to present the divergence theorem and
formula for integration by parts on composite domains, also in a
variant for integrands satisfying the normality condition
\eqref{eq:normality}:

\begin{lemma}[{{\bf Divergence theorem for composite domains and
  normality condition}}]\label{lem:divthmComposite} Let
  $V=\bigcup_{k=1}^{k_0} V_k \sbs \RR^n$ be a $\Lip$-composite domain
  \eqref{eq:composite} with interior boundary $\Sig=\d
  V\setminus\d\overline{V}$. If $f \in H^1(V)^{m\times n}$ then
\begin{equation}\label{eq:partdivthm}
\int_V\div f\:\dvol
=\int_{\d\ovl V}f\cdot\nu\:\dsurf-\int_{\Sig}[f]_-^+\cdot\nu\:\dsurf.
\end{equation}
If in addition
$h\in H^1(\ovl V)^m$, this leads to the following variant of the formula for integration by parts:
$\int_V\sum_{j=1}^nf_{ij}(\d_j h_i)\:\dvol=-\int_V\sum_{j=1}^n(\d_jf_{ij}) h_i\:\dvol
+\int_{\d\ovl V}\sum_{j=1}^n{f_{ij}}\:\nu_j h_i\:\dsurf-\int_{\Sig}\sum_{j=1}^n[f_{ij}]_-^+\:\nu_j h_i\:\dsurf$
or
\begin{equation}\label{eq:IBC_calc}
\int_V f\cdot Dh\:\dvol=-\int_V h\cdot\div f\:\dvol
+\int_{\d\ovl V}h\cdot f\cdot\nu\:\dsurf-\int_{\Sig}h\cdot[f]_-^+\cdot\nu\:\dsurf.
\end{equation}
Equation \eqref{eq:IBC_calc} also holds if 
$h\in H^1_{\Sig}(V)^{n}$ and $m=n$,
provided that $f\in H^1(V)^{n\times n}$ is isotropic: $f_{ij}=f^0\de_{ij}$ with $f^0\in H^1(V)$ (or in a neighborhood of $\Sig$).
If instead $f$ has no particular symmetry but satisfies the normality condition \eqref{eq:normality} on the $-$-side of $\Sig$,  
\eqref{eq:IBC_calc} still is true with $h\in H^1_{\Sig}(V)^{n}$, but one has to replace $h$ by $h^+$ in the last term: $\int_{\Sig}h^+\cdot[f]_-^+\cdot\nu\:\dsurf$.
\end{lemma}

As in the classical version (\ref{eq:divthm}) one has to interpret the
surface integrals as duality of the respective traces of $H^1$ in
$H^{\frac{1}{2}}$ and the normal vectors in $L^{\infty}\sbs L^2\sbs
H^{-\frac{1}{2}}$ on $\d\ovl V$ or $\Sig$. Moreover it again suffices
to assume $f\in H_{\div}$ instead of $H^1$. In the application $h$
will play the role of a test function.

\begin{proof}
Equation (\ref{eq:partdivthm}) is obtained by using the divergence theorem (\ref{eq:divthm}) for every region $V_k$ and summing up, using the fact that the two normal vectors on interior boundaries are antiparallel with $\nu^\pm=\mp\nu$, see (\ref{eq:normal}). For (\ref{eq:IBC_calc}) note that on every $V_k$ the product rule 
$$\d_j(f_{ij} h_i)=(\d_jf_{ij}) h_i+f_{ij}(\d_j h_i)$$ holds, implying
$f_{ij}(\d_j h_i)=(\d_jf_{ij}) h_i-\d_j(f_{ij} h_i)$ on
$\bigcup_{k=1}^{k_0} V_k$. If $h\in H^1(\ovl V)^m$ the function
$f_{ij} h_i$ is piecewise $H^1$ corresponding to the interior
boundaries $\Sig$ (that is $(f_{ij} h_i)_{j=1}^n=h\cdot f\in
H^1(V)^m$), hence integrating and applying (\ref{eq:partdivthm}) to
the second term yields the result. If $h\in H^1_{\Sig}(V)^{m}$, $h$
can not directly be pulled out from $[h\cdot f]_-^+\cdot\nu$ in the
surface integral even if $m=n$ (this is only possible if one assumes
that $f$, near $\Sigma$, is proportional to the identity matrix
$f_{ij}=f^0\de_{ij}$). However, the Leibniz rule for the jump
\eqref{eq:Leibnizjump} implies
$$[h\cdot f]_-^+\cdot\nu=h^+\cdot [f]_-^+\cdot\nu+[h]_-^+\cdot
f^-\cdot\nu.$$ The last term vanishes which will prove the claim:
Indeed, $[h]_-^+\cdot\nu=0$ implies that $[h]_-^+$ is parallel to
$\Sig$. But, due to normality of $f$ we have $f^-\cdot\nu=(\nu\cdot
f^-\cdot\nu)\nu$ which is normal to $\Sig$. Hence, their product is
zero.
\end{proof}

The surface divergence theorem is the classical divergence theorem but formulated on the surface as a bounded manifold, see e.g.\  \cite[Thm.\ III.7.5, p.\ 152]{Chavel:06} for a proof in the smooth case. The surface divergence $\sdiv$ is defined in \eqref{eq:surfdiv}.

\begin{lemma}[{{\bf Divergence theorem for surfaces}}]\label{lem:surfdivthm}
Let $S$ be a $\Lip$-hypersurface of $\RR^n$ and denote by $\dint\la$
the line element which is orthogonal to the surface boundary $\d
S$. If $f\in\Lip(U)^n$ for $U$ a neighborhood of $S$, then
\begin{equation}\label{eq:surfdivthm}
   \int_S \sdiv f\:\dsurf = \int_{\d S}f\cdot\dint\la.
\end{equation}
\end{lemma}

\begin{remark}
In $\RR^3$, the surface divergence theorem (\ref{eq:surfdivthm}) is
the classical theorem of Stokes, stating $\int_S
\operatorname{rot}(\nu\times f)\:\dsurf=\int_{\d S}(\nu\times
f)\cdot\dline$.  Indeed, by straightforward calculation we have
$\operatorname{rot}(\nu\times f)=\sdiv f$, and for the integrand on
the right-hand side, observe that $(\nu\times
f)\cdot\dline=f\cdot(\dline\times \nu)$. Since $\dline$ denotes the
classical line element which is parallel to the boundary curve $\d S$,
it follows that $\dline\times \nu=:\dint\la$ is normal to $\d S$.
\end{remark}

We will apply the surface divergence theorem also to the interior
boundary $\Sig$ of a $\Lip$-composite domain. In this case, we again
have to interpret it via summing up all integrals over $\Sig_{kk'}$
and $\d\Sig_{kk'}$. However, if there are no junctions where an odd
number of interior regions meet, the contributions of
$\d\Sig_{kk'}\setminus\d\ovl{V}$, that is, parts which lie in the
interior of $\ovl{V}$, cancel and in view of the specific form
(\ref{eq:composite_bdintbd}) of the interior boundary $\Sig$ we are
left with those from $\d\Sig$. Thus the surface divergence theorem
(\ref{eq:surfdivthm}) also holds for the interior boundary $\Sig$ of a
$\Lip$-composite domain.

We conclude with the well-known distributional analogue of integration
by parts:

\begin{lemma}[{{\bf Distributional integration by parts formula}}]\label{lem:div_dist}
Let $V\sbs\RR^n$ be open, $f\in\cD'(V)^{m\times n}$ be a matrix with
distributional components, and $h\in\cD(V)^m$ be a vector of test
functions. Then $D h \in\cD (V)^{m\times n}$, $\div f \in\cD'(V)^m$,
and we have
\begin{equation}\label{eq:IBC_calc_dist}
  \lara{f,D h} = - \lara{\div f,h}.
\end{equation}
The same equation holds with appropriate distributional dualities on a
$\Lip$-hypersurface $S$ in $\RR^n$, provided that the surface
differential operators are employed and we have $f\in
H^{\frac{1}{2}}(S)^{m\times n}$ and $h\in H^{\frac{1}{2}}(S)^m$.
\end{lemma}

\subsection{Geometry and kinematics}

Kinematics is the study of motion, regardless of the forces causing
it. The primitive concepts are position, time and body. 

\begin{figure}
\centering
\includegraphics[width=120mm]{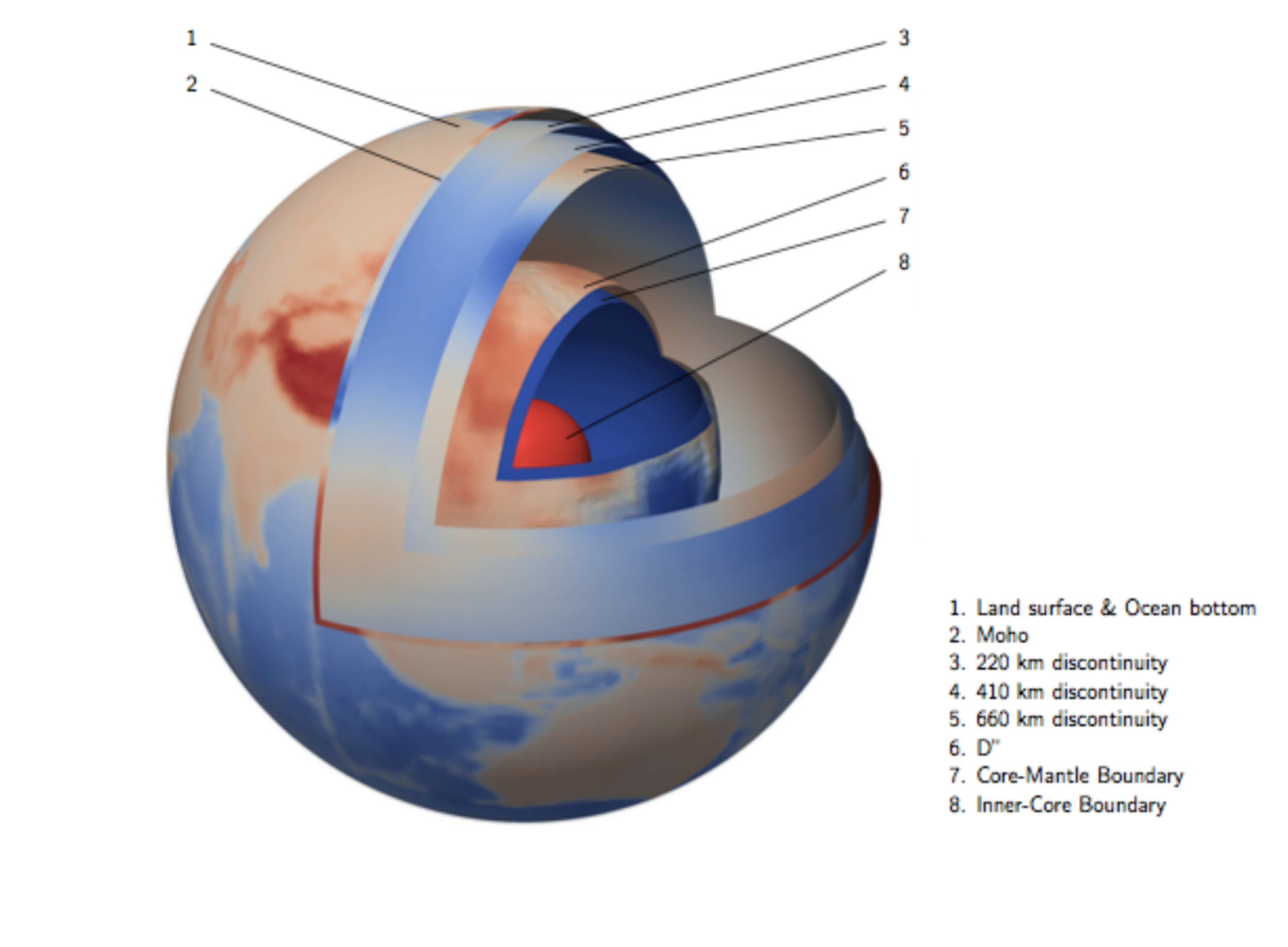}
\caption{Solid and fluid regions and major discontinuities in the
  earth. Color indicates topography (courtesy Ruichao Ye).}
\label{fig:1}
\end{figure}

\subsubsection{The composite fluid-solid earth model and its motion}

We consider a uniformly rotating, elastic and self-gravitating earth
model, subdivided into solid and fluid regions. Adopting the setting
of continuum mechanics, we study the moving earth as a bounded
continuous body that moves and deforms in space as time passes. In
particular, the general earth model is described within the framework
of classical Newtonian physics, where space-time is represented by
$\RR^3\times\RR$. The space $\RR^3$ is equipped with the usual
Euclidean norm $|.|$. By $\nabla=(\d_i)_{i=1}^3$ we denote
differentiation with respect to the space points, by
$\snabla=(\sd_i)_{i=1}^3$ the surface gradient which is defined as in
(\ref{eq:surfdiff}), and by $\d_t$ (or an over dot) the derivative
with respect to time (see \ref{ssec:basic}). Furthermore, for a
function $q$ defined on (a subset of) $\RR^3\times\RR$, we write
$q_\tim$ for the partial function $q(\:.\:,t)$ defined on (a subset
of) $\RR^3$.

We describe the earth as a family $\{\earth_\tim\}_{t\in I}$ of bodies, which in accordance with the class (\ref{eq:bodiesLip}) are $\Lip$-domains $\earth_\tim\sbs\RR^3$. 
The parameter $t$ is time, which varies in the closed interval
\begin{equation}
I:=[t_0,t_1]\sbs\RR.
\end{equation}
We write $I^\circ$ for the open interval $(t_0,t_1)$. The open set $\earth_\tim$ models the volume occupied by the material of the earth at time $t\in I$ and is referred to as a configuration of the earth. As a {\bf reference configuration} we identify the earth with the closure $\ovl\earth$ of the volume of its configuration at time $t_0$,
\begin{equation}
\earth:=\earth_\timO\sbs\RR^3.
\end{equation}
Thus the set $\earth$ is open and bounded with $\Lip$-boundary $\d\earth$. The different configurations $\earth_\tim$ are related by 
\begin{equation}
\earth_\tim=\vphi_\tim(\earth)
\end{equation}
where 
\begin{equation}\label{eq:motion} 
\vphi\colon\ovl\earth\times I\to\RR^3
\end{equation}
is the {\bf motion} of the earth. Due to elasticity of the material of the earth, there exists a natural equilibrium state which we assume to be adopted in the reference configuration at initial time $t_0$. That  is, $\earth$ is the equilibrium shape of the earth,  $\vphi(X,\timO) = X$ for every $X \in \ovl{B}$, or (by slight abuse of notation) 
\begin{equation}\label{eq:motion_initial} 
\vphi_\timO=\Id_{\ovl\earth}
\end{equation}
and the point  
\begin{equation}
x=\vphi(X,t)=\vphi_\tim(X)\in\RR^3
\end{equation} 
gives  the position of the material particle $X\in\ovl\earth$ of the earth at time $t\in I$. 

Uniform rotation of the earth means that the angular velocity of the
earth's rotation
\begin{equation}
   \Om \in \RR^3
\end{equation}
is independent of time and space. This approximation is valid in the
seismic frequency range. We establish a co-rotating reference frame of
Cartesian basis vectors in $\RR^3$, whose origin lies in the earth's
center of mass with respect to an inertial frame.

The earth model is subdivided into solid and fluid regions with
$\Lip$-continuous interior boundaries. This model is called a
composite fluid-solid earth model; see Figure~\ref{fig:1} for an
illustration. The color indicates topography of the relevant
(interior) boundaries. The Moho defines the bottom of the crust. The
discontinuities at an average depth of 410 km and 660 km and the top
of the so-called $D''$ layer correspond with major phase
transitions. We also indicated the oceans, the core-mantle boundary
(CMB) and the inner-core boundary (ICB) noting that the inner core is
solid and the outer core is fluid.

\begin{definition}[{{\bf Composite fluid-solid earth model}}]\label{def:earthmodel}
A composite fluid-solid earth model is a body $\earth\sbs\RR^3$ according to (\ref{eq:bodiesLip}), that is, a $\Lip$-domain, 
such that there exists a $\Lip$-composite domain $(\bigcup_{k\in K^\fluid}V_k^{\fluid})\cup(\bigcup_{k\in K^{\solid}}V_k^{\solid})$ as in Definition \ref{def:composite} consisting of (disjoint) fluid and solid interior regions $V_k^{\fluid}$ and $V_k^\solid$ respectively, and with interior boundary $\Sig$, such that
\begin{equation}\label{eq:earthmodel}
\earth=(\bigcup_{k\in K^\fluid}V_k^{\fluid})\cup(\bigcup_{k\in K^{\solid}}V_k^{\solid})\cup\Sig
\end{equation}
as disjoint union. Note that this decomposition corresponds to equation (\ref{eq:intbd_decomp}). 
The {\bf exterior boundary} of the earth thus is given by $\d B=\d\ovl{B}$.
The interior boundary $\Sig$ splits into {\bf solid-solid}, {\bf fluid-solid}, and {\bf fluid-fluid interior boundaries} denoted by $\Sig^{\sxs}$, $\Sig^\fxs$, and $\Sig^{\fxf}$ respectively, that is, we have the disjoint union
\begin{equation}
\Sig=\Sig^{\sxs}\cup\Sig^\fxs\cup\Sig^{\fxf}.
\end{equation}
\end{definition}

Different interior regions represent regions in the real earth of
different values or symmetry properties of physical parameters
(density, elasticity coefficients). In the seismic frequency range
(that is $>3$ mHz, \cite[p.\ 13]{LaWa:95}) where we can neglect
viscosity, the earth's crust, mantle, and inner core consist of solid
material, whereas the outer core and the oceans are fluid. The
solid-solid interior boundaries $\Sig^{\sxs}$ comprise discontinuities
in the crust, mantle and inner core such as the Moho, the LAB, and the
discontinuities in the transition zone. Our earth model, here, does
not contain fault surfaces, which would correspond to slipping
solid-solid interior boundaries. However, we will address ruptures in
Section~\ref{sec:ruptures}. The fluid-solid boundaries $\Sig^\fxs$
model the inner core boundary, the core-mantle boundary, and the ocean
bottom. Fluid-fluid interior boundaries $\Sig^{\fxf}$ separate
different fluid layers in the inner core or the ocean, but we will not
consider them in the following.  Interior boundaries in our model thus
only divide in welded interior boundaries $\Sig^{\sxs}$ and,
neglecting viscosity, slipping interior boundaries $\Sig^\fxs$. The
boundary of the interior boundaries $\d\Sig\subseteq\d\earth$, see
also (\ref{eq:composite_bdintbd}), in particular contains the coast
lines of oceans $\d\Sig^{\fxs}$.

\begin{remark}[{{\bf No welded $\Sig^{\sxs}$ for $L^{\infty}$ material parameters}}]
The concept of welded solid-solid boundaries $\Sig^{\sxs}$ is superfluous if the material parameters are modeled as $L^{\infty}$-functions of space (and one defines the boundary only based on jumps in the values and not on the changing symmetry behavior of e.g.\ elastic coefficients). Essentially, welded boundaries can only be defined if one assumes the material parameters to be at 
least piecewise continuous.
Indeed, a welded boundary between two different material types within the earth is any surface, across which the jump
$[.]_-^+$ (\ref{eq:jump}) of one of the material parameters is nonzero. Since the evaluation of $[.]_-^+$  involves calculating limits from both sides of the surface, defining a boundary surface amounts to requiring piecewise continuity of the material parameters. However, one may weaken this requirement by considering Sobolev spaces with mixed regularity, which are continuous only in one direction \cite[Appendix B, Def.\ B.1.10, p.\ 477]{Hoermander:V3}. 
\end{remark}

In the following we will formally introduce some physical quantities derived from the motion $\vphi\colon\ovl\earth\times I\to\RR^3$  of the earth. Sufficient regularity conditions are given later. The partial time derivative of the motion defines the {\bf velocity} field
\begin{equation}\label{eq:velocity}
v:=\dot\vphi.
\end{equation}
The derivative of $\vphi$ with respect to $X\in\earth$ is called the
{\bf deformation gradient} $\nabla\vphi$.  More precisely, for all
$t\in I$, $(\nabla\vphi)_\tim:=D(\vphi_\tim)$, and the Cartesian
components of $\nabla\vphi$ are given by $(\nabla\vphi)_{ij} =
\d_j\vphi_i$ for $1\leq i,j\leq 3$. The deformation gradient measures
the difference of final positions of initially adjacent material
particles. Since the motion $\vphi$, taking values in space $\RR^3$,
is differentiated with respect to points in the equilibrium earth
$\earth$, the deformation gradient is a so-called two-point tensor
field (see \cite[p.\ 48 and p.\ 70]{MaHu:83rep}). The {\bf Jacobian
  determinant} of $\vphi$ is the determinant of its deformation
gradient
\begin{equation}\label{eq:Jacobian}
J:=\det(\nabla\vphi).
\end{equation}
The {\bf material strain tensor} $e$ (also called the Lagrangian or Green-St.Venant strain tensor \cite{Ciarlet:88}) associated to $\vphi$ is a time-dependent symmetric two-tensor field on $\earth$ defined by
\begin{equation}\label{eq:nlstraintensor}
e:=\frac{1}{2}\Big(\nabla\vphi^T\cdot\nabla\vphi-1_{3\times 3}\Big) ,
\end{equation}
that is, in components,
$e_{ij}=\frac{1}{2}((\d_i\vphi_k)(\d_j\vphi_k)-\de_{ij})$ for $1\leq
i,j\leq 3$. The material strain tensor quantifies the associated
deformation of the earth, that is, changes in lengths and angles due
to the motion (\cite{Ciarlet:88}, \cite[p.\ 30]{DaTr:98}).

We discuss the relation between the {\bf spatial} and the {\bf
  material representation} (the latter is also referred to as {\bf
  referential description}). Physical fields are time and position
dependent (real or complex) quantities, that is, functions of space
and time taking values in a set $Q\sbs\CC^m$ ($m\in\NN$). When
describing physical fields we will often switch between the spatial
(or Eulerian) and the material (or Lagrangian) representation denoted
by $q^s$ and $q$ respectively. For every $t\in I$, the spatial
quantity $q_\tim^s$ depends on the space points
$x\in\vphi_\tim(\ovl\earth)$, that is
\begin{equation}
q^s\colon\bigcup_{t\in I} \big( \vphi_\tim(\ovl\earth)\times\{t\} \big) \to Q,
\end{equation}
where the collection of all configurations of the earth $\bigcup_{t\in I}(\vphi_\tim(\ovl\earth)\times\{t\})$ has the structure of a bundle on $I$ with fibers $\vphi_\tim(\ovl\earth)$ for $t\in I$.
Its material counterpart $q_\tim$ is defined as a function of the material points $X\in\ovl\earth$, that is
\begin{equation}
q\colon\ovl\earth\times I\to Q.
\end{equation}
Spatial and material representations are related by the motion $\vphi$: For all $t\in I$ we have 
\begin{equation}\label{eq:materialspatial}
  q_\tim=q_\tim^s\circ\vphi_\tim,
 \end{equation}
that is, $q(X,t)=q^s(\vphi(X,t),t)$ for all $(X,t)\in\ovl\earth\times I$ (cf.\ \cite[p.\ 27]{MaHu:83rep}).  The spatial  representation of the motion itself gives $x = \vphi(X,t) = \vphi^s(\vphi(X,t),t) = \vphi^s(x,t) = \vphi_t^s(x)$, that is, 
$\vphi_\tim^s=\Id_{\vphi_\tim(\ovl\earthi)}$.

Recalling $\vphi_\timO=\Id_{\ovl\earthi}$,  we obtain that the spatial and the material representation of a physical field coincide at equilibrium time and we may introduce the notation
\begin{equation}
q^0:=q_\timO=q_\timO^s.
\end{equation}
Note that a spatial quantity typically is defined on the whole space,
that is, $q^s\colon\RR^3\times I\to Q$.  However, starting from the
material representation $q$, one only gets and needs information about
the restriction of $q^s$ to the bundle $\bigcup_{t\in
  I}(\vphi_\tim(\ovl\earth)\times\{t\})$.  The situation at fixed time
$t$ is illustrated by the following diagram
\begin{eqnarray}
&& \vphi_\tim(\ovl\earth)\hookrightarrow\:\RR^3 \nn\\
&& \:\:{\scriptstyle \vphi_\tim}\!\Big\uparrow   
\qquad\:\:\;\Big\downarrow{\scriptstyle q^s_\tim} \nn \\
&& \quad\ovl\earth\:\underset{q_t}\longrightarrow\: Q \nn 
\end{eqnarray}

Note that in case of vector or tensor fields the transition from $q^s$ to $q$ does not give the differential geometric pull-back, since only the base point transformation is involved. One of the consequences to keep in mind is the following: For a $\Lip$-surface $S\sbs\earth$ with unit normal vector $\nu$ we denote the corresponding spatial unit normal vector of the surface $\vphi_\tim(S)$  by 
\begin{equation}\label{eq:spatialnormal}
\nu_\tim^s\colon\vphi_\tim(S)\to\RR^3.
\end{equation} 
Even  when $\vphi_t$ is a diffeomorphism, we typically encounter $\nu_\tim^s\neq\nu\circ\vphi_\tim^{-1}$ and the correct transformation turns out to be
$\nu_\tim^s=\frac{(J_\tim(\nabla\vphi_\tim)^{-T}\cdot\nu)}{|(J_\tim(\nabla\vphi_\tim)^{-T}\cdot\nu)|}\circ\vphi_\tim^{-1}$ (see \cite[p.\ 41]{Ciarlet:88}).

\subsubsection{Regularity of the motion and kinematical interface conditions}\label{sssec:regularmotion}

A reasonable minimal assumption for the regularity of the motion of
the earth is
\begin{equation}
   \vphi\in\cC^0(I,L^{\infty}(\earth))^3.
\end{equation}
First, continuity with respect to time is required for the initial
condition (\ref{eq:motion_initial}) to be defined. Second, requiring
$\vphi_\tim\in L^{\infty}(\earth)^3$ for all $t\in I$ is natural,
since the moving earth does not disperse to an unbounded set in space.
However, to exclude catastrophic phenomena like tearing or
inter-penetration of the material of the earth, we need to introduce
additional regularity conditions, which we define for a motion of a
set $\ovl V\sbs\RR^n$, where $V$ is bounded and open. We anticipate
that due to the possible slip along fluid-solid boundaries, the motion
of the composite fluid-solid earth can not be globally injective, see
Definition \ref{def:Am}.

\begin{definition}[{{\bf Classes of regular motions $\Reg$,
                         $\Reg_{\cC^k}$, $\Reg_\Lip$}}]
\label{def:Reg}
Let $I\sbs\RR$ be a closed interval, $t\in I$, $V\sbs\RR^n$ open and
bounded, and $\vphi\colon \ovl V\times I\to\RR^n$ be a motion such
that $\vphi\in\cC^0(I,L^{\infty}(V))^n$ with $\vphi_\timO=\Id_{\ovl
  V}$.
\begin{enumerate}[label=(\roman*)]
\item $\vphi$ is called {\bf regular}, $\vphi\in\Reg(\ovl V\times I)$,
  if $\vphi_\tim(V)\sbs\RR^n$ is open, $\vphi_\tim(\ovl V)\sbs\RR^n$
  is closed, and $\vphi_\tim\colon \ovl V\to\RR^n$ is injective.
\item $\vphi$ is called {\bf $\cC^k$-regular},
  $\vphi\in\Reg_{\cC^k}(\ovl V\times I)$, for $k\in\NN_0$, if $\vphi$
  is regular, $\vphi\in\cC^k(\ovl V \times I)^n$, and for all $t\in I$
  the inverse $\vphi_\tim^{-1}$ is $\cC^k$ as a map $\vphi_t(\ovl{V})
  \to \ovl{V}$ (that is, is the restriction of a $\cC^k$ map defined
  on an open neighborhood of $\vphi_t(\ovl{V})$).
\item $\vphi$ is called {\bf $\Lip$-regular}, $\vphi\in\Reg_\Lip(\ovl
  V\times I)$, if $\vphi$ is regular, $\vphi\in\Lip(\ovl V \times
  I)^n$, and for all $t\in I$ the inverse $\vphi_\tim^{-1}$ is
  Lipschitz-continuous $\vphi_t(\ovl{V}) \to \ovl{V}$.
\end{enumerate}
\end{definition}

By definition, we have the inclusion $\bigcup_{1\leq
  k\leq\infty}\Reg_{\cC^k}\sbs\Reg_{\Lip}\sbs\Reg_{\cC^0}\sbs\Reg$ on
$\ovl V\times I$. Regular and $\cC^k$-regular motions are also
discussed in \cite[Definition 1.4, p.\ 27]{MaHu:83rep}.

\begin{remark}[{{\bf Properties of regular motions}}]\label{rem:regmotion}
\begin{enumerate}[label=(\roman*)]
\item For $t\in I$, the injectivity of $\vphi_\tim\colon\ovl V\to\RR^n$ for a regular motion
$\vphi\in\Reg(\ovl V\times I)$ prohibits inter-penetration as well as self-contact of material, because two initially different points cannot be mapped to one and the same point by the motion. 

\item A $\cC^0$-regular motion prevents the material of the body from tearing, since under continuous mappings connected portions of material stay connected. In particular, as is shown in \cite[p.\ 16]{Ciarlet:88}, the conditions
$$
V\sbs\RR^n\:\textrm{ open},\quad{\ovl V}^{\:\circ}=V,\quad\textrm{and}\quad\vphi_\tim\in\cC^0(\ovl V)^n\:\textrm{ injective}
$$
imply
$$
\vphi_\tim(\ovl V)=\ovl{\vphi_\tim(V)},\quad\vphi_\tim(V)=(\vphi_\tim(V))^\circ,
\quad\textrm{and}\quad\vphi_\tim(\d V)=\d (\vphi_\tim(V))=\d (\vphi_\tim(\ovl V)).
$$
Since ${\ovl V}^{\:\circ}=V$ is true for $\Lip$-domains $V$, $\cC^0$-regular motions preserve the interior, the exterior and the boundary of a $\Lip$-domain.

\item
By definition, a $\cC^0$-regular motion $\vphi\in\Reg_{\cC^0}(\ovl
V\times I)$ is continuous up to the boundary $\d V$, hence we can
deduce continuity of $t\mapsto\vphi_\tim^{-1}(x)$ for all
$x\in\bigcup_{t\in I}\vphi_\tim(\ovl{V})$. Indeed, if
$\vphi_\tim\colon\ovl{V}\to\RR^n$ is continuous and injective for
$t\in I$ (closed), then $\wtil{\vphi}\colon\ovl{V\times
  I}\to\wtil{\vphi}(\ovl{V\times I}) =\bigcup_{t\in
  I}(\vphi_\tim(\ovl{V})\times\{t\})\sbs\RR^n\times I$ defined by
$\wtil{\vphi}(x,t):=(\vphi_\tim(x),t)$ is a continuous and bijective
map between compact sets. Consequently, its inverse given by
$(x,t)\mapsto\wtil{\vphi}^{-1}(x,t):=(\vphi_\tim^{-1}(x),t)$ is
continuous, which implies the continuity of
$t\mapsto\vphi_\tim^{-1}(x)$.
\item
$\cC^k$-regularity guarantees the continuity of derivatives of the
  motion up to order $k$ and includes that $\vphi_\tim\colon
  V\to\vphi_\tim(V)$ is a $\cC^k$-diffeomorphism.
\end{enumerate}
\end{remark}

Injectivity of $\vphi$ is a global requirement, that is, it depends on the motion as a map of the entire body.  A related pointwise condition is positive orientation of $\vphi$, defined in terms of the determinant of the deformation gradient:

\begin{definition}[{{\bf Positively oriented $\Lip$-regular motions $\Reg_\Lip^+$}}]
\label{def:Jpos}
Let $V\sbs\RR^n$ be open and  bounded and $I\sbs\RR$ be a closed interval. A motion $\vphi\in\Reg_\Lip(\ovl V\times I)$ is called positively oriented (or orientation preserving) if $J>0$ holds for every $t \in I$ almost everywhere. In this case we write $\vphi\in\Reg_\Lip^+(\ovl V\times I)$.
\end{definition}

\begin{remark}[{{\bf Orientation and injectivity}}]\label{rem:Jpos}
A $\cC^1$-regular motion $\vphi\colon\ovl V\times I\to\RR^n$ automatically satisfies $J>0$, since $J$ is continuous with respect to time, $J(X,t)\neq 0$ for all $(X,t)$ (due to injectivity), and $J^0(X)=J(X,t_0)=1$ (due to $\vphi_\timO=\Id_{\ovl V}$). Conversely, by the inverse function theorem, if $\vphi$ is $\cC^1$ and $J>0$ then $\vphi$ is locally invertible. However, this is generally not true under $\Lip$-regularity.
Further consequences of $J>0$ are discussed in Remark \ref{rem:rho_reg}.
\end{remark}

We have defined regular motions for a general domain $\ovl V\sbs\RR^n$. Considering the motion $\vphi\colon\ovl\earth\times I\to\RR^3$ of the earth, for an earth model without interior boundaries we  may simply take $V=\earth$. However, for a composite fluid-solid earth model (Definition \ref{def:earthmodel}) the regularity properties will be assumed to hold on every interior region, that is, for each of the restrictions $\vphi|_{\ovl{V_k^\solid}\times I}$  and  $\vphi|_{\ovl{V_k^\fluid}\times I}$  separately, but not necessarily for $\vphi$ globally. On the interior boundaries $\Sig$, we have to impose additional continuity conditions on the motion, which are referred to as {\bf kinematic interior boundary conditions}. The bracket $[.]_-^+$ denotes the jump of the enclosed quantity as defined in (\ref{eq:jump}). In view of the properties of welded solid-solid and slipping fluid-solid and fluid-fluid interior boundaries of a composite fluid-solid earth model, for all $t\in I$ we thus have (spatial) continuity of the motion across $\Sig^{\sxs}$, that is
\begin{equation}\label{eq:IBC_SS}
[\vphi_\tim]_-^+=0\qquad\textrm{on}\quad\Sig^{\sxs}
\end{equation}
and continuity of the normal component of the spatial velocity $v^s_\tim=v_\tim\circ\vphi_\tim^{-1}$ across the current $\Sig^\fxs$ (and also $\Sig^{\fxf}$)-boundaries $\vphi_\tim(\Sig^\fxs)$, that is
\begin{equation}\label{eq:IBC_FS}
[v_\tim^s]_-^+\cdot \nu_\tim^s=0 \qquad\textrm{on}\quad\vphi_\tim(\Sig^\fxs),
\end{equation}
with $\nu_\tim^s$  the spatial normal vector field (\ref{eq:spatialnormal}) on $\vphi_\tim(\Sig^\fxs)$, for all $t\in I$ (see also \cite[pp.\ 48, 67, and 71]{DaTr:98}).
 
Due to continuity \eqref{eq:IBC_SS} of the motion across $\Sig^\sxs$, a function that is $\Lip$-continuous on two adjacent solid interior regions $\ovl{V_k^{\solid}}$ and $\ovl{V_{k'}^{\solid}}$ is $\Lip$-continuous on the closure of their union $\ovl{V_k^{\solid}\cup V_{k'}^{\solid}}$. 
To define the regularity of the motion, it thus suffices to consider the composite earth model as a union of the open solid and fluid interior regions $\earth^\solid$ and $\earth^\fluid$ obtained by glueing together adjacent regions:  
\begin{equation}\label{eq:earth_fs}
\earth^\solid:= \Big( \bigcup_{k\in K^\solid} \ovl{V_k^{\solid}} \Big)^{\:\circ}
\qquad\textrm{and}\qquad
\earth^\fluid:= \Big( \bigcup_{k\in K^\fluid} \ovl{V_k^{\fluid}} \Big)^{\:\circ}.
\end{equation}
We recall that $\earth$, $V_k^\solid$ $(k\in K^\solid)$, and
$V_k^\fluid$ $(k\in K^\fluid)$ are $\Lip$-domains. The sets
$\earth^\solid$ and $\earth^\fluid$ are finite unions of
$\Lip$-domains, as they are obtained as finite unions of
$\Lip$-domains where possibly resulting interior boundaries are
removed by the closure. Note that cusps are ruled out by Definition
\ref{def:composite} of $\Lip$-composite domains. By construction we
have
\begin{equation}
\ovl\earth=\ovl{\earth^\solid}\cup\ovl{\earth^\fluid}, \qquad \Sig^\fxs\setminus\d B=\ovl{\earth^\solid}\cap\ovl{\earth^\fluid}=\d \earth^\solid\cap\d \earth^\fluid,
\end{equation} 
and the disjoint union 
\begin{equation}\label{eq:earth_fs_union}
\earth=\earth^\solid\cup\earth^\fluid\cup(\Sig^\fxs\setminus\d\earth),
\end{equation}
which simplifies the original decomposition $\earth=(\bigcup_{k\in K^\fluid}V_k^{\fluid})\cup(\bigcup_{k\in K^{\solid}}V_k^{\solid})\cup\Sig$ of Definition \ref{def:earthmodel}.
For convenience we further introduce the notation 
\begin{equation}\label{eq:domain_sfc}
\earthfs:=\earth^\solid\cup\earth^\fluid=\earth\setminus\Sig^\fxs, \quad
\earth^{\cpl}:=\RR^3\setminus\ovl{\earth},
\quad
\earthfsc:=\earth^\solid\cup\earth^\fluid\cup\earth^\cpl=\RR^3\setminus(\Sig^\fxs\cup\d\earth).\quad
\end{equation}

We are ready to define the class of admissible motions of a composite fluid-solid earth model. Essentially, an admissible motion   consists of $\Lip$-regular positively oriented motions on fluid and solid parts separately and satisfies additional compatibility conditions, including slip along $\Sig^\fxs$:

\begin{definition}[{{\bf Admissible motions $\Am$ of a composite fluid-solid earth model}}]\label{def:Am}
Let $\earth$ be a composite fluid-solid earth model as in Definition \ref{def:earthmodel} with
$\earth^\solid$, $\earth^\fluid$ defined by (\ref{eq:earth_fs}), and interior boundary $\Sig^\fxs$.
We define the associated class of admissible motions by 
$$
\Am(\ovl\earth\times I):=\Big\{\vphi\colon \ovl\earth\times I\to\RR^3:
\:\vphi\in\cC^0(I,L^{\infty}(\earth))^3,\:\vphi_\timO=\Id_{\ovl\earthi},\:\vphi\:
\textrm{satisfies }\mathrm{(i)}\textrm{ to }\mathrm{(iv)}\Big\}
$$
where
\begin{enumerate}[label=(\roman*)]
\item \label{def:Am_i} {\bf Global conditions:} $\forall\:t\in
  I:\:\vphi_\tim(\ovl\earth)$ is homeomorphic to $\ovl\earth$ and
  $\bigcup_{t\in I}\vphi_\tim(\earth)$ is bounded,
\item \label{def:Am_ii} {\bf Piecewise $\Lip$-regularity:}
  $\vphi|_{\earthi^\solid\times I}$ can be extended to
  $\Reg_\Lip^+(\ovl{\earth^\solid}\times I)$ and
  $\vphi|_{\earthi^\fluid\times I}$ to
  $\Reg_\Lip^+(\ovl{\earth^\fluid}\times I)$,
\item \label{def:Am_iii} {\bf No interpenetration:} $\forall\:t\in
  I:\:\vphi_\tim(\earth^\solid)\cap\vphi_\tim(\earth^\fluid)=\emptyset$,
  and
\item \label{def:Am_iv} {\bf Slip condition:} \eqref{eq:IBC_FS},
  $[v_\tim^s]_-^+\cdot \nu_\tim^s=0$ on $\vphi_\tim(\Sig^\fxs)$, holds
  for almost every $t\in I$.
\end{enumerate}
\end{definition}

The class of admissible motions $\Am(\ovl\earth\times I)$ thus
consists of functions $\vphi$ that preserve the connectedness
properties as well as boundedness of the earth by \ref{def:Am_i},
possess a positively oriented $\Lip$-regular extension to the closure
of each interior region by \ref{def:Am_ii}, prohibit interpenetration
of different interior regions by \ref{def:Am_iii}, satisfy the
slipping interior boundary conditions (\ref{eq:IBC_FS}) on fluid-solid
interior boundaries by \ref{def:Am_iv}, and satisfy the welded
interior boundary conditions (\ref{eq:IBC_SS}) on solid-solid interior
boundaries by construction using the domain $\earth^\solid$ introduced
in (\ref{eq:earth_fs}).

Thus a motion in $\Am(\ovl\earth\times I)$ is continuous across
$\Sig^{\sxs}$ but possibly discontinuous across $\Sig^\fxs$. We
accounted for this discontinuity in \ref{def:Am_ii} by demanding
only the existence of a positively oriented $\Lip$-regular extension
for each interior region instead of requiring $\vphi\in \Reg_\Lip^+$
on its closure. Since slip on interior surfaces is allowed, the
earth's motion is not globally injective. Hence, it is not a
$\cC^0$-regular motion of $\ovl\earth$ and the confinement condition
\ref{def:Am_i} needs to be imposed to preserve connectedness and
boundedness of $\vphi_\tim(\ovl\earth)$. However $\vphi$ is injective
in the interior of each interior region $\earth^\solid$ and
$\earth^\fluid$ separately by property \ref{def:Am_ii}. Since these
sets are finite unions of $\Lip$-domains, an admissible motion
satisfies
$\vphi_\tim(\ovl{\earth^\solid})=\ovl{\vphi_\tim(\earth^\solid)}$ and
$\vphi_\tim(\ovl{\earth^\fluid})=\ovl{\vphi_\tim(\earth^\fluid)}$ for
all $t\in I$ (cf.\ Remark \ref{rem:regmotion}). Combined with
property \ref{def:Am_i} we thus have
\begin{equation}\label{eq:motionclosure}
\vphi_\tim(\ovl\earth)=\ovl{\vphi_\tim(\earth)}.
\end{equation} 
For a discussion of further regularity properties, let $V=\earth^\solid$ or $\earth^\fluid$  for the moment. 
Then $\vphi\in\Am(\ovl\earth\times I)$ implies that the restriction $\vphi|_{V\times I}$
 can be extended to yield a motion in  $\Reg_\Lip^+(\ovl V\times I)\sbs\Lip(\ovl V \times I)^3$.
Motions in $\Am(\ovl\earth\times I)$ are thus Lipschitz continuous functions with respect to time and
take values in the space of all piecewise Lipschitz continuous functions of $\earth$.
Moreover, since $\Reg_\Lip^+(\ovl V\times I)\sbs H^1(V\times I^\circ)^3$ it follows
\begin{equation}\label{eq:SobolevRegpLip}
\Am(\ovl\earth\times I)\sbs H^1(\earthfs\times I^\circ)^3.
\end{equation}
The regularity of  $\vphi\in\Am(\ovl\earth\times I)$  implies that the first-order derivatives of the
$\Lip$-regular extension of $\vphi|_{V\times I}$ to $\ovl{V}$, denoted by $\vphi^{\ovl{V}}$,
satisfy $\nabla\vphi^{\ovl{V}}\in\Lip(I,L^{\infty}(V))^{3\times 3}$ and $\d_t\vphi^{\ovl{V}}\in L^{\infty}(I^\circ,\Lip(\ovl V))^3$.
Consequently, the components of $v$, $\nabla\vphi$, $e$, and $J$, defined by (\ref{eq:velocity}), (\ref{eq:Jacobian}),
and (\ref{eq:nlstraintensor}), are essentially bounded functions on $\earthfs\times I^\circ$. Thus, 
by $\earth=\earthfs\cup(\Sig^\fxs\setminus\d\earth)$ (\ref{eq:earth_fs_union}) and since $\Sig^\fxs$ 
has zero volume, they can be extended to $L^{\infty}(\earth\times I^\circ)$. 
In particular we note that the fluid-solid boundary condition $[v_\tim^s]_-^+\cdot \nu_\tim^s=0$ \eqref{eq:IBC_FS} (or condition \ref{def:Am_iv} in Definition \ref{def:Am}) makes sense within the regularity setting specified in \ref{def:Am_ii}, which shows consistency of definition of $\Am(\ovl\earth\times I)$: Indeed using the $\Lip$-regular
extension $\vphi_\tim^{\ovl{V}}$ yields $v_\tim^{\ovl{V}}=\d_t\vphi_\tim^{\ovl{V}}\in\Lip(\ovl{V})^3$ and thus
$(v_\tim^{\ovl{V}})^s=v_\tim^{\ovl{V}}\circ(\vphi_\tim^{\ovl{V}})^{-1}\in\Lip(\vphi_\tim^{\ovl{V}}(\ovl{V}))^3$ for almost
every $t\in I$ with uniform Lipschitz constant. Therefore $(v_\tim^{\ovl{V}})^s$ can be restricted to the moving boundary
$\vphi_\tim(\d V)=\d(\vphi_\tim(V))$ for almost every $t\in I$. This implies that $v_\tim^s$ can be restricted to
moving interior boundaries $\vphi_\tim(\Sig^\fxs)$ as in (\ref{eq:partintpart}). Hence, its jump $[v_\tim^s]_-^+$ is a continuous function on $\vphi_\tim(\Sig^\fxs)$ that can be multiplied with the bounded spatial unit normal $\nu_\tim^s$.

\begin{remark}[{{\bf Derivatives of discontinuous functions: Surface
measures}}] Note that the spatial derivatives occurring in the
  definitions of $\nabla\vphi$, $e$, and $J$ need not be identical to
  the global distributional spatial derivatives of the function
  $\vphi$, which may contain additional surface measures at
  $\Sig^\fxs$ (cf.\ \cite[Thm.\ 3.1.9, p.\ 60]{Hoermander:V1}).
\end{remark}

The following lemma relates the regularity of the material to that of
the spatial representations, establishing equivalence between
essential boundedness of a spatial quantity $q^s$ and of the
corresponding material quantity $q$ for an admissible motion of a
composite earth model.

\begin{lemma}[{{\bf Regularity of material and spatial representations
for admissible motions}}]\label{lem:Linfty} Let $\earth$ be a
  composite fluid-solid earth model as in Definition
  \ref{def:earthmodel}, $\vphi\in\Am(\ovl{\earth}\times I)$ be an
  admissible motion, and the maps $q \colon \ovl{B} \times I \to
  \RR^3$ and $q^s \colon \RR^3 \times I \to \RR^3$ be related by
  $q_\tim =q_\tim^s\circ\vphi_\tim$ for every $t\in I$. Then we have
$$q^s\in\cC^0(I,L^{\infty}(\RR^3))\quad\Longleftrightarrow\quad q\in\cC^0(I,L^{\infty}(\earth)).$$
\end{lemma}

\begin{proof}
  We first establish the equivalence of the $L^{\infty}$-condition with respect the spatial
variables. Since $\vphi$ preserves interior boundaries which are of measure zero it suffices to show
$q_\tim^s\in L^{\infty}(\vphi_\tim(V))$ $\Longleftrightarrow$ $q_\tim\in L^{\infty}(V)$ for all interior regions $V=\earth^\solid$
or $\earth^\fluid$. We denote the restrictions of $q_\tim^s$ and $q_\tim$ to $\vphi_\tim(V)$ and $V$ again by $q_\tim^s$
and $q_\tim$. We start with the implication from left to right. First note that if $q_\tim^s\in L^{\infty}(\vphi_\tim(V))$
boundedness of $q_\tim=q_\tim^s\circ\vphi_\tim$
is clear. Measurability is guaranteed by the Lipschitz-continuity of $\vphi_\tim^{-1}$
on the corresponding interior regions $\vphi_\tim(V)$.
More precisely, for a Borel set $B$ in $\RR$ with $B\sbs q_\tim(V)$
consider the pre-image $q_\tim^{-1}(B)=\vphi_\tim^{-1}((q_\tim^s)^{-1}(B))$.
Since $(q_\tim^s)^{-1}(B)$
is Lebesgue measurable by assumption and Lipschitz-maps preserve Lebesgue measurability
(e.g.\ \cite[Lemma 3.6.3 on p.\ 192]{Bogachev:07}), we have the Lebesgue measurability of $q_\tim^{-1}(B)$.
Hence, $q_\tim\in L^{\infty}(V)$.
The other implication follows similarly by changing the roles of $\vphi_\tim$ and $\vphi_\tim^{-1}$.
Finally, continuity of $q_\tim:=q_\tim^s\circ\vphi_\tim$ resp.\ $q_\tim^s:=q_\tim\circ\vphi_\tim^{-1}$ with
respect to the time variable follows from the continuity of $t\mapsto\vphi_\tim$ resp.\ $t\mapsto\vphi_\tim^{-1}$
(see Remark \ref{rem:regmotion}, (iii)) for $t\in I$.
\end{proof}

For $\vphi\in\Am(\ovl\earth\times I)$, the spatial counterpart of the material velocity field
$v_\tim\in L^{\infty}(\earth)^3$ is given by $v_\tim^s:=v_\tim\circ\vphi_\tim^{-1}\in L^{\infty}(\RR^3)$ (extended to $\RR^3$ by zero) 
for almost all $t\in I$. Similarly, since $J$ and the components of $e$ are in $L^{\infty}(\earth\times I^\circ)$ it follows that $J_\tim^s$ and the components of $e_\tim^s$ are in $L^{\infty}(\RR^3)$ (extending again by zero).

Material and spatial representations of volume and surface integrals are related by the following transformation formulas:

\begin{lemma}[{{\bf Integrals in material and spatial representation}}]
\label{lem:intvolsurf}
Let $\earth$ be a composite fluid-solid earth model as in Definition \ref{def:earthmodel} and
$\vphi\in\Am(\ovl\earth\times I)$, $t\in I$. Then the following hold:  
\begin{enumerate}[label=(\roman*)]
\item For an open subset $A\sbs\earth$ and $f_\tim\in L^{\infty}(A)$,
\begin{equation}\label{eq:intvol}
\int_{\vphi_\tim(A)}f_\tim^s\:\dvol=\int_Af_\tim J_\tim\:\dvol.
\end{equation}
\item If $\vphi$ is $\cC^1$-regular, $S\sbs\earth$ a $\Lip$-surface and $g_\tim\in\cC^0(S)$, then 
\begin{equation}\label{eq:intsurf}
\int_{\vphi_\tim(S)}g_\tim^s\:\nu_\tim^s\dsurf=\int_Sg_\tim J_\tim\:(\nabla\vphi)_\tim^{-T}\cdot \nu\:\dsurf.
\end{equation}
\end{enumerate}
\end{lemma}

In the physics literature 
(see e.g.\ \cite[p.\ 31 and p.\ 40]{Ciarlet:88}, \cite[(4.5.24) and (4.5.29b)]{Malvern:69}, \cite[(2.30) and (2.37)]{DaTr:98})
the relations (\ref{eq:intvol}) and (\ref{eq:intsurf}) are usually stated in the form
\begin{equation}\label{eq:volsurf_elements}
\dvol_\tim=J_\tim\dvol\qquad\textrm{and}\qquad\nu_\tim^s\dsurf_\tim=J_\tim(\nabla\vphi)^{-T}_\tim\cdot\nu\:\dsurf.
\end{equation}
\begin{proof}
By Lemma \ref{lem:Linfty} we have $f_\tim^s=f_\tim\circ\vphi_\tim^{-1}\in L^{\infty}(\vphi_\tim(A))$ and, recalling that
$S$ is part of a $\Lip$-boundary, $g_\tim^s=g_\tim\circ\vphi_\tim^{-1}\in \cC^0(\vphi_\tim(S))$. For $A\sbs\earth^\solid$
or $A\sbs\earth^\fluid$ the first equation (\ref{eq:intvol}) follows from the substitution law for integrals and the
positivity of $J_\tim$. 
By the decomposition (\ref{eq:earth_fs_union}) together with the mapping properties of $\vphi\in\Am(\ovl\earth\times I)$
(properties (ii) and (iii) of Definition \ref{def:Am}) the result is true for a general open subset $A\sbs\earth$.
For a proof of the second equation (\ref{eq:intsurf}) see \cite[Thm.\ 1.7-1, p.\ 39]{Ciarlet:88}. 
\end{proof}

\subsection{Physical fields and material parameters}

\subsubsection{Density} \label{ssec:density}
The {\bf mass} $M(B)$ of a continuous body $B$ in $\RR^3$ describes the body's resistance to acceleration when a force acts on it.   
Following \cite{Antman:05}, mass has  the essential properties of a positive measure $M$ on the set of all continuous bodies, that is, $M(\emptyset)=0$, $B_1\sbs B_2$ implies $M(B_1)\leq M(B_2)$, and $M(\bigcup_k B_k)=\sum_k M(B_k)$, if $B_1, B_2, \ldots$ are disjoint.
In continuum mechanics one disregards point-, line- or surface-masses and assumes that if a subbody $A\sbs B$ has zero volume (that is $\int_A\dvol=0$
with respect to Lebesgue measure), then it also has zero mass, $M(A)=0$. In other words, $M$ is absolutely
continuous with respect to Lebesgue measure. The Radon-Nikodym theorem then implies the existence of a
nonnegative function $\rho^0\in L^1_\loc(B)$, the initial {\bf density} of mass, such that
\begin{equation}
M(B)=\int_B\rho^0\:\dvol.
\end{equation}
The existence of the spatial density $\rho^s_\tim\in L^1_\loc(B_\tim)$ with $B_\tim=\vphi_\tim(B)$ is deduced along the same lines:
\begin{equation}
M(B_\tim)=\int_{B_\tim}\rho^s_\tim\:\dvol.
\end{equation}
Moreover, assuming boundedness of mass density implies that a body's mass can be estimated by its volume, since  
$\rho^0\in L^\infty(B)$ implies $M(B)\leq \|\rho^0\|_{L^\infty(B)}\int_B\:\dvol$ (analogously in spatial formulation). 
 
In physics terms, density is a measure of mass per unit volume. The
spatial density of the earth is a non-negative function $\rho^s \colon
\RR^3 \times I\to\RR$ which is compactly supported by the actual
earth, that is,
\begin{equation}
\supp(\rho^s_\tim)=\ovl{\vphi_\tim(\ovl\earth)}=\vphi_\tim(\ovl\earth)
\end{equation}
for all $t\in I$, where we invoked (\ref{eq:motionclosure}). Thus we assume that  
\begin{equation}\label{eq:basic_rho_reg}
\rho^s\in  \cC^0(I,L_\cpt^\infty(\RR^3))\sbs\cC^0(I,L^{\infty}(\RR^3)).
\end{equation}
The corresponding material density $\rho\colon\earth\times I\to\RR$ is given by
$\rho_\tim
:=\rho^s_\tim\circ\vphi_\tim\colon\earth\to\RR$ for all $t\in I$, see \eqref{eq:materialspatial}.
Setting $\rho_\tim(X):=0$ if $X\notin\earth$ yields the extension $\rho\colon\RR^3\times I\to\RR$, with compact support 
$\supp(\rho_\tim)=\ovl\earth$ for all $t\in I$. If $\vphi\in\Am(\ovl\earth\times I)$ the condition
$\rho^s\in\cC^0(I,L^{\infty}(\RR^3))\sbs L^{\infty}(\RR^3\times I)$ implies
$\rho\in L^{\infty}(\earth\times I)$, or $\rho\in L^{\infty}(\RR^3\times I)$ with support of every $\rho_t$ being contained in a fixed compact set upon extension by zero.

The principle of conservation of mass states that during the motion of a body $V\sbs\RR^3$ mass can
neither be produced nor annihilated. For $\vphi\in\Reg_\Lip^+(\ovl V\times I)$ and
$\rho_\tim^s\in L^{\infty}(\vphi_\tim(V))$ for almost all $t\in I$
conservation of mass is equivalent to $M(\vphi_\timd(A))=M(\vphi_\timdd(A))$, that is,
\begin{equation}\label{eq:conservationofmass}
   \int_{\vphi_\timd(A)} \rho_\timd^s \: \dvol
       = \int_{\vphi_\timdd(A)} \rho_\timdd^s \: \dvol
\end{equation}
to hold for all open sets $A\sbs V$ and almost all times $t'$, $t''\in I$. Note that the principle of
conservation of mass often is formulated as $\frac{d}{dt}\int_{\vphi_\tim(A)}\rho_\tim^s(x)\:\dvol(x)=0$,
see e.g.\ \cite{MaHu:83rep}, which however requires differentiability of $\rho^s$ with respect to time.
The following lemma shows that conservation of mass allows us to relate the current material density $\rho_\tim$
to the initial density $\rho^0$ via $J_\tim$. It thus gives a necessary condition for conservation of mass in
a composite earth model $B$.

\begin{lemma}[{{\bf Conservation of mass}}]\label{lem:conservationofmass}
Let $\vphi\in\Am(\ovl\earth\times I)$ and $\rho_\tim^s\in L^{\infty}(\vphi_\tim(\earth))$
for (almost all) $t\in I$. Then conservation of mass implies the equation
\begin{equation}\label{eq:rho0rhoJ}
\rho^0=\rho_\tim J_\tim
\end{equation}
to hold almost everywhere on $\earth$ and for (almost all) $t\in I$.
\end{lemma}

\begin{proof}
Set $t'=t_0$, $t''=t\in I$ and let $V=\earth^\solid$ or $\earth^\fluid$.
The principle of conservation of mass for $\vphi|_{V\times I}$ then reads
$\int_A\rho^0(x)\:\dvol(x)=\int_{\vphi_\tim(A)}\rho_\tim^s(x)\:\dvol(x)$ for all $A\sbs V$ open, where on the left-hand
side we used $\vphi_\timO(A)=\Id_{\earth}(A)=A$ and $\rho_\timO^s=\rho_\timO=\rho^0$. Lipschitz-continuity of $\vphi_\tim$
on $A$ for all $t\in I$ allows to use change of variables on the right-hand side
(see \cite[Theorem 3.2.3, p.\ 243]{Federer:69}) which yields
$\int_A\rho^0(x)\:\dvol(x)=\int_A\rho_\tim^s(\vphi_\tim(X))\:|J_\tim(X)|\:\dvol(X)$.
Since $\rho_\tim^s\circ\vphi_\tim=\rho_\tim$ and $J_\tim$ is positive, $\int_A(\rho^0-\rho_\tim J_\tim)(X)\:\dvol(X)=0$ 
for every open subset $A\sbs V$, which is equivalent to $\rho^0=\rho_\tim J_\tim$ to hold
almost everywhere on each $\earthfs$ and thus almost everywhere on $\earth$, completing the proof.
\end{proof}

As a corollary of the Lemma \ref{lem:intvolsurf} and Lemma \ref{lem:conservationofmass} above we obtain
equations for the integral of a function times density in spatial and material representation.

\begin{corollary}
\label{cor:rho0int}
Under the assumptions of Lemma \ref{lem:intvolsurf} and Lemma \ref{lem:conservationofmass}, conservation of mass implies
\begin{equation}
\int_{\vphi_\tim(A)}f_\tim^s\rho_\tim^s\:\dvol=\int_Af_\tim\rho^0\:\dvol
\end{equation}
for all $\Lip$-domains $A\sbs \earth$ and (almost all) $t\in I$.
\end{corollary}
\begin{proof} By Lemma \ref{lem:Linfty}, the condition $f_t\in L^{\infty}(\earth)$ implies that
$f_\tim^s=f_t\circ\vphi_\tim\in L^{\infty}(\vphi_\tim(\earth))$. Thus the assertion follows
  upon substitution $x=\vphi_\tim(X)$
 from Lemma \ref{lem:intvolsurf} and Lemma \ref{lem:conservationofmass}, since
$\int_{\vphi_\tim(A)}(f_\tim^s\rho_\tim^s)(x)\:\dvol(x)=\int_{A}f_t(X)\rho_\tim(X)\:J_\tim(X)\:\dvol(X)=\int_A(f_t\rho^0)(X)\:\dvol(X)$.
\end{proof}

In particular, by Lemma \ref{lem:conservationofmass}, conservation of mass relates the regularity of
density to the regularity of the motion:

\begin{remark}[{{\bf Improved regularity of material density by conservation of mass}}]
\label{rem:rho_reg}
If $\vphi\in\Am(\ovl\earth\times I)$ and $\rho_\tim^s\in L^{\infty}(\vphi_\tim(\earth))$, Lemma
\ref{lem:conservationofmass} shows that conservation of mass implies $\rho^0=\rho_\tim J_\tim$. Thus,
given the initial density $\rho^0$, this equation expresses the current density $\rho_\tim$ in terms of derivatives
of the motion $\vphi$. In particular, improving the regularity of $\vphi$ directly improves the regularity of
$t\mapsto\rho_\tim(x)=\rho^0(x)/J_\tim(x)$ for almost all $x$:

\begin{enumerate}[label=(\roman*)]
\item  If  $\vphi\in\Reg_{\cC^1}(\ovl\earth\times I)\sbs \cC^1(\ovl\earth \times I)^3$
(and thus positively oriented, see Remark \ref{rem:Jpos}) then we know that $J\in\cC^1(I,\cC^0(\ovl\earth))$ and
that $J$ is positive. Therefore, by conservation of mass it follows that
$\rho\in\cC^1(I,L^{\infty}(\earth))$ if $\rho^0\in L^{\infty}(\earth)$ or even $\rho\in\cC^1(I,\cC^0(\earth))$ if
$\rho^0\in\cC^0(\earth)$.

\item  For $\vphi\in\Am(\ovl\earth\times I)$ and $V=\earth^\solid$ or $\earth^\fluid$
we have $\vphi|_{V\times I}\in \Lip(V \times I )^3$ with $J$ is positive and bounded away from zero on $V$.
Consequently, if $\rho^0\in L^{\infty}(\earth)$, conservation of mass implies
$\rho|_{V\times I}\in\Lip(I,L^{\infty}(V))$ which yields $\rho\in\Lip(I,L^{\infty}(\earth))$. 

\item Note  that conservation of mass does not necessarily improve the regularity of the spatial density function $\rho^s\in\cC^0(I,L^{\infty}(\RR^3))$ for $\vphi\in\Am(\ovl\earth\times I)$, since existence of the time derivative 
$\d_t(\rho^s(x,t))=\d_t(\rho(\vphi_\tim^{-1}(x),t))=\d_t\rho(\vphi_\tim^{-1}(x),t)+\nabla\rho(\vphi_\tim^{-1}(x),t)\cdot\d_t\vphi_\tim^{-1}(x)$
would imply boundedness of $\nabla\rho$, which cannot be expected in the general case with merely  $\rho_\tim \in L^{\infty}(\earth)$.
\end{enumerate}
\end{remark}

\subsubsection{Gravity}\label{ssec:gravity}

The force of gravity is the mutual force of attraction of mass. The
corresponding acceleration field is conservative and thus equal to the
gradient of the {\bf gravitational potential} $\Phi^s$. The latter is
determined by the earth's actual density distribution
$\rho^s\in\cC^0(I,L^{\infty}(\RR^3))$ as the distributional solution
of Poisson's
equation\footnote{$\triangle:=\nabla\cdot\nabla=\sum_{l=1}^3\d_{x_l}^2$
  is the Laplacian and $G=6.67\cdot
  10^{-11}\:\text{m}^3\text{s}^{-2}\text{kg}^{-1}$ denotes the
  gravitational constant.}
\begin{equation}\label{eq:poisson}
\triangle \Phi_\tim^s=4\pi G\rho_\tim^s
\end{equation}
in $\RR^3$, which vanishes at infinity, that is, with $\lim_{|x|\to\infty}\Phi_\tim^s(x)=0$, for $t\in I$.
The material gravitational potential $\Phi$ on $\earth\times I$ is given by
$\Phi_\tim
:=\Phi_\tim^s\circ\vphi_\tim$ for all $t \in I$, see \eqref{eq:materialspatial}.

In the next lemma we review general properties of distributional solutions of Poisson's equation in $\RR^3$.
We introduce the solution set 
\begin{equation}
Y(\RR^3):=\{y\in\cD'(\RR^3):\:\triangle y\in\cE'(\RR^3),\:\lim_{|x|\to\infty}y(x)=0\}
\end{equation}
and its subset
\begin{equation}
Y^\infty(\RR^3):=\{y\in\cD'(\RR^3):\:\triangle y\in L_\cpt^\infty(\RR^3),\:\lim_{|x|\to\infty}y(x)=0\}.
\end{equation}
Let $E_3\in\cD'(\RR^3)$ denote the unique radial fundamental solution of $\triangle$ in $\RR^3$ (that is
$\triangle E_3=\de$) that vanishes at infinity ($E_3$ is given by $E_3(x)=-\frac{1}{4\pi|x|}$ for $x\neq 0$).
In the following $ D^\al$ for $\al\in\NN_0^3$, $\d_i$, and $\d_i\d_j$ for $i,j=1,2,3$ denote distributional derivatives
(with respect to $x\in\RR^3$).

\begin{lemma}[{{\bf Newtonian potential, decay conditions, and regularity}}]
\label{lem:poisson1} $\:$
\begin{enumerate}[label=(\roman*)]
\item If $y\in Y(\RR^3)$, then $y$ can be written as the Newtonian potential of its Laplacian, that is, as the
convolution $y=E_3\ast\triangle y$.

\item If $y\in Y(\RR^3)$, then $y$ satisfies
$ D^\al y(x)=\lara{\triangle y,1}( D^\al E_3)(x)+\mathcal O(1/|x|^{2+|\al|})$
as $|x|\to\infty$ for all multi-indices $\al\in\NN_0^3$.

\item  If $y\in Y^\infty(\RR^3)$, then  we have $y\in\bigcap_{1\leq p<\infty} W_\loc^{2,p}(\RR^3)\sbs \cC^1(\RR^3)$ $(1 \leq i,j \leq 3)$.
\end{enumerate}
\end{lemma}

\begin{proof}
The assertions (i) and (ii) follow by distributional solution theory of Poisson's equation
(cf.\ \cite[ch.\ II.3]{DL:V1}). For the Newtonian potential (i) see \cite[ch.\ II.3 Proposition 3, p.\ 279]{DL:V1}
and for the decay conditions (ii) see
\cite[ch.\ II.3 Proposition 2, p.\ 278]{DL:V1} for the case $n=3$ respectively.
To prove (iii), note that the regularity properties of $y\in Y^\infty(\RR^3)$ are a consequence of the
ellipticity of the Laplacian: By definition
$\triangle y\in L_\cpt^\infty(\RR^3)\sbs L_\loc^\infty(\RR^3)\sbs L_\loc^p(\RR^3)$ for every $p\geq 1$.
Thus (local) elliptic regularity for the $L^p$-based Sobolev spaces (cf.\ \cite[Theorem
7.9.7 p.\ 246 and Theorem 4.5.13 p.\ 123]{Hoermander:V1}) 
implies $ D^\al y\in\bigcap_{1\leq p<\infty} L_\loc^p(\RR^3)$ for $|\al|\leq 2$, in other words
$y\in\bigcap_{1\leq p<\infty} W_\loc^{2,p}(\RR^3)$.  Furthermore, by the Sobolev
embedding theorem (cf.\ \cite[Theorem 5.4 (part I case C), p.\ 97]{Adams:75}) $W_\loc^{2,p}(\RR^3)\sbs\cC^k(\RR^3)$
for $k\in\NN_0$ if $0\leq k<2-3/p$. Since this inequality is satisfied for $k=1$ and $p>3$ we have
$y\in\bigcap_{1\leq p<\infty} W_\loc^{2,p}(\RR^3)\sbs\cC^1(\RR^3)$, completing the proof.
\end{proof}

\begin{remark}
Note that the fact that $Y^\infty(\RR^3)\sbs\cC^1(\RR^3)$ in  Lemma \ref{lem:poisson1} (iii) can alternatively be
proven as follows: Let $y\in Y^\infty(\RR^3)$, then $y=E_3\ast\triangle y$ and $\d_iy=\d_iE_3\ast\triangle y$ by
Lemma \ref{lem:poisson1} (i). Since $E_3$ and $\d_iE_3\in L^1_\loc(\RR^3)$ and, by assumption,
$\triangle y\in L^{\infty}_\cpt(\RR^3)$, the relation $L^1_\loc\ast L^{\infty}_\cpt\sbs\cC^0$
(cf.\ \cite[ch.\ II.3 Lemma 3, p.\ 284]{DL:V1}) yields the result.
\end{remark}

Based on these results and the elliptic regularity of $\triangle$, we obtain a decomposition of
$y\in Y^\infty(\RR^3)$ which will be useful in constructions below. The idea is to separate the
monopole term of $y$ from its more rapidly decaying remainder at large distances to $\supp(\triangle y)$.

\begin{lemma}[{{\bf Decomposition of $Y^\infty(\RR^3)$}}]
\label{lem:poisson2} 
Let $y\in Y^\infty(\RR^3)$, let $R>0$ and sufficiently large to ensure $\supp(\triangle y)\sbs B_R(0)$, and
let $\chi\in\cC^\infty_\cpt(\RR^3)$ be a smooth cutoff with $\chi(x) = 1$ when $|x| < R$ and $\chi(x) = 0$
when $|x| > 2R$. Define the function $m_y\colon\RR^3\to\RR$ by $m_y:=\lara{\triangle y,1}E_3(1-\chi)$.
Then $m_y\in\cC^\infty(\RR^3)\cap L^{\infty}(\RR^3)$, $\supp(m_y)\sbs\RR^3\setminus\supp(\triangle y)$,
$\triangle m_y\in\cC_\cpt^\infty(\RR^3)$, and we have $y=m_y+\wtil y$ where $\wtil y\in H^2(\RR^3)$.
\end{lemma}

\begin{proof}
The regularity and the support property of $m_y$ is clear by definition. Using $\triangle E_3=\de$ a direct
calculation yields $\triangle m_y\in\cC_\cpt^\infty(\RR^3)$. Set $\wtil y:=y-m_y\in\cD'(\RR^3)$. Noting
that $\lara{\triangle y,1}E_3(x)\chi(x)$
vanishes when $|x| > 2R$ we obtain thanks to the decay conditions in Lemma \ref{lem:poisson1} (ii) that
$\wtil y(x)=y(x)-m_y(x)
=\lara{\triangle y,1}E_3(x)-\lara{\triangle y,1}E_3(x)(1-\chi(x))+{\mathcal O}(1/|x|^2)
=\lara{\triangle y,1}E_3(x)\chi(x)+{\mathcal O} (1/|x|^2)={\mathcal O} (1/|x|^2)$
as $|x|\to\infty$, which shows square integrability of $\wtil y$ outside $B_{2R}(0)$.
The proof of Lemma \ref{lem:poisson1} (iii) for $p=2$ gives $y\in H^2_\loc(\RR^3) \sbs L^2_\loc(\RR^3)$
(this also follows by applying (local) elliptic regularity for the $L^2$-based Sobolev spaces
(cf.\ \cite[Lemma 9.26, p.\ 307]{Folland:99})
of $\triangle$ to $\triangle y\in L_\cpt^\infty(\RR^3)\sbs L^2_\loc(\RR^3)$).
Boundedness of $m_y$ thus yields $\wtil y\in L^2_\loc(\RR^3)$. Consequently, $\wtil y\in L^2(\RR^3)$, which finally
allows to use (global) elliptic regularity (cf.\ \cite[Lemma 9.25, p.\ 307]{Folland:99}) to obtain
$\wtil y\in H^2(\RR^3)$.
\end{proof}

As concerns the gravitational potential, if $\rho_\tim^s\in L^{\infty}_\cpt(\RR^3)$, we have
$\Phi_\tim^s\in Y^\infty(\RR^3)$ for $t\in I$.
Thus by Lemma \ref{lem:poisson1}, $\Phi_\tim^s$ can be expressed as the Newtonian potential
\begin{equation}
\label{eq:poisson_newton}
\Phi_\tim^s(x)=(E_3\ast4\pi G\rho_\tim^s)(x)=-G\int_{\RR^3}\frac{\rho_\tim^s(x')}{|x-x'|}\dvol(x')
\qquad (x\in\RR^3).
\end{equation}
Moreover $\Phi_\tim^s\in \cC^1(\RR^3)$,
$\d_i\Phi_\tim^s\in\bigcap_{1\leq p<\infty} W_\loc^{1,p}(\RR^3)$, and
$\d_i\d_j\Phi_\tim^s\in\bigcap_{1\leq p<\infty} L_\loc^p(\RR^3)$ ($i,j=1,2,3$).
Furthermore, if $x\neq 0$, then
\begin{equation}
\lara{\triangle\Phi_\tim^s,1}( D^\al E_3)(x)=\lara{4\pi G\rho_\tim^s,1}( D^\al E_3)(x)
=4\pi G\:M(\earth)\:( D^\al E_3)(x).
\end{equation}
Here we have replaced the deformed earth's total mass
$M(\vphi_\tim(\earth))=\int_{\RR^3}\rho_\tim^s(x)\dvol(x)=\lara{\rho_\tim^s,1}$ by $M(\earth)$ which is possible
due to conservation of mass. Hence, $\Phi_\tim^s$ satisfies the asymptotic condition
\begin{equation}
\label{eq:poisson_decay}
 D^\al\Phi_\tim^s(x)=-G\:M(\earth)\: D^\al(1/|x|)+\mathcal O(1/|x|^{2+|\al|})
\qquad (|x|\to\infty,\:\al\in\NN_0^3).
\end{equation}
In particular, $\Phi_\tim^s(x)=-G\:M(\earth)/|x|+\mathcal O(1/|x|^{2})$ as $|x|\to\infty$, which is consistent with the
multipole expansion (see \cite[p.\ 62]{Torge:03}).
To interpret the result of Lemma \ref{lem:poisson2}, let $\chi\in\cC^\infty_\cpt(\RR^3)$ be a cutoff for the ball
$B_R(0)\sps\bigcup_{t\in I}\vphi_\tim(\earth)$
that vanishes outside $B_{2R}(0)$ 
(such a ball exists by condition (\ref{def:Am_i}) in Definition \ref{def:Am} of admissible motions).  
As $M(\earth)$ is constant, conservation of mass implies that the function
\begin{equation}
\label{eq:poisson_monopolefunction}
m^s(x):=m_{\Phi_\tim^s}(x)=-\frac{G\:M(\earth)}{|x|}\:(1-\chi(x))
\qquad (x\in\RR^3)
\end{equation}
is independent of $\Phi_\tim^s$ ($m^s$ depends only on the fixed total mass of the earth $M(\earth)$ and the choice
of the cutoff $\chi$).  
Therefore we have a decomposition of the gravitational potential $(t\in I)$  
\begin{equation}
\label{eq:poisson_monopole}
\Phi_\tim^s=m^s+\wtil{\Phi}_\tim^s, 
\end{equation}  
 where $m^s\in\cC^\infty(\RR^3)\cap L^{\infty}(\RR^3)$
with $\supp(m^s)\sbs\RR^3\setminus\bigcup_{t\in I}\vphi_\tim(\earth)$ represents the far field monopole term of the earth's
gravitational field at large distances outside the earth ($|x|\geq 2R$) and 
$\wtil{\Phi}_\tim^s\in H^2(\RR^3)$
contains the physical information and consists of all the higher order multipole
terms that model the earth's near and interior gravitational field.

Up to now we have not imposed any regularity condition on the gravitational potential with respect to time.
The reason is that due to Poisson's equation the temporal regularity of $\Phi^s$ is determined by that of $\rho^s$.
We assume that $\rho^s\in\cC^0(I,L^{\infty}(\RR^3))$ with compact support of $\rho^s_t$ in $\RR^3$ as in \eqref{eq:basic_rho_reg}. Thus, in view of Poisson's equation
$\triangle\Phi^s=4\pi G\rho^s$ or the representation of $\Phi^s$ as Newtonian potential $\Phi^s=E_3\ast(4\pi G\rho^s)$,
we have
\begin{equation}\label{eq:basic_Phi_reg}
\Phi^s\in\cC^0(I,Y^\infty(\RR^3))
\end{equation}
which implies $\wtil\Phi^s\in\cC^0(I,H^2(\RR^3))$.

\subsubsection{Conservative external field of force}\label{sssec:conforces}

In general, {\bf volume forces} (also called {\bf body forces}) are modeled by vector fields
\begin{equation}
f^s\colon\RR^3\times I\to\RR^3
\end{equation}
or $f\colon\earth\times I\to\RR^3$ in the material formulation. A
conservative external field of force can be naturally incorporated in
the calculus of variation \cite[p. 82]{Ciarlet:88}. Such a force can
be expressed as density times the gradient of a scalar potential
$F^s\colon\RR^3\times I\to\RR$,
\begin{equation}
f^s=\rho^s\nabla F^s
\end{equation}
($\nabla F^s$ is the associated acceleration field).
We assume that
\begin{equation}\label{eq:basic_force_reg}
F^s\in\cC^0(I,\Lip(\RR^3))
\qquad\text{which implies}\qquad
f^s\in\cC^0(I,L^{\infty}(\RR^3)^3)
\end{equation}
if $\rho^s\in\cC^0(I,L^{\infty}(\RR^3)^3)$.
In the calculus of variations, we need otherwise to introduce a
source via initial conditions, and then use Duhamel's
principle, or introduce friction in the originally frictionless
solid-fluid boundary integral contribution.

\subsection{Elasticity}\label{ssec:elasticity} 

We assume the material of the earth to be elastic, that is,
characterized by the existence of a time-independent relationship
between stress and strain, the constitutive relation.  The function
expressing stress in terms of strain is called the response function.
In case of hyperelasticity, the response function is the gradient of a
strain energy function $U$ with respect to strain, see
(\ref{eq:PKconstitutive}).

\subsubsection{The concept of stress}

We recall that strain, representing the deformation of the material,
may be quantified by the material strain tensor $e$
(\ref{eq:nlstraintensor}). Stress is a tensor defined by relating the
{\bf surface force} to the oriented area element on which it acts: Let
$S$ be a (hyper-)surface in a general continuum $V\sbs\RR^n$ with
sufficiently regular motion $\vphi\colon\ovl{V}\times I\to\RR^n$.  The
spatial {\bf Cauchy stress tensor} $T^s$ gives the relation of the
spatial surface force density, $\tau^s$, referred to as {\bf
  traction}, to the current area element with unit normal $\nu^s$
across which it acts:
\begin{equation}
\tau^s(\nu^s)=T^s\cdot \nu^s \qquad\text{on}\qquad \vphi_\tim(S).
\end{equation}
As is shown in \cite[Chapter 2, Theorem 2.10]{MaHu:83rep} or
\cite[Theorem 2.3-I]{Ciarlet:88}, the Cauchy stress tensor $T^s$ is
symmetric, if the axioms of conservation of mass, balance of momentum,
and balance of moment of momentum hold.

The {\bf first Piola-Kirchhoff stress tensor} $T^{\PK}$ relates $\tau^s$  to the original undeformed area element
(see \cite[(16.5) and (43 A.1)]{TrNo:04}, \cite[p.\ 71]{Ciarlet:88}). 
\begin{equation}
\tau^s(\nu^s)=T^{\PK}\cdot \nu \qquad\text{on}\qquad S.
\end{equation}
Occasionally  the stress tensor is defined in a transposed variant with the convention $\tau(\nu)=\nu\cdot T$ instead of  $\tau(\nu)=T\cdot \nu$, see e.g.\ \cite[(3.2.8) and (5.3.19)]{Malvern:69}). Note that, by construction, $T^{\PK}$ is a two-point tensor (see \cite[p.\ 48 and p.\ 70]{MaHu:83rep})
depending on space and time.   The total surface force acting on a deformed surface
$\vphi_\tim(S)$ is given by
\begin{equation}\label{surfforce}
\int_{\vphi_\tim(S)}\tau^s(\nu^s)\:\dsurf
=\int_{\vphi_\tim(S)}T_\tim^s\cdot \nu^s_\tim\:\dsurf
=\int_S J_\tim \, T_\tim \cdot(\nabla\vphi)_\tim^{-T}\cdot \nu\:\dsurf,
\end{equation}
where the  last equality follows from (\ref{eq:intsurf}) (under suitable regularity of $\vphi$ and $S$) and we have introduced $T_\tim :=T_\tim^s\circ\vphi_\tim$ as the material Cauchy stress tensor.  Varying $S$, we obtain for every $t\in I$ the relation
\begin{equation}\label{eq:TPK_Piola}
  T^{\PK}_\tim = J_\tim\, (T_\tim^s \circ \vphi_t) \cdot {\nabla\vphi_\tim}^{-T}  
  = J_\tim\, T_\tim \cdot {\nabla\vphi_\tim}^{-T},
\end{equation}
showing that  $T^{\PK}$ is the \textbf{Piola transform} of $T^s$ (\cite[pp.\ 38-39]{Ciarlet:88}).
We note that $J_\tim\, {\nabla \vphi_t}^{-1}\cdot T_\tim$  signifies the nominal stress while $T^{\PK}$ is its transpose.
Combining the divergence theorem, the integral identities (\ref{eq:intvol}) and \eqref{surfforce}, and \eqref{eq:TPK_Piola}, we obtain for an arbitrary $\Lip$ subbody $A\sbs V$, 
\begin{eqnarray}
  \int_{A}J_\tim\:\div T_\tim\:\dvol\nn = \int_{\vphi_\tim(A)}\div T_\tim^s\:\dvol
  &=& \int_{\d\vphi_\tim(A)}T_\tim^s\cdot \nu^s_\tim\:\dsurf\\
  =\int_{\d A}T_\tim^{\PK}\cdot \nu\:\dsurf
  &=&\int_{A}\div T_\tim^{\PK}\:\dvol.
\end{eqnarray}
Since  $A$ was arbitrary, we obtain the simple relation $\div T^{\PK} = J \, \div T$ (cf.\ \cite[Theorem 1.7-I]{Ciarlet:88}).

\subsubsection{Constitutive theory for elastic solids and elastic fluids}

The mechanical response of a material to an applied force depends on
the specific physical properties of the material and allows one a
classification of materials.  A body $V\sbs\RR^3$ is {\bf elastic}, if
there exists a {\bf constitutive function (response function)} $r
\colon V \times GL(3,\RR) \to \RR^{3\times 3}$, where $GL(3,\RR)$
denotes the so-called general linear group of invertible $3 \times
3$-matrices, such that the Cauchy stress tensor $T_t^s$ at $\vphi_t
(X)$ can be expressed as the value of $r$ at $X\in V$ and the
deformation gradient $(\nabla\vphi)(X,t)$, that is,
\begin{equation}
T_t^s (\vphi_t(X)) = r(X,(\nabla\vphi)(X,t))
\end{equation}
\cite[Chapter 4, (4) on p.129]{Chadwick:1998}. If thermal effects are
included, then $r$ also depends on the specific entropy.

Material symmetries are reflected by the invariance of the response function under certain changes of the reference configuration for an arbitrary deformation gradient. An arbitrary invertible matrix $A \in GL(3,\RR)$, representing such a change of reference configuartion, can be expressed as the product $A = (|\det A|^{1/3} \, 1_{3\times 3}) \cdot (|\det A|^{-1/3} A)$, where the first factor is a pure dilation and the second factor is a unimodular matrix. Thus, it
suffices to consider changes of the reference configurations associated with the unimodular group $\RR^{3 \times 3}_1:=\{A\in\RR^{3\times 3}:\:|\det A|=1\}$.
The {\bf isotropy group} $I(r)$ is the collection of all static
density-preserving deformations at $X\in V$ that cannot be detected by
experiment and is a subgroup of the unimodular group
$\RR^{3\times 3}_1$.
With respect to the reference configuration $\vphi_\timO=\Id_V$, the isotropy group is defined by
\begin{equation}
   I(r):=\{H\in\RR^{3\times 3}_1:  
   \:r(X,F)=r(X,F\cdot H)\:\:\forall\:F\in GL(3,\RR),\: X\in V\}.
\end{equation} 
The isotropy group $I(r)$ encodes the symmetry behavior of the constitutive function $r$ and allows to  classify the material as solid or as fluid (see \cite[p.\ 77 ff.]{TrNo:04}). In particular, a compressible inviscid fluid is an elastic material with maximal symmetry possessing no preferred configuration. This can be expressed more precisely in the form of the result, the material representation version of which is shown in \cite[Chapter 4, Section 5]{Chadwick:1998}, that an elastic material whose isotropy group is the entire unimodular group necessarily is a compressible inviscid fluid and has a Cauchy stress, which is a scalar multiple of the identity matrix. This may be turned into a definition for an elastic fluid. 

\begin{definition}[{{\bf Elastic fluid in terms of the isotropy group}}]\label{def:fluids} 
An elastic material whose isotropy group coincides with the unimodular
group is a compressible inviscid fluid and satisfies $T_t^s (x) =
- h(x,\rho^s_t(x)) \, 1_{3\times 3}$ with a scalar function $h$. 
\end{definition}

We recall from Remark \ref{rem:rho_reg}(ii) that an admissible motion
$\vphi$ will automatically be orientation preserving, that is, $J =
\det (\nabla \vphi) > 0$. Therefore, it suffices henceforth to
consider in the matrix argument $F$, representing $\nabla \vphi_t
(X)$, of the response functions $r(X,F)$ only elements from the group
$GL_+(3,\RR) :=\{F\in GL(3,\RR):\:\det F>0\}$ of invertible
orientation-preserving real $3\times 3$-matrices.

A {\bf hyperelastic material} (cf.\ \cite[p. 141]{Ciarlet:88} or
\cite[pp.\ 210-211]{MaHu:83rep}) is defined to be an elastic material
such that the response function is given in terms of a {\bf stored
  energy function} $U \colon V \times GL_+(3,\RR) \to \RR$, see
(\ref{geomresp}) below. If we assume the motion to be isentropic (that
is reversible), $U$ is given by the \textbf{internal energy density}
per unit mass (in the isothermal case this, up to a constant, also
equals the free energy).

It can be shown that the principle of material frame indifference (see
\cite[p.\ 146]{Ciarlet:88} or \cite[Chapter 3, Theorem
  2.10]{MaHu:83rep}) implies that the stored energy function $U$
depends on $F = \nabla\vphi$ only through the \textbf{right
  Cauchy-Green (or deformation) tensor}
$$
   C := \nabla\vphi^T\cdot\nabla\vphi
$$
or, equivalently, in terms of the material strain tensor $e = (C -
1_{3\times 3})/2$ introduced in (\ref{eq:nlstraintensor}).  In fact,
the theory of elasticity and its tensorial quantities are more
naturally treated in the context of metric components corresponding to
the Euclidean metric in more general coordinates describing (parts of)
the body $V$ and its deformed state $\vphi_t(V)$, or even with $V$ and
$\vphi_t(V)$ being submanifolds of arbitrary Riemannian manifolds
(cf.\cite{MaHu:83rep}). In such geometric context, the adjoint
appearing in the definition of the right Cauchy-Green tensor $C$ has
to be interpreted in terms of the Riemannian metrics on the tangent
spaces of $V$ and $\vphi_t(V)$. Then $C$ can be described in terms of
the pullback to $V$ of the \textbf{Riemannian metric} $g_t$ on the
deformed body $\vphi_t(V)$, that is,
$$
   C = \vphi_t^* g_t
$$
(strictly speaking, this pull-back relation holds for the associated
   tensor $C^\flat$, with the first index lowered, but we adopt the
   notational convention of keeping $C$ as in \cite[Remark on
     p.\ 194]{MaHu:83rep}). Therefore, $C$ and $U$ depend on the
   metric $g_t$ as well, that is, we have $C$ as a tensor field $(X,t)
   \mapsto C(\nabla \vphi_t(X), g_t(\vphi_t(X))) = (\vphi_t^*
   g_t)(X)$, and in a simplified setting one may read all formulae
   involving these quantities with $g_t$ equal to the Euclidean metric
   in standard Cartesian coordinates.

We consider the basic regularity and notational convention for the
internal energy. In the model of a fluid-solid composite earth model
$\earth$ with a positively oriented piecewise $\Lip$-regular motion
$\vphi\in\Am(\ovl\earth\times I)$ we follow \cite{Ciarlet:88} in
assuming that $U \colon B \times GL_+(3,\RR) \to \RR$ is a bounded
function of its first argument and continuously differentiable with
respect to its second argument, more precisely, we require that
\begin{multline}\label{eq:basic_U_reg}
   \text{the map}\ (X,F) \mapsto U(X,F)\ \text{belongs to}\
           L^{\infty}(\earth,\cC^1(GL_+(3,\RR)) ,
\\ 
   \text{which implies that}\ (X,t) \mapsto U(X,(\nabla\vphi)(X,t))\
   \text{is in}\ \cC^0(I,L^{\infty}(\earth)) .
\end{multline}
In the sense of the latter functional composition, we will in the
sequel often consider $U$ also as a function $\earth\times I\to\RR$,
or by even further abuse of notation, always write $U$ while
considering it interchangeably as function of $X$ and $\nabla \vphi_t$
or $C$ or $g_t$ or $e$ or $\rho$.

We recall that $g_t$ is the Riemannian metric after deformation under
$\varphi_t$ and that the transition between material and spatial
representations involves base point transformations only, hence we
have (with $x = \varphi_t(X)$)
\begin{multline}\label{UtransUs}
  U^s(x,\nabla \vphi_t(\vphi_t^{-1}(x)), g_t(x)) = 
  U^s(\vphi_t(X), \nabla \vphi_t(X), g_t(\varphi_t(X)))
\\[0.25cm]
  = U(X, C(\nabla \vphi_t(X), g_t(\varphi_t(X)))) 
  = U(X, (\varphi_t^* g_t)(X)).
\end{multline}
Observing the formulae for pull backs of the
$\left(\begin{smallmatrix} 0\\2 \end{smallmatrix}\right)$ tensor field
$A$, push forwards of the $\left(\begin{smallmatrix}
  2\\0 \end{smallmatrix}\right)$ tensor field $\frac{\d U}{\d C}$, and
basic properties of the trace defining the matrix inner product, we
obtain
\begin{multline*}
   \frac{\d U^s}{\d g_t} \colon A = 
    \frac{\d U}{\d C} \colon \left( \frac{\d C}{\d g_t} \cdot A \right) =
    \frac{\d U}{\d C} \colon \left( {\vphi_t}^* A \right) =
    \frac{\d U}{\d C} \colon \left( \nabla \vphi_t^T \cdot A \cdot \nabla \vphi_t \right) \\
   = \left( \nabla \vphi_t  \cdot \frac{\d U}{\d C} 
     \cdot \nabla \vphi_t^T \right) \colon A
    = \left( {\vphi_t}_*\, \frac{\d U}{\d C} \right) \colon A,
\end{multline*}
which we may state in the brief form
\begin{equation}
   \frac{\d U^s}{\d g_t} = {\vphi_t}_*\, \frac{\d U}{\d C}.
\end{equation}

\begin{remark} Note that the explicit mapping between $U^s$ and $U$ given in \eqref{UtransUs} above is implicitly contained in \cite[Equation (2.98)]{DaTr:98}, the latter  follows from \eqref{UtransUs} and Corollary \ref{cor:rho0int}.
\end{remark}

In the context of a general metric $g_t$ on the deformed body,  the isothermal hyperelastic response function is of the form 
\begin{equation}\label{geomresp}
r(X,\nabla \vphi_t(X))=\rho^0(X) \, g_t^\sharp(\vphi_t(X)) \cdot \frac{\d U}{\d (\nabla \vphi_t)}(X,\nabla \vphi_t(X)), 
\end{equation}
where $g_t^\sharp$ implements the raising of the first index to
produce a $\left(\begin{smallmatrix}
  2\\0 \end{smallmatrix}\right)$-tensor field from one of type
$\left(\begin{smallmatrix} 1\\1 \end{smallmatrix}\right)$. It is shown
by a sequence of arguments in \cite[Chapter 3, Theorem 2.4, the
  transition from $\hat{P}$ to $S$ on the bottom part of p.\ 195, and
  Propositions 2.11-12]{MaHu:83rep}, that the latter relation is
equivalent to
\begin{equation} \label{eq:Cconstitutive} T_t^s = 2
  \rho_t^s \, \frac{\d U^s}{\d g_t} = 2 \rho_t^s \, {\vphi_t}_*
  \frac{\d U}{\d C} = \rho_t^s \, {\vphi_t}_* \frac{\d U}{\d e}.
\end{equation}
 
Applying $T_\tim =T_\tim^s\circ\vphi_\tim$  in the relation \eqref{eq:TPK_Piola}, we obtain from \eqref{eq:Cconstitutive} the following expression for the first Piola-Kirchhoff stress tensor (in agreeement with \cite[Equation (2.141)]{DaTr:98} and \cite[Chapter 3, Proposition 2.6 and Notation 4.7]{MaHu:83rep} and suppressing the subscript $t$) 
\begin{equation}\label{eq:PKconstitutive}
  T^{\PK} =  \rho^0\, \nabla \vphi \cdot  \frac{\d U}{\d e}   =  
  \rho^0\, \frac{\d U}{\d(\nabla\vphi)},
\end{equation} 
where the second inequality follows from $\frac{\d e}{\d(\nabla
  \vphi)}\cdot A = \frac{1}{2} (A^T \cdot \nabla \vphi + \nabla
\vphi^T \cdot A)$, which implies that
\begin{multline*}
   \frac{\d U}{\d(\nabla \vphi)}\colon
   A = \frac{\d U}{\d e}\colon \frac{1}{2} (A^T \cdot \nabla \vphi +
\nabla \vphi^T \cdot A)
\\ = \frac{1}{2} \left( \nabla \vphi \cdot \frac{\d
  U}{\d e}^T \right) \colon (A^T)^T + \frac{1}{2} \left( \nabla \vphi \cdot
\frac{\d U}{\d e}^T \right) \colon A = (\nabla \vphi \cdot \frac{\d U}{\d e})
\colon A
\end{multline*}
by the properties of the matrix inner product (in terms of the trace)
and the symmetry of $\frac{\d U}{\d e}$ (thanks to that of $T^s$).

We briefly discuss the constitutive equations for
elastic fluids in the regions $\earth^\fluid$ of the earth. In a general elastic
continuum, (the spatial) \textbf{pressure} is defined by 
\begin{equation} 
  p^s=-\frac{1}{3} \, \tr\: T^s. 
\end{equation}
In case of an elastic fluid, Theorem \ref{def:fluids} implies $p_t^s(x) = h(x,\rho^s_t(x))$ and  therefore, recalling the general geometric form of the response function from \eqref{geomresp}, 
$$
   T_t^s = - p_t^s \, g_t^\sharp, 
$$   
or, in case of the ambient space being Euclidean with Cartesian coordinates, simply 
\begin{equation}\label{eq:fluidCauchystress} 
   T^s = - p^s \, 1_{3\times 3}.
\end{equation} 

In the hyperelastic case, the isotropy group can equivalently be determined from the invariance properties of the internal elastic energy in place of the response function, therefore the energy of an isentropic fluid depends only on its density, that is, in place of $U(X,C(\ldots))$ in \eqref{UtransUs} we have 
$$
   U(X, \rho_t(X)) = U^s(\vphi_t(X),\rho_t(X), g_t(\vphi_t(X))) 
   = U^s(\vphi_t(X),\rho_t^s(\vphi_t(X)), g_t(\vphi_t(X))), 
$$   
which (upon further abuse of notation) implies $\frac{\d U}{\d \rho}(X) = \frac{\d U^s}{\d \rho^s} (\vphi_t(X))$ and thus $(\frac{\d U}{\d \rho})^s = \frac{\d U^s}{\d \rho^s}$. Furthermore, 
we observe that 
$$
   \displaystyle \frac{\d U}{\d(\nabla\vphi)} = \frac{\d U}{\d \rho} \frac{\d \rho}{\d (\nabla\vphi)} =\frac{\d U}{\d \rho} \left( - \frac{\rho^0}{J^2} \frac{\d J}{\d(\nabla\vphi)}\right) = - \rho \, \frac{\d U}{\d \rho} (\nabla\vphi)^{-T},
$$ 
where we have used $\rho^0 = \rho J$ and the fact that $\frac{\d
  J}{\d(\nabla\vphi)} = J \, (\nabla\vphi)^{-T}$ (as noted also in
\cite[p.\ 10]{MaHu:83rep}). Taking all this into account, we may
determine the spatial pressure
$p_\tim^s\colon\vphi_\tim(\earth^\fluid)\to\RR$ (in agreement with
\cite[eq.\ (44.9), p.\ 246]{Gurtin:09}) directly from $U^s$ and the
spatial density $\rho^s$ in view of \eqref{eq:fluidCauchystress},
\eqref{eq:PKconstitutive}, and \eqref{eq:TPK_Piola} as follows:
\begin{equation}\label{eq:pressureUs}
    p^s =  - \left( J^{-1} \rho^0\, \frac{\d U}{\d \rho}
        \left(-\frac{\rho^0}{J}\right)\right)^s   
        = (\rho^s)^2 \, \frac{\d U^s}{\d \rho^s} .
\end{equation} 

Due to the conservation of mass we can express the dependence of
pressure on density again by a dependence on the deformation
gradient. This allows one to treat hyperelastic solid and fluid
regions in a unified way by prescribing $U$.

\subsubsection{Dynamical interface conditions}

We  come to a description of the dynamical interface conditions, constraining the traction on $\Lip$-surfaces $S$ inside the earth (recall Definition \ref{def:Lipdom} for $\Lip$-surfaces).  Note that an analysis carried out in Subsection \ref{nonlinearstationarity} below reveals that these interface conditions can be obtained as Euler-Lagrange equations or nonlinear boundary conditions from variational methods.
 By Newton's third law the spatial traction vector must satisfy \cite[(1.2)]{BeSa:00}
\begin{equation}
\tau^s(-\nu^s)=-\tau^s(\nu^s) 
\end{equation} 
on any spatial  surface $\vphi_t(S)$. Here the insertion of $\pm\nu^s$ also means evaluating $\tau^s$ on the $\pm$-side of the surface $\vphi_\tim(S)$. With the convention \eqref{eq:normal} that $\nu^s$ points from the $-$ to the $+$-side we identify $\nu^s=\nu^{s,-}$. The definition of the Cauchy stress then implies $-T^{s,-}\cdot\nu^s=T^{s,-}\cdot(-\nu^s)=\tau^s(-\nu^s)=-\tau^s(\nu^s)=-T^{s,+}\cdot\nu^s$.
Spatial traction is thus continuous across interior boundaries, that is, the  
 spatial jump condition 
\begin{equation}\label{eq:CauchyIBC}
[T^s]_-^+\cdot\nu^s=0\qquad\text{on}\qquad \vphi_\tim(S)
\end{equation} 
holds for all $t\in I$. 
In case of {\bf perfect slip}, that is, a {\bf frictionless surface}, $T^s$ also satisfies 
the normality condition \eqref{eq:normality} in the spatial formulation, while neglecting viscosity of the fluids inside the earth, that is  
\begin{equation}\label{eq:Cauchynofriction}
T^s\cdot\nu^s=(\nu^s\cdot T^s\cdot\nu^s)\nu^s \quad\text{on}\quad \vphi_t(\Sig^\fxs).
\end{equation}
If \eqref{eq:CauchyIBC} is considered on the earth's exterior boundary, if atmospheric stresses are neglected (that is the earth is considered as an elastic body in vacuum), it  gives the dynamical boundary condition
\begin{equation}\label{eq:CauchyIBC_ext}
T^s\cdot\nu^s=0\qquad\text{on}\qquad \vphi_\tim(\d\earth).
\end{equation}
In  combination, the conditions \eqref{eq:CauchyIBC} and \eqref{eq:Cauchynofriction} guarantee the absence of the fluid-solid interface and boundary integral contribution to the energy balance in the spatial representation.

\subsubsection{Prestress}

The stress at equilibrium time $t_0$ is called the {\bf prestress} of the medium:
\begin{equation}\label{eq:prestress}
T^0:=T^s_\timO=T_\timO=T_\timO^{\PK}.
\end{equation}
Evaluating \eqref{eq:CauchyIBC} at $t=t_0$ immediately gives the continuity of $T^0$ in directions normal to any surface $S\sbs \earth$,
\begin{equation}\label{eq:IBC_FST0}
[T^0]_-^+\cdot\nu=0\qquad\text{on}\qquad S.
\end{equation}
Similarly, \eqref{eq:CauchyIBC_ext} leads to the equilibrium zero-traction condition on $\d\earth$
\begin{equation}\label{eq:NBC_T0}
T^0\cdot\nu=0\qquad\text{on}\qquad \d\earth.
\end{equation}
The prestress in a fluid material is a pure pressure, 
\begin{equation}\label{eq:fluidprestress}
T^0=-p^0 1_{3\times 3},
\end{equation}
where $p^0$ is called the equilibrium or initial {\bf hydrostatic pressure}.
On a fluid-solid boundary $T^0$ satisfies the normality condition \eqref{eq:normality} 
\begin{equation}\label{eq:surfpressure}
T^0\cdot \nu=-p^0\nu\qquad\text{with}\qquad p^0=-\nu\cdot T^0\cdot \nu\qquad\text{on}\qquad \Sig^\fxs.
\end{equation}
Thus we have 
\begin{equation}\label{eq:IBC_FSp0}
\qquad [p^0]_-^+=0\qquad\text{on}\qquad \Sig^\fxs
\end{equation}
with $p^0=-\frac{1}{3}\tr\: T^0$. The trace-free difference 
\begin{equation}\label{eq:T0dev}
T^0_{\mathrm{dev}}:=T^0-\frac{1}{3}(\tr\: T^0)1_{3\times 3}=T^0-(-p^0) 1_{3\times 3}=T^0+p^0 1_{3\times 3}
\end{equation} 
is called {\bf deviatoric prestress}. As $T^0_{\mathrm{dev}}$ is symmetric (due to symmetry of $T^0$), only three of its components
are independent. It vanishes in the fluid regions but is generally non-zero in the solid parts of the earth.
  
\section{The action integral}\label{sec:3}

\subsection{The structure of the action integral of the composite fluid-solid earth model}
\label{ssec:actionstructure}
The state of the uniformly rotating, elastic, and self-gravitating earth model is characterized by specifying its
motion $\vphi$, its gravitational potential $\Phi^s$, and its density $\rho^s$ for given elastic properties which
are encoded in the internal energy density $U$, and given force potential $F^s$. In accordance with the
regularity conditions of Definition \ref{def:Am}, (\ref{eq:basic_Phi_reg}), and (\ref{eq:basic_rho_reg})
$\vphi$, $\Phi^s$, and $\rho^s$ are modeled as elements of the following basic configuration spaces:
\begin{definition}[{{\bf Basic configuration spaces}}]\label{def:nonlin_reg}
\begin{eqnarray}
W_\motion&:=&\{\vphi\in\cC^0(I,L^{\infty}(\RR^3))^3:\:\vphi|_{\ovl\earthi\times I}\in\Am(\ovl\earth\times I)\},\nn\\
W_\grav&:=&\cC^0(I,Y^\infty(\RR^3)),\nn\\
W_\density&:=&\{\rho^s\in\cC^0(I,L^{\infty}(\RR^3)):\:\supp(\rho_\tim^s)\sbs\vphi_\tim(\ovl\earth)\:\forall\:t\in I\}.\nn
\end{eqnarray}
\end{definition}
We further assume that $U$ satisfies (\ref{eq:basic_U_reg}), that is $U\in L^{\infty}(\earth,\cC^1(GL_+(3,\RR)))$.  By Lemma \ref{lem:Linfty} we obtain from $\vphi\in W_\motion$ that 
$U^s\in\cC^0(I,L^{\infty}(\RR^3))$ with $\supp(U_\tim^s) \subseteq \vphi_\tim(\ovl\earth)$ for all $t\in I$.
Finally, $F^s$ is required to satisfy (\ref{eq:basic_force_reg}), namely $F^s\in\cC^0(I,\Lip(\RR^3))$, which, combined with $\rho^s\in W_\density$, implies that the force $f^s=\rho^s\nabla F^s$ is in $\cC^0(I,L^{\infty}(\RR^3))^3$
(with compact support contained in $\vphi_\tim(\ovl\earth)$ at time $t\in I$).

By Hamilton's principle the earth's configuration $(\vphi,\Phi^s,\rho^s)\in W_\motion \times W_\grav \times W_\density$
is a stationary point for the action  functional
\begin{equation}
\action\colon W_\motion \times W_\grav \times W_\density \to \RR,
\end{equation}
which we suppose to be of the following basic structure:
\begin{equation}\label{eq:action_uncon}
\action(\vphi,\Phi^s,\rho^s)=\int_I\left(\int_{\vphi_\tim(\earthi)}L^s(x,t)\:\dvol(x)
+\int_{\vphi_\tim(\Sig^\fxs)}L^s_{\vphi_\tim(\Sig^\fxs)}(x,t)\:\dsurf(x)\right)\mathrm{d}t.
\end{equation}
Here, $L^s$ is a short-hand notation  for the volume Lagrangian density, which may depend explicitly on $(x,t)$ besides being a function of the space- and time-derivatives of the state-variables $(\vphi,\Phi^s,\rho^s)$. In detail, the integrand $L^s$ is to be understood as a  function
$$
(x,t)\mapsto L^s\left(x,t,\vphi(x,t),\Phi^s(x,t),\rho^s(x,t),\nabla\vphi(x,t),\nabla\Phi^s(x,t),\nabla\rho^s(x,t),
\dot\vphi(x,t),\dot\Phi^s(x,t),\dot\rho^s(x,t)\right).
$$
The  part of $\action$ representing the surface action consists of a temporally integrated surface integral over all fluid-solid boundaries within the earth model:
\begin{equation}\label{eq:action_surf}
\action_{\Sig^\fxs}(\vphi):=\int_I\action_{\Sig^\fxs,\tim}(\vphi)\:\mathrm{d}t
:=\int_I\int_{\vphi_\tim(\Sig^\fxs)}L^s_{\vphi_\tim(\Sig^\fxs)}(x,t)\:\dsurf(x)\:\mathrm{d}t
=\int_I\int_{\Sig^\fxs}L_{\Sig^\fxs}(X,t)\:\dsurf(X)\:\:\mathrm{d}t.
\end{equation} 
Accounting
for the mutual interaction of fluid and solid regions, $\action_{\Sig^\fxs}(\vphi)$ 
  only occurs if fluid regions are present in the earth model. 
We will see later that $\action_{\Sig^\fxs}$ is, in fact, independent of $\Phi^s$ and $\rho^s$, that is, $L^s_{\vphi_\tim(\Sig^\fxs)}$ is a function 
$$
(x,t)\mapsto L^s_{\vphi_\tim(\Sig^\fxs)}\left(x,t,\vphi(x,t),\nabla\vphi(x,t),\dot\vphi(x,t)\right).
$$
As will be seen in the  following section, in the linearized model the purely second-order surface action does not vanish, even in the frictionless case. Its Lagrangian density is given by \eqref{eq:Lsurf2}, an approximation which implies that the second-order surface energy terms account for the work done by slip of material at fluid-solid boundaries against the initial traction due to prestress. This interpretation is consistent with \cite[p.\ 96, (3.232)]{DaTr:98}.

Observe that the action does not contain an integral over the
(exterior) boundary of the earth $\d \earth$.  This will correspond to
the zero-traction (homogeneous Neumann) boundary condition
(\ref{eq:NBC_TPK}): $T^{\PK} \cdot\nu = 0$ on $\d\earth$. The explicit
form (\ref{eq:Lfull_surf}) of the material surface Lagrangian
$L_{\Sig^\fxs}$ will be obtained as a consequence of energy balance in
\ref{ssec:energybalance}.

Let us further specify the volume Lagrangian density by expressing it in terms of physical quantities.
Since we consider elastic, gravitational, and also other (internal or external) conservative forces, 
according to Hamilton's principle we have
\begin{equation}
L^s=E_\kin^s-E_\pot^s\qquad\textrm{with}\qquad
E_\pot^s=E_\elast^s+E_\grav^s+E_\ext^s.
\end{equation}
Here $E_\kin^s$, $E_\elast^s$, $E_\grav^s$, and $E_\ext^s$ are the kinetic, elastic, gravitational
energy densities of the earth and of the other (internal or external) conservative forces respectively, given by
\begin{equation}
E_\kin^s=\frac{1}{2}(v^s+\Omega\times x)^2\rho^s,\quad
E_\elast^s=U^s\rho^s, \quad
E_\grav^s=\frac{1}{2}\:\Phi^s\rho^s, \quad\textrm{and}\quad E_\ext^s=F^s\rho^s. \nn
\end{equation}
The factor $\frac{1}{2}$ in $E_\grav^s$ is due to self-gravitation
(cf.\ \cite{DaTr:98}, \cite{WoDe:07}). Thus
\begin{equation}\label{spatialLagrangianfirst}
L^s=\Big(\frac{1}{2}(v^s+\Omega\times x)^2-U^s-\frac{1}{2}\:\Phi^s-F^s\Big)\rho^s.
\end{equation}
Note that the additional term $\Omega\times x$ in the formula of $E_\kin^s$ is due to the adoption of a
co-rotating coordinate system \cite[p.\ 247]{MaRa:94}. 
The $i$th component of $\Omega\times x$ is $\eps_{ijk}\Omega_jx_k$  (with $\eps_{ijk}$ denoting the Levi-Civita permutation symbol and recalling summation convention), thus
we have $(v^s+\Omega\times x)^2=(v^s_i)^2 + 2\:v^s_i\eps_{ijk}\Omega_jx_k+\Omega_l^2x_k^2-(\Omega_kx_k)^2$. The expression $v^s_i\eps_{ijk}\Omega_jx_k$ represents the Coriolis force
(more precisely, the Coriolis acceleration is given by $2\:v^s\times\Omega$), while
the last two terms of the right-hand side correspond to the centrifugal force. Therefore, upon introducing the centrifugal potential $\Psi^s\colon\RR^3\to\mathbb R$,
\begin{equation}\label{eq:Psi}
\Psi^s(x)=-\frac{1}{2}\Big(\underbrace{\Omega^2x^2-(\Omega\cdot x)^2}_{= (\Omega \times x)^2}\Big)
=-\frac{1}{2}\Big(\Omega_l^2x_k^2-(\Omega_kx_k)^2\Big)
\qquad (x\in\RR^3),
\end{equation}
we may write $L^s$  in the form
\begin{equation}
\label{eq:Lagrangian_uncon}
L^s=\Big(\frac{1}{2}\:(v^s)^2 +v^s\cdot(\Omega\times x)-\Psi^s-U^s-\frac{1}{2}\:\Phi^s-F^s\Big)\rho^s.
\end{equation}
We note for later purpose that 
\begin{equation}\label{psinote}
  \nabla \Psi^s(x) = \Omega \times (\Omega \times x). 
\end{equation}

\begin{remark}[{{\bf Geometric formulation}}]\label{geom1} In terms of a more general $n$-dimensional Riemannian manifold structure on the deformed body we  would have to consider an invariant form of the Lagrange density $L^s  dV$ with $dV$ being the Riemannian volume form, which in coordinates reads $\sqrt{\det ((g_t(x))_{ij})_{1 \leq i,j \leq n}} \; dx_1 \wedge \ldots \wedge dx_n$. The first two terms in $L^s$ would be written invariantlty as $\frac{1}{2} g_t(v^s,v^s)$ and $g_t(v^s,\omega^s)$, where $\omega^s$ denotes the velocity field of the rotational  motion.  Moreover, $\Psi^s$ would read $g_t(\omega^s,\omega^s)$ in invariant form and $\Phi^s$ has to be obtained from geometric constructions with Green functions (or distributions) for the Laplacian corresponding to the metric $g_t$ replacing the convolution with the Newtonian potential in \eqref{eq:poisson_newton}. Thus, the geometric invariant version of the Lagrangian reads
\begin{equation}
\label{eq:Lagrangian_uncon_geom}
L^s=\Big(\frac{1}{2}\: g_t(v^s,v^s) + g_t(v^s,\omega^s) -
g_t(\omega^s,\omega^s) -U^s-\frac{1}{2}\:\Phi^s-F^s\Big)\rho^s.
\end{equation}
Note that the action involves a surface integral, which we also
discuss briefly here in terms of a Riemannian manifold $(M,g)$ with
boundary $\d M$, embedded as submanifold via $\iota_{\d M} \colon \d M
\hookrightarrow M$. The ``surface measure'' $\dsurf$ on $\d M$ is, in
fact, the volume form of the induced Riemannian metric $\tilde{g} =
\iota^*_{\d M} (g)$ on $\d M$. By \cite[Corollary 15.34]{Lee:13}, we
have $\dsurf = \iota^*_{\d M}(\nu \lrcorner \,\dvol)$, where $\nu$ is
the outward unit normal vector field on $\d M$ and $\nu \lrcorner
\,\dvol$ denotes contraction, that is, $\nu \lrcorner \,\dvol(Y_1,
\ldots, Y_{n-1}) = \dvol(\nu, Y_1, \ldots, Y_{n-1})$. Existence and
uniqueness of $\nu$ is guaranteed by \cite[Propositions
  15.33]{Lee:13}, where the construction of $\nu$ is based on a
function $f \colon M \to \RR$, with $d f \neq 0$ on $\d M$, defining
the boundary as level set $\d M = f^{-1}(0)$ (cf.\ \cite[Propositions
  5.43]{Lee:13}) in the form $\nu = - \nabla f / \sqrt{g(\nabla f,
  \nabla f)}$ (with $\nabla$ denoting the metric gradient).
\end{remark}

Since the functions $v^s$, $U^s$, $\rho^s$, $\Phi^s$, $F^s$ are
elements of $\cC^0(I,L^{\infty}(\RR^3))$, $\rho_t^s$ is compactly
supported for every $t \in I$, and $I$ is a bounded interval, we
obtain that $L^s$ belongs to $L^{\infty}(\RR^3 \times I)$ with compact
support $\supp(L_\tim^s)\sbs\vphi_\tim(\ovl\earth)$ for almost all
$t\in I$. The regularity conditions on $\vphi$ and $U$ imply
$L^s_{\Sig^\fxs}\in\cC^0(I,L^{\infty}(\bigcup_{t\in
  I}\vphi_\tim(\Sig^\fxs)))$, thus also the surface Lagrangian defines
an integrable function.  Consequently, the action functional
$\action(\vphi,\Phi^s,\rho^s)$ is defined on all of $W_\motion \times
W_\grav \times W_\density$.

\begin{remark}
The material quantities $U$, $\rho$ and $\Phi$ are functions in
$\cC^0(I,L^{\infty}(\earth))$ by Lemma \ref{lem:Linfty}. In
particular, since $\vphi\in\Am(\ovl\earth\times I)$ and $\rho^0\in
L^{\infty}(\earth)$, conservation of mass implies that
$\rho=\rho^0/J\in\Lip(I,L^{\infty}(\earth))$ (see Remark
\ref{rem:rho_reg}).
\end{remark}

In the following sections we will modify $\action$ and $L$ in order to incorporate self-gravitation and eliminate the dependence on density via conservation of mass. We thereby arrive at the material formulation of the action $\action''\colon W_\motion \times W_\grav\to\RR$ in \eqref{eq:action_full}.
However, due to the nonlinearity of $U$, the interrelation of spatial and material quantities, and the fact that $W_\motion \times W_\grav$ is not a Banach space  (more precisely, not normable), a rigorous mathematical framework for the linearization and a calculus lies beyond the basic notion of Fr\'echet differentiability or related concepts and the derivation of EL, NBC, and NIBC in (\ref{eq:Ham_EL}) to (\ref{eq:Ham_IBCS}) with rigorous proofs in full generality is left for a possible future study in finding the most appropriate concept among the infinite dimensional calculi in locally convex vector spaces.   
We will introduce a physically reasonable approximation of the action integral yielding a linearization in Subsection 
\ref{sec:approxaction}, which enables us to apply the calculus of variations in a Sobolev framework and eventually lead to
linear governing equations presented in Section \ref{sec:5}.

\subsection{A variational problem constrained by self-gravitation}\label{ssec:actionconstrained}
In addition to the stationarity of the action $\action$ (\ref{eq:action_uncon})
with respect to variations in $\vphi$, $\Phi^s$ and $\rho^s$, the fields $\Phi^s$ and $\rho^s$ are linked via Poisson's
equation (\ref{eq:poisson}) $\triangle \Phi^s=4\pi G\rho^s$ in $\RR^3\times I$.
Furthermore, $\vphi$ and $\rho^s$ are coupled through conservation of mass.
Therefore we have to interpret the stationarity of $\action$ to hold under these constraints. We begin with
incorporating Poisson's equation via a Lagrange multiplier method in Lemma \ref{prop:con_var1} below
(a similar result is also used in \cite[p.\ 34]{WoDe:07} or in \cite[p.\ 88]{DaTr:98}). We will
apply the Lagrange multiplier theorem for constrained variational problems in a Hilbert space setting to a
conveniently reformulated constrained variational problem. The result will be the modified action integral
$\action'$ (\ref{eq:action_mod}). Afterwards we will incorporate conservation of mass when
transforming the spatial Lagrangian density to its material representation, to finally arrive at the
action $\action''$ (\ref{eq:action_full}).

We begin with some preparatory observations. Let
$(\vphi,\Phi^s,\rho^s)\in W_\motion \times W_\grav \times W_\density$
(see Definition \ref{def:nonlin_reg}).  First, note that for $t\in I$,
$L_\tim^s=L^s(.,t)$ is proportional to $\rho_\tim^s$ and thus
supported in $\vphi_\tim(\earth)$, that is $L_\tim^s\in
L^{\infty}_\cpt(\RR^3)$ with $\supp(L_\tim^s) \subseteq
\vphi_\tim(\ovl\earth)$, and clearly also
$L^s_{\vphi_\tim(\Sig^\fxs)}$ vanishes outside
$\vphi_\tim(\ovl\earth)$. Hence, we may integrate $L_{\tim}^s$ over
$\RR^3$ without changing the definition of $\action$
(\ref{eq:action_uncon}). Second, we may neglect $L^s_{\Sig^\fxs}$ for
the moment, because $\Phi^s$ and $\rho^s$ do not directly contribute to
the surface action, and consider
\begin{equation}
   \action(\vphi,\Phi^s,\rho^s)
        = \int_I\int_{\RR^3}L^s(x,t)\:\dvol(x)\mathrm{d}t.
\end{equation} 
We now have $\RR^3\times I$ as both, the integration domain and the
domain of the constraint equation. Third, we observe that, if we
replace $\Phi_\tim^s$ by $m^s+\wtil{\Phi}_\tim^s$ according to
equation (\ref{eq:poisson_monopole}), $m^s$ does not contribute to
$\action$, since by construction (which implicitly uses conservation
of mass) its support is disjoint from $\supp(\rho_\tim^s) \supseteq
\supp(L_\tim^s)$. Therefore, $\action(\vphi,\Phi^s,\rho^s) =
\action(\vphi,\wtil\Phi^s,\rho^s)$. Fourth, since there is no
explicit dependence on time in Poisson's equation, it suffices to
investigate the constrained variational problem for the spatially
integrated part of the action integral
\begin{equation}
\action_\tim(\vphi_\tim,\Phi_\tim^s,\rho_\tim^s):=\int_{\RR^3}L^s(x,t)\:\dvol(x),
\end{equation}
 considered as a functional of $\vphi_\tim$, $\rho_\tim^s$, and $\Phi_\tim^s$ for fixed $t\in I$.
Since Poisson's equation does not involve $\vphi$, we can consider the constrained functional
$\action_\tim$ as a functional of $\rho_\tim^s$ and $\Phi_\tim^s$ only and keep $\vphi_\tim$ fixed (by abuse of notation we
use the same symbol $\action_\tim$).

The observations above thus show that for $t\in I$ fixed, $\action_\tim$ can be written in the form
\begin{equation}
\action_\tim(\Phi_\tim^s,\rho_\tim^s)=\action_\tim(\wtil\Phi_\tim^s,\rho_\tim^s)
=\int_{\RR^3}\Big(a_\tim^s(x)-\frac{1}{2}\wtil\Phi_\tim^s(x)\Big)\rho_\tim^s(x)\dvol(x)
=\lara{a_\tim^s-\frac{_1}{^2}\:\wtil\Phi_\tim^s|\rho_\tim^s}_{L^2(\RR^3)}.
\end{equation}
Here we use the abbreviation
\begin{equation}
a_\tim^s(x):=\frac{1}{2}(v_\tim^s(x))^2+v_\tim^s(x)\cdot(\Omega\times x)-\Psi^s(x)-U_\tim^s(x)-F_\tim^s(x).
\end{equation}
Note that  $a_\tim^s\in L^{\infty}_\cpt(\RR^3)\sbs L^2(\RR^3)$ (upon extending $\Psi^s$ by zero outside the current earth) and does not depend on $\Phi_\tim^s$ and $\rho_\tim^s$. The bracket
$\lara{.|.}_{L^2(\RR^3)}$ denotes the $L^2$ inner product.
Consequently,  we may identify $\action_\tim$  with the functional
\begin{equation}
\label{eq:action_Poisson}
\scJ:H^2(\RR^3)\times L^2(\RR^3)\to\RR\:,\quad
\scJ(\wtil{\Phi}^s,\rho^s)\:=\:\lara{a^s-\frac{_1}{^2}\:\wtil{\Phi}^s|\rho^s}_{L^2(\RR^3)},
\end{equation}
with $a^s\in L^2(\RR^3)$. Poisson's equation can be stated in the form $g=0$ with
\begin{equation}
\label{eq:action_g}
g:H^2(\RR^3)\times L^2(\RR^3)\to L^2(\RR^3)\:,\quad
g(\wtil{\Phi}^s,\rho^s)\:=\:\triangle\wtil{\Phi}^s+\triangle m^s-4\pi G\rho^s.
\end{equation}
Note that we have omitted the explicit time dependence to simplify the notation. The functions $\wtil\Phi^s$,
$\rho^s$, and $a^s$ in $\scJ$ and $g$ correspond to $\wtil\Phi_\tim^s$, $\rho_\tim^s$, and $a_\tim^s$ in $\action_\tim$ and
Poisson's equation $\triangle(\wtil\Phi_\tim^s+m^s)=4\pi G\rho_\tim^s$ for $t\in I$ as above.
Finally, the constrained variational problem for $\action$ (\ref{eq:action_uncon}) and Poisson's equation
(\ref{eq:poisson}) can be formulated as follows:
For $\scJ$ given by (\ref{eq:action_Poisson}) and $g$ given by (\ref{eq:action_g}), find
$(\wtil{\Phi}^s_\ast,\rho^s_\ast)\in H^2(\RR^3)\times L^2(\RR^3)$ such that $\scJ$ is stationary under the constraint $g=0$.

In the following lemma we deduce a necessary
condition for $(\wtil{\Phi}^s_\ast,\rho^s_\ast)$ by the Lagrange multiplier method (see \ref{rem:LagMult} for a review of the corresponding Banach space theory).

\begin{lemma}
\label{prop:con_var1}
Let $\scJ:H^2(\RR^3)\times L^2(\RR^3)\to\RR$ given in (\ref{eq:action_Poisson}) be stationary at
$(\wtil{\Phi}^s_\ast,\rho^s_\ast)\in H^2(\RR^3)\times L^2(\RR^3)$ under the constraint $g=0$ with
$g:H^2(\RR^3)\times L^2(\RR^3)\to L^2(\RR^3)$ given in (\ref{eq:action_g}).
Then there exists a Lagrange multiplier $\la\in L^2(\RR^3)$ such that the modified functional
$\scJ^\la:=\scJ+\lara{\la|g(.)}_{L^2(\RR^3)}:H^2(\RR^3)\times L^2(\RR^3)\to\RR$,
that is
$$
\scJ^\la(\wtil{\Phi}^s,\rho^s)=\scJ(\wtil{\Phi}^s,\rho^s)
+\lara{\la|\triangle(\wtil{\Phi}^s+m^s)-4\pi G\rho^s}_{L^2(\RR^3)},
$$
is stationary at $(\wtil{\Phi}^s_\ast,\rho^s_\ast)$.
Moreover, in this case $\la$ satisfies the equation $\triangle\la=\frac{1}{2}\:\rho^s_\ast$ in $\cD'(\RR^3)$.
\end{lemma}

\begin{proof}
To apply the Lagrange multiplier theorem we just have to show differentiability of $\scJ$ and $g$ and verify that $(\wtil{\Phi}^s_\ast,\rho^s_\ast)$ is
a regular point of $g$ (note that $\ker Dg(\wtil{\Phi}^s_\ast,\rho^s_\ast)$ automatically has a topological
complement in the Hilbert space $H^2(\RR^3)\times L^2(\RR^3)$).
The functional $\scJ$ is differentiable on $H^2(\RR^3)\times L^2(\RR^3)$, because it can be written as a sum
of a linear functional and a quadratic form associated to a bounded linear operator on
$H^2(\RR^3)\times L^2(\RR^3)$: 
$$
\scJ(\wtil{\Phi}^s,\rho^s)= 
\Lara{\left(\arr{cc}{0 \\ a}\right)\Big|\left(\arr{cc}{\wtil{\Phi}^s\\ \rho^s}\right)}_{L^2\times L^2}
-\frac{1}{2}\Lara{\left(\arr{cc}{0 & \Id_{L^2}\\ 0 & 0}\right)\cdot \left(\arr{cc}{\wtil{\Phi}^s\\ \rho^s}\right)
\Big|\left(\arr{cc}{\wtil{\Phi}^s\\ \rho^s}\right)}_{L^2\times L^2}
$$
Here, $\lara{.|.}_{L^2\times L^2}$ denotes the $L^2(\RR^3)\times
L^2(\RR^3)$-inner product and boundedness of $\scJ$ follows from the
estimate $\norm{.}{L^2\times L^2}\leq\norm{.}{H^2\times L^2}$.  Since
$g$ is a continuous (affine) linear operator its differentiability is
clear. The derivative $h := Dg(\wtil{\Phi}^s,\rho^s):H^2(\RR^3)\times
L^2(\RR^3)\to L^2(\RR^3)$ at $(\wtil{\Phi}^s,\rho^s)\in
H^2(\RR^3)\times L^2(\RR^3)$ reads $h(y,z) = \triangle y-4\pi Gz$.  We
show that $h$ is surjective. Indeed, let $r\in L^2(\RR^3)$ and define
$f:=-\mathscr F^{-1}(\frac{\mathscr F r(\xi)}{|\xi|^2+4\pi G})\in
H^2(\RR^3)$. Then $\mathscr F h(f,f)(\xi)=\mathscr F(\triangle f-4\pi
Gf)(\xi)=-(|\xi|^2+4\pi G)\mathscr F f(\xi)=\mathscr F r(\xi)$ on
$\cS'(\RR^3)$, that is $h(f,f)=r$.  Thus, every point in
$H^2(\RR^3)\times L^2(\RR^3)$ is a regular point of $g$ and we can
apply the Lagrange multiplier theorem: If
$(\wtil{\Phi}^s_\ast,\rho^s_\ast)\in H^2(\RR^3)\times L^2(\RR^3)$ is a
stationary point for $\scJ$ under the constraint $g=0$, then there
exists a linear functional $\la':L^2(\RR^3)\to\RR$ such that
$(\wtil{\Phi}^s_\ast,\rho^s_\ast)$ is stationary for the modified
functional $\scJ^{\la'}:H^2(\RR^3)\times L^2(\RR^3)\to\RR$,
$\scJ^{\la'}:=\scJ-\la'\circ g$.  Using $L^2$-duality we may replace
the action of $\la'$ by an inner product with $\la\in L^2(\RR^3)$ such
that we can rewrite $\scJ^{\la'}$ in the form $\scJ^\la$ claimed in
the proposition: for $(\wtil{\Phi}^s,\rho^s)\in H^2(\RR^3)\times
L^2(\RR^3)$,
$\scJ^\la(\wtil{\Phi}^s,\rho^s)\:=\:\scJ(\wtil{\Phi}^s,\rho^s)
+\lara{\la|\triangle(\wtil{\Phi}^s+m^s)-4\pi G\rho^s}_{L^2(\RR^3)}$ or
$\scJ^{\la}:=\scJ+\lara{\la|g(.)}_{L^2(\RR^3)}$ in short-hand
notation.

To prove $\triangle\la=\frac{1}{2}\:\rho^s_\ast$ in the sense of
distributions, observe that stationarity of $\scJ^\la$ at
$(\wtil{\Phi}^s_\ast,\rho^s_\ast)\in H^2(\RR^3)\times L^2(\RR^3)$ is
equivalent to $D\scJ^\la(\wtil{\Phi}^s_\ast,\rho^s_\ast)(y,z)=0$ for
all $(y,z)\in H^2(\RR^3)\times L^2(\RR^3)$ and differentiation yields
$0=D\scJ^\la(\wtil{\Phi}^s_\ast,\rho^s_\ast)(y,z)=\lara{a^s-\frac{1}{2}\:\wtil{\Phi}^s_\ast|z}_{L^2(\RR^3)}
-\lara{\frac{1}{2}\rho^s_\ast|y}_{L^2(\RR^3)}+\lara{\la|\triangle
  y-4\pi G z}_{L^2(\RR^3)}
=\lara{a^s-\frac{1}{2}\:\wtil{\Phi}^s_\ast-4\pi G \la|z}_{L^2(\RR^3)}
-\lara{\frac{1}{2}\rho^s_\ast|y}+\lara{\la|\triangle
  y}_{L^2(\RR^3)}$.  Setting $z=0$, which corresponds to considering
only the stationarity of $\scJ^\la$ with respect to its first variable
$\wtil{\Phi}^s$, it follows that
$\lara{\frac{1}{2}\rho^s_\ast|y}_{L^2(\RR^3)}=\lara{\la|\triangle
  y}_{L^2(\RR^3)}$ holds for every $y\in H^2(\RR^3)$.  Since
$\cD(\RR^3)\sbs H^2(\RR^3)$ this implies the
$\lara{\cD'(\RR^3),\cD(\RR^3)}$-duality
$\lara{\frac{1}{2}\rho^s_\ast,y}=\lara{\la,\triangle
  y}=\lara{\triangle\la,y}$, hence
$\triangle\la=\frac{1}{2}\:\rho^s_\ast$ holds in $\cD'(\RR^3)$.
\end{proof}

The modified functional $\scJ^\la$ in Lemma \ref{prop:con_var1} corresponds to the modified action integral
\begin{equation}
\action_\tim^{\la_\tim}(\vphi_\tim,\Phi_\tim^s,\rho_\tim^s):=\int_{\RR^3}\left(
L_\tim^s+\la_\tim(\triangle\Phi_\tim^s-4\pi G\rho_\tim^s)\right)(x)\:\dvol(x)
\end{equation}
for $t\in I$. Here, the time-dependence of the Lagrange multiplier $\la\in L^2(\RR^3)$ in the proposition is
indicated by $\la_\tim$.
As the proof shows, the equation $\triangle\la=\frac{1}{2}\:\rho^s$ in $\cD'(\RR^3)$ is a consequence of
the stationarity of the modified action integral $\scJ^{\la}$ with respect to variations in $\wtil\Phi^s$ solely.
The Fr\'echet derivatives of $\scJ^{\la_\tim}$ and of $\action_\tim^{\la_\tim}$ are of the exact same structure
(although the functionals act on different spaces). This suggests to consider the equation
\begin{equation}\label{eq:PoissonLag}
\triangle\la_\tim=\frac{1}{2}\:\rho_\tim^s\quad\textrm{in}\quad \cD'(\RR^3)
\end{equation}
as a necessary condition for the stationarity of the original action
integral $\action$ constrained by Poisson's equation
$\triangle\Phi_\tim^s=4\pi G\rho_\tim^s$.  The next lemma shows that
we can solve $\triangle\la_\tim=\frac{1}{2}\:\rho_\tim^s$ and identify
the Lagrange multiplier within the setting of Definition
\ref{def:nonlin_reg}. We recall that the gravitational potential is
given by $\Phi_\tim^s=4\pi G\, \rho_\tim^s\ast E_3$.

\begin{lemma}
\label{lem:con_var2} 
Let $\rho_\tim^s\in L^{\infty}_\cpt(\RR^3)$,  then  $\la_\tim := \frac{1}{2}\:\rho_\tim^s\ast E_3=\Phi_\tim^s/8\pi G$ is the unique solution to (\ref{eq:PoissonLag}) in $Y(\RR^3)$.
\end{lemma}

\begin{proof}
Since $\rho_\tim^s\in L^{\infty}_\cpt(\RR^3)$ and  $\triangle\la_\tim=\frac{1}{2}\:\rho_\tim$ holds by construction, we have $\la_\tim\in Y(\RR^3)$ (even in $Y^\infty(\RR^3)$). It only remains to note that $y \in Y(\RR^3)$ and $\triangle y = 0$ implies $y = 0$ (due to analyticity of $y$ and the decay condition in $Y(\RR^3)$).
\end{proof}

Thus the Lagrange multiplier is proportional to the gravitational potential.
Consequently, for deriving the equations governing the state of our earth model, it is sufficient to study the
unconstrained stationarity of the action integral
$\action'(\vphi,\Phi^s,\rho^s):=\int_I\action'_t(\vphi_\tim,\Phi_\tim^s,\rho_\tim^s)\mathrm{d}t$
with, still omitting $\action_{\Sig^\fxs}(\vphi)$,
\begin{equation}
\action'_\tim(\vphi_\tim,\Phi_\tim^s,\rho_\tim^s):=\action_\tim^{\Phi_\tim^s/8\pi G}(\vphi_\tim,\Phi_\tim^s,\rho_\tim^s)
=\int_{\RR^3}\Big(L_\tim^s+\frac{1}{8\pi G}\Phi_\tim^s\triangle\Phi_\tim^s-\frac{1}{2}\:\Phi_\tim^s\rho_\tim^s\Big)(x)\:\dvol(x).
\end{equation}
We now rewrite the term involving the Laplacian using integration by
parts.  Since $\Phi_\tim^s\in Y^\infty(\RR^3)$ Lemma
\ref{lem:poisson1} yields $\Phi_\tim^s\in L^{\infty}(\RR^3)$,
$\Phi_\tim^s(x)=\mathcal O(1/|x|)$ as $|x|\to\infty$,
$\d_{x_l}^2\Phi_\tim^s\in L^1_{\loc}(\RR^3)$, and
$\d_{x_l}^2\Phi_\tim^s(x)=\mathcal O(1/|x|^3)$ as $|x|\to\infty$,
hence $\Phi_\tim^s\:\d_{x_l}^2\Phi_\tim^s\in L^1(\RR^3)$.  Therefore,
when studying the
term $$\int_{\RR^3}(\Phi_\tim^s\:\triangle\Phi_\tim^s)(x)\dvol(x)=
\sum_{l=1}^3\int_{\RR^3}(\Phi_\tim^s\:\d_{x_l}^2\Phi_\tim^s)(x)\dvol(x)
,$$ by Fubini's theorem, it is sufficient to consider the
one-dimensional integrals
$\int_{\RR}(\Phi_\tim^s\:\d_{x_l}^2\Phi_\tim^s)(x)\:dx_l$ for $1\leq
l\leq 3$.  Since by Lemma \ref{lem:poisson1} $\Phi_\tim^s$ and
$\d_{x_l}\Phi_\tim^s\in W_\loc^{1,1}(\RR^3)$, which consists of
functions with a representative that is absolutely continuous on
(bounded intervals of) almost all lines parallel to the coordinate
axes (cf.\ \cite[Theorem 2.1.4, p.\ 44]{Ziemer:89}), integration by
parts may be applied (cf.\ \cite[Theorem 3.35, p.\ 106]{Folland:99}).
Together with the decay conditions of $\Phi_\tim^s$ and
$\d_{x_l}\Phi_\tim^s$ it leads to
$\int_{\RR}(\Phi_\tim^s\:\d_{x_l}^2\Phi_\tim^s)(x)\:dx_l=-\int_{\RR}(\d_{x_l}\Phi_\tim^s)^2(x)\:dx_l$
for $1\leq l\leq 3$. Here, $(\d_{x_l}\Phi_\tim^s)^2=(\nabla\Phi^s)^2$
also is in $L^1(\RR^3)$ and, again by Fubini, we arrive at
\begin{equation}
\action_\tim'(\vphi_\tim,\Phi_\tim^s,\rho_\tim^s)=\int_{\RR^3}\left(
L_\tim^s-\frac{1}{2}\:\Phi_\tim^s\rho_\tim^s-\frac{1}{8\pi G}(\nabla\Phi_\tim^s)^2\right)(x)\:\dvol(x).
\end{equation}
Thus, integrating over the time interval and again including the surface action $\action_{\Sig^\fxs}(\vphi)$ (\ref{eq:action_surf}),
\begin{eqnarray}
&& \action'(\vphi,\Phi^s,\rho^s)=
\int_I\int_{\RR^3}\bigg(\Big(\frac{1}{2}\:(v^s)^2 +v^s\cdot(\Om\times x)
-U^s-(\Phi^s+\Psi^s)-F^s\Big)(x,t)\rho^s(x,t)\nn\\
&&\qquad\qquad\qquad\qquad\qquad\qquad
 -\frac{1}{8\pi G}(\nabla\Phi^s)^2(x,t)\bigg)\:\dvol(x)\mathrm{d}t+ \action_{\Sig^\fxs}(\vphi).
\label{eq:action_mod}
\end{eqnarray}

\subsection{The full action integral incorporating conservation of mass}\label{sssec:action_full}
Finally, we use conservation of mass and formulate $\action'$ as integral over material
quantities.  This step is crucial, as the material description is natural for the variational formulation of field equations based on Hamilton's principle. Moreover, this substitution will eliminate the dependence of the action integral on the variation of density with respect to time.

\subsubsection{The action integral in the material formulation}

For $t\in I$, we decompose $\action'_\tim$ and write,
\begin{eqnarray}
&&\action'_\tim(\vphi_\tim,\Phi_\tim^s,\rho_\tim^s)=\int_{\vphi_\tim(\earthi)}\Big(\frac{1}{2}\:(v^s)^2 +v^s\cdot(\Om\times x)
-U^s-(\Phi^s+\Psi^s)-F^s\Big)(x,t)\rho^s(x,t)\:\dvol(x)\nn\\
&&\qquad\qquad\qquad\qquad\qquad\qquad -\int_{\RR^3}\frac{1}{8\pi G}(\nabla\Phi^s)^2(x,t)\:\dvol(x)
+ \action_{\Sig^\fxs,\tim}(\vphi).
\end{eqnarray}
Note  that we have made use of the fact that $\supp(\rho_\tim^s) \subseteq \vphi_\tim(\ovl\earth)$. In the first integral we substitute $x=\varphi(X,t)$ and make the transitions from the spatial to the material representations $\dot\vphi$, $U$,
$\Phi$, $\Psi$, and $\rho$ according to the general rule 
$q_\tim=q_\tim^s\circ\vphi_\tim$, in particular, recall Equation \eqref{UtransUs} relating $U$ and $U^s$.
 Conservation of mass means, thanks to Corollary \ref{cor:rho0int}, that we may apply the relation  
$\int_{\vphi_\tim(\earthi)}(h_\tim^s\rho_\tim^s)(x)\:\dvol(x)=\int_\earthi(h_\tim\rho^0)(X)\:\dvol(X)$ to the function $h^s=\frac{1}{2}\:(v^s)^2 +v^s\cdot(\Om\times
x)-U^s-(\Phi^s+\Psi^s)-F^s$ and thereby reduce the dependence on $\rho^s$ to one on $\rho^0$. Moreover, $\supp(\rho^0) \subseteq \ovl\earth$, hence we may extend the domain of integration to all of $\RR^3$, which allows us to combine both integrals into one (upon renaming the variable in the second integral). To summarize, we may rewrite $\action'(\vphi,\Phi^s,\rho^s)$ as $\action''(\vphi,\Phi^s)$, where
\begin{equation} 
  \action''\colon W_\motion \times W_\grav\to\RR 
\end{equation} 
is given by
\begin{equation}\label{eq:action_full}
  \action''(\vphi,\Phi^s) =  \int_I \Big(
     \underbrace{\int_{\RR^3}L''(X,t)\:\dvol(X)}_{=: \action''_{\text{vol}}(t)} 
     \,+\, 
     \underbrace{\int_{\Sig^{\fxs}}L_{\Sig^\fxs}''(X,t)\:\dsurf(X)}_{=: \action''_{\Sig^{\fxs}}(t)}
     \Big)\mathrm{d}t,
\end{equation} 
with material volume Lagrangian density \begin{equation}\label{eq:Lfull}
L''(X,t):=\Big(\frac{1}{2}\:\dot\vphi^2
+\dot\vphi\cdot(\Om\times\vphi)-U-(\Phi+\Psi)-F\Big)(X,t)\rho^0(X)
-\frac{1}{8\pi G}(\nabla\Phi^s)^2(X,t), 
\end{equation} 
and material surface Lagrangian density $L_{\Sig^\fxs}''=L_{\Sig^\fxs}=-E_{\Sig^{\fxs}}''$  (the surface energy density $E_{\Sig^{\fxs}}''$ is introduced in \eqref{eq:energy_surf} below)
as 
\begin{equation}\label{eq:Lfull_surf} 
  L_{\Sig^\fxs}''(X,t)  =
  -\left[\int_{t_0}^t\dot\vphi_\timd(X)\cdot T^{\PK}_\timd(X)\:\dint
  t'\right]_-^+\cdot \nu(X), 
\end{equation}
which can be obtained from the energy balance, as we discuss with several aspects in a  little detour in \ref{ssec:energybalance} below.

As will be noted later in \eqref{Asurfdis} we may (at least in the sense of surface densities) write the part contributing to the surface action in the form 
$$
  \action''_{\Sig^{\fxs}}(t)
  = -  \int_{t_0}^t  \int_{\Sig^{\fxs}_\timd}  \left[v_\timd^s \cdot T_\timd^s   \right]_-^+ \cdot \nu_\timd^s \, dS_\timd \, \dint t' ,
$$
where $\Sig^{\fxs}_\timd = \varphi_\timd(\Sig^{\fxs})$, and therefore
obtain the full action $\action''$ in the pure spatial form
\begin{equation}\label{eq:action_full_spatial}
 \action''(\vphi,\Phi^s) =  \int_I \Big(\int_{\RR^3}L''^s(x,t)\:\dvol(x) \,+\, \action''_{\Sig^{\fxs}}(t)\Big)\mathrm{d}t,
\end{equation} 
where the spatial volume Lagrangian density reads
\begin{multline}\label{Lzweifach}
L''^s(x,t) :=  \Big(\frac{1}{2}\:(v^s)^2 + v^s \cdot(\Om\times \Id_{\RR^3})
    -  U^s  -  (\Phi^s+\Psi^s)\Big)(x,t) \rho^s(x,t) 
\\
   - \frac{1}{8\pi G}(\nabla\Phi^s)^2(x,t)
   = \Big(\frac{1}{2}\:(v^s + \Om\times \Id_{\RR^3})^2 
     - U^s - \Phi^s\Big)(x,t) \rho^s(x,t)
                        - \frac{1}{8\pi G}(\nabla\Phi^s)^2(x,t).
\end{multline}

\subsubsection{Determining the surface action from energy balance}
\label{ssec:energybalance}

The specific form of the surface action may be motivated from the following energy considerations. Our starting point is energy conservation, which takes the integrated material form (cf.\ \cite[p.\ 143; without $\Sig^{\fxs}$-term]{MaHu:83rep} or \cite[(3.201), p.\ 91; linearized setting and without $\d\earth$-term]{DaTr:98}; recall the convention $\nu=\nu^-$ (\ref{eq:normal}))
\begin{equation}\label{eq:energyconservation}
\frac{d}{dt}\left(\int_{\earth^\fxs}E''_\tim\:\dvol\right)
=\int_{\d\earthi}\dot\vphi\cdot T^{\PK}\cdot \nu\:\dsurf
-\int_{\Sig^{\fxs}}[\dot\vphi\cdot T^{\PK}]_-^+\cdot \nu\:\dsurf.
\end{equation}
The left-hand side is the time derivative at time $t$ of the total
energy corresponding to the material volume Lagrangian in its final
form $L''_\tim$ (\ref{eq:Lfull}). Due to the regularity required of
the gravitational potential, $\Phi^s\in
W_\grav\sbs\cC^0(I,\cC^1(\RR^3))$, we have no jump terms caused by
gravity.  Moreover, for simplicity, we consider the case of zero
external forces $F=0$. The integral over the exterior boundary
$\d\earth$ vanishes due to the zero traction natural boundary
condition (\ref{eq:NBC_TPK}) at the free surface.  Rewriting
\begin{equation}
\frac{d}{dt}\left(\int_{\earth^\fxs}E''_\tim\:\dvol
+\int_{t_0}^t\int_{\Sig^{\fxs}}[\dot\vphi_\timd\cdot T^{\PK}_\timd]_-^+\cdot \nu\:\dsurf\dint t'\right)=0,
\end{equation}
we can identify the second term as the work against traction due to slip at fluid-solid boundaries that occurred in the time interval $[t_0,t]$ (see also \cite[p.\ 96]{DaTr:98} in the linearized setting). Thus, upon interchanging time and surface integration, energy conservation suggests that we have to introduce a surface energy density given by
\begin{equation}\label{eq:energy_surf}
E_{\Sig^{\fxs}}''(X,t)
=\left[\int_{t_0}^t\dot\vphi_\timd(X)\cdot T^{\PK}_\timd(X)\:\dint t'\right]_-^+\cdot \nu(X)
\end{equation}
and get $\int_{\Sig^{\fxs}}E_{\Sig^{\fxs}}''\dsurf$ as an additional
part of the energy, corresponding to the work which slip performs
against the normal traction. The complete Lagrangian density of the
earth model is given as kinetic minus potential energy density (see
Section \ref{ssec:actionstructure}). Since we do not consider
dissipative forces, the surface energy density obtained above can be
viewed as a part of potential energy (it clearly does not have the
form of a kinetic energy). Consequently, the complete Lagrangian
density has to include an additional surface density given by
\begin{equation}\label{eq:Legendre_surf}
L_{\Sig^{\fxs}}''=-E_{\Sig^{\fxs}}''
\end{equation}
as was announced in (\ref{eq:Lfull_surf}).

\begin{remark}[{{\bf Validity of the Legendre transform for
surface densities}}]\label{rem:Legendre}
The volume energy density $E$ of a dynamical system in the variable
$q(x,t)$ is generally obtained via the Legendre transform of the
volume Lagrangian density
\begin{equation}\label{eq:Legendre}
L\mapsto E:=\dot q\, \frac{\d L}{\d\dot q}-L.
\end{equation}
If the system also involves densities on a discontinuity surface $S$,
their Legendre transform is given by $L_S\mapsto E_S:=-L_S$, which
coincides with \eqref{eq:Legendre_surf}.  This form of the surface
Legendre transform is clear, if $L_S$ does not depend on $\dot q$
(which is true at least in the linear case, cf.\ the second-order
surface Lagrangian density \eqref{eq:Lsurf2} or see
\cite[p.\ 90]{DaTr:98}; in particular, note that (\ref{eq:Lfull_surf})
depends linearly on $\dot\vphi$).  However, as is discussed in a
general nonsmooth geometric setting in \cite[Lemma 3.1]{FeMaWe:03},
elements in the tangent space of the configuration $q$ have zero
normal jump across the discontinuity surface.  Consequently, the
derivative of the surface Lagrangian density with respect to $\dot q$
vanishes and Legendre transform (\ref{eq:Legendre}) is reduced to
$E_S=-L_S$, when applied to a surface density.
\end{remark}

\begin{remark}[{{\bf Local energy balance by Noether's theorem}}]
If $L$ does not explicitly depend on time, then by Noether's theorem (see \cite[p.\ 283]{MaHu:83rep}) the volume energy density $E$, defined by the Legendre transform (\ref{eq:Legendre}) of the volume Lagrangian density $L$, satisfies the local energy balance equation
\begin{equation}
0=\dot E+\div\left(\dot q\cdot\frac{\d L}{\d(\nabla q)}\right).
\end{equation}
For the elastic, self-gravitating earth model we have $q=(\vphi,\Phi^s)$ and $L=L''$ defined in (\ref{eq:Lfull}).
Inserting the volume energy density obtained from (\ref{eq:Legendre})
\begin{eqnarray}
E''&=&\dot\vphi\frac{\d L''}{\d\dot\vphi}+{\dot\Phi}^s\frac{\d L''}{\d{\dot\Phi}^s}-L''
\:=\:\Big(\frac{1}{2}\:\dot\vphi^2 +U+\Phi+\Psi+F\Big)\rho^0+\frac{1}{8\pi G}(\nabla\Phi^s)^2,
\end{eqnarray} 
the energy balance equation reads
\begin{eqnarray}\label{eq:energyconservation_local}
0&=&\dot E''+\div\Big(\dot\vphi\cdot\frac{\d L''}{\d(\nabla\vphi)}+\dot{\Phi}^s\frac{\d L''}{\d(\nabla\Phi^s)}\Big)
\:=\:\dot E''-\div\Big(\dot\vphi\cdot T^{\PK}+\frac{1}{4\pi G}\:\dot{\Phi}^s\:\nabla\Phi^s\Big).
\end{eqnarray}
Integrating over $\earth^\fxs$, using the divergence theorem for composite domains 
(Lemma \ref{lem:divthmComposite}) and the regularity properties of the gravitational potential,
gives (\ref{eq:energyconservation}). 
\end{remark}

\subsection{Stationarity of the action and nonlinear dynamical equations}\label{nonlinearstationarity}

Stationarity of the action $\action''$ \eqref{eq:action_full} with
respect to variations in the state variables $(\vphi,\Phi^s)$ gives
the nonlinear dynamical equations describing the motion of a uniformly
rotating, nonlinear elastic, self-gravitating composite fluid-solid
continuum with initial density $\rho^0\in L^{\infty}(\earth)$ and
internal elastic energy function $U\in
L^{\infty}(\earth,\cC^1(\RR^{3\times 3}))$.  By considering two
classes of variations separately, one class with support avoiding to
the fluid-solid boundaries and the other class with support near these
boundaries, we may consider the contributions of
$\action''_{\text{vol}}$ and $\action''_{\Sig^{\fxs}}$ independently
in deriving Euler-Lagrange equations from stationarity of the total
action functional. While the investigation of stationarity of
$\action''_{\text{vol}}$ is classical, in case of
$\action''_{\Sig^{\fxs}}$ we have to use a concept of weak (or
distributional) stationarity described in detail in \ref{weakstat}
below.  For a discussion in detail of the Euler-Lagrange equations for
linearized elasticity theory with quadratic Lagrangian densities in a
Sobolev space setting we may refer to Subsection
\ref{sssec:varQuadratic}.

\subsubsection{Stationarity of the nonlinear volume action}

\paragraph{Material description}
The nonlinear dynamical equations are the classical Euler-Lagrange
equations (EL) \eqref{eq:Ham_EL}
$$\d_t(\d_{\dot y}L)+\nabla\cdot(\d_{\nabla y}L)-\d_yL=0$$
for the state variable $y = (\vphi,\Phi^s)$ and the volume Lagrangian
$L = L''$ \eqref{eq:Lfull} in the form
\begin{equation}
L''= \left(\frac{1}{2}\:\dot\vphi^2+\dot\vphi\cdot(\Om\times\vphi)-U-(\Phi^s+\Psi^s+F^s)\circ\vphi\right) \rho^0
-\frac{1}{8\pi G}(\nabla\Phi^s)^2.
\end{equation}
We have
\begin{eqnarray}
  \d_{\dot\vphi}L'' &=& \rho^0\left(\dot\vphi+\Om\times\vphi\right),\nn\\
  \d_{\nabla \vphi}L'' &=& -\rho^0\d_{\nabla \vphi}U=-T^{\PK},\nn\\
  \d_{\vphi} L'' &=& \rho^0\left(\dot\vphi\times\Omega+\left(\nabla(\Phi^s+\Psi^s+F^s)\right)\circ\vphi\right),\nn
\end{eqnarray}
where we applied the constitutive relation (\ref{eq:PKconstitutive}) in the second equation.
Variations of $\action''_{\text{vol}}$ in \eqref{eq:action_full} with respect to $\vphi$ thus result in the EL
\begin{multline}
\d_t(\d_{\dot\vphi}L'')+\nabla\cdot(\d_{\nabla \vphi}L'')-\d_{\vphi} L''\\
= \rho^0\left(\ddot\vphi+\Om\times\dot\vphi\right)-\nabla\cdot T^{\PK}
-\rho^0\left(\dot\vphi\times\Omega+\left(\nabla(\Phi^s+\Psi^s+F^s)\right)\circ\vphi\right)=0.
\end{multline}
The terms 
\begin{equation}
g=-(\nabla\Phi^s)\circ\vphi \qquad \text{and} \qquad (\nabla\Psi^s)\circ\vphi=\Omega\times(\Omega\times\vphi)
\end{equation} 
are the material gravitational acceleration and centrifugal
acceleration respectively. Consequently, we have obtained the full
nonlinear material momentum equation \cite[(2.39)]{DaTr:98}
\begin{equation}\label{materialeqm}
\rho^0\left(\ddot\vphi+2\Om\times\dot\vphi+\Omega\times(\Omega\times\vphi)\right)=\nabla\cdot T^{\PK}+\rho^0 g+f,
\end{equation}
that is generalized to include an additional conservative forcing term
$f=\rho^0(\nabla F^s)\circ\vphi$, see \ref{sssec:conforces}.

Variations with respect to $\Phi^s$ give Poisson's equation
\eqref{eq:poisson}:
$$\triangle\Phi^s=4\pi G\rho^s.$$
Indeed, as we have seen in \eqref{eq:action_full_spatial} the action $\action''$ does not change if we write the volume Lagrangian $L''$ in spatial form: 
\begin{equation}
L''^s=  \Big(\frac{1}{2}\:(v^s)^2 + v^s \cdot(\Om\times x)
    -  U^s  -  (\Phi^s+\Psi^s + F^s)\Big) \rho^s  - \frac{1}{8\pi G}(\nabla\Phi^s)^2.
\end{equation}
As a result, we get $\d_{\dot\Phi^s}L''^s=0$, $\d_{\nabla \Phi^s}L''^s= -\frac{1}{4\pi G}(\nabla\Phi^s)$, $\d_{\Phi^s} L ''^s= - \rho^s$
and the corresponding EL read
$$\d_t(\d_{\dot\Phi^s}L''^s)+\nabla\cdot(\d_{\nabla \Phi^s}L''^s)-\d_{\Phi^s} L''^s
= -\frac{1}{4\pi G}\triangle\Phi^s+\rho^s=0 .$$

In addition to the equations of motion and for the gravitational
potential, we obtain the external and welded solid-solid
  dynamical boundary conditions as consequences of stationarity as
well.  As is clear from \eqref{eq:Ham_NBC} and \eqref{eq:Ham_IBC},
stationarity of $\action''_{\text{vol}}$ with respect to $\vphi$
directly implies the zero-traction natural boundary condition
\begin{equation}\label{eq:NBC_TPK}
T^{\PK}\cdot\nu=0
\end{equation}
on $\d\earth$ and the dynamical interface condition of continuity of normal traction
\begin{equation}\label{eq:IBC_TPK}
[T^{\PK}]_-^+\cdot \nu=0
\end{equation}
across all welded solid-solid interior boundaries surfaces $S\sbs B^\text{S}$. The natural and interior boundary conditions obtained from variations of $\action''$ with respect to $\Phi^s$ reduce to the condition
\begin{equation}\label{eq:IBC_Phis}
[\nabla\Phi^s]_-^+\cdot \nu=0,
\end{equation}
valid across all surfaces.
This is equivalent to continuity of the normal component of spatial gravitational acceleration $g^s=-\nabla\Phi^s$ and thus is no new constraint (we even have $\Phi^s_\tim\sbs\cC^1(\RR^3)$ and thus $g^s_\tim\sbs\cC^0(\RR^3)^3$ a.e.\ $t\in I$, see \ref{ssec:gravity}).

\paragraph{Spatial description}
A geometric framework for a completely spatial variational principle
in continuum mechanics is presented in \cite{GBMR:12} and in
\cite{AOS:11} including self-gravitation (an earlier attempt for a
spatial formulation based on a Lagrange multiplier argument can be
found in \cite{SeWh:68}). According to \cite[Subsection 3.2, Equation
  (3.13)]{GBMR:12} (where we have to neglect any additional symmetry
group action in our case), the independent variables in the spatial
action are $v^s$, $\Sig^{\fxs}_\tim$, spatial density, deformation (or
the deformation tensor), and the Riemannian metric $g_t$ on the
deformed earth $\vphi_t(B)$. As described in \cite[Equation (3.16) and
  Theorem 3.2]{GBMR:12}, the metric $g_t$ is tied to the variation of
$\Sig^{\fxs}_\tim$ corresponding to the equation of boundary movement
in the form
\begin{equation}\label{boundary_movement}
  \d_t \Sig^{\fxs}_\tim = g_t(v^s,\nu_t^s).
\end{equation}
 
We start from the Lagrange density in the form given in the second line of \eqref{Lzweifach}, thus writing the spatial volume Lagrangian  in the form 
$$
  L''^s  =  \Big(\frac{1}{2}\:(v^s + \Om\times x)^2 
    -  U^s  -  \Phi^s\Big) \rho^s  - \frac{1}{8\pi G}(\nabla\Phi^s)^2.
$$
Compared to the Lagrange density $\Lambda$ used in \cite[Subsection 2.2]{AOS:11}  (upon identifications of the notation used in \cite{AOS:11} as follows: $f$ with $\vphi_t^{-1}$, $-\tau$ with $T^s$, $U$ with $-\Phi^s$, $(v_1,v_2,v_3)$ with $\dot\vphi^s_t$, $\rho$ with $\rho^s$, $n \varepsilon$ with $\rho^s U^s$) the only change is in the term corresponding to the kinetic energy, where the spatial velocity $v^s$ is replaced by $v^s + \Om\times x$ to incorporate rotation.

The equation of motion derived in \cite[Equation (2.20)]{AOS:11} from variation of $\Lambda$ is
\begin{equation}\label{inertialeqm}
   \rho^s (\dot v^s +  \nabla v^s \cdot v^s) = 
        \nabla \cdot T^s - \rho^s \,\nabla \Phi^s,
\end{equation}
where the terms $\dot v^s +  \nabla v^s \cdot v^s$ correspond to $v_\mu \d_\mu v_i = \d_t v_i + v_j \d_j v_i$  in \cite[Equation (2.20)]{AOS:11}, that is, the material time derivative  denoted by $D_t v^s$ (observe that the notation $D_t a^s$ in \cite{DaTr:98}  for the derivative $a^s$ corresponds to $\d_t a^s + \nabla a^s \cdot v^s$ in our setting). 

The material time derivative of a field $(x,t)\mapsto a^s(x,t)$ in inertial space can be written as
\begin{equation}\label{rotequ0}
D_ta^s=D^{\text{rot}}_t a^s+\Omega\times a^s.
\end{equation}
where relates $(.)^{\text{rot}}$ indicates the coordinates with respect to the rotating reference frame. Analogously, as spatial velocity is obtained as time derivative of the motion $v^s_t\circ\vphi_t=v_t=\dot\vphi_t$, we have
\begin{equation}\label{rotequ}
v^s_t\circ\vphi_t=v_t=\dot\vphi_t=(\dot\vphi_t)^{\text{rot}}+\Omega\times\vphi_t=(v^s_t\circ\vphi_t)^{\text{rot}}+\Omega\times\vphi_t=(v^s_t)^{\text{rot}}\circ\vphi_t+\Omega\times\vphi_t.
\end{equation}
Consequently, as 
$$
(D_ta^s_t)\circ\vphi_t=\d_t(a^s_t\circ\vphi_t),
$$ 
the spatial acceleration in the inertial frame reads 
\begin{eqnarray}
\ddot\vphi_t  = \d_t(v^s_t\circ\vphi_t) &=&  (D_t v^s_t)\circ\vphi_t\nn\\
\text{\small [by \eqref{rotequ0}]} &=&\left(D^{\text{rot}}_tv^s_t+\Omega\times v^s_t\right)\circ\vphi_t\nn\\
&=&\left(D^{\text{rot}}_tv^s_t\right)\circ\vphi_t+\left(\Omega\times v^s_t\right)\circ\vphi_t\nn\\
&=& \left(\d_t(v^s_t\circ\vphi_t)\right)^{\text{rot}} + 
      \Omega\times (v^s_t\circ\vphi_t)\nn\\
\text{\small [by \eqref{rotequ}]} &=& \left(\d_t((v^s_t)^{\text{rot}}\circ\vphi_t+\Omega\times\vphi_t)\right)^{\text{rot}}
     +\Omega\times \left((v^s_t)^{\text{rot}}\circ\vphi_t+\Omega\times\vphi_t\right)\nn\\
&=&\left(D^{\text{rot}}_t(v^s_t)^{\text{rot}}+2\Omega\times (v^s_t)^{\text{rot}}\right)\circ\vphi_t
+\Omega\times(\Omega\times\vphi_t),\nn
\end{eqnarray}
that is, in rotating coordinates, $D_tv^s$ has to be replaced by
\begin{equation}
D_tv^s+2\Omega\times v^s_t+\Omega\times(\Omega\times x),
\end{equation}
where we now drop the label $(.)^{\text{rot}}$. Therefore, in place of \eqref{inertialeqm}
 we arrive at an equation of motion in the form
$$
  \rho^s (D_t v^s_t + 2\Omega\times v^s_t
   +\Omega\times(\Omega\times x)) = \nabla \cdot T^s - \rho^s\, \nabla \Phi^s.
$$
Upon recalling from \eqref{psinote} that $\nabla \Psi^s(x) = \Omega \times (\Omega \times x)$ this yields 
$$
    \rho^s (\dot v^s +  \nabla v^s \cdot v^s + 
     2 \Omega \times v^s) =  \nabla \cdot T^s - \rho^s\, \nabla (\Phi^s+\Psi^s),
$$
which is exactly the nonlinear spatial equation of motion as stated in
\cite[Equation (2.117)]{DaTr:98}.

We obtain the boundary conditions in spatial form, essentially, from
\cite[Theorem 3.2]{GBMR:12} or \cite[Equation (2.19c)]{AOS:11}, but
recalling that in our case of a merely piecewise Lipschitz continuous
motion pointwise conditions on $\dot\vphi$ with only bounded
measurable component functions would not be meaningful (e.g., a
statement like $\dot\vphi_t^s = 0$ on $\vphi_t(\d\earth)$ would not
make sense, in general). However, apart from Equation
\eqref{boundary_movement} mentioned already above, we obtain the
material counterpart of the traction boundary condition
\begin{equation}
   T_t^s\cdot\nu_t^s=0\qquad\textrm{on}\qquad \vphi_t(\d\earth),
\end{equation}
and, taking into account the lower regularity of $T^s$ in our case, we have in addition 
\begin{equation}
[T_t^s]_-^+\cdot \nu_t^s=0
\end{equation}
across all deformations $\vphi_t(S) \subset \vphi_t(B^\text{S})$ of welded solid-solid interior boundaries surfaces $S \subset B^\text{S}$. 

\subsubsection{Weak stationarity of the nonlinear surface action}
\label{weakstat}

We investigate the implications by the contribution of $\action''_{\Sig^{\fxs}}(\vphi,\Phi^s)$, defined by \eqref{eq:action_full} in the form $\action''_{\Sig^{\fxs}}(\vphi,\Phi^s) =   \int_I \action''_{\Sig^{\fxs}}(t)\, dt$ with
$$
  \action''_{\Sig^{\fxs}}(t) = \int_{\Sig^{\fxs}} L_{\Sig^\fxs}''(X,t)\:\dsurf(X),
$$ 
to the action functional at the fluid-solid boundaries. At fixed time $t$, we interpret $\action''_{\Sig^{\fxs}}(t)$ as the action of a distribution on $\RR^3$ with support on a two-dimensional surface, that is, the surface integral is considered as action $\lara{\action''_{\Sig^{\fxs}}(t),h}$ on a smooth compactly supported test function $h$ on $\RR^3$ satisfying $h \!\mid_{\Sig^{\fxs}} = 1$ (which bears some similarity with the concept of mass of a $2$-current as \cite[Subsection 7.2]{KP:08}). According to \eqref{eq:Lfull_surf}, the distributional action is given by integration over $\Sig^{\fxs}$ in the form
\begin{multline*}
   \lara{\action''_{\Sig^{\fxs}}(t),h} = 
   \int_{\Sig^{\fxs}} h(X)\, L_{\Sig^\fxs}''(X,t)\, dS(X)\\ =
  -  \int_{\Sig^{\fxs}} h(X) \left[\int_{t_0}^t\dot\vphi_\timd(X)\cdot T^{\PK}_\timd(X)\:\dint
  t'\right]_-^+\cdot \nu(X) \, dS(X)\\
  = - \int_{t_0}^t \underbrace{\int_{\Sig^{\fxs}} h(X)\,
  \left[\dot\vphi_\timd(X)\cdot T^{\PK}_\timd(X)\right]_-^+\cdot \nu(X) \, dS(X)}_{=: \Lara{\left[\dot\vphi_\timd \cdot T^{\PK}_\timd \cdot
        \nu \, dS \right]_-^+, h}}\, \dint t'
  =  - \int_{t_0}^t \Lara{\left[\dot\vphi_\timd \cdot T^{\PK}_\timd \cdot
        \nu \, dS \right]_-^+, h} \, \dint t',
\end{multline*}
where the distribution $\left[\dot\vphi_\timd \cdot T^{\PK}_\timd
  \cdot \nu \, dS \right]_-^+$ on $\RR^3$ with support in
$\Sig^{\fxs}$ coincides with $\left[\dot\vphi_\timd \cdot
  T^{\PK}_\timd\right]_-^+ \cdot \nu \, dS$ as a distributional
density on the surface (alternatively denoted by
$\left[\dot\vphi_\timd \cdot T^{\PK}_\timd\right]_-^+ \cdot \nu \,
\delta_{\Sig^{\fxs}}$, e.g., \cite[Appendix, \S 1, 4.4, pages
  487-488]{DL:V2}; see also \cite[Equation (3.1.5) and the comment
  about extension to Lipschitz surfaces in the second paragraph on
  page 61; moreover, Theorem 8.1.5 and Example
  8.2.5]{Hoermander:V1}). Our notation here is chosen to make results
better comparable with geophysics literature as in
\cite{DaTr:98}. Note that $\timd \mapsto \left[\dot\vphi_\timd \cdot
  T^{\PK}_\timd \cdot \nu \, dS \right]_-^+$ is a weakly measurable
and bounded map, hence weakly integrable over any bounded interval of
time. Therefore, $$t \mapsto \int_{t_0}^t \left[\dot\vphi_\timd \cdot
  T^{\PK}_\timd \cdot \nu \, dS \right]_-^+ \, \dint t'$$ is weakly
absolutely continuous, in particular, almost everywhere weakly
differentiable.

Applying the Piola transform \eqref{eq:TPK_Piola} yields 
$$
  \Lara{\action''_{\Sig^{\fxs}}(t), h}
  = - \int_{t_0}^t  \Lara{ \left[\dot\vphi_\timd \cdot T^{\PK}_\timd \cdot
        \nu \, dS \right]_-^+, h}\, \dint t'\\ =
  - \int_{t_0}^t  \Lara{ \left[\dot\vphi_\timd \cdot T_\timd J_\timd \nabla \vphi_\timd^{-T} \cdot \nu \, dS \right]_-^+, h} \, \dint t'
$$
and by  an  interpretation of Lemma \ref{lem:intvolsurf}(ii) in the sense of a distributional pull-back we obtain (with $h^s \circ \varphi_\timd = h$ and the spatial velocity $v^s = \dot\vphi^s$)
\begin{equation}\label{Asurfdis}
  \Lara{\action''_{\Sig^{\fxs}}(t), h}
  = - \int_{t_0}^t  \Lara{ \left[v_\timd^s \cdot T_\timd^s  \cdot \nu_\timd^s \, dS_\timd \right]_-^+, h^s}\, \dint t'.
\end{equation}
Here,  $\left[v_\timd^s \cdot T_\timd^s  \cdot \nu_\timd^s \, dS_\timd \right]_-^+$  is a spatial distributional density with support on $\Sig^{\fxs}_\timd$, which again should be understood in the sense $\left[v_\timd^s \cdot T_\timd^s  \right]_-^+ \cdot \nu_\timd^s \, dS_\timd$ or $\left[v_\timd^s \cdot T_\timd^s  \right]_-^+ \cdot \nu_\timd^s \, \delta_{\Sig^{\fxs}_\timd}$ (one may read the explanation following Equation (3.68) in \cite{DaTr:98} also in that way), and is weakly integrable  with respect to $\timd \in I$. Thus, we may consider \eqref{Asurfdis} as action with \emph{surface Lagrangian in spatial representation} at fixed $t$. As we will show in the lines following, the variation of $\action''_{\Sig^{\fxs}}$ with respect to the spatial velocity $v^s$ (still treating $\Sig^{\fxs}_\tim$ as independent variable in the spatial action according to \cite[Subsection 3.2, Equation (3.13)]{GBMR:12}) reproduces the classical fluid-solid boundary conditions \eqref{eq:CauchyIBC} in spatial representation in the form \eqref{eq:FSIBC_TPK}.

For an arbitrary test function $\tilde{h}$ on $I \times \RR^3$ and gravitational potential $\Phi^s$,  we have the surface action contribution as a map $v^s \mapsto \lara{\action''_{\Sig^{\fxs}}(\vphi,\Phi^s), \tilde{h}} = \int_I \lara{\action''_{\Sig^{\fxs}}(t), \tilde{h}(t,.)} \, dt$ whose stationarity with respect to $v^s$ may be investigated at a fixed but arbitrary test function and gravitational potential. In this sense, we obtain a \emph{weak stationarity condition} on $\action''_{\Sig^{\fxs}}$ playing the role of Euler-Lagrange equations. By density of tensor products, it suffices to consider test functions of the form $\tilde{h}= h_0 \otimes h \colon (t,X) \mapsto h_0(t) h(X)$ and we have to investigate the stationarity of the map
$$
    \lara{\action''_{\Sig^{\fxs}}, h_0 \otimes h} \colon 
    v^s \mapsto 
    \int_I h_0(t) \lara{\action''_{\Sig^{\fxs}}(t), h} \, dt \\ =
    - \int_I h_0(t) \int_{t_0}^t   \Lara{\left[v_\timd^s \cdot T_\timd^s  
        \cdot \nu_\timd^s \, dS_\timd \right]_-^+,h^s}\, \dint t' \, dt.
$$
As noted above, $t \mapsto \int_{t_0}^t   \Lara{\left[v_\timd^s \cdot T_\timd^s  
        \cdot \nu_\timd^s \, dS_\timd \right]_-^+,h^s}\, \dint t'$ is differentiable almost everywhere and has an integrable derivative. Moreover, $v^s \mapsto \Lara{\left[v_\timd^s \cdot T_\timd^s  \cdot \nu_\timd^s \, dS_\timd \right]_-^+,h^s}$ defines a family of linear continuous maps $L^\infty(I \times B)^3 \to \RR$ depending measureably on $\timd$. Therefore, $v^s \mapsto \lara{\action''_{\Sig^{\fxs}}(t), h} = - \int_{t_0}^t   \Lara{\left[v_\timd^s \cdot T_\timd^s  \cdot \nu_\timd^s \, dS_\timd \right]_-^+,h^s}\, \dint t'$ is Fr\'{e}chet differentiable with derivative acting  by 
$$
    \frac{\d \lara{\action''_{\Sig^{\fxs}}(t), h}}{\d v^s} \cdot w = 
    - \Lara{\left[ w_t \cdot T_t^s  \cdot \nu_t^s \, dS_t \right]_-^+, h^s}
$$
and the weak stationarity of $\action''_{\Sig^{\fxs}}$ means
$$
     \int_I h_0(t) \Lara{\left[ w_t \cdot T_t^s  \cdot \nu_t^s \, dS_t \right]_-^+, h^s} \, dt
     = 0
$$
for all test functions $h_0 \otimes h$ on $I\times \RR^3$ and for all $w \in L^\infty(I \times B)^3$. Thus, we arrive at the condition
\begin{equation}\label{eq:FSIBC_TPK} 
   [T_t^s\cdot \nu_t^s \, dS_t]_-^+ = 0
\end{equation}
to hold across the fluid-solid boundary $\Sig^{\fxs}_t$ (in agreement with \cite[Equation (3.68)]{DaTr:98}). Note that the above condition coincides with \eqref{eq:CauchyIBC} obtained by Newton's third law.


\section{Linearization}\label{sec:4}

In this section, we linearize the earth model around the reference
state at time $t_0$, which is assumed to be an equilibrium state. To
be more specific, the dynamical variables $\vphi$ and $\Phi^s$ are
decomposed into reference values (Section \ref{ssec:equilibrium}) and
perturbations (Sections \ref{ssec:displacement} and
\ref{ssec:lingravity}). Linearization is justified within the seismic
regime, where the elastic motion of the earth under self-gravitation
results only in a slight departure from equilibrium
(e.g.\ \cite[p.\ 56]{DaTr:98}).
 
Formally, the nonlinear governing equations \eqref{materialeqm} and
\eqref{eq:poisson} split into the equilibrium equations
(\ref{eq:staticeq}) and (\ref{eq:Phi0_Poisson}) constraining the
reference fields and the linear governing equations (\ref{eq:eqm}) and
(\ref{eq:Phi1_Poisson}) determining the evolution of the first-order
perturbations. In this splitting, one neglects quadratic or
higher-order contributions of the perturbations or of their
derivatives.  The boundary and interface conditions are treated in a
similar way. The linear equations and boundary/interface conditions
derived in this way will be said to hold ``correct to first order'' in
the perturbations. However, as will be established in Section
\ref{sec:5}, the governing equations as well as the dynamical
interface and boundary conditions are rigorously obtained from the
first variation of a corresponding action (as Euler-Lagrange equations
EL and associated natural boundary/interface conditions NBC/NIBC
respectively). Consequently, in the linearized setting, action,
Lagrangian, and energy require an approximation correct to second
order in the perturbations (Section \ref{sec:approxaction}).

The omission of terms of order $k+1$ in the perturbed quantities or in their derivatives will be indicated by $\approx_k$. In other words, 
\begin{equation}
a\approx_k b
\end{equation} 
means ``$a$ is equal to $b$ up to $k$-th order'' or ``correct to order $k$''.

\subsection{The prestressed equilibrium state}\label{ssec:equilibrium}

We assume that the earth is in equilibrium at initial time $t_0$, which we use as a reference configuration. Thus the equilibrium state is characterized by specifying earth's rotation $\Omega$, initial density $\rho^0$, whose support gives $\ovl\earth$, and prestress $T^0$. 

Since we use some properties of the equilibrium state when discussing perturbed quantities, which correspond to stationarity of the second-order terms of the approximated action, we announce the equilibrium results which follow from stationarity of the action in the first-order approximation as shown in Section \ref{sec:5}:
The Euler-Lagrange equations \eqref{eq:Ham_EL} consist of the static equilibrium equation (\ref{eq:staticeq}) 
$$\rho^0\nabla(\Phi^0+\Psi^s)-\nabla\cdot T^0=0$$
and the equilibrium Poisson's equation (\ref{eq:Phi0_Poisson})
$$\triangle\Phi^0=4\pi G\rho^0.$$
These equations are obtained from setting $t = t_0$ in the nonlinear
equations \eqref{materialeqm} and \eqref{eq:poisson}. The interior and
exterior boundary conditions are the continuity conditions
\eqref{eq:grav_IBC}, $[\Phi^0]_-^+=0$ and
$[\nabla\Phi^0]_{-}^{+}\cdot\nu=0$ on arbitrary surfaces $S\sbs\RR^3$,
for the initial gravitational potential $\Phi^0$, as well as the
continuity and zero-traction conditions for prestress $T^0$,
$[T^0]_-^+\cdot\nu=0$ \eqref{eq:IBC_FST0} across all surfaces
$S\sbs\earth$ (reducing to $[p^0]_-^+=0$ \eqref{eq:IBC_FSp0} across
fluid-solid interior boundaries) and $T^0\cdot \nu=0$
\eqref{eq:NBC_T0} at the free surface $\d\earth$.

The equilibrium earth model is thus set up as follows: First recall that $\Psi^s$ is fully determined by the angular velocity $\Omega$ from \eqref{eq:Psi}. Then, for given $\rho^0$, the solution of the equilibrium Poisson equation \eqref{eq:Phi0_Poisson} yields $\Phi^0$. Finally, the static equilibrium (\ref{eq:staticeq}) constrains $\nabla\cdot T^0$.

\begin{remark}[{{\bf Constraining the components of prestress}}] 
Note that $T^0$ is not fully determined by the governing equations.
Formally, given $\rho^0$ and $\Omega$, the three components of the static equilibrium equation \eqref{eq:staticeq} (plus the interface and boundary conditions) constrain only three out of the six independent components of $T^0$. The remaining three components, specifically those of the initial deviatoric prestress $T^0_{\mathrm{dev}}$ \eqref{eq:T0dev},  need to be treated as additional material parameters that have to be specified independently \cite[p.\ 100]{DaTr:98}. A new method of parametrization of the equilibrium stress is discussed in \cite{AlWo:10}. In particular, there it is shown that the equilibrium stress field with minimum deviatoric component in terms of a given norm corresponds to the solution of a steady-state incompressible viscous flow problem.
\end{remark}

We summarize the regularity properties  in the equilibrium earth model: 
The conditions for the time-dependent density, gravitational potential, and stress tensor found in Section \ref{sec:2} also hold at initial time $t_0$,  in particular. Consequently, by Definition \ref{def:nonlin_reg} of the general configuration spaces for a composite fluid-solid earth, our general assumption is 
\begin{equation} 
\rho^0\in L^\infty(\RR^3)
\end{equation} 
with compact support contained in $\ovl{\earth}$.
As we have seen in Lemma \ref{lem:poisson1}, Poisson's equation then implies that 
\begin{equation}
\Phi^0\in\cC^1(\RR^3).
\end{equation} 
Hence, the static equilibrium equation $\nabla\cdot T^0=\rho^0\nabla(\Phi^0+\Psi^s)$ (\ref{eq:staticeq}) is valid in $L^\infty(\RR^3)$. As $T^0$ is compactly supported on $\ovl{\earth}$, boundedness of $\nabla\cdot T^0$ implies $\nabla\cdot T^0\in L^2(\earthfsc)$
and we thus obtain 
\begin{equation}
T^0\in H_\text{div}(\earthfsc)^{3\times 3}.
\end{equation}
Then also the boundary conditions for $T^0$ make sense, since $H_\text{div}$ \eqref{eq:Hdiv} has traces in $H^{1/2}$.

\subsection{Displacement}\label{ssec:displacement}  

We write the motion $\vphi\in W_\motion$ (cf.\ Definition
\ref{def:nonlin_reg}) in the form of a perturbation of its equilibrium at time $t_0$,
\begin{equation}\label{eq:displacement}
\vphi (X,t) = X + u(X,t) \qquad (X \in \earth)
\end{equation}
with
\begin{equation}
\vphi_\timO (X) = X.
\end{equation}
The displacement $u$ depends on space and time and, by definition, vanishes at time $t_0$,
\begin{equation}
u_\timO=0.
\end{equation}

Note that $u$ and $\vphi$ have the same differentiability properties on $\earth$ and we are free to set $u$ equal to zero outside the earth. Since $\vphi|_{\ovl\earthi\times I}\in\Am(\ovl\earth\times I)$ we have $u\in\Lip(I,L^{\infty}(\RR^3))^3$ with
\begin{equation}
\supp\:u_\tim \sbs \ovl\earth
\end{equation}
for all $t\in I$, and $u|_{\earthi^\solid\times I}$, $u|_{\earthi^\fluid\times I}$, and $u|_{\earthi^{\cpl}\times I}(=0)$
have Lipschitz continuous extensions to $\ovl{\earth^\solid}\times I$, $\ovl{\earth^\fluid}\times I$, and
$\ovl{\earth^\cpl}\times I$ respectively. Hence, in accordance with (\ref{eq:SobolevRegpLip})
\begin{equation}\label{eq:u_H1}
u|_{\earthfsci\times I^\circ}\in H^1(\earthfsc\times I^\circ)^3.
\end{equation}

Continuity of $\vphi$ (\ref{eq:IBC_SS}) across welded solid-solid interior boundaries directly implies continuity of $u$:
\begin{equation}\label{eq:weldedfirstorder}
[u]_-^+=0\qquad \text{on}\qquad \Sig^{\sxs}.
\end{equation}
We recall that the slipping kinematical interface condition across
fluid-solid interior boundaries, (\ref{eq:IBC_FS}), expresses the
continuity of the normal component of the spatial velocity
$v^s_\tim\cdot \nu_\tim^s$ across $\vphi_\tim(\Sig^\fxs)$. However,
its material counterpart is nonlinear in $u$.  Correct to first order
in $u$ we have the {\bf linear tangential slip condition}
\begin{equation}\label{eq:slipfirstorder}
[u]_-^+\cdot \nu\approx_1 0\qquad \text{on}\qquad \Sig^\fxs
\end{equation}
whose structure is similar to the exact spatial slip
condition. On the linear level,
\begin{equation}\label{eq:H1Sig}
u\in H^1_{\Sig^\fxs}(\earthfsc\times I^\circ)^3:=\{u\in H^1(\earthfsc\times I^\circ)^3:\:[u]_-^+\cdot\nu=0\:\text{ on }\Sig^\fxs\}.
\end{equation}
This is the final space for the variational formulation.

The {\bf second-order tangential slip condition} reads  \cite[(3.95)]{DaTr:98}
\begin{equation}\label{eq:slipsecondorder}
\left[u\cdot \nu -u\cdot\snabla (u\cdot \nu)+\frac{1}{2}\:u\cdot\snabla\nu\cdot u\right]_-^+\approx_20
\qquad \text{on}\qquad \Sig^\fxs.
\end{equation} 
As $u\cdot\snabla (u\cdot \nu)=u\cdot\snabla\nu\cdot u+\nu\cdot\snabla
u\cdot u$ the condition may equivalently be written as
\begin{equation}\label{eq:sslipsecondorder}
\left[u\cdot \nu -\nu\cdot\snabla u\cdot u-\frac{1}{2}\:u\cdot\snabla\nu\cdot u\right]_-^+\approx_20.
\end{equation}

We give a proof of \eqref{eq:slipsecondorder}  along the lines of \cite[p.\ 72]{DaTr:98}:
In the first step we establish the separate first-order expansions  \cite[(3.31) and (3.30)]{DaTr:98} of the spatial surface element $\nu^s\:\dsurf^s$,
\begin{equation}\label{eq:spatialsurfapprox}
\nu^s\approx_1 \nu-(\snabla u)^T\cdot \nu=\nu-\nu\cdot \snabla u
\qquad\text{and}\qquad
\dsurf^s\approx_1(1+\snabla\cdot u)\:\dsurf.
\end{equation}
They are deduced from the relation for material and spatial surface elements (Lemma \ref{lem:intvolsurf}), the first-order expansions  
\begin{equation}\label{eq:Jfristorder}
J=\det(\nabla\vphi)=\det(1_{3\times 3}+\nabla u)\approx_1 1+\nabla \cdot u\quad\textrm{and}\quad
(\nabla\vphi)^{-T}\approx_1 1_{3\times 3}-(\nabla u)^{-T},
\end{equation} 
and the identities
$\nabla\cdot u=\snabla\cdot u+\nu\cdot\nabla u\cdot\nu$ and $\nabla u=\snabla u+(\nabla u\cdot\nu)\nu$:
\begin{eqnarray}
\nu^s\:\dsurf^s=J\:(\nabla\vphi)^{-T}\cdot \nu\:\dsurf
&\approx_1&(1+\nabla\cdot u)(1_{3\times 3}-(\nabla u)^{T})\cdot \nu\:\dsurf\nn\\
&\approx_1&(\nu+(\nabla\cdot u)\nu -(\nabla u)^{T}\cdot \nu)\dsurf\nn\\
&=&(\nu+(\nabla\cdot u)\nu -\nu\cdot\nabla u)\dsurf\nn\\
&=&(\nu+(\snabla\cdot u)\nu -\nu\cdot\snabla u)\dsurf\nn\\
&=&(\nu+(\snabla\cdot u)\nu -(\snabla u)^{T}\cdot \nu)\dsurf\nn\\
&\approx_1&(\nu-(\snabla u)^{T}\cdot \nu)(1+\snabla\cdot u)\dsurf.\nn
\end{eqnarray}
In the second step we consider a slipping surface $\Sig^\fxs$. The situation is illustrated in  \cite[Figure 3.3, p.\ 68]{DaTr:98}. We observe that for all
$x\in\vphi_t(\Sig^\fxs)$ there exist $X',X''\in\Sig^\fxs$ such that 
\begin{equation}\label{eq:surfpoints}
   X'+u_t^+(X') = \vphi_t^+(X') = x = \vphi_t^-(X'') = X''+u_t^-(X''),
\end{equation}
where $\vphi_t^+(X')$ denotes the limit from the $+$-side and
$\vphi_t^-(X'')$ the limit from the $-$-side of $\vphi_t(\Sig^\fxs)$,
analogoulsy for $u_t^\pm$. (The existence of these limits is
guaranteed by the piecewise $\Lip$-regularity of the motion, see
\ref{def:Am_ii} of Definition \ref{def:Am}.)  Omitting the dependence
on $t$ for the moment, if $x=\vphi(X)$ we have $v^s(x)=
\dot{u}(X)$. Together with \eqref{eq:spatialsurfapprox}, this implies
\begin{equation}
(v^s\cdot\nu^s)(x)\approx_2\dot u(X)\cdot(\nu-(\snabla u)^T\cdot \nu)(X)
=(\dot u\cdot\nu-\dot u\cdot(\snabla u)^T\cdot \nu)(X).
\end{equation}
From \eqref{eq:surfpoints} and the exact spatial slip condition $[v^s]_-^+\cdot\nu^s=[v^s\cdot\nu^s]_-^+=0$ we then get
$$
(\dot u^+\cdot\nu-\dot u^+\cdot(\snabla u^+)^T\cdot \nu)(X')
\approx_2(v^s\cdot\nu^s)^+(x)=(v^s\cdot\nu^s)^-(x)
\approx_2(\dot u^-\cdot\nu-\dot u^-\cdot(\snabla u^-)^T\cdot \nu)(X''),
$$
that is,
\begin{equation}\label{eq:tangslipeq}
(\dot u^+\cdot\nu)(X')-(\dot u^-\cdot\nu)(X'')
\approx_2(\dot u^+\cdot(\snabla u^+)^T\cdot \nu)(X')-(\dot u^-\cdot (\snabla u^-)^T\cdot \nu)(X'').
\end{equation}
In particular, up to first order we obtain
$$
(\dot u^+\cdot\nu)(X')-(\dot u^-\cdot\nu)(X'')  \approx_1 0,
$$ 
which upon inserting $X''=X'+u^+(X')-u^-(X'')$ from \eqref{eq:surfpoints} implies
\begin{equation}\label{eq:slipudot}
   \left[ \dot u \cdot\nu\right]^+_- (X') = (\dot u^+\cdot\nu)(X')-(\dot u^-\cdot\nu)(X') 
       \approx_1  (\dot u^+\cdot\nu)(X')-(\dot u^-\cdot\nu)(X'') \approx_1 0.
\end{equation}
We observe that integrating with respect to time yields the
first-order tangential slip condition \eqref{eq:slipfirstorder}.

To obtain the second-order condition, we use again $X''=X'+u^+(X')-u^-(X'')$ and Taylor expansion applied to the second term of the right-hand side in \eqref{eq:tangslipeq}:
\begin{eqnarray}
&&(\dot u^+\cdot(\snabla u^+)^T\cdot \nu)(X')-(\dot u^-\cdot(\snabla u^-)^T\cdot \nu)(X'')\nn\\
&&\quad =(\dot u^+\cdot(\snabla u^+)^T\cdot \nu)(X')-(\dot u^-\cdot(\snabla u^-)^T\cdot \nu)(X'+u^+(X')-u^-(X''))\nn\\
&&\quad \approx_2(\dot u^+\cdot(\snabla u^+)^T\cdot \nu)(X')-(\dot u^-\cdot(\snabla u^-)^T\cdot \nu)(X')
=\left[\dot u\cdot(\snabla u)^T\cdot \nu\right]_-^+(X').\nn
\end{eqnarray}
For the second-order approximation of the left-hand side of \eqref{eq:tangslipeq}, also the first-order terms of the Taylor expansion have to be taken into account:
\begin{eqnarray}
&&(\dot u^+\cdot\nu)(X')-(\dot u^-\cdot\nu)(X'')\nn\\
&&\quad =(\dot u^+\cdot\nu)(X')-(\dot u^-\cdot\nu)(X'+u^+(X')-u^-(X''))\nn\\
&&\quad \approx_2 (\dot u^+\cdot\nu)(X')-(\dot u^-\cdot\nu)(X') - \nabla(\dot u^-\cdot\nu)(X')\cdot\left(u^+(X')-u^-(X'')\right) \nn\\
&&\quad = \left[\dot u\cdot\nu\right]_-^+(X') - \nabla(\dot u^-\cdot\nu)(X')\cdot u^+(X')+\nabla(\dot u^-\cdot\nu)(X')\cdot u^-(X''). \nn
\end{eqnarray}
We rewrite the second term and observe that $\nabla[\dot u\cdot \nu]_-^+(X')\cdot u^+(X') \approx_2 0$ holds, since $[\dot u\cdot \nu]_-^+\approx_1 0$ by \eqref{eq:slipudot}:
\begin{eqnarray}
&&-\nabla(\dot u^-\cdot\nu)(X')\cdot u^+(X')\nn\\
&&\quad =-\nabla(\dot u^-\cdot\nu)(X')\cdot u^+(X')
+\nabla(\dot u^+\cdot\nu)(X')\cdot u^+(X')-\nabla(\dot u^+\cdot\nu)(X')\cdot u^+(X')\nn\\
&&\quad =\nabla[\dot u\cdot \nu]_-^+(X')\cdot u^+(X')-\nabla(\dot u^+\cdot\nu)(X')\cdot u^+(X')\nn\\
&&\quad \approx_2-\nabla(\dot u^+\cdot\nu)(X')\cdot u^+(X'). \nn
\end{eqnarray}
The third term, again upon Taylor expansion, may be written in the form 
$$
\nabla(\dot u^-\cdot\nu)(X')\cdot u^-(X'')\approx_2 \nabla(\dot u^-\cdot\nu)(X')\cdot u^-(X').
$$
The left-hand side of \eqref{eq:tangslipeq} can thus be approximated as follows 
$$
  (\dot u^+\cdot\nu)(X')-(\dot u^-\cdot\nu)(X'') 
  \approx_2 \left[\dot u\cdot\nu\right]_-^+(X') 
       -\nabla(\dot u^+\cdot\nu)(X')\cdot u^+(X')
       +\nabla(\dot u^-\cdot\nu)(X')\cdot u^-(X'),
$$
that is
$$
(\dot u^+\cdot\nu)(X')-(\dot u^-\cdot\nu)(X'')
\approx_2\left[\dot u\cdot\nu-u\cdot\nabla(\dot u\cdot\nu)\right]_-^+(X')
=\left[\dot u\cdot\nu-u\cdot\snabla(\dot u\cdot\nu)\right]_-^+(X').
$$
Combining these results, \eqref{eq:tangslipeq} implies the identity \cite[(3.93)]{DaTr:98}
\begin{equation}
\left[\dot u\cdot\nu-u\cdot\snabla(\dot u\cdot\nu)\right]_-^+
\approx_2\left[\dot u\cdot(\snabla u)^T\cdot \nu\right]_-^+.
\end{equation}
As 
$$\dot u\cdot(\snabla u)^T\cdot \nu=\nu\cdot\snabla u\cdot\dot u
=\snabla(u\cdot\nu)\cdot\dot u-u\cdot\snabla\nu\cdot \dot u
=\dot u\cdot\snabla(u\cdot\nu)-u\cdot\snabla\nu\cdot \dot u$$ 
and due to the symmetry of $\snabla\nu$, we can modify the right-hand side and pull out the time derivative:
$$0\approx_2\left[\dot u\cdot\nu-u\cdot\snabla(\dot u\cdot\nu)-\dot u\cdot\snabla(u\cdot\nu)
+u\cdot\snabla\nu\cdot \dot u\right]_-^+
=\d_t\left[u\cdot\nu-u\cdot\snabla(u\cdot\nu)+\frac{1}{2}\:u\cdot\snabla\nu\cdot u\right]_-^+.$$
Since $u_{t_0}=0$, integration 
from $t_0$ to $t$ finally leads to the tangential slip condition (\ref{eq:slipsecondorder}).

\subsection{Linearized fields}\label{ssec:linfields}

\subsubsection{Gravitational potential perturbation}\label{ssec:lingravity}

The linearization of the (spatial) gravitational potential $\Phi^s$ takes the form
\begin{equation}
\Phi^s=\Phi^0+\Phi^1
\end{equation}
with the equilibrium field given by
\begin{equation}
\Phi^s_\timO=\Phi^0.
\end{equation}
The gravitational potential perturbation (incremental gravitational potential, mass-redistribution
potential) $\Phi^1$ depends on space and time and, by definition, vanishes in the equilibrium state at time $t_0$, that is
\begin{equation}
\Phi^1_\timO=0.
\end{equation}
Since the equilibrium gravitational potential $\Phi^0$ is time-independent,  $\Phi^s$ and $\Phi^1=\Phi^s-\Phi^0$
have the same regularity with respect to time. Moreover, $x=\vphi(X,t)=X+u(X,t)$ and the mean value theorem give
\begin{multline*}
\Phi^1(x,t) = \Phi^s(x,t) - \Phi^0(x) = \Phi^s(X+u(X,t),t)-\Phi^0(X+u(X,t))\\
=\Phi^s(X,t) - \Phi^0(X) + \int_0^1 \left( \nabla\Phi^s(X + r u(X,t),t) 
    - \nabla\Phi^0(X + r u(X,t))\right) \cdot u(X,t) \, \dint r,
\end{multline*}
which shows that  the spatial regularity of $\Phi^1$ will be one order lower compared to that of $\Phi^s$ or $\Phi^0$. Recalling the decay and regularity
results in Lemma \ref{lem:poisson1} (ii) and (iii), we have $\Phi^0\in Y^\infty(\RR^3)\sbs\cC^1(\RR^3)$ with $\Phi^0(x)=\mathcal O(1/|x|)$ as $|x|\to\infty$ and
$\Phi^1\in\cC^0(I,D Y^\infty(\RR^3))\sbs\cC^0(I,\bigcap_{1\leq p<\infty} W_\loc^{1,p}(\RR^3))$
with $\Phi^1_\tim(x)=\mathcal O(1/|x|^2)$ as $|x|\to\infty$, where
$D Y^\infty(\RR^3):=\{f\in\cD'(\RR^3):\:\exists\:g\in Y^\infty(\RR^3),\:\exists\:\al\in\NN_0^3,\:|\al|=1:f= D^\al g\}$.
It follows that
\begin{equation}\label{eq:Phi1_H1}
\Phi^1\in\cC^0(I,H^1(\RR^3)).
\end{equation}
By the Sobolev embedding theorem we have $W_\loc^{1,p}(\RR^3)\sbs\cC^k(\RR^3)$ if $0\leq k<1-3/p$. Since this is satisfied
for $k=0$ and $p>3$ we conclude that $\Phi^1_\tim$ is continuous, that is $\Phi^1\in\cC^0(\RR^3 \times I)$.
Continuity of $\Phi^1_\tim$ and continuous differentiability of $\Phi^0$ imply the conditions
\begin{equation}\label{eq:grav_IBC}
[\Phi^0]_-^+=0,\quad [\nabla\Phi^0]_-^+\cdot\nu=0,\quad\text{and}\quad [\Phi^1]_-^+=0
\end{equation}
to hold on every surface $S$ in $\RR^3$ and in particular on earth's interior and exterior boundaries $\Sig\cup\d\earth$.

\subsubsection{The external force potential in the linearized setting}\label{ssec:linsources}

The force potential $F^s$ may be considered as a perturbative term in the Lagrangian (\ref{eq:Lfull}), similarly as $u$ is viewed as a perturbation of the spatial position in $x = X + u_t(X)$. We observe that $F_\tim=F^s_\tim\circ \vphi_\tim$  implies the Taylor approximation
\begin{equation}\label{eq:forcelin}
F\approx_2 F^s+u\cdot\nabla F^s,
\end{equation} 
where the first term corresponds to first-order approximation, while the second term, being a product of two first-order perturbations, corresponds to second-order approximation.

\subsection{Linear elasticity}\label{ssec:linelastic}

\subsubsection{Linear elasticity including prestress}

In the presence of large ambient stresses encoded in the prestress
$T^0$ (\ref{eq:prestress}), it is customary to use the following
second-order approximation of the internal elastic energy density
\cite[p.\ 76 eq.\ (3.115)]{DaTr:98} 
\begin{equation}
\rho^0 U\approx_2 T^0:e+\frac{1}{2}e:\Xi:e= T^0_{ij}e_{ij}+\frac{1}{2}\:\Xi_{ijkl}e_{ij}e_{kl}.
\end{equation}
Here, $e$ is the full (nonlinear) material strain tensor which was introduced in (\ref{eq:nlstraintensor}). Since $e$ is symmetric, the
time-independent tensor fields $T^0$ and $\Xi$ without loss of generality possess the symmetries
\begin{equation}\label{eq:symmetries}
T^0_{ij}=T^0_{ji}\qquad\text{and}\qquad \Xi_{ijkl}=\Xi_{klij}=\Xi_{jikl}=\Xi_{ijlk}.
\end{equation}

As a consequence of the regularity (\ref{eq:basic_U_reg}) of $U$,
implying $U\in\cC^0(I,L^{\infty}(\earth))$, the components of $T^0$
and $\Xi$ need to be bounded, that is $T^0_{ij}$ and $\Xi_{ijkl}\in
L^{\infty}(\earth)$ ($i,j,k,l=1,2,3$).  The associated stress-strain 
relation is given by (\ref{eq:PKconstitutive}), which, upon recalling $\nabla\vphi = 1_{3\times 3} + \nabla u$ and observing that $\frac{\d \widetilde{U}}{\d (\nabla u)} = \frac{\d U}{\d (\nabla\vphi)}$ in case $\widetilde{U}(X,\nabla u) := U(X, 1_{3\times 3} + \nabla u)$, we may (by common abuse of notation) write in the form 
\begin{equation}\label{eq:PKlinconstitutive}
T^{\PK}=\rho^0\frac{\d U}{\d(\nabla\vphi)}=\rho^0\frac{\d U}{\d(\nabla u)}\approx_1 T^0+T^{\PK 1}.
\end{equation}
We express the perturbation $T^{\PK 1}$ of the first Piola-Kirchhoff stress tensor in terms of $u$.
For the material strain tensor (\ref{eq:nlstraintensor}) we have
\begin{equation}
e=\frac{1}{2}(\nabla\vphi^T\cdot\nabla\vphi-1_{3\times 3})=\frac{1}{2}(\nabla u+\nabla u^T+\nabla u^T\cdot\nabla u)
\end{equation}
or in components
$e_{ij}=\frac{1}{2}(\d_iu_j+\d_ju_i)+\frac{1}{2}(\d_iu_k)(\d_ju_k)$. Hence, the second-order expansion
of the elastic energy density reads
\begin{eqnarray} \label{eq:U_u}
\rho^0U&\approx_2&T^0_{ij}e_{ij}+\frac{1}{2}\:\Xi_{ijkl}e_{ij}e_{kl}\nn\\
&\approx_2&T^0_{ij}\d_ju_i+\frac{1}{2}\left(T^0_{ij}(\d_iu_k)(\d_ju_k)+\Xi_{ijkl}(\d_ju_i)(\d_lu_k)\right)\nn\\
&=&T^0_{ij}\d_ju_i+\frac{1}{2}\left(T^0_{jl}(\d_ju_k)(\d_lu_k)+\Xi_{ijkl}(\d_ju_i)(\d_lu_k)\right)\nn\\
&=&T^0_{ij}\d_ju_i+\frac{1}{2}\left(\de_{ik}T^0_{jl}+\Xi_{ijkl}\right)(\d_ju_i)(\d_lu_k).
\end{eqnarray}
We introduce the {\bf prestressed elasticity tensor} $\La^{T^0}$ as
\begin{equation}\label{eq:Lambda}
\La^{T^0}_{ijkl}:=\de_{ik}T^0_{jl}+\Xi_{ijkl}
\end{equation}
for $i,j,k,l\in\{1,2,3\}$. It follows that $\La^{T^0}_{ijkl}\in L^{\infty}(\RR^3)$, is supported by $\ovl\earth$,
and satisfies the symmetry relations
\begin{equation}
\La^{T^0}_{ijkl}=\La^{T^0}_{klij}.
\end{equation}
Thus we found
\begin{equation}\label{eq:elastenergyapprox}
\rho^0U\approx_2 T^0:\nabla u+\frac{1}{2}\nabla u:\La^{T^0}:\nabla u
\end{equation}
as second-order approximation of the internal energy density in terms of the displacement gradient. Consequently the linearization of the first Piola-Kirchhoff stress tensor is $T^{\PK}\approx_1 T^0+T^{\PK 1}$ with
\begin{equation}\label{eq:Hooke}
T^{\PK 1}:=\La^{T^0}:\nabla u,
\end{equation}
or in components 
$$T^{\PK 1}_{ij}:=\La^{T^0}_{ijkl}\d_lu_k$$ 
for
$i,j=1,2,3$ \cite[(3.120)\footnote{Note that \cite{DaTr:98} employ transposed conventions for $\nabla u$ and $T^{\PK}$. Thus, instead of $\La^{T^0}_{ijkl}=\de_{ik}T^0_{jl}+\Xi_{ijkl}$ \eqref{eq:Lambda} they  replace $i\leftrightarrow j$, $k\leftrightarrow l$ and write $\La^{T^0}_{ijkl}=\de_{jl}T^0_{ik}+\Xi_{jilk}=\de_{jl}T^0_{ik}+\Xi_{ijkl}$ \cite[(3.122)]{DaTr:98}.}]{DaTr:98}.  The linear stress-strain relation
(\ref{eq:Hooke}) generalizes {\bf Hooke's law} to a prestressed
elastic material. 
In fact, the elastic tensor $\Xi$ may also be written in the more general form \cite[(3.135)]{DaTr:98} 
\begin{equation}
\Xi^{T^0}_{ijkl}
:=\Ga_{ijkl}+a(T^0_{ij}\de_{kl}+T^0_{kl}\de_{ij})
+b(T^0_{ik}\de_{jl}+T^0_{jk}\de_{il}+T^0_{il}\de_{jk}+T^0_{jl}\de_{ik}),
\end{equation}
with parameters $a,b\in\RR$ and $\Ga$ satisfying the classical symmetries $\Ga_{ijkl}=\Ga_{klij}=\Ga_{jikl}=\Ga_{ijlk}$.
Upon replacing $i\leftrightarrow j$, $k \leftrightarrow l$, $(i,j)\leftrightarrow (k,l)$ and by employing the symmetry of $T^0$, one easily verifies that $\Xi^{T^0}$ also possesses the desired symmetries \eqref{eq:symmetries}. 
 
In the absence of prestress, $\Ga$ is equal to the classical elasticity tensor (typically denoted by $c$).
In a solid medium, $\Ga$ takes the form \cite[(3.147)]{DaTr:98}
\begin{equation}
\Ga_{ijkl}=(\kappa-\frac{_2}{^3}\mu)\de_{ij}\de_{kl}+\mu(\de_{ik}\de_{jl}+\de_{il}\de_{jk})+\Ga^{\mathrm{a}}_{ijkl},
\end{equation}
where $\kappa$ denotes the bulk modulus (incompressibility), $\mu$ the shear modulus, and $\Ga^\mathrm{a}$ contains the purely anisotropic part. Thus, the generalized prestressed elasticity tensor \eqref{eq:Lambda} reads
\begin{eqnarray}
\La^{T^0}_{ijkl}&=&T^0_{jl}\de_{ik}+\Xi^{T^0}_{ijkl}\\
&=&\Ga_{ijkl}+a(T^0_{ij}\de_{kl}+T^0_{kl}\de_{ij})
+(1+b)T^0_{jl}\de_{ik}+b(T^0_{ik}\de_{jl}+T^0_{jk}\de_{il}+T^0_{il}\de_{jk}).\nn
\end{eqnarray}

\subsubsection{First perturbation of the Lagrangian stress tensor}

The incremental Lagrangian stress tensor is given by
\begin{equation}
   T^{1} \approx_1 T-T^0.
\end{equation} 
Its relation to $T^{\PK1}$ is obtained based on the Piola transform \eqref{eq:TPK_Piola} and the expansions \eqref{eq:Jfristorder}:
\begin{eqnarray}
T^0+T^{\PK1}\approx_1 T^{\PK}=J T\cdot(\nabla\vphi)^{-T}
&\approx_1&(1+\nabla\cdot u)(T^0+T^{1})\cdot(1_{3\times 3}-(\nabla u)^{-T})\nn\\
&\approx_1& T^0+T^0(\nabla\cdot u)-T^0\cdot(\nabla u)^T+T^{1}.
\end{eqnarray}
Thus we established \cite[(3.36)]{DaTr:98}
\begin{equation}\label{eq:TPK1T1}
T^{1}\approx_1 T^{\PK 1}+T^0\cdot(\nabla u)^T-T^0(\nabla\cdot u)=\Upsilon^{T^0}:\nabla u,
\end{equation}
or in coordinates,
$$
T^{1}_{ij}\approx_1 T^{\PK 1}_{ij}+T^0_{il}\d_lu_j-T^0_{ij}\d_ku_k=\Upsilon^{T^0}_{ijkl}\d_lu_k,
$$
where (cf.\ \cite[(3.123) with $i\leftrightarrow j$ and $k\leftrightarrow l$]{DaTr:98})
\begin{eqnarray}
\Upsilon^{T^0}_{ijkl}&=&\La^{T^0}_{ijkl}+T^0_{il}\de_{kj}-T^0_{ij}\de_{kl}
\label{eq:Upsilon}\\
&=&\Xi^{T^0}_{ijkl}+T^0_{jl}\de_{ik}+T^0_{il}\de_{kj}-T^0_{ij}\de_{kl}\nn\\
&=&\Ga_{ijkl}+(a-1)T^0_{ij}\de_{kl}+aT^0_{kl}\de_{ij}
+(b+1)(T^0_{il}\de_{jk}+T^0_{jl}\de_{ik})+b(T^0_{ik}\de_{jl}+T^0_{jk}\de_{il}).\nn
\end{eqnarray}
Symmetry of $T^1$ corresponds to the symmetry relation 
\begin{equation}
\Upsilon^{T^0}_{ijkl}=\Upsilon^{T^0}_{jikl}.
\end{equation}
In accordance with \cite[p.\ 80]{DaTr:98} we will choose $a=-b=1/2$. The advantage of this convention is that it renders $T^{1}$ independent of $p^0=-\frac{1}{3}\tr T^0$. Indeed, if $a=-b=1/2$ then
\begin{equation}
\Upsilon^{T^0}_{ijkl}=\Ga_{ijkl}+\frac{1}{2}\left(-T^0_{ij}\de_{kl}+T^0_{kl}\de_{ij}
+T^0_{il}\de_{jk}+T^0_{jl}\de_{ik}-T^0_{ik}\de_{jl}-T^0_{jk}\de_{il}\right)
\end{equation}
and thus (now writing $=$ instead of $\approx_1$)
\begin{equation}
T^{1}_{ij}=\Ga_{ijkl}\d_l u_k+\frac{1}{2}\left(-T^0_{ij}\d_k u_k+T^0_{kl}\d_l u_k\de_{ij}
+T^0_{il}\d_l u_j+T^0_{jl}\d_l u_i-T^0_{ik}\d_j u_k-T^0_{jk}\d_i u_k\right).
\end{equation}
Applying the decompositions $T^0=-p^01_{3\times 3}+T^0_{\mathrm{dev}}$ \eqref{eq:T0dev} for initial stress and 
\begin{equation}\label{eq:linstrain}
\nabla u=\varepsilon+\omega
\qquad\text{with}\qquad
\varepsilon:=\frac{1}{2}(\nabla u+\nabla u^T)=\varepsilon^T
\quad\text{and}\quad
\omega:=\frac{1}{2}(\nabla u-\nabla u^T)=-\omega^T
\end{equation}
for strain, as well as by symmetry of $\Ga$ and $T^0$, we deduce
\begin{align}
T^{1}&=\Ga:\nabla u + \frac{1}{2} (
\underbrace{-T^0(\nabla\cdot u)+(T^0:\nabla u)1_{3\times 3}}_{
=-T^0 (\tr\:\varepsilon)+(T^0:\varepsilon)1_{3\times 3}}
+\underbrace{T^0\cdot(\nabla u)^T+\nabla u\cdot T^0-T^0\cdot\nabla u-(\nabla u)^T\cdot T^0}_{
=T^0\cdot((\nabla u)^T-\nabla u)+(\nabla u-(\nabla u)^T)\cdot T^0
=2(-T^0\cdot\omega+\omega\cdot T^0)})\nn\\
&=\Ga:\varepsilon +\frac{1}{2}\left(-T^0(\tr\:\varepsilon)+(T^0:\varepsilon)1_{3\times 3}\right)
-T^0\cdot\omega+\omega\cdot T^0\nn\\
&=\Ga:\varepsilon +\frac{1}{2}\left(p^0 (\tr\:\varepsilon)1_{3\times 3}-T^0_{\mathrm{dev}}(\tr\:\varepsilon)
-p^0(\tr\:\varepsilon)1_{3\times 3}+(T^0_{\mathrm{dev}}:\varepsilon)1_{3\times 3}\right)
-T^0_{\mathrm{dev}}\cdot\omega+\omega\cdot T^0_{\mathrm{dev}}\nn\\
&=\Ga:\varepsilon + \frac{1}{2}\left(-T^0_{\mathrm{dev}}(\tr\:\varepsilon)
+(T^0_{\mathrm{dev}}:\varepsilon)1_{3\times 3}\right)
-T^0_{\mathrm{dev}}\cdot\omega+\omega\cdot T^0_{\mathrm{dev}},
\label{eq:TL1}
\end{align}
which does not involve $p^0$. 
In a hydrostatic earth model where the deviatoric prestress $T^0_{\mathrm{dev}}$ vanishes, $T^{1}=\Ga:\varepsilon$ and thus $T^{1}$ is even independent of $T^0$.
Finally, we note that the choice $a=-b=1/2$ yields
\begin{equation}
\La^{T^0}_{ijkl}=\Ga_{ijkl}+\frac{1}{2}\left(T^0_{ij}\de_{kl}+T^0_{kl}\de_{ij}+T^0_{jl}\de_{ik}
-T^0_{ik}\de_{jl}-T^0_{jk}\de_{il}-T^0_{il}\de_{jk}\right)
\end{equation}
and
\begin{equation}
\Xi^{T^0}_{ijkl}=\Ga_{ijkl}+\frac{1}{2}\left(T^0_{ij}\de_{kl}+T^0_{kl}\de_{ij}
  - T^0_{jl}\de_{ik} -T^0_{ik}\de_{jl}-T^0_{jk}\de_{il}-T^0_{il}\de_{jk}\right).
\end{equation}

\subsubsection{Stress tensors in fluid regions}
In a perfect elastic fluid, shear resistance and anisotropy vanish and $\Ga$ reduces to
\begin{equation}
\Ga_{ijkl}=\kappa\de_{ij}\de_{kl}.
\end{equation}
In fluid regions of the earth we also have that   $T^0=-p^01_{3\times 3}$, that is, $T^0_{\mathrm{dev}}=0$. Thus \eqref{eq:TL1} leads to the following simple form of the first perturbation of the Lagrangian stress tensor in elastic prestressed fluids:
$$
T^{1}_{ij}=\kappa(\d_k u_k)\de_{ij}
$$
or
\begin{equation}\label{eq:TL1fluid}
T^{1}=\kappa(\nabla\cdot u)1_{3\times 3}.
\end{equation}
Furthermore, we have
\begin{eqnarray}
\La^{T^0}_{ijkl}&=&\kappa\de_{ij}\de_{kl}-\frac{1}{2}p^0(\de_{ij}\de_{kl}+\de_{kl}\de_{ij}+\de_{jl}\de_{ik}
-\de_{ik}\de_{jl}-\de_{jk}\de_{il}-\de_{il}\de_{jk})\nn\\
&=&\kappa\de_{ij}\de_{kl}-p^0(\de_{ij}\de_{kl}-\de_{jk}\de_{il})\nn\\
&=&p^0(\ga-1)\de_{ij}\de_{kl}+p^0\de_{jk}\de_{il}
\end{eqnarray}
where 
$\ga:=\kappa/p^0$ 
is the fluid's adiabatic index. Hence, the first perturbation of the first Piola-Kirchhoff stress tensor in elastic prestressed fluids is given by
$$
T^{\PK 1}_{ij}=\La^{T^0}_{ijkl}\d_lu_k=p^0(\ga-1)(\d_ku_k)\de_{ij}+p^0\d_iu_j
$$ 
or 
\begin{equation}\label{eq:TPK1fluid}
T^{\PK 1}=p^0(\ga-1)(\nabla\cdot u)1_{3\times 3}+p^0(\nabla u)^T.
\end{equation}
Similarly, 
\begin{eqnarray}
\Xi^{T^0}_{ijkl}&=&\kappa\de_{ij}\de_{kl}-\frac{1}{2}p^0(\de_{ij}\de_{kl}+\de_{kl}\de_{ij}
   - \de_{jl}\de_{ik}
-\de_{ik}\de_{jl}-\de_{jk}\de_{il}-\de_{il}\de_{jk})\nn\\
&=&\kappa\de_{ij}\de_{kl}-p^0(\de_{ij}\de_{kl}-\de_{jk}\de_{il} - \de_{jl}\de_{ik})\nn\\
&=&p^0(\ga-1)\de_{ij}\de_{kl}+p^0(\de_{jk}\de_{il} + \de_{jl}\de_{ik}).
\end{eqnarray}
Inserting $\Xi$ and $T^0$  in Equation \eqref{eq:U_u} yields the following second-order approximation for the elastic energy density in fluid regions of the earth (compare with the discussion of the nonlinear fluid case in \ref{ssec:elasticity} below Equation \eqref{eq:PKconstitutive})
\begin{equation}
  \rho^0 U \approx_2 - p^0 (\nabla \cdot u) + 
   \frac{p^0}{2} \left( (\gamma -1) (\nabla \cdot u)^2 + 
     \nabla u : (\nabla u)^T \right).
\end{equation}
The latter is consistent with the above equation for $T^{\PK 1}$ and \eqref{eq:PKlinconstitutive} in fluids, since a direct computation shows
$$
   \rho^0 \frac{\d U}{\d (\nabla u)} \approx_1 - p^0 1_{3\times 3} + 
     p^0(\ga - 1) (\nabla \cdot u) 1_{3\times 3} +
     p^0 (\nabla u)^T = T^0 + T^{\PK 1}.
$$

\subsubsection{Linearized dynamical interface condition} \label{444}

We consider a sufficiently regular, e.g.\ $\cC^{1,1}$, surface $S$
inside the earth; $S$ may be welded or perfectly slipping.  The exact
dynamical jump condition \eqref{eq:CauchyIBC}, $[T^s]_-^+\cdot\nu^s=0$
on $\vphi_\tim(S)$, takes the linearized form \cite[(3.73)]{DaTr:98}
\begin{equation}\label{eq:linCauchyIBC}
\left[T^{\PK 1}\cdot\nu-\snabla\cdot((T^0\cdot\nu)u)\right]_-^+=0\quad\text{on}\quad S,
\end{equation}
that is, 
$[T^{\PK 1}_{ij}\nu_j-\sd_k(T^0_{ij}\nu_ju_k)]_-^+=0$. We present a derivation of \eqref{eq:linCauchyIBC} along the lines of \cite[3.4.2]{DaTr:98} by techniques similar to the proof of \eqref{eq:slipsecondorder} based on $X'$ and $X''$ as in \eqref{eq:surfpoints}: We rewrite \eqref{eq:spatialsurfapprox}, $\dsurf^s\approx_1(1+\snabla\cdot u)\:\dsurf$,  in the form $\dsurf  \approx_1(1+\snabla\cdot u)^{-1}\:\dsurf^s \approx_1(1-\snabla\cdot u)\:\dsurf^s$; the material counterpart of \eqref{eq:CauchyIBC}, 
$$
   \left((T^0 + T^{\PK 1}) \cdot \nu \, \dint S\right)^+ (X') = 
   \left((T^0 + T^{\PK 1}) \cdot \nu \, \dint S\right)^- (X''),
$$
thus reads
$$
    \left((T^0 + T^{\PK 1}) \cdot \nu \,(1-\snabla\cdot u) \right)^+ (X') \,\dint S^s(x)
    \approx_1 
   \left((T^0 + T^{\PK 1}) \cdot \nu \,(1-\snabla\cdot u) \right)^- (X'') \,\dint S^s(x),
$$
which implies 
\begin{multline*}
    \underbrace{\left(T^{\PK 1} \cdot \nu\right)^+ (X') - 
       \left(T^{\PK 1} \cdot \nu\right)^- (X'')}_{=: A} +
    \underbrace{\left(T^0 \cdot \nu\right)^+ (X') - 
       \left(T^0 \cdot \nu\right)^- (X'')}_{=: B}\\  
    \underbrace{-\left( (T^0 \cdot \nu) (\snabla\cdot u) \right)^+ (X') +
       \left((T^0 \cdot \nu) (\snabla\cdot u) \right)^- (X'')}_{=: C}
      \approx_1 0.
\end{multline*}
A Taylor approximation yields $A \approx_1 [T^{\PK 1} \cdot \nu]_-^+
(X')$ and $C \approx_1 -[(T^0 \cdot \nu) (\snabla\cdot u)]_-^+
(X')$. We find $B \approx_1 -[ \snabla (T^0 \cdot \nu) \cdot u]_-^+
(X')$ by employing, in addition, $[T^0]_-^+\cdot\nu=0$
(cf.~\eqref{eq:IBC_FST0}) and the kinematical slip condition
$[u]_-^+\cdot\nu=0$ \eqref{eq:slipfirstorder}. Combining these terms
we arrive at $[T^{\PK 1}\cdot\nu-\snabla\cdot((T^0\cdot\nu)u)]_-^+
\approx_1 0$, proving \eqref{eq:linCauchyIBC}.

If the surface is welded like a solid-solid boundary $S=\Sig^\sxs$,
the conditions $[T^0]_-^+\cdot\nu=0$ \eqref{eq:IBC_FST0} and
$[u]_-^+=0$ \eqref{eq:weldedfirstorder} immediately give
\begin{equation}
\left[T^{\PK 1}\right]_-^+\cdot\nu=0\quad\text{on}\quad \Sig^\sxs.
\end{equation}
In contrast, when restricting to perfectly slipping fluid-solid boundaries $S=\Sig^\fxs$ we have $T^0 \cdot \nu = -p^0 \nu$ by \eqref{eq:surfpressure}. Thus  
$$
-\snabla\cdot((T^0\cdot\nu)u)=\snabla\cdot(p^0\nu\: u)=\nu\snabla\cdot(p^0u)+p^0\snabla\nu\cdot u.
$$
Due to $\snabla\nu\cdot a^\parallel=-\nu\cdot\snabla a^\parallel$ \eqref{eq:Wswap}, $[p^0]_-^+=0$ \eqref{eq:IBC_FSp0}, and the kinematical slip condition $[u]_-^+\cdot\nu=0$ \eqref{eq:slipfirstorder}
the last term can be replaced by
$[p^0\snabla\nu\cdot u]_-^+=p^0\snabla\nu\cdot [u]_-^+=-p^0\nu\cdot\snabla [u]_-^+=[-p^0\nu\cdot\snabla u]_-^+$. 
Consequently, the dynamical slip condition \eqref{eq:linCauchyIBC} reduces to
$$
\left[T^{\PK 1}\cdot\nu+\nu\snabla \cdot(p^0u)-p^0\nu\cdot\snabla  u\right]_-^+=0\quad\text{on}\quad \Sig^\fxs.
$$
Note that this condition coincides with Equation \eqref{eq:tractionIBC}, which is derived independently in Section \ref{sec:5}  from 
the variational principle  applied to the second-order approximation 
\eqref{eq:action_approx} of the action integral. There the second-order surface Lagrange density is \eqref{eq:Lsurf2} and the linear dynamical jump condition \eqref{eq:tractionIBC} is obtained as natural interior boundary condition (NIBC) according to \eqref{eq:Ham_IBCS}.
 
\subsection{Second-order approximation of the action integral}\label{sec:approxaction}

Aiming at dynamical equations which are linear in $(u,\Phi^1)$, we
approximate the action integral by a quadratic expression in these
perturbations and their first-order derivatives. Inserting
$\vphi=\Id_{\RR^3}+u$, $\Phi^s=\Phi^0+\Phi^1$, and the second-order
relation (\ref{eq:elastenergyapprox}) for prestressed linear
elasticity in the action integral $\action''(\vphi,\Phi^s)$
(\ref{eq:action_full}) yields the approximation
\begin{eqnarray}
\action''(\vphi,\Phi^s)&=&\action''(\Id_{\RR^3}+u,\Phi^0+\Phi^1)\nn\\
&\approx_2&\action_{[0]}''(u,\Phi^1)+\action_{[1]}''(u,\Phi^1)+\action_{[2]}''(u,\Phi^1)
\label{eq:action_approx}
\end{eqnarray}
where $\approx_2$ indicates the omission of terms of third order in $(u,\nabla u,\dot u,\Phi^1,\nabla\Phi^1,\dot\Phi^1)$, with
\begin{equation}\label{eq:action_approxi}
\action_{[i]}''(u,\Phi^1):=\int_I\left(\int_{\RR^3}L_{[i]}''\:\dvol+\int_{\Sig^\fxs}L_{\Sig^\fxs,[i]}''\:\dsurf\right)\mathrm{d}t
\qquad (i=0,1,2).
\end{equation}
As a general principle, we always assume the validity of all first-order results when we modify second-order terms. In particular, we frequently employ the kinematic slip condition $[u]_-^+\cdot \nu\approx_1 0$ \eqref{eq:slipfirstorder} as well as the equilibrium stress interface condition $[T^0]_-^+\cdot \nu=0$ \eqref{eq:IBC_FST0}, which reduces to continuity of $p^0=-\nu\cdot T^0\cdot\nu$ across fluid-solid boundaries.

\subsubsection{Volume Lagrangian densities}\label{ssec:vol_lag}
We show that the approximated Lagrangian densities $L_{[0]}''$, $L_{[1]}''$, and $L_{[2]}''$
of order zero, order one, and order two respectively read:
\begin{eqnarray}
L_{[0]}''&=&-\rho^0(\Phi^0+\Psi^s)-\frac{1}{8\pi G}\:(\nabla\Phi^0)^2,
\label{eq:L0}\\
&&\nn\\
L_{[1]}''&=&\rho^0\dot u\cdot R_\Om\cdot x-\:T^0:\nabla u-\rho^0u\cdot\nabla(\Phi^0+\Psi^s)
\nn\\ &&
-\rho^0\Phi^1-\frac{1}{4\pi G}\:\nabla\Phi^0\cdot\nabla\Phi^1-\rho^0 F^s,
\label{eq:L1}\\
&&\nn\\
L_{[2]}''&=&\frac{1}{2}\:\rho^0\dot u^2+\rho^0\dot u\cdot R_\Om\cdot u-\frac{1}{2}\:\nabla u:\La^{T^0}:\nabla u
-\frac{1}{2}\:\rho^0u\cdot(\nabla\nabla(\Phi^0+\Psi^s))\cdot u\nn\\
&& -\rho^0u\cdot\nabla\Phi^1-\frac{1}{8\pi G}(\nabla\Phi^1)^2-\rho^0u\cdot\nabla F^s.
\label{eq:L2}
\end{eqnarray}
Here we have written the vector product of $\Om$ with $y\in\RR^3$ in a matrix multiplication form:
\begin{equation}
R_\Om\cdot y=\Om\times y \quad\textrm{for}\quad R_\Om:=(\eps_{ijk}\Om_j)_{i,k=1}^3.
\end{equation}
The Lagrangian densities are functions of $x$, $t$, $u(x,t)$, $\Phi^1(x,t)$, $\nabla u(x,t)$, $\dot u(x,t)$, and
$\nabla\Phi^1(x,t)$ (note their independence of $\dot\Phi^1(x,t)$). More precisely, $L_{[0]}''$ is independent of the
variables $(u,\Phi^1)$ and their derivatives as well as of time. Furthermore, $L_{[1]}''$ and
$L_{[2]}''$ are first- and second-order polynomials in $u$, $\Phi^1$, $\nabla u$, $\dot u$, $\nabla\Phi^1$, the force potential $F^s$, and its gradient $\nabla F^s$ with time-independent coefficients (the appearance  of the force terms is explained in \ref{ssec:linsources}). Apart from the inclusion of the external force, twice the second-order Lagrangian density $L_{[2]}''$ coincides
with the volume Lagrangian given by \cite[(3.190), p.\ 89]{DaTr:98} of their ``displacement-potential variational principle''.

The zero-, first-, and second-order actions (\ref{eq:action_approxi}) with Lagrangians (\ref{eq:L0}) to (\ref{eq:L2})
are defined within the regularity setting of the linearized fields on the $\Lip$-composite fluid-solid earth model,
which is in complete consistence with the regularity conditions of Definition \ref{def:nonlin_reg} for the nonlinear case
making the exact action integral (\ref{eq:action_full}) defined.
Indeed, with $\vphi=\Id_{\RR^3}+u$, the kinetic energy density, the Coriolis term, and the centrifugal potential term, respectively read
\begin{eqnarray}
v^2&=&\dot u^2,\\
v\cdot R_\Om\cdot\vphi&=&\dot u\cdot R_\Om\cdot x+\dot u\cdot R_\Om\cdot u,\\
\Psi&=&\Psi^s+u\cdot\nabla\Psi^s+\frac{1}{2}u\cdot\nabla\nabla\Psi^s\cdot u,
\end{eqnarray}
where the last equality holds since $\Psi(X,t)=\Psi^s\circ\vphi_\tim(X)=\Psi^s(X+u_\tim(X))$ and $\Psi^s$
is a second-order polynomial. Upon inserting $\Phi^s=\Phi^0+\Phi^1$, the squared norm of the spatial gradient
$|\nabla\Phi^s|^2=(\d_k\Phi^s)^2$ is
\begin{equation}
(\nabla\Phi^s)^2=\left(\nabla\Phi^0+\nabla\Phi^1\right)^2=(\nabla\Phi^0)^2+2\nabla\Phi^0\cdot\nabla\Phi^1+(\nabla\Phi^1)^2.
\end{equation}
The second-order approximation for the material gravitational potential
\begin{equation}
\Phi_\tim(X)=(\Phi_\tim^s|_{\vphi_\tim(\earthi)})\circ\vphi_\tim(X)
=\Phi_\tim^s(X+u_\tim(X))=\Phi^0(X+u_\tim(X))+\Phi^1_t(X+u_\tim(X))
\end{equation}
reads
\begin{eqnarray}
\Phi&\approx_2&\Phi^0+(u\cdot\nabla\Phi^0+\Phi^1)+\left(\frac{1}{2}\:u\cdot\nabla\nabla\Phi^0\cdot u+u\cdot\nabla\Phi^1\right)
\end{eqnarray}
(where $\approx_2$ indicates the omission of terms of third order involving $(u,\Phi^1)$ or their derivatives).
Inserting these approximations and (\ref{eq:elastenergyapprox}) for linear prestressed elasticity into the
Lagrangian density (\ref{eq:Lfull}) and comparing with the powers of $(u,\Phi^1)$ and their derivatives in (\ref{eq:action_approx}) gives equations (\ref{eq:L0}) to (\ref{eq:L2}).

\subsubsection{Surface Lagrangian densities}\label{ssec:surf_lag}

We will show below that the approximated surface Lagrangian densities $L_{\Sig^\fxs,[0]}''$, $L_{\Sig^\fxs,[1]}''$, $L_{\Sig^\fxs,[2]}''$
of orders zero up to two are given by the following expressions (or several equivalent formulae for $L_{\Sig^\fxs,[2]}''$ in Remark \ref{RemEquiv}):
\begin{eqnarray}
L_{\Sig^\fxs,[0]}''&=&L_{\Sig^\fxs,[1]}''=0,
\label{eq:Lsurf01}\\
&&\nn\\
L_{\Sig^\fxs,[2]}''&=&-p^0\left[\nu\cdot\snabla u\cdot u
+\frac{1}{2}u\cdot\snabla \nu\cdot u\right]_-^+.
\label{eq:Lsurf2}
\end{eqnarray}
The surface Lagrangian $L_{\Sig^\fxs,[2]}''$, in the equivalent form (\ref{eq:Lsurf2v4}) and recalling that the Jacobi matrix in \cite{DaTr:98} is transposed, coincides with the surface density given in \cite[(3.165), p.\ 89]{DaTr:98}.

\begin{remark}[{{\bf Equivalent representations of the surface Lagrangian}}]\label{RemEquiv} 
Since $\nu\cdot\snabla u\cdot u=u\cdot\snabla (u\cdot \nu)-u\cdot\snabla\nu\cdot u$
we may rewrite \eqref{eq:Lsurf2} as 
\begin{equation}\label{eq:Lsurf2v1}
L_{\Sig^\fxs,[2]}''=-p^0\left[u\cdot\snabla(\nu\cdot u)
-\frac{1}{2}u\cdot\snabla \nu\cdot u\right]_-^+
\end{equation}
and obtain, upon integration,
\begin{eqnarray}
\int_{\Sig^{\fxs}}L_{\Sig^\fxs,[2]}''\:\dint S
&=&\int_{\Sig^{\fxs}}\left[(\nu\cdot u)\snabla\cdot(p^0u)
+\frac{1}{2}p^0u\cdot\snabla \nu\cdot u\right]_-^+\dint S\label{eq:Lsurf2v3}\\
&=&\frac{1}{2}\int_{\Sig^{\fxs}}\left[(\nu\cdot u)\snabla\cdot(p^0u)
-p^0\nu\cdot\snabla u\cdot u\right]_-^+\dint S\label{eq:Lsurf2v4}.
\end{eqnarray}
Here, we used the identities
\begin{eqnarray}
p^0u\cdot\snabla(\nu\cdot u)
&=&\snabla\cdot\left(p^0u(\nu\cdot u)\right)-(\nu\cdot u)\snabla\cdot(p^0u)\nn\\
p^0u\cdot\snabla \nu\cdot u
&=&\snabla\cdot(p^0u(\nu\cdot u))-(\nu\cdot u)\snabla\cdot(p^0u)
-p^0\nu\cdot\snabla u\cdot u\nn
\end{eqnarray}
and applied the surface divergence theorem to the corresponding surface action, where the line integrals vanish as $p^0=0$ on $\d\Sig^{\fxs}\subseteq\d\earth$.
\end{remark}

The derivation of the second-order surface action is more involved than in case of the volume action. It requires in addition that the earth's fluid-solid interior boundaries $\Sig^{\fxs}$ are at least $\cC^{1,1}$. The second-order surface action is obtained as approximation of the exact surface action
$\action_{\Sig^\fxs}(\vphi)=\action_{\Sig^\fxs}''(\vphi)=\int_I\action_{\Sig^\fxs,\tim}''(\vphi)\:\mathrm{d}t$ (\ref{eq:action_surf})
with surface Lagrangian (\ref{eq:Lfull_surf}), that is
\begin{equation}
\action_{\Sig^\fxs,\tim}''(\Id_{\RR^3}+u)=\int_{\Sig^{\fxs}}L_{\Sig^\fxs,\tim}''\dsurf
=-\int_{t_0}^t\int_{\Sig^{\fxs}}\left[\dot u_\timd\cdot T^{\PK}_\timd \right]_-^+\cdot \nu\:\dsurf\:\mathrm{d}t'.
\end{equation}
In the sequel we discuss three variants to derive the approximation of the action up to second order: (A) starting from the linearized dynamical interface condition \eqref{eq:tractionIBC} as presented in \ref{444}, (B) by direct expansion of the surface Lagrangian density \eqref{eq:Lfull_surf}, and (C) from energy conservation and the second-order tangential slip condition \eqref{eq:sslipsecondorder}.

\paragraph{Variant (A) - derivation via linearized dynamical interface condition} We split the surface action $\action_{\Sig^\fxs,\tim}''$ into two parts corresponding to $T^\PK \approx_1 T^0+T^{\PK1}$,
\begin{equation}
\int_{\Sig^{\fxs}}L_{\Sig^\fxs,\tim}''\dsurf \approx_2
-\int_{t_0}^t\int_{\Sig^{\fxs}}\left[\dot u_\timd\cdot
  T^0\right]_-^+\cdot \nu\:\dsurf\:\mathrm{d}t'
-\int_{t_0}^t\int_{\Sig^{\fxs}}\left[\dot u_\timd\cdot T^{\PK1}_\timd
  \right]_-^+\cdot \nu\:\dsurf\:\mathrm{d}t',
\end{equation}
and consider the first-order term. Pulling out the time derivative 
and employing the normality condition $T^0\cdot\nu=-p^0\nu$ \eqref{eq:surfpressure}, the continuity of $p^0$ \eqref{eq:IBC_FSp0}, and the first-order tangential slip condition $[u]_-^+\cdot\nu=0$ \eqref{eq:slipfirstorder} on $\Sig^\fxs$, we deduce 
\begin{equation}
L_{\Sig^\fxs,[1]}''=
-\int_{t_0}^t\left[\dot u\cdot T^0\right]_-^+\cdot \nu\:\mathrm{d}t'
=-\int_{t_0}^t\frac{\d}{\d t}\left[u\cdot T^0\right]_-^+\cdot \nu\:\mathrm{d}t'
=-\left[u\cdot T^0\right]_-^+\cdot \nu
=p^0\left[u\right]_-^+\cdot \nu=0.
\end{equation}
We proceed to investigate the second-order term. We employ \cite[Lemma 3.1]{dHHP:15}, which is based on the linear dynamical slip condition \eqref{eq:tractionIBC} and the specific form  of $T^{\PK1}$ in fluids \eqref{eq:TPK1fluid}. Therefore, the integrand of the second-order contribution can be rewritten in the following form:
\begin{eqnarray}
-\int_{\Sig^{\fxs}}\left[\dot u\cdot T^{\PK1} \right]_-^+\cdot \nu\:\dsurf
&=&-\int_{\Sig^{\fxs}}p^0\left([u]_-^+\cdot\snabla(\dot u\cdot\nu)+[\dot u]_-^+\cdot\snabla(u\cdot\nu)
-\left[u\cdot\snabla\nu\cdot\dot u\right]_-^+\right)\dsurf\nn\\
&=&-\frac{\d}{\d t}\int_{\Sig^{\fxs}}\left(p^0[u]_-^+\cdot\snabla(u\cdot\nu)
-\frac{1}{2}p^0\left[u\cdot\snabla\nu\cdot u\right]_-^+\right)\dsurf,\nn
\end{eqnarray}
where the second equality follows from the product rule and the symmetry of the Weingarten operator $\snabla\nu$.
Application of the surface divergence theorem in combination with the zero traction condition $p^0=0$ on $\d\Sig^\fxs\subseteq\d\earth$ yields 
\begin{eqnarray}
\int_{\Sig^\fxs}L_{\Sig^\fxs,[2]}''\dsurf
&=&-\int_{t_0}^t\int_{\Sig^\fxs}\left[\dot u\cdot T^{\PK1}\right]_-^+\cdot \nu\:\dsurf\:\mathrm{d}t'\nn\\
&=&-\int_{\Sig^{\fxs}}\left(p^0[u]_-^+\cdot\snabla(u\cdot\nu)
-\frac{1}{2}p^0\left[u\cdot\snabla\nu\cdot u\right]_-^+\right)\dsurf\nn \\
&=&\int_{\Sig^{\fxs}}\left((u\cdot\nu)\snabla\cdot(p^0[u]_-^+)
+\frac{1}{2}p^0\left[u\cdot\snabla\nu\cdot u\right]_-^+\right)\dsurf.
\end{eqnarray}
Finally, continuity of $p^0$ and $[u]_-^+\cdot\nu=0$ across $\Sig^\fxs$ allow to pull out the jump operator, thereby showing that the second-order term $\int_{\Sig^\fxs}L_{\Sig^\fxs,[2]}''\dsurf$ coincides with \eqref{eq:Lsurf2v3}.

\paragraph{Variant (B) - direct expansion}
We first apply the Piola transform \eqref{eq:TPK_Piola} $T^\PK=J T\cdot(\nabla\vphi)^{-T}$ and linearize each factor: 
\begin{equation}
T\approx_1 T^0+T^1,\qquad
J\approx_1 1+\nabla\cdot u
\qquad\text{and}\qquad 
(\nabla\vphi)^{-T}\approx_1 1_{3\times 3}-(\nabla u)^{T}.
\end{equation}
Consequently,
\begin{eqnarray}
T^\PK=J T\cdot(\nabla\vphi)^{-T}&\approx_1&(1+\nabla\cdot u)(T^0+T^1)\cdot( 1_{3\times 3}-(\nabla u)^{T})\nn\\
&\approx_1& T^0+(\nabla\cdot u)T^0+T^1-T^0\cdot(\nabla u)^{T}.
\end{eqnarray}
Thus, to second order and noticing $T^0\cdot(\nabla u)^{T}=(\nabla u\cdot T^0)^T$ due to symmetry of $T^0$,
\begin{eqnarray}
\dot u\cdot T^{\PK}\cdot\nu &\approx_2&
\dot u\cdot T^0\cdot\nu+(\nabla\cdot u)\dot u\cdot T^0\cdot\nu
+\dot u\cdot T^1\cdot\nu-\dot u\cdot T^0\cdot(\nabla u)^{T}\cdot\nu\\
&=&\dot u\cdot T^0\cdot\nu+(\nabla\cdot u)\dot u\cdot T^0\cdot\nu
+\dot u\cdot T^1 \cdot\nu-\nu\cdot(\nabla u\cdot T^0)\cdot\dot u.
\end{eqnarray}
Rewriting the second and the last term in terms of surface derivatives yields
\begin{multline*}
  (\nabla\cdot u)\dot u\cdot T^0\cdot\nu-\nu\cdot(\nabla u\cdot T^0)\cdot\dot u \\
  = (\snabla\cdot u)\dot u\cdot T^0\cdot\nu-\nu\cdot (\snabla u\cdot T^0)\cdot\dot u
	 + (\nu\cdot\nabla u\cdot \nu)\dot u\cdot T^0\cdot\nu - 
	   (\nu\cdot\nabla u\cdot\nu)\nu\cdot T^0\cdot\dot u\\
		= (\snabla\cdot u)\dot u\cdot T^0\cdot\nu-\nu\cdot (\snabla u\cdot T^0)\cdot\dot u,
\end{multline*}
since $T^0$ is symmetric. As we will consider the jump across the fluid solid boundary and the normal component of $u$ is continuous, we argue that we may approximate the material Cauchy stress perturbation $T^1$ by the surface derivative of $T^0$ in direction of $u$, that is, 
\begin{equation}\label{eq:T1approx}
T^1\approx_1 \snabla T^0\cdot u.
\end{equation}
Indeed, by \cite[Equation (3.16)]{DaTr:98}, $T^1 \approx_1 T^{s1} + \nabla T^0 \cdot u =  T^{s} - T^0 + \nabla T^0 \cdot u$, hence $[\dot u\cdot T^1 \cdot\nu]^+_- \approx_2 [\dot u\cdot T^{s} \cdot\nu]^+_- - [\dot u\cdot T^0 \cdot\nu]^+_- + [\dot u\cdot (\nabla T^0 \cdot u) \cdot\nu]^+_- = [\dot u\cdot (\snabla T^0 \cdot u) \cdot\nu]^+_-$, since $T^s$ and $T^0$ satisfy the normality condition \eqref{eq:normality} and $[u\cdot\nu]^+_- = 0$. 
Consequently, 
\begin{eqnarray}
\dot u\cdot T^{\PK}\cdot\nu &\approx_2&
   \dot u\cdot T^0\cdot\nu+(\snabla\cdot u)\dot u\cdot T^0\cdot\nu
+\dot u\cdot T^1 \cdot\nu-\nu\cdot (\snabla u\cdot T^0)\cdot\dot u\nn\\
&=&\dot u\cdot T^0\cdot\nu+(\snabla\cdot u)\dot u\cdot T^0\cdot\nu
+\dot u\cdot(\snabla T^0\cdot u)\cdot\nu-\nu\cdot (\snabla u\cdot T^0)\cdot\dot u.
\end{eqnarray}

Successive application of the product rule and symmetry of $T^0$ and
$\snabla\nu$ allows to write the sum of the second and the third term
in the form
\begin{eqnarray}
   \dot u \cdot ((\snabla\cdot u)T^0+(\snabla T^0\cdot u)) \cdot \nu
   &=&\dot u_i ((\sd_k u_k)T^0_{ij}+(\sd_k T^0_{ij}) u_k)) \nu_j
\nn\\
&=&\dot u_i\sd_k(T^0_{ij}u_k)\nu_j
\nn\\
&=&\sd_k(\dot u_iT^0_{ij}u_k\nu_j)-(\sd_k\dot u_i)T^0_{ij}u_k\nu_j-\dot u_iT^0_{ij}u_k(\sd_k\nu_j)
\nn\\
&=&\snabla\cdot ((\dot u\cdot T^0\cdot\nu)u)
-(T^0\cdot\nu)\cdot\snabla\dot u\cdot u-\dot u\cdot (T^0\cdot\snabla\nu)\cdot u
\nn\\
&=&\snabla\cdot ((\dot u\cdot T^0\cdot\nu)u)
-\nu\cdot(T^0\cdot\snabla\dot u)\cdot u-u\cdot (T^0\cdot\snabla\nu)\cdot \dot u.
\end{eqnarray}
By the surface divergence theorem and the zero traction boundary condition, $T^0\cdot\nu=0$ on $\d\Sig^\fxs\subseteq\d\earth$, the integral of the surface divergence term over $\Sig^\fxs$ vanishes. Thus, 
\begin{eqnarray}
\dot u\cdot T^{\PK}\cdot\nu &\approx_2&
\dot u\cdot T^0\cdot\nu
-u\cdot (T^0\cdot\snabla\nu)\cdot \dot u
-\nu\cdot(T^0\cdot\snabla\dot u)\cdot u
-\nu\cdot (\snabla u\cdot T^0)\cdot\dot u.\nn\\
&=&\frac{\d}{\d t}\left(u\cdot T^0\cdot\nu-\frac{1}{2}\:u\cdot (T^0\cdot\snabla\nu)\cdot u \right)
-\nu\cdot (\snabla u\cdot T^0\cdot\dot u+T^0\cdot\snabla \dot u\cdot u)
\end{eqnarray}
Considering the deviatoric initial stress $T^0_\mathrm{dev}$ \eqref{eq:T0dev} as an effect of first order (see \cite[p. 102]{DaTr:98}), we have
\begin{equation}\label{eq:devstresszero}
T^0\approx_0-p^01_{3\times 3}.
\end{equation} 
Therefore, up to second order, also the last term is a time derivative
\begin{equation}
-\nu\cdot (\snabla u\cdot T^0\cdot\dot u+T^0\cdot\snabla \dot u\cdot u)
\approx_2\nu\cdot p^0(\snabla u\cdot\dot u+\snabla \dot u\cdot u)=p^0\frac{\d}{\d t}(\nu\cdot\snabla u\cdot u)
\end{equation}
and with 
$u\cdot T^0\cdot\nu-\frac{1}{2}\:u\cdot (T^0\cdot\snabla\nu)\cdot u
\approx_2-p^0(u\cdot\nu-\frac{1}{2}\:u\cdot\snabla\nu\cdot u)$ 
we arrive at
\begin{equation}
\dot u\cdot T^{\PK}\cdot\nu \approx_2
-p^0\frac{\d}{\d t}\left(u\cdot\nu-\nu\cdot\snabla u\cdot u-\frac{1}{2}\:u\cdot \snabla\nu\cdot u\right).
\end{equation}
Consequently, on a fluid-solid interface, the up to second-order expansion of the surface Lagrangian \eqref{eq:Lfull_surf} is given by
\begin{eqnarray}
L_{\Sig^\fxs}''
&\approx_2& \left[p^0\left(u\cdot\nu-\nu\cdot\snabla u\cdot u-\frac{1}{2}\:u\cdot \snabla\nu\cdot u\right)\right]_-^+
\label{eq:Lsurf12}\\
&=&[p^0(u\cdot\nu)]_-^+
-\left[p^0\left(\nu\cdot\snabla u\cdot u+\frac{1}{2}\:u\cdot \snabla\nu\cdot u\right)\right]_-^+\nn \\[0.2cm]
&=&L_{\Sig^\fxs,[0]}''+L_{\Sig^\fxs,[1]}''+L_{\Sig^\fxs,[2]}''. \nn
\end{eqnarray}
We observe that there is no constant part: $L_{\Sig^\fxs,[0]}''=0$. The linear part vanishes to first order, due to continuity of $p^0$ \eqref{eq:IBC_FSp0} and  the first-order tangential slip condition $[u]_-^+\cdot\nu=0$ \eqref{eq:slipfirstorder}: 
$L_{\Sig^\fxs,[1]}'':=[p^0(u\cdot\nu)]_-^+=p^0[u]_-^+\cdot\nu\approx_1 0$. 
The second-order part $L_{\Sig^\fxs,[2]}''$ also coincides with \eqref{eq:Lsurf2}. 
      
As an additional result, \eqref{eq:Lsurf12} incorporates the linear plus quadratic approximation of the surface energy density on fluid-solid interfaces, since by \eqref{eq:Legendre_surf} $E_{\Sig^{\fxs}}''=-L_{\Sig^{\fxs}}''$,
\begin{eqnarray}\label{eq:Esurf12}
E_{\Sig^{\fxs}}''&\approx_2&
-p^0\left[u\cdot\nu-\nu\cdot\snabla u\cdot u-\frac{1}{2}\:u\cdot \snabla\nu\cdot u\right]_-^+\\
&=&E_{\Sig^\fxs,[0]}''+E_{\Sig^\fxs,[1]}''+E_{\Sig^\fxs,[2]}''.\nn
\end{eqnarray}
      
\paragraph{Variant (C) - derivation from energy considerations}
Energy conservation in the nonlinear setting was discussed in \ref{ssec:energybalance}.
In the linearized setting, as is shown in \cite[Appendix C]{dHHP:15}, energy conservation implies that 
the second-order surface energy must be equal to the work due to slip against initial stress on $\Sig^\fxs$, whereas the zero and first-order contribution  to the surface energy vanishes: 
\begin{eqnarray}
E_{\Sig^{\fxs},[0]}''&=&E_{\Sig^{\fxs},[1]}''=0,\nn\\
E_{\Sig^{\fxs},[2]}''&=&[u\cdot T^0]_-^+\cdot \nu. \label{eq:Esurf2}
\end{eqnarray}
With the second-order tangential slip condition 
$\left[u\cdot \nu -\nu\cdot\snabla u\cdot u-\frac{1}{2}\:u\cdot\snabla\nu\cdot u\right]_-^+\approx_20$ \eqref{eq:sslipsecondorder} we obtain
\begin{eqnarray}
\int_{\Sig^{\fxs}}E_{\Sig^{\fxs},[2]}''\:\dint S 
=-\int_{\Sig^{\fxs}}[u\cdot T^0]_-^+\cdot \nu\:\dint S 
&=& \int_{\Sig^{\fxs}}p^0[u]_-^+\cdot \nu\dint S \nn\\
&\approx_2&
\int_{\Sig^{\fxs}}p^0\left[\nu\cdot\snabla u\cdot u+\frac{1}{2}\:u\cdot \snabla\nu\cdot u\right]_-^+\dint S.\nn
\end{eqnarray}
Consequently, in agreement with the second-order terms of \eqref{eq:Esurf12}, we have
\begin{equation}
E_{\Sig^{\fxs},[2]}''=p^0\left[\nu\cdot\snabla u\cdot u+\frac{1}{2}\:u\cdot \snabla\nu\cdot u\right]_-^+
\end{equation}
and the corresponding second-order surface Lagrangian
$L_{\Sig^{\fxs},[2]}''=-E_{\Sig^{\fxs},[2]}''$ coincides with \eqref{eq:Lsurf2}. 

\subsection{Hydrostatic prestress}

In a {\bf hydrostatic equilibrium earth model}, prestress is given by
\eqref{eq:fluidprestress},
\begin{equation*}
   T^0 = -p^0 1_{3\times 3},
\end{equation*}
and the static equilibrium equation \eqref{eq:staticeq} reduces to
{\bf hydrostatic balance}
\begin{equation}\label{eq:hydrostaticeq}
\nabla p^0=-\rho^0\nabla(\Phi^0+\Psi^s).
\end{equation}
In a hydrostatic equilibrium earth, pressure gradients are parallel to
the gravitational plus centrifugal acceleration terms.  This also
follows by taking the cross product of \eqref{eq:hydrostaticeq} with
$\nabla p^0$, which results in $\nabla
p^0\times\nabla(\Phi^0+\Psi^s)=0$.  Let us consider regions where
$\rho^0$ is sufficiently smooth such that all terms in
\eqref{eq:hydrostaticeq} are at least $\cC^1$. This is the case, e.g.,
if $\rho^0$ is $\cC^{1,\alpha}$ for $0<\alpha < 1$, because solutions
$\Phi^0$ of Poisson equation then are $\cC^{3,\alpha}$ by elliptic
regularity on H\"older-Zygmund spaces \cite[Notes at the end in
  Section 8.6, referring to Corollary 8.4.7]{Hoermander:97}. The
properties of these spaces with respect to multiplication (see
\cite[Proposition 8.6.8]{Hoermander:97}) guarantee that $\nabla p^0$
is $\cC^{1,\alpha}$.

Taking the curl of \eqref{eq:hydrostaticeq} yields
$\nabla\rho^0\times\nabla(\Phi^0+\Psi^s)=0$.
Consequently, in a hydrostatic earth, 
\begin{equation}
\nabla \rho^0\: \parallel \: \nabla p^0\: \parallel \: \nabla(\Phi^0+\Psi^s).
\end{equation}
In other words, the level surfaces of the density $\rho^0$, the
pressure $p^0$ and the geopotential $\Phi^0+\Psi^s$ are parallel. Due
to the zero-traction condition \eqref{eq:NBC_T0} we have $p^0=0$ on
the exterior boundary $\d\earth$, that is, $\d\earth$ is a level set
for $p^0$. Therefore also $\rho^0$ and $\Phi^0+\Psi^s$ must be
constant on $\d\earth$.  Furthermore, in the hydrostatic model one
assumes that all interior boundaries $\Sig$ of a composite fluid-solid
earth are level surfaces of $\rho^0$, $p^0$ and $\Phi^0+\Psi^s$. On
$\cC^1$-discontinuity surfaces of $\rho^0$, which are contained in
$\Sig\cup\d\earth$, the $\pm$-traces, $(\nabla\rho^0)_\pm$, $(\nabla
p^0)_\pm$, and $(\nabla(\Phi^0+\Psi^s))_\pm$, must also be parallel
\cite[Lemma 2.1]{dHHP:15}. Consequently, the respective surface
gradients vanish \cite[(3.258)]{DaTr:98}:
\begin{equation}\label{eq:snabla0}
   \snabla \rho^0=0, \quad \snabla p^0=0,\quad \snabla(\Phi^0+\Psi^s)=0
\quad \text{on}\quad \Sig\cup\d\earth.
\end{equation}
The hydrostatic assumption severely restricts the possible equilibrium
earth models. In absence of rotation, hydrostatic equilibrium implies
spherical symmetry of the planet, that is, its material parameters are
functions of the radial coordinate only. Uniform rotation leads to
rotational ellipsoidal symmetry (up to first order in the
centrifugal-to-gravitational-force ratio, which is sufficiently
accurate in case of the relatively slowly rotating earth
\cite[p.\ 597]{DaTr:98}). Thus, in the hydrostatic equilibrium earth
model, all level sets and interior boundaries are oblate
ellipsoids. In particular, the fields $\rho^0$, $p^0$, and
$\Phi^0+\Psi^s$ depend only on the ellipsoidal radial distance from
the center of the earth but are constant in lateral (that is\ latitude
or longitude) directions. Consequently, the presence of any lateral
heterogeneity in $\rho^0$ requires a non-zero initial deviatoric
stress $T^0_\text{dev}$ for its support, so that the hydrostatic
assumption $T^0=-p^0 1_{3\times 3}$ does no longer hold. In practice
this discrepancy is often ignored and one assumes validity of the
hydrostatic equations even if $\rho^0$ is laterally
heterogeneous. However, as $T^0_{\text{dev}}$ is small compared to the
earth's rigidity, this so-called {\bf quasi-hydrostatic assumption} is
indeed justified \cite[p. 102]{DaTr:98}.

The advantage of the hydrostatic assumption is that it simplifies the
elastic-gravitational equations and interface conditions, especially
if the Lagrangian stress perturbation $T^1$ is employed. As by
\eqref{eq:TL1} $T^1$ is independent of initial pressure and
$T^0_\text{dev}=0$, we immediately get the constitutive law of
classical (linearized) elasticity, recalling the definition of the
linearized strain tensor $\varepsilon=\frac{1}{2}(\nabla u+(\nabla
u)^T)$:
\begin{equation}
T^1=\Gamma:\varepsilon.
\end{equation} 
The linearized dynamical boundary conditions simplify to \cite[(3.265)-(3.267)]{DaTr:98}
\begin{equation}\label{eq:BChyd}
T^1\cdot\nu=0\quad\text{on}\quad\d\earth,\qquad
\left[T^1\right]_-^+\cdot\nu=0\quad\text{on}\quad\Sig=\Sig^\sxs\cup\Sig^\fxs, 
\end{equation}
and 
the additional normality condition \begin{equation}\label{eq:BCnorhyd}
T^1\cdot\nu=(\nu\cdot T^1\cdot\nu)\nu\quad\text{on}\quad\Sig^\fxs.
\end{equation}
The identities follow from the relation \eqref{eq:TPK1T1} between
$T^{\PK 1}$ and $T^1$, $T^0=-p^01_{3\times 3}$, and the isobaric
condition $\snabla p^0=0$ \eqref{eq:snabla0}.  In consistence with the
variational derivation of the dynamical interface conditions as NIBC
\eqref{eq:Ham_IBCS}, the second-order surface energy vanishes if
prestress is hydrostatic. This immediately follows from the
representation \eqref{eq:Esurf2} of the surface energy, continuity of
$p^0$ \eqref{eq:IBC_FSp0}, and the first-order tangential slip
condition $[u]_-^+\cdot\nu=0$ \eqref{eq:slipfirstorder}:
\begin{equation}
E_{\Sig^{\fxs},[2],\text{hyd}}''=[u\cdot T^0]_-^+\cdot \nu
=-[p^0 u]_-^+\cdot \nu=-p^0[u]_-^+\cdot \nu=0. 
\end{equation}
Consequently, also the second-order surface Lagrangian vanishes,
\begin{equation}
L_{\Sig^{\fxs},[2],\text{hyd}}''=-E_{\Sig^{\fxs},[2],\text{hyd}}''=0.
\end{equation}
The second-order action integral in \eqref{eq:action_approxi} thus contains only a volume integral:
\begin{equation}
\action_{[2],\text{hyd}}''(u,\Phi^1)
:=\int_I\left(\int_{\RR^3}L_{[2],\text{hyd}}''\:\dvol\right)\mathrm{d}t
\end{equation}
with \cite[(3.275)]{DaTr:98}
\begin{eqnarray}
L_{[2],\text{hyd}}''&=&\frac{1}{2}\:\rho^0\dot u^2+\rho^0\dot u\cdot R_\Om\cdot u
-\frac{1}{2}\:\varepsilon:\Ga:\varepsilon
-\frac{1}{2}\rho^0\nabla (\Phi^0+\Psi^s)\cdot \left(\nabla u\cdot u-(\nabla\cdot u)u\right)
\nn\\
&& 
-\frac{1}{2}\:\rho^0u\cdot(\nabla\nabla(\Phi^0+\Psi^s))\cdot u
-\rho^0u\cdot\nabla\Phi^1-\frac{1}{8\pi G}(\nabla\Phi^1)^2-\rho^0u\cdot\nabla F^s.
\label{eq:L2hyd}
\end{eqnarray}
Indeed, by hydrostatic balance \eqref{eq:hydrostaticeq} twice the fourth term is equal to $\nabla p^0\cdot (\nabla u\cdot u-(\nabla\cdot u)u)$. The divergence theorem with $p^0=0$ on $\d\earth$ implies $-p^0\nabla\cdot (\nabla u\cdot u-(\nabla\cdot u)u)$, which by the product rule and the symmetry of $\nabla\nabla u$ reduces to $-p^0\left((\nabla u):(\nabla u)^T-(\nabla\cdot u)^2\right)$. As $\Ga:\varepsilon=T^1=\Upsilon^{T_0}:\nabla u$, where $\Upsilon^{T_0}$ and $\Lambda^{T_0}$ 
are related by  \eqref{eq:Upsilon}, and as $\Ga$ has all classical symmetries, one may show that the sum of the third and the fourth term coincides with $\frac{1}{2}\nabla u:\La^{T^0}:\nabla u$. This proves the equivalence of \eqref{eq:L2hyd} to the general form of the second-order Lagrangian density \eqref{eq:L2} in case of hydrostatic prestress.

\section{Variational principle}\label{sec:5} 

We have derived a second-order approximation of
the full action integral (\ref{eq:action_full}) by replacing the
configuration variables $(\vphi,\Phi^s)$ by equilibrium fields
$(\Id_{\RR^3},\Phi^0)$ plus perturbations $(u,\Phi^1)$.  Stationarity
of the approximated action integral
$\action_{[0]}''+\:\action_{[1]}''+\action_{[2]}''$
(\ref{eq:action_approx}) with respect to variations in $(u,\Phi^1)$ is
understood as stationarity of the actions of each order $0,1,2$
separately. Then the equilibrium equations coupling
$\rho^0$, $\Phi^0$, and $T^0$ will follow from the variation of
$\action_{[1]}''$, whereas the variation of $\action_{[2]}''$ will yield the
dynamical equations for $u$ and $\Phi^1$ (Section \ref{ssec:EL}).  However, before studying
stationarity we need to establish a suitable function space setting
for the different approximations of the action in which we can
calculate the Fr\'echet derivatives.

\subsection{Stationarity of the action for the composite fluid-solid earth model}
\label{ssec:Ham}

\subsubsection{Regularity}\label{ssec:Hamreg}

First we observe that in the approximated action
(\ref{eq:action_approxi}), we may instead of $\RR^3$ integrate over
the unbounded $\Lip$-composite domain $\earthfsc$
(\ref{eq:domain_sfc}), since its interior boundaries
$\Sig^\fxs\cup\d\earth$ have measure zero. This restriction allows us
to benefit from the higher regularity inside the interior regions of
the earth. Indeed the displacement associated to a piecewise
$\Lip$-continuous motion is just $L^\infty$ on $\earth$ but $H^1$ on
fluid and solid interior regions separately. Furthermore, spatial
derivatives of the motion have been defined as $L^{\infty}$-functions
on $\earth$ by simply neglecting the possible discontinuity surface
$\Sig^\fxs$, and not globally as distributional derivatives.  Further
note that since $\supp\rho^0=\supp u_\tim=\ovl\earth$ the only
non-zero contributions to the volume integral exterior to $\earth$ are
integrands associated to the squared gradient of the gravitational
potential, namely $-\frac{1}{8\pi G}\:(\nabla\Phi^0)^2$,
$-\frac{1}{4\pi G}\:(\nabla\Phi^0)\cdot(\nabla\Phi^1)$, and
$-\frac{1}{8\pi G}(\nabla\Phi^1)^2$ appearing in $L_{[0]}''$,
$L_{[1]}''$, and $L_{[2]}''$ respectively (see Subsection
\ref{ssec:vol_lag}).

Assumption~\ref{ass:lin_reg} below collects the regularity conditions
on the material parameters $c$, $\rho^0$, the reference fields $T^0$,
$\Phi^0$, the perturbations $u$, $\Phi^1$ of the linearized
setting, and of the force potential $F^s$, which will be sufficient to
define the approximated action integrals (\ref{eq:action_approxi}) and
study their Fr\'echet differentiability.

\begin{assumption}[{{\bf Regularity conditions for linearized setting}}] \label{ass:lin_reg}
We assume that
$$
\rho^0,\:T^0_{ij},\:c_{ijkl}\in L^{\infty}(\RR^3)\:\textrm{ and compactly supported by }\:\ovl\earth
$$
for $i,j,k,l\in\{1,2,3\}$, whence by (\ref{eq:Lambda}) this is also true for $\La^{T^0}$,
$$
\Phi^0\in Y^\infty(\RR^3),\quad \Phi^1\in H^1(\earthfsc\times I^\circ),\quad 
u\in H^1_{\Sig^\fxs}(\earthfsc\times I^\circ)^3
\:\textrm{ with }\:\supp(u_\tim)=\ovl\earth\quad(t\in I),
$$
and $F^s\in L^2(I^\circ,H^1(\RR^3))^3$. Furthermore we assume that
$\Sig^\fxs$ locally is a $\mathcal{C}^{1,1}$-surface.
\end{assumption}

We emphasize that these conditions in turn are deduced from the conditions on $\vphi$, $\Phi^s$, $\rho^0$, $U$, and $F^s$
(see Definition \ref{def:nonlin_reg}) which are sufficient to define the full action in the nonlinear setting:
More precisely, the regularity conditions on the reference quantities $\Phi^0$, $T^0$, the material parameters $\rho^0$, $c$,
and the force potential $F^s$ directly follow from the regularity conditions of Definition \ref{def:nonlin_reg}, namely
(\ref{eq:basic_Phi_reg}), (\ref{eq:basic_rho_reg}), (\ref{eq:basic_U_reg}), and (\ref{eq:basic_force_reg}) of $\Phi^s$, $\rho^s$, $U$,
and $F^s$ respectively. The requirement that $u$ is piecewise $H^1$ is a consequence of $\vphi$ being a piecewise
$\Lip$-regular motion of the composite fluid-solid earth, as we have seen in equation (\ref{eq:u_H1}).
For technical reasons we make a slightly stronger additional assumption concerning the regularity of $\Phi^1$
with respect to time. Indeed, the regularity (\ref{eq:basic_Phi_reg}) of $\Phi^s$, only implies continuity of $\Phi^1$ with respect
to time, see also equation (\ref{eq:Phi1_H1}). However the time derivative of $\Phi^1$ does neither appear in the action integral
nor other calculations below. Nonetheless, the higher temporal regularity for $\Phi^1$ will be justified from the existence of solutions \cite{dHHP:15}. 

The conditions summarized in Assumption~\ref{ass:lin_reg} allow us to study the variational problem for the approximated action
(\ref{eq:action_approx}) in a Sobolev space setting. To simplify the notation and draw the attention to the structure of
the action, we combine $u$ and $\Phi^1$ to a single variable and set
\begin{equation}\label{eq:yuPhi1}
y:=(u,\Phi^1)^T\in H
\quad \text{with} \quad H:=H^1_{\Sig^\fxs}(\earthfsc\times I^\circ)^3\times H^1(\earthfsc\times I^\circ),
\end{equation}
where $H^1_{\Sig^\fxs}$ was defined in \eqref{eq:H1Sig}.
The first derivative is given by
\begin{equation}
D y=\left(\arr{cc}{\nabla u & \dot u \\ (\nabla\Phi^1)^T & \dot\Phi^1}\right)\in L^2(\earthfsc\times I^\circ)^{4\times 4}.
\end{equation}

\subsubsection{Quadratic Lagrangians in a Sobolev-space setting}\label{sssec:varQuadratic}

The classical theory of calculus of variations for functionals of the
form $\func(y) = \int_V F(.,y,D y)\:\dvol$ is reviewed in Appendix
\ref{app:var}. Based on Sobolev space techniques we will show that
with an integrand $F$ of the action functional in the form of a
second-order polynomial in $y$ and $D y$, the first variation
$\de\func(y_0,h)$ may still be determined using formula
\eqref{eq:firstvar_scalc} for the weak EL, even with regularity lower
than the classical $\cC^2$-condition. This will allow to extend the
validity of the classical results to functions of lower regularity, in
particular, the EL (\ref{eq:Ham_EL}) and NBC (\ref{eq:Ham_NBC}), as
well as the NIBC (\ref{eq:Ham_IBC}) for composite domains and their
generalization (\ref{eq:Ham_IBCS}) for functionals including a surface
term.

We start with the volume action and consider the functional $\func$ in
(\ref{eq:func}) of $y \colon V \to \RR^m$ for $V$ open (and possibly
unbounded) with integrand of the form of a quadratic polynomial
$F:=F_{[0]}+F_{[1]}+F_{[2]}$, that is
\begin{equation}\label{eq:func012_quad}
\func(y):
=(\func_{[0]}+\func_{[1]}+\func_{[2]})(y)
\end{equation}
with $\displaystyle{\func_{[i]}(y):=\int_VF_{[i]}(x,y(x),Dy(x))\:\dvol(x)}$ for $i=0,1,2$ and
\begin{eqnarray}
F_{[0]}(x,y(x),D y(x))&:=&f_0(x)
\label{eq:func0_quad}\\
F_{[1]}(x,y(x),D y(x))&:=&\lara{f_1(x)|(y(x),D y(x))}
\label{eq:func1_quad}\\
F_{[2]}(x,y(x),D y(x))&:=&\lara{f_2^{0}(x)\cdot y(x)|y(x)}+\lara{f_2^{1}(x)\cdot y(x)|D y(x)}
\nn\\
&&+\lara{f_2^{2}(x)\cdot D y(x)|D y(x)}.
\label{eq:func2_quad}
\end{eqnarray}
Here, by abuse of notation, we suppress explicit switching between
vector and matrix notation and simply write $\lara{.|.}$ for inner
products in $\RR^p$ or $\RR^{q\times r}$, for $p,q,r\in\NN$. We will
also identify $\RR^{m\times n}$ with $\RR^{mn}$.

The constituents of the Lagrangian density in (\ref{eq:func0_quad}) - (\ref{eq:func2_quad}) correspond to \eqref{eq:L0} - \eqref{eq:L2}: $f_0=L_{[0]}''$ and $f_1$ as well as $f_2^{j}$ for $j=0,1,2$ are obtained from $L_{[1]}''$ and $L_{[2]}''$  respectively. Under the regularity conditions 
\begin{equation}\label{eq:Lagreg}
\text{$f_0\in L^1(V)$, $f_1\in L^{2}(V)^{m+ mn}$, 
$f_2^{0}\in L^{\infty}(V)^{m\times m}_\sym$,
$f_2^{1}\in L^{\infty}(V)^{(mn)\times m}$, $f_2^{2}\in L^{\infty}(V)^{(mn)\times (mn)}_\sym$},
\end{equation}
we easily verify
\begin{equation}
\func:H^1(V)^m\to\RR.
\end{equation}
Indeed, integrability of $F_{[i]}(.,y,D y)$ ($i=0,1,2$) on $V$ is a consequence of the inclusions
$H^1\sbs L^2$, $L^{\infty}\cdot L^2\sbs L^2$, $L^2\cdot L^2\sbs L^1$, and $L^{\infty}\cdot L^1\sbs L^1$.
However, application requires the condition
$$
f_2^{0}\in\bigcap_{1\leq p<\infty} L^p(V)^{m\times m}_\sym.
$$
Thus proving integrability of this term is more subtle, since the components of $f_2^{0}$ are in $L^p$
for $1\leq p<\infty$ but not for $p=\infty$. We use Lemma \ref{lem:L246H1} below to show that
$\lara{f_2^{0}y|y}$ is integrable at least on $V\sbs\RR^4$ or $\RR^3$. Based on the Sobolev embedding
theorem it gives a result on $L^p$-multipliers that map $H^1$ to $L^2$, if $V$ is open and satisfies the cone property
(see \cite[p.\ 65]{Adams:75}). Note that this is true if $V$ is a $\Lip$-domain or a $\Lip$-composite domain.

\begin{lemma}[{{\bf Products of $L^p$ and $H^1$-functions via Sobolev embeddings}}]\label{lem:L246H1}
Let $V\sbs\RR^n$ be open and have the cone property. Then $H^1(V)$ is
embedded in $L^p(V)$ with $\norm{y}{L^p(V)}\leq c\norm{y}{H^1(V)}$ for
a constant $c>0$ and $p\leq \frac{2n}{n-2}$.  In particular, the
estimate holds for $p\leq 4$ if $n=4$ and for $p\leq 6$ if
$n=3$. Moreover
\begin{enumerate}[label=(\roman*)]
\item  If $n=4$, $y\in H^1(V)$, and $f\in L^4(V)$, then $f\:y\in L^2(V)$ with
\begin{equation}\label{eq:L4H1}
\norm{fy}{L^2(V)}\leq c\norm{f}{L^4(V)}\norm{y}{H^1(V)}
\end{equation}
and $f\:y^2\in L^1(V)$ with $\norm{fy^2}{L^1(V)}\leq c\norm{f}{L^4(V)}\norm{y}{H^1(V)}^2$.
\item  If $n=3$, $y\in H^1(V)$, and $f\in L^6(V)\cap L^2(V)$, then $f\:y\in L^2(V)$ with
\begin{equation}\label{eq:L62H1}
\norm{fy}{L^2(V)}\leq c^{\frac{3}{4}}(\norm{f}{L^6(V)}^3\norm{f}{L^2(V)})^{\frac{1}{4}}\norm{y}{H^1(V)}
\end{equation}
and $f\:y^2\in L^1(V)$ with
$\norm{fy^2}{L^1(V)}\leq  c^{\frac{3}{4}}(\norm{f}{L^6(V)}^3\norm{f}{L^2(V)})^{\frac{1}{4}}\norm{y}{H^1(V)}^2$.
\end{enumerate}
\end{lemma}

\begin{proof}
The Sobolev embedding result is found for instance in \cite[Theorem
  5.4 (part I case A), p.\ 97]{Adams:75}.  For (i) note that for
$h,g\in L^4(V)$ we have $hg\in L^2(V)$ since by the Cauchy-Schwarz
inequality
\[
   \norm{h g}{L^2(V)}^2=\int_V|h|^2|g|^2
     \leq \norm{|h|^2}{L^2(V)}\norm{|g|^2}{L^2(V)}
               = \norm{h}{L^4(V)}^2\norm{g}{L^4(V)}^2 .
\]
Therefore, $f y\in L^2(V)$ and by Sobolev embedding of $H^1(V)$ in
$L^4(V)$ the inequality $\norm{fy}{L^2(V)}\leq
c\norm{f}{L^4(V)}\norm{y}{H^1(V)}$ holds.  Furthermore, using again
Cauchy-Schwarz and the inequality above, gives
\begin{multline*}
   \norm{fy^2}{L^1(V)}
   = \int |fy| |y|
   \leq \norm{fy}{L^2(V)} \norm{y}{L^2(V)}
\\
   \leq c \norm{f}{L^4(V)} \norm{y}{H^1(V)} \norm{y}{L^2(V)}
   \leq c \norm{f}{L^4(V)} \norm{y}{H^1(V)}^2 ,
\end{multline*}
completing the proof of (i).

In a similar way, (ii) is a consequence of Cauchy-Schwarz and using
the Sobolev embeddings $H^1\sbs L^2$ and $H^1\sbs L^6$ for
$n=3$. Indeed, if $f,g\in L^6(V)\cap L^2(V)$ then $fg\in L^2(V)$ since
\begin{multline*}
   \norm{fg}{L^2(V)}^2=\norm{|f|^2|g|^2}{L^1(V)}
   \leq \norm{|f|^2}{L^2(V)}\norm{|g|^2}{L^2(V)}
      = (\norm{|f|^3|f|}{L^1(V)}
         \norm{|g|^3|g|}{L^1(V)})^{\frac{1}{2}}
\\
   \leq (\norm{|f|^3}{L^2(V)}
         \norm{|f|}{L^2(V)}
         \norm{|g|^3}{L^2(V)}
         \norm{|g|}{L^2(V)})^{\frac{1}{2}}
\\
   = \norm{f}{L^6(V)}^{\frac{3}{2}}
        \norm{f}{L^2(V)}^{\frac{1}{2}}
        \norm{g}{L^6(V)}^{\frac{3}{2}}
        \norm{g}{L^2(V)}^{\frac{1}{2}}
\end{multline*}
by twice applying Cauchy-Schwarz. Combining this result with the
embedding inequalities $\norm{y}{L^2(V)}\leq \norm{y}{H^1(V)}$ and
$\norm{y}{L^6(V)}\leq c\norm{y}{H^1(V)}$ for $y\in H^1(V)$ yields the
assertion (ii).
\end{proof}

We are ready to prove Fr\'echet differentiability of the functional
(\ref{eq:func012_quad}).

\begin{proposition}[{{\bf Fr\'echet derivatives for volume integrals}}]
\label{prop:stat_Frechet}
Let $V\sbs\RR^n$ be open and $F_{[i]}$ ($i=0,1,2$) be given by \eqref{eq:func0_quad} to \eqref{eq:func2_quad}
for coefficients $f_0$, $f_1$, $f_2^0$, $f_2^1$, and $f_2^2$ with regularity \eqref{eq:Lagreg}.
Then the functionals
$\func_{[i]}$ $(i=0,1,2)$ in \eqref{eq:func012_quad}
are Fr\'echet-differentiable on $H^1(V)^m$ with derivatives
\begin{eqnarray}
D\func_{[0]}(y)(h)&=&0,\\[0.2cm]
D\func_{[1]}(y)(h)&=&\int_V\lara{f_1|(h,D h)}\:\dvol=\func_{[1]}(h),\\
D\func_{[2]}(y)(h)
&=&\int_V\Big(\lara{2f_2^{0}\cdot y+(f_2^{1})^T\cdot D y|h}+\lara{f_2^{1}\cdot y+2f_2^{2}\cdot D y|D h}\Big)\:\dvol
\end{eqnarray}
for $y,h\in H^1(V)^m$. Furthermore, instead of assuming boundedness of $f_2^{0}$,
the result also holds if $f_2^{0}\in L^4(V)^{m\times m}_\sym$ for $V\sbs\RR^4$ with cone property
or if $f_2^{0}\in L^2(V)^{m\times m}_\sym\cap L^6(V)^{m\times m}_\sym$ for $V\sbs\RR^3$ with cone property.
\end{proposition}

\begin{proof}
The functional $\func_{[0]}$ is constant and $\func_{[1]}$ linear and
bounded, because by the Cauchy-Schwarz inequality
$|\func_{[1]}(y)|\leq\norm{f_1^0}{L^2}\norm{y}{L^2}+\norm{f_1^1}{L^2}\norm{D
  y}{L^2}\leq c\norm{y}{H^1}$.  Thus their Fr\'echet derivatives are
well defined (see Example \ref{rem:frechet_linquad}).  For
$\func_{[2]}$ we have
\begin{eqnarray}
\func_{[2]}(y+h)&=&
\int_V\Big(\lara{f_2^{0}\cdot (y+h)|y+h}+\lara{f_2^{1}\cdot(y+h)|(D y+D h)}\nn\\
&&\qquad+\lara{f_2^{2}\cdot(D y+D h)|(D y+D h)}\Big)\:\dvol
=\func_{[2]}(y)+D\func_{[2]}(y)(h)+\func_{[2]}(h)\nn
\end{eqnarray}
with the Fr\'echet derivative $D\func_{[2]}(y)(h)$ as claimed above. Thus $D\func_{[2]}(y)$ is linear and continuity
follows from the estimate $|D\func_{[2]}(y)(h)|\leq c(y)\norm{h}{H^1(V)^m}^2$ obtained from the Cauchy-Schwarz inequality
(and possibly using the inequalities (\ref{eq:L4H1}) or (\ref{eq:L62H1}) of Lemma \ref{lem:L246H1} for estimating the term involving $f_2^{0}$).
Similarly, the remainder satisfies $|\func_{[2]}(h)|\leq c\norm{h}{H^1(V)^m}^2$
implying $|\func_{[2]}(h)|/\norm{h}{H^1(V)^m}\to 0$ as $h\to 0$ in
$H^1(V)^m$. Hence, we have also proven Fr\'echet-differentiability of $\func_{[2]}$.
\end{proof}

Thus the  Fr\'echet derivative of $\func$ at $y_0\in H^1(V)^m$ has the same form as formula (\ref{eq:firstvar_calc})
and, using the notation $F_0$ introduced in \eqref{eq:evalLag}, the weak EL read
\begin{eqnarray}
D\func(y_0)(h)&=&\int_V\Big(\sum_{i=1}^m(\d_{y_i}F_0)h_i+\sum_{i=1}^m\sum_{j=1}^n(\d_{\d_jy_i}F_0)\d_jh_i\Big)\:\dvol\nn\\
&=&\int_V\Big(\lara{\d_{y}F_0|h}_{\RR^m}+\lara{\d_{D y}F_0|D h}_{\RR^{m\times n}}\Big)\:\dvol=0
\label{eq:wELquad}
\end{eqnarray}
for all $h\in H^1(V)^m$. For $h\in\cD(V)^m$ we can write these equations as distributional duality
\begin{equation}
\lara{\d_{y}F_0,h}_{\RR^m}+\lara{\d_{D y}F_0,D h}_{\RR^{m\times n}}=0
\end{equation}
implying
\begin{equation}
\div(\d_{D y}F_0)-\d_{y}F_0=0
\end{equation}
in $\cD'(V)^m$. These are the distributional EL which formally
coincide with the classical EL (\ref{eq:Ham_EL}). The distributional
EL follow from the weak EL by the distributional analogue of
integration by parts with the divergence theorem, Lemma
\ref{lem:divthm}, for matrix-valued functions and distributions, Lemma
\ref{lem:div_dist}.

We discuss Fr\'echet differentiability of a quadratic functional of the form of a surface integral
\begin{equation}
\func_S(y):=\int_{S}F_{S}(x,y(x),\wtil{D}y(x))\:\dsurf(x)
\end{equation}
for a closed $\Lip$-hypersurface $S$ in the interior of $V$ and with the special surface Lagrangian
\begin{equation}\label{eq:funcsurf_quad}
F_{S}(x,y(x),\wtil{D} y(x)):=[\lara{g^{0}(x)\cdot y(x)|y(x)}+\lara{g^{1}(x)\cdot y(x)|\wtil{D} y(x)}]_-^+.
\end{equation}
Here $y$ has to be understood in the sense of the trace of functions in $H^1(V)^m$ in $H^{\frac{1}{2}}(S)^m$, which possibly consists
of different values $y^+$ and $y^-$ when $S$ is approached from its two different sides. The surface gradient operator
$\wtil{D}$ \eqref{eq:surfdiff} acts on the different traces of $y$.
If 
\begin{equation}
\text{$g^{0}\in L^{\infty}(S)^{m\times m}_{\textrm{sym}}$ and $g^{1}\in L^{\infty}(S)^{(mn)\times m}$},
\end{equation}
then clearly
\begin{equation}
\func_S:H^1(V)^m\to\RR.
\end{equation}
Indeed we have for the trace $y\in H^{\frac{1}{2}}(S)^m$,
$\wtil{D}y\in H^{-\frac{1}{2}}(S)^m$, and thus for the surface
integrand $F_{S}(.,y,\wtil{D} y)\in L^{\infty}\cdot
H^{\frac{1}{2}}\cdot H^{-\frac{1}{2}}\sbs L^1$ on $S$ (see
\cite{Wloka:82,KiOd:88}).  The same estimates and interchanging the
jump operator $[.]_-^+$ with derivatives allow to establish Fr\'echet
differentiability of $\func_S$ on $H^1(V)^m$.

However, by comparing with the surface Lagrangian
$L_{\Sig^\fxs,[2]}''$ \eqref{eq:Lsurf2} for a composite fluid-solid
earth model, it turns out that the nonzero components of $g^{0}$ and
$g^{1}$ are proportional to $p^0 \wtil{D}\nu$ and $p^0\nu$
respectively, where $p^0=-\nu\cdot T^0\cdot\nu$ for $T^0$ a bounded
two-tensor field, see (\ref{eq:surfpressure}).  If $S$ is a
$\Lip$-hypersurface then $\nu\in L^{\infty}(S)^n$ and $\wtil{D}\nu\in
H^{-1}(S)^{n\times n}$. Thus we must allow $g^{0}\in
H^{-1}(S)^{m\times m}_{\textrm{sym}}$ and $g^{1}\in
L^{\infty}(S)^{(mn)\times m}$.  This implies in particular that for
$y\in H^1(V)^m$ the term $\int_{S}\lara{g^{0}\cdot y|y}\dsurf$ is of
the form $\lara{H^{-1}\cdot H^{\frac{1}{2}},H^{\frac{1}{2}}}$ which
does not make sense as
$\lara{H^{-\frac{1}{2}},H^{\frac{1}{2}}}$-duality.

We may compensate the lower regularity of the coefficients in the
surface integral by requiring higher regularity from
$y$. Specifically, if we restrict $\func_S$ to $y\in H^2(V)^m$ whose
trace is in $H^{\frac{3}{2}}(S)^m\sbs L^2(S)^m$, both terms can be
interpreted by duality on $S$ as follows: $\int_{S}\lara{g^{0}\cdot
  y|y}\dsurf\in \lara{H^{-1}\cdot
  H^{\frac{3}{2}},H^{\frac{3}{2}}}\sbs\lara{H^{-1},H^{1}}$ and
$\int_{S}\lara{g^{1}\cdot y|\wtil{D} y}\dsurf\in \lara{L^{\infty}\cdot
  H^{\frac{3}{2}},H^{\frac{1}{2}}}\sbs\lara{L^2,L^2}$.

Alternatively, we can stay with $\func_S$ defined on $H^1(V)^m$ but
may assume that $S$ is more regular, which will improve the regularity
of the coefficients: If $S$ is a $\cC^{1,1}$-surface then
$\nu\in\Lip(S)^n$ hence $\wtil{D}\nu\in L^{\infty}(S)^{n\times n}$,
implying $g^{0}\in L^{\infty}(S)^{m\times m}_{\textrm{sym}}$ and
$g^{1}\in\Lip(S)^{(mn)\times m}$ with $m=n$.  Consequently,
$\int_{S}\lara{g^{0}\cdot y|y}\dsurf\in \lara{L^{\infty}\cdot
  H^{\frac{1}{2}},H^{\frac{1}{2}}}$ and $\int_{S}\lara{g^{1}\cdot
  y|\wtil{D} y}\dsurf\in \lara{\Lip\cdot H^{\frac{1}{2}},
  H^{-\frac{1}{2}}}$ are defined as
$\lara{H^{-\frac{1}{2}},H^{\frac{1}{2}}}$- and
$\lara{H^{\frac{1}{2}},H^{-\frac{1}{2}}}$-duality respectively. Again,
the same estimates allow to establish Fr\'echet differentiability on
$H^1(V)^m$ (and consequently also on $H^2\sbs H^1$).  However, since
we aim at dealing with functionals which are the sum of a volume and a
surface integral, we will prefer the case in which $S$ is $\cC^{1,1}$
which allows one to use the larger domain $H^1(V)^m$ for the combined
functional.

We summarize our findings about Fr\'echet differentiability of the surface integral $\func_S$ in the following proposition.

\begin{proposition}[{{\bf Fr\'echet derivatives for surface integrals}}]
\label{prop:stat_SFrechet}
Let $V\sbs\RR^n$ be open, $S$ a $\Lip$-hypersurface in $V$, and $F_S$
be given by \eqref{eq:funcsurf_quad} with $g^{0}\in
L^{\infty}(S)^{m\times m}_{\textrm{sym}}$ and $g^{1}\in
L^{\infty}(S)^{(mn)\times m}$.  Then $\mathcal{J}_S$ is
Fr\'echet-differentiable on $H^1(V)^m$ with derivative
\begin{equation}
D\func_S(y)(h)=\int_{S} [\lara{2g^{0}\cdot y + (g^{1})^T\cdot\wtil{D} y|h}
+\lara{g^{1}\cdot y|\wtil{D} h}]_-^+\:\dsurf
\end{equation}
for $y,h\in H^1(V)^m$. Furthermore, instead of assuming boundedness of $g^{0}$, $g^{1}$
the result also holds if $g^{0}\in L^{\infty}(S)^{m\times m}_{\textrm{sym}}$, $g^{1}\in\Lip(S)^{(mn)\times m}$
and $S$ a $\cC^{1,1}$-surface (or for $y\in H^2(V)^m$ if $g^{0}\in H^{-1}(S)^{m\times m}_{\textrm{sym}}$,
$g^{1}\in L^{\infty}(S)^{(mn)\times m}$).
\end{proposition}

Thus the corresponding weak surface EL at $y_0\in H^1(V)^m$ are of the form
\begin{eqnarray}
D\func_S(y_0)(h)&=&\int_{S}\Big(\sum_{i=1}^m(\d_{y_i}F_{S0}) h_i
+\sum_{i=1}^m\sum_{j=1}^n(\d_{\sd_jy_i}F_{S0})\sd_j h_i\Big)(x)\:\dsurf\nn\\
&=&\int_{S} (\lara{\d_{y}F_{S0}|h}_{\RR^m}+\lara{\d_{\wtil{D}y}F_{S0}|\wtil{D} h}_{\RR^{m\times n}})\:\dsurf=0
\label{eq:wELSquad}
\end{eqnarray}
for all $h\in H^1(V)^m$. By the surface version of Lemma \ref{lem:div_dist} and the fundamental lemma 
we deduce the associated distributional surface EL in $H^{-\frac{1}{2}}(S)^m$,
\begin{equation}
\label{eq:dELSquad}
\sdiv(\d_{\wtil{D}y}F_{S0})-\d_{y}F_{S0}=0.
\end{equation}
In total we have thus established Fr\'echet differentiability of a functional $\func:H^1(V)^m\to\RR$
comprising a volume and a surface integral with integrands which are quadratic polynomials in $y$ and $D y$ or $\wtil{D}y$.
Using the Sobolev versions of the divergence theorem (Lemma \ref{lem:divthm} or Lemma \ref{lem:divthmComposite})
we can establish the validity of the EL (\ref{eq:Ham_EL}) in $L^2(V)$ and the NBC (\ref{eq:Ham_NBC}) in
$H^{-\frac{1}{2}}(\d V)$. If 
$\d_{D y}F_0\in H_{\div}\left(\bigcup V_k\right)$, 
we get the NIBC (\ref{eq:Ham_IBC}) in
$H^{-\frac{1}{2}}(\Sig\setminus S)$, or, if in addition $\d_yF_{S0}\in H^{-\frac{1}{2}}(S)$ and $\d_{\wtil{D}y}F_{S0}\in H^{\frac{1}{2}}(S)$,
the NIBC (\ref{eq:Ham_IBCS}) in $H^{-\frac{1}{2}}(S)$.

These results extend to a $\Lip$-composite domain $V$ with interior boundary $\Sig$ (and if $m=n$), where stationarity only is required on 
\begin{equation}
H^1_S(V)^n=\{h\in H^1(V)^n:\:[h]_-^+\cdot\nu=0\:\text{ on }S\sbs\Sig\}
\end{equation} 
introduced in \eqref{eq:H1Sig}, provided that the corresponding normality condition \eqref{eq:normality} holds on $S$:

\begin{corollary}[{{\bf Fr\'echet differentiability, EL, NBC, and NIBC for $\Lip$-composite domains}}]\label{cor:Frechet}
Let $V$ be a $\Lip$-composite domain with interior boundary $\Sig$ such that $S\sbs\Sig$.
Then the results of Proposition \ref{prop:stat_Frechet} for $\func$ and of Proposition \ref{prop:stat_SFrechet} for $\func_S$ hold, if one considers the restricted functional $\func\colon H^1_S(V)^n\to\RR$ or $\func_S\colon H^1_S(V)^n\to\RR$. 
If 
\begin{equation}
\d_{D y}F_0\in H_{\div}(V)^{n\times n}
\end{equation}
then stationarity of $\func$ implies the strong EL \eqref{eq:Ham_EL} which hold in $L^2(V)^n$, the NBC 
\eqref{eq:Ham_NBC} valid in $H^{-\frac{1}{2}}(\d V)^n$,  as well as the  NIBC \eqref{eq:Ham_IBC} valid in $H^{-\frac{1}{2}}(\Sig\setminus S)^n$. 
Moreover, if in addition $\d_{D y}F_0$ 
satisfies the normality condition \eqref{eq:normality} on $S$, stationarity of $\func$ implies the  NIBC \eqref{eq:Ham_IBC} valid in $H^{-\frac{1}{2}}(S)^n$. Stationarity of the combined functional 
\begin{equation}
\func+\func_S\colon H^1_S(V)^n\to\RR
\end{equation} 
leads to the NIBC \eqref{eq:Ham_IBCS} in $H^{-\frac{1}{2}}(S)^n$, provided that
$[\d_{D y} F_0]_-^+\cdot \nu+\sdiv(\d_{\wtil{D}y}F_{S0})-\d_{y}F_{S0}\in H^{-\frac{1}{2}}(S)^n$
and satisfies the normality condition \eqref{eq:normality} on $S$.
\end{corollary}

\begin{proof}
Restricting $\func$ to $H^1_S\sbs H^1$ \eqref{eq:H1Sig} formally does
not alter its derivative, so the Fr\'echet derivatives (the weak EL)
take the same form.  The validity of the EL, NBC, and NIBC in their
strong form follow from the last part of Lemma
\ref{lem:divthmComposite}.
\end{proof}

\subsubsection{Fr\'echet derivatives}\label{sssec:Frechet}

The Lagrangians $L_{[i]}''$ in \eqref{eq:L0}-\eqref{eq:L2} and
$L_{\Sig^\fxs,[i]}''$ in \eqref{eq:Lsurf01}-\eqref{eq:Lsurf2} have the
same structure and regularity of coefficients as the functionals with
quadratic Lagrangians $F_{[i]}$ discussed in \ref{sssec:varQuadratic}
(except of terms arising from the force potential, but which are of
lower order). Consequently, with $V=\earthfsc\times I^\circ$ and
$S=\Sig^\fxs\times I^\circ\sbs\Sig\times I^\circ$ and the regularities
collected in Assumption \ref{ass:lin_reg}, we have integrability of
the integrals defining $\action_{[0]}''$, $\action_{[1]}''$, and
$\action_{[2]}''$ in (\ref{eq:action_approxi}) on $\earthfsc\times
I^\circ$. Fr\'echet differentiability of the action integrals on
$H^1(\earthfsc\times I^\circ)^4$ is established by Propositions
\ref{prop:stat_Frechet} and \ref{prop:stat_SFrechet}. By Corollary
\ref{cor:Frechet} the result also holds on
$H=H^1_{\Sig^\fxs}(\earthfsc\times I^\circ)^3\times
H^1(\earthfsc\times I^\circ)$ \eqref{eq:yuPhi1}.  Hence, for $y, h \in
H$ the Fr\'echet derivatives are given by the formulas
\begin{eqnarray}
&&D\action_{[0]}''(y)(h)=0,\\
&&D\action_{[1]}''(y)(h)=\int_I\int_{\earthfsci}\Big(\lara{\d_yL_{[1]}''|h}_{\RR^4}
+\lara{\d_{D y}L_{[1]}''|D h}_{\RR^{4\times 4}}\Big)\:\dvol\:\mathrm{d}t,
\label{eq:Frechet_action1}\\
&&\nn\\
&&D\action_{[2]}''(y)(h)=\int_I\int_{\earthfsci}\Big(\lara{\d_yL_{[2]}''|h}_{\RR^4}
+\lara{\d_{D y}L_{[2]}''|D h}_{\RR^{4\times 4}}\Big)\:\dvol\:\mathrm{d}t\nn\\
&&\hspace{2cm}
\quad +\int_I\int_{\Sig^{\fxs}} \Big(\lara{\d_y L_{\Sig^\fxs,[2]}''|h}_{\RR^4}
+\lara{\d_{\wtil{D} y} L_{\Sig^\fxs,[2]}''|\wtil{D}h}_{\RR^{4\times 4}}\Big) \dsurf\:\mathrm{d}t.
\label{eq:Frechet_action2}
\end{eqnarray}
The derivatives of the Lagrangians at $y=(u,\Phi^1)^T$ read
\begin{eqnarray}
\d_yL_{[1]}''&=&
\left(\arr{c}{-\rho^0\nabla(\Phi^0+\Psi^s)\\ -\rho^0}\right), \qquad
\d_{D y}L_{[1]}''
=\left(\arr{cc}{-T^0 & \rho^0R_\Om\cdot x\\ -\frac{1}{4\pi G}(\nabla\Phi^0)^T & 0}\right),\\
&&\nn\\
\d_yL_{[2]}''
&=&\left(\arr{c}{-\rho^0\nabla F^s-\rho^0\nabla\nabla(\Phi^0+\Psi^s)\cdot u+\rho^0R_\Om^T\cdot\dot u-\rho^0\nabla\Phi^1 \\ 0}\right),\\
\d_{D y}L_{[2]}''&=&
\left(\arr{cc}{-\La^{T^0}:\nabla u & \rho^0R_\Om\cdot u+\rho^0\dot u\\ (-\rho^0u-\frac{1}{4\pi G}\nabla\Phi^1)^T& 0}\right),\\
\d_yL_{\Sig^\fxs,[2]}''
&=&\left(\arr{c}{-p^0 [\snabla\nu\cdot u+\nu\cdot\snabla u]_-^+\\ 0}\right)
\quad\textrm{and}\quad
\d_{\wtil{D} y}L_{\Sig^\fxs,[2]}''
=\left(\arr{cc}{-p^0\:[\nu u]_-^+ & 0_{3\times 1}\\ 0_{1\times 3} & 0}\right) .
\label{eq:surfcoeff}
\end{eqnarray}
According to Hamilton's principle the system of dynamical equations
for $y\in H$ follow from the stationarity of the regional action
integrals, expressed by the conditions
\begin{equation}\label{eq:stationarity}
D\action_{[1]}''(y)(h)=0\qquad\textrm{and}\qquad D\action_{[2]}''(y)(h)=0
\end{equation}
for all $h\in H$. These equations are the linearized Euler-Lagrange
equations for the composite fluid-solid earth model in their most
general weak form, coinciding with the principle of virtual work
\eqref{eq:Ham_virtual}.

\subsection{Euler-Lagrange equations}\label{ssec:EL}

\subsubsection{Spatially weak formulation}\label{sssec:ELweak}
 
We separate space and time components of derivatives and dot products
and study stationarity on $\cD'(I^\circ, H^1(\earthfsc)^4)$ by
inserting a test function $h(x,t)=z(x)\psi(t)$ for $(x,t)\in
\earthfsc\times I^\circ$, $z\in H^1(\earthfsc)^4$, and
$\psi\in\cD(I^\circ)$.  The (temporally distributional) spatially weak
EL for $\action_{[1]}''$, $\action_{[2]}''$ in $\cD'(I^\circ)$ which
must hold for all $z\in H^1(\earthfsc)^4$ are then given by
\begin{equation}
0=\int_{\earthfsci}\Big(\lara{\d_yL_{[1]}''|z}_{\RR^4}+\lara{\d_{\nabla y}L_{[1]}''|\nabla z}_{\RR^{4\times 3}}\Big)\dvol
\label{eq:wELintLag1}
\end{equation}
and
\begin{eqnarray}
0&=&\frac{d}{dt}\Big(\int_{\earthfsci}\lara{\d_{\dot y}L_{[2]}''|z}_{\RR^{4}}\:\dvol
+\int_{\Sig^{\fxs}}\lara{\d_{\dot y}L_{\Sig^\fxs,[2]}''|z}_{\RR^{4}}\:\dsurf\Big)\nn\\
&&-\int_{\earthfsci}\Big(\lara{\d_{y}L_{[2]}''|z}_{\RR^4}+\lara{\d_{\nabla y}L_{[2]}''|\nabla z}_{\RR^{4\times 3}}\Big)\dvol\nn\\
&&-\int_{\Sig^{\fxs}} \Big(\lara{\d_{y}L_{\Sig^\fxs,[2]}'' |z}_{\RR^4}
+\lara{\d_{\snabla y}L_{\Sig^\fxs,[2]}''|\snabla  z}_{\RR^{4\times 3}}\Big) \dsurf ,
\label{eq:wELintLag2}
\end{eqnarray}
respectively. We further split the spatial test function $z$, in
accordance with the components $u$ and $\Phi^1$ of $y$, into
\begin{equation}
z=(v,w). 
\end{equation}
Then the spatially weak equilibrium and dynamical EL (\ref{eq:wELintLag1}) and (\ref{eq:wELintLag2}) take the form
\begin{equation}
0=\int_{\earthfsci}\Big(-T^0:\nabla v-\rho^0\nabla(\Phi^0+\Psi^s)\cdot v-\frac{1}{4\pi G}\nabla\Phi^0\cdot\nabla w-\rho^0w\Big)\dvol
\label{eq:wELint1}
\end{equation}
and
\begin{eqnarray}
0&=&\frac{d}{dt}\Big(\int_{\earthfsci}\rho^0\dot u\cdot v\:\dvol\Big)
+\int_{\earthfsci}\rho^0\Big(2\Om\times \dot u+\nabla\nabla(\Phi^0+\Psi^s)\cdot u+\nabla\Phi^1+\nabla F^s\Big)\cdot v\:\dvol\nn\\
&&+\int_{\earthfsci}(\La^{T^0}:\nabla u):\nabla v\:\dvol+\int_{\earthfsci}(\rho^0u+\frac{1}{4\pi G}\nabla\Phi^1)\cdot\nabla w\:\dvol\nn\\
&&+\int_{\Sig^{\fxs}}p^0\Big[u\cdot (\snabla \nu)\cdot v
+\nu\cdot(\snabla  u)\cdot v
+\nu\cdot(\snabla  v)\cdot u\Big]_-^+\dsurf
\label{eq:wELint2}
\end{eqnarray}
$$\forall \: v\in H^1_{\Sig^\fxs}(\earthfsc)^3,\quad w\in H^1(\earthfsc).$$ 
To arrive at the surface integrand we have used the symmetry of $\snabla \nu$ and continuity of $p^0$ across $\Sig^\fxs$ (\ref{eq:IBC_FSp0}).
The weak dynamical EL (\ref{eq:wELint2}) are in accordance with the weak form given in
\cite[eq.\ (36)]{Valette:86}. 
Finally, we emphasize again that the spatially weak equilibrium equations (\ref{eq:wELint1}) follow from $D\action_{[1]}''(u,\Phi^1)=0$ and the
spatially weak dynamical equations (\ref{eq:wELint2}) follow from $D\action_{[2]}''(u,\Phi^1)=0$ on $\cD'(I^\circ, H^1(\earthfsc)^4)$.

We observe that including already at this stage the fluid-solid version \eqref{eq:tractionIBC}  of the linear dynamical slip condition \eqref{eq:linCauchyIBC} in Section \ref{444} would allow to apply \cite[Lemma 3.1]{dHHP:15} and to write the surface integrals exactly as in \cite[Equation (4.1)]{dHHP:15}, that is, 
\begin{multline*}
  \int_{\Sig^{\fxs}}p^0\Big[u\cdot (\snabla \nu)\cdot v
+\nu\cdot(\snabla  u)\cdot v
+\nu\cdot(\snabla  v)\cdot u\Big]_-^+\dsurf =\\
  \int_{\Sig^{\fxs}}p^0 \Big( - [u\cdot (\snabla \nu)\cdot v]^+_-
+  [u]^+_- \cdot \snabla  (\nu\cdot v)
+ [v]^+_- \cdot \snabla  (\nu\cdot u) \Big) \dsurf.
\end{multline*}
We note that the derivation of \eqref{eq:tractionIBC} in Section \ref{sec:5} is independent of the above observation.

\begin{remark}[{{\bf Geometric formulation}}]
We recall the geometric invariant form
\eqref{eq:Lagrangian_uncon_geom} of the Lagrangian density in the
nonlinear theory, with similar adaptations for the linearized theory
(with $g$ denoting the Riemannian metric in the reference
configuration). In particular, the spatially weak EL
\eqref{eq:wELint2} in the setting of Riemannian manifolds will read
\begin{eqnarray}
0&=&\frac{d}{dt}\Big(\int_{\earthfsci} g(\rho^0\dot u,
v)\:\dvol\Big)\nn\\ && +\int_{\earthfsci} \left( g(\rho^0
({\operatorname{Cor}}(\dot u) +\nabla\Phi^1+\nabla F^s),v) +
\operatorname{Hess}(\Phi^0+\Psi^s)(u,v) \right)
\:\dvol\nn\\ &&+\int_{\earthfsci}(\La^{T^0}:\nabla u):\nabla
v\:\dvol+\int_{\earthfsci} g(\rho^0u+\frac{1}{4\pi
  G}\nabla\Phi^1,\nabla
w)\:\dvol\nn\\ &&+\int_{\Sig^{\fxs}}p^0\Big[\snabla \nu (u,v) +
  \snabla u (\nu,v) + \snabla v (\nu,u)\Big]_-^+\dsurf,
\end{eqnarray}
where $\dvol$ and $\dsurf$ denote the Riemannian volume and surface
forms (see Remark \ref{geom1}), $\nabla$ and $\snabla$ are the metric
gradients and induced metric surface gradients (that is, $\snabla$ is
the gradient with respect to the metric induced on the surface as in
Remark \ref{geom1}; a coordinate expression can be read off from
Appendix \ref{app:surf}), $\operatorname{Cor}(\dot u)$ denotes
Coriolis acceleration (\cite[Box 3.1 in Chapter 2]{MaHu:83rep}),
$\operatorname{Hess}$ denotes the Hessian, that is, the second
covariant derivative, of a scalar function (\cite[Chapter 3,
  Definition 48]{ONeill:83}), $A:B$ denoting contraction of tensors
$A$ and $B$ (natural in case of contra- and covariant indices -- as
with $(\Lambda^{T^0} : \nabla u)_i^j = (\La^{T^0})_{i k}^{j l} \, \d_l
u^k$ -- and involving metric coefficients otherwise \cite[Chapter 2,
  pages 40-42, and Chapter 3, pages 81-84]{ONeill:83}), the surface
derivatives $\snabla \nu$, $\snabla v$, $\snabla u$ act as
two-tensors, the test function $w$ belongs to the Riemannian version
of $H^1(\earthfsc)$ (cf.\ \cite{Hebey:96}), and the test vector field
$v$ is an $H^1$ section of the tangent bundle $T \earthfsc$ with
additional continuity in directions normal to $\Sig^{\fxs}$. As for
the geometrically invariant tensorial definition of $\La^{T^0}$, we
note that Equation \eqref{eq:elastenergyapprox} implies $\La^{T^0} =
\rho^0 \frac{\d^2 U}{\d (\nabla u)^2}$, that is, in coordinates
$(\La^{T^0})_{i k}^{j l} = \rho^0 \frac{\d^2 U}{\d (\d_j u^i) \d (\d_l
  u^k)}$ (in analogy to \cite[Chapter 3, Proposition
  4.4(b)]{MaHu:83rep}).
\end{remark}

\subsubsection{Strong formulation}\label{sssec:ELstrong}

Clearly, if the components of $v,w$ are in $\cD(\earthfsc)$, the
surface term vanishes and the spatially weak EL
(\ref{eq:wELint1}), (\ref{eq:wELint2}) reduce to the strong, that is,
distributional EL (\ref{eq:staticeq}), (\ref{eq:Phi0_Poisson}) and
(\ref{eq:eqm}), (\ref{eq:Phi1_Poisson}). The strong form of the EL is
obtained by applying the distributional integration by parts formula
\eqref{eq:IBC_calc_dist} to the general weak EL with
$h\in\cD(\earthfsc\times I^\circ)^4$. The component of the equations
$D\action_{[1]}''(u,\Phi^1)=0$ that corresponds to variations with
respect to $u$ is the static equilibrium equation
\begin{equation}\label{eq:staticeq}
\rho^0\nabla(\Phi^0+\Psi^s)-\nabla\cdot T^0=0
\end{equation}
and variation with respect to $\Phi^1$ yields equilibrium Poisson's equation
\begin{equation}\label{eq:Phi0_Poisson}
\triangle\Phi^0=4\pi G\rho^0.
\end{equation}
The dynamical equations $D\action_{[2]}''(u,\Phi^1)=0$ consist of the
linear equation of motion
\begin{equation}\label{eq:eqm}
\rho^0\left(\ddot u+2\Om\times\dot u\:+\:\nabla\nabla(\Phi^0+\Psi^s)\cdot u+\nabla\Phi^1\right)
-\nabla\cdot(\La^{T^0}:\nabla u)=\rho^0\nabla F^s
\end{equation}
when varying with respect to $u$. Here we have inserted
$R_\Om\cdot\dot u=\Om\times\dot u$. Stationarity with respect to
variations of $\Phi^1$ implies Poisson's equation for $\Phi^1$
\begin{equation}\label{eq:Phi1_Poisson}
\triangle\Phi^1=-4\pi G\nabla\cdot(\rho^0u).
\end{equation}
Alternatively, the system (\ref{eq:staticeq}) to (\ref{eq:Phi1_Poisson}) can be written in index notation $(i=1,2,3)$:
\begin{eqnarray*}
\rho^0\:\d_i(\Phi^0+\Psi^s)-\d_jT_{ij}^0&=&0 ,\\
\d_k^2\Phi^0&=&4\pi G\rho^0 ,\\
\rho^0(\ddot u_i+2\epsilon_{ijk}\Omega_j\dot u_k+u_j\d_j\d_i(\Phi^0+\Psi^s)+\d_i\Phi^1)
-\d_j(\La^{T^0}_{ijkl}(\d_lu_k))&=&\rho^0\d_i F^s ,\\
\d_k^2\Phi^1&=&-4\pi G\:\d_j(\rho^0u_j) .
\end{eqnarray*}
Thus on $\earthfsc$ the state of the earth model of material parameters $\rho^0$, $\La^{T^0}$ and prestress
$T^0$ is given by the solution $(u,\Phi^1)$ of the dynamical equations (\ref{eq:eqm}) and (\ref{eq:Phi1_Poisson}) with equilibrium fields
$\Phi^0$ and $T^0$ satisfying the equilibrium equations (\ref{eq:staticeq}) and (\ref{eq:Phi0_Poisson}). More precisely,
the equilibrium and dynamical equations hold in the sense of $\cD'(\earthfsc)$ and $\cD'(\earthfsc\times I^\circ)$
respectively. In particular, outside the earth the only nontrivial equations are the two Laplace equations
\begin{equation}
\triangle\Phi^0=0\qquad\textrm{and}\qquad\triangle\Phi^1=0
\end{equation}
valid in $\cD'(\earth^\cpl)$ and $\cD'(\earth^\cpl\times I^\circ)$
respectively.

As was discussed in Section \ref{sssec:varQuadratic}
the strong form of the EL, which holds in $\cD'(\earthfsc\times I^\circ)^4$, formally coincides with the classical EL \eqref{eq:Ham_EL}:
\begin{equation}
\d_t(\d_{\dot y}L)+\nabla\cdot(\d_{\nabla y}L)-\d_yL=0
\quad\text{in}\quad \earthfsc\times I
\end{equation}
with 
$y=(u,\Phi^1)^T$ as in \eqref{eq:yuPhi1}. Under the higher regularity conditions 
$\d_{(\dot u, \nabla u)}L\in H_{\div}(\earthfsc\times I^\circ)^4$, 
as was established in Corollary \ref{cor:Frechet}, these dynamical equations are valid in 
$L^2(\earthfsc\times I^\circ)^4$.
Application to $L=L''_{[1]}$ \eqref{eq:L1} leads to the equilibrium equations (\ref{eq:staticeq}) and (\ref{eq:Phi0_Poisson}).
For $L=L''_{[2]}$ \eqref{eq:L2} the EL consist of the dynamical equations (\ref{eq:eqm}) and (\ref{eq:Phi1_Poisson}), which take the form \cite[(3.169) and (3.192)]{DaTr:98}
\begin{equation}
\d_t(\d_{\dot u}L''_{[2]})+\nabla\cdot(\d_{\nabla u}L''_{[2]})-\d_uL''_{[2]}=0
\quad\text{and}\quad
\nabla\cdot(\d_{\nabla \Phi^1}L''_{[2]})=0.
\end{equation}
We announce that the dynamical exterior boundary conditions, which hold in $\cD'(I^\circ,H^{-\frac{1}{2}}(\d\earth))^4$ if $\d_{(\nabla u,\nabla\Phi^1)}L''_{[2]}$ is of $H_{\div}$-regularity, will coincide with the classical NBC \eqref{eq:Ham_NBC}
\begin{equation} 
(\d_{\nabla y}L)\cdot\nu=0
\quad\textrm{on}\quad\d \earth\times I.
\end{equation}
The dynamical interface conditions, which hold in $\cD'(I^\circ,H^{-\frac{1}{2}}(\Sig^\fxs))^4$, are equivalent to the NIBC \eqref{eq:Ham_IBCS},
\begin{equation} 
\d_{y}{L_{\Sig^\fxs}}-\d_t(\d_{\dot{y}}{L_{\Sig^\fxs}})-\snabla\cdot(\d_{\snabla y}{L_{\Sig^\fxs}})-[\d_{\nabla y}L]_-^{+}\cdot\nu=0
\quad\textrm{on}\quad \Sig^\fxs\times I.
\end{equation}
A detailed discussion of these boundary and interface conditions is given in Subsection \ref{ssec:IBC}.

\subsubsection{Hydrostatic assumption}

Since the second-order Lagrangian density $L_{[2],\text{hyd}}''$
\eqref{eq:L2hyd} is a special case of $L_{[2]}''$ \eqref{eq:L2} the
EL, NBC, and NIBC also hold within the hydrostatic setting and under
the same regularity conditions. In particular, the EL giving the
dynamical equations read \cite[(3.270) and (3.276)]{DaTr:98}
\begin{equation}
\d_t(\d_{\dot u}L''_{[2],\text{hyd}})+\nabla\cdot(\d_{\nabla u}L''_{[2],\text{hyd}})-\d_uL''_{[2],\text{hyd}}=0
\quad\text{and}\quad
\nabla\cdot(\d_{\nabla \Phi^1}L''_{[2],\text{hyd}})=0,
\end{equation}
both valid in $\earthfsc\times I$. The contributions involving $\Phi^1$ remain unchanged by adopting the hydrostatic assumption. In the hydrostatic case the equation of motion \eqref{eq:eqm} thus reduces to \cite[(3.259)]{DaTr:98}
\begin{equation}\label{eq:eqmhyd}
\rho^0\left(\ddot u+2\Om\times\dot u+\nabla\Phi^1\right)
-\left(\nabla\cdot(\rho^0 u)\right)\nabla(\Phi^0+\Psi^s)
+\nabla\left(\rho^0 \nabla(\Phi^0+\Psi^s)\cdot u\right)
-\nabla\cdot(\Ga:\nabla u)=\rho^0\nabla F^s.
\end{equation}
This is equivalent to the equation of motion expressed in terms of $T^1$ \cite[(3.56)]{DaTr:98}
\begin{equation}\label{eq:eqmT1}
\rho^0\left(\ddot u+2\Om\times\dot u+\nabla\Phi^1\right)
-\left(\nabla\cdot(\rho^0 u)\right)\nabla(\Phi^0+\Psi^s)
+\nabla\cdot(\nabla T^0\cdot  u)
-\nabla\cdot(\Upsilon^{T^0}:\nabla u)=\rho^0\nabla F^s
\end{equation}
restricted to zero deviatoric initial stress, where we recall from \eqref{eq:TL1} that $T^1=\Upsilon^{T^0}:\nabla u=\Gamma : \varepsilon$  if $T^0_\text{dev}=0$. 
The fact that, correct to first order, \eqref{eq:eqmT1} coincides with the equation of motion \eqref{eq:eqm} that involves $T^{\PK 1}=\La^{T^0}:\nabla u$ can be seen from the identity \cite[(3.61)]{DaTr:98}
\begin{equation}
\nabla\cdot T^{1}-\nabla\cdot (\nabla T^0\cdot u)+\left(\nabla\cdot(\rho^0 u)\right)\nabla(\Phi^0+\Psi^s)
\approx_1 \nabla\cdot T^{\PK 1}- \rho^0\nabla\nabla(\Phi^0+\Psi^s)\cdot u,
\end{equation}
which is a consequence of $T^{1}\approx_1 T^{\PK 1}+T^0\cdot(\nabla u)^T-T^0(\nabla\cdot u)$ \eqref{eq:TPK1T1}, the static equilibrium equation $\nabla\cdot T^0=\rho^0\nabla(\Phi^0+\Psi^s)$ \eqref{eq:staticeq}, and the product rule.

The NBC are \cite[(3.271)]{DaTr:98}
\begin{equation}
(\d_{\nabla u}L''_{[2],\text{hyd}})\cdot\nu=0
\end{equation}
on $\d \earth\times I$. As $L''_{\Sig^\fxs,[2],\text{hyd}}=0$ the NIBC reduce to \cite[(3.272)]{DaTr:98}
\begin{equation} 
[\d_{\nabla u}L''_{[2],\text{hyd}}]_-^{+}\cdot\nu=0
\end{equation}
on $\Sig^\fxs\times I$. The last two equations coincide with the
dynamical exterior boundary condition $T^1\cdot\nu=0$ and the
dynamical interface condition $\left[T^1\right]_-^+\cdot\nu=0$ given
in \eqref{eq:BChyd}.

\subsection{Boundary and interface conditions}\label{ssec:IBC}

Under higher regularity conditions on the material parameters
$\rho^0$, $\La^{T^0}$ and equilibrium fields $T^0$ and $\Phi^0$ one
can derive the NBC and the NIBC from the stationarity
(\ref{eq:wELint1}) and (\ref{eq:wELint2}) of the action integrals.
Let $S'$ be an interior $\Lip$-regular surface in $\earthfsc$
separating it into connected components $V_k$.  In particular, we can
choose
\begin{equation}
S'=\Sig\cup\d\earth=\Sig^\sxs\cup\Sig^\fxs\cup \Sig^\fxf\cup\d\earth
\end{equation} 
such that $\Sig^\fxs\sbs S'$. We apply Lemma
\ref{lem:divthmComposite} for integration by parts on composite
domains. If the test functions $v$ and $w$ are in $H^1$ on
$\earthfsc$, they also are in $H_\div(\bigcup_kV_k)$ (the more general
case where testfunctions are restricted to $v\in H^1_{\Sig^\fxs}$
\eqref{eq:H1Sig} is covered by Corollary \ref{cor:Frechet}, leading to
the normality conditions \eqref{eq:norT0} and \eqref{eq:nortau} given
below). Consequently, if the components of $T^0$ and $\nabla\Phi^0$
are in $H_\div(\bigcup_kV_k)$ (which follows from solvability of the
static equilibrium equation \eqref{eq:staticeq}, see Section
\ref{ssec:equilibrium}), the weak EL (\ref{eq:wELint1}) imply the
equilibrium NIBC
\begin{equation}
[T^0]_{-}^{+}\cdot\nu=0
\qquad\textrm{and}\qquad
[\nabla\Phi^0]_{-}^{+}\cdot\nu=0\qquad\textrm{on}\quad S'
\end{equation}
and, as a special case, the equilibrium NBC
\begin{equation}
T^0\cdot\nu=0\qquad\textrm{on}\quad\d\earth.
\end{equation}
Since there is no first-order contribution from the surface action, slip at fluid-solid boundaries does not affect the
equilibrium equations. 

If the components of the first-order perturbation of Piola-Kirchhoff stress tensor $T^{\PK 1}=\La^{T^0}:\nabla u$
and $\rho^0u+\frac{1}{4\pi G}\nabla\Phi^1$  are in $H_\div(\bigcup_kV_k)$, then (\ref{eq:wELint2}) gives the dynamical NIBC
\begin{equation}
[T^{\PK 1}]_{-}^{+}\cdot\nu=0\qquad\textrm{on}\quad S'\setminus\Sig^\fxs
\qquad\text{and}
\qquad 
[4\pi G\rho^0u+\nabla\Phi^1]_{-}^{+}\cdot\nu=0\qquad\textrm{on}\quad S'
\end{equation}
and, as a special case, the dynamical NBC
\begin{equation}
T^{\PK 1}\cdot\nu=0\qquad\textrm{on}\quad\d\earth.
\end{equation}
On $\Sig^\fxs$ the surface term in (\ref{eq:wELint2}) gives an additional contribution to the dynamical NIBC, as is already clear from \eqref{eq:Frechet_action2}  or \eqref{eq:wELintLag2}. 
We evaluate the NIBC \eqref{eq:Ham_IBCS} for the second-order Lagrangians \eqref{eq:Lsurf2},
\begin{equation}
\d_t(\d_{\dot{ u}}{L_{\Sig^\fxs,[2]}''})+\snabla\cdot(\d_{\snabla u}{L_{\Sig^\fxs,[2]}''})-\d_{u}{L_{\Sig^\fxs,[2]}''}+[\d_{\nabla u}L_{[2]}'']_-^+\cdot \nu=0
\end{equation}
valid in $H^{-\frac{1}{2}}(\Sig^\fxs)$. Since 
\begin{eqnarray}
\d_{\nabla u}L_{[2]}''&=&-T^{\PK 1} ,\nn\\[0.25cm] 
\d_{\dot{ u}}{L_{\Sig^\fxs,[2]}''}&=&0 ,\nn\\ 
\d_{u}{L_{\Sig^\fxs,[2]}''}&=&-p^0\left[\nu\cdot\snabla u+u\cdot \snabla\nu\right]_-^+
=-p^0\nu\cdot\left[\snabla u\right]_-^+-p^0\snabla\nu\cdot [u]_-^+ ,\nn\\ &&\nn\\
\snabla\cdot(\d_{\snabla u}{L_{\Sig^\fxs,[2]}''})&=&\snabla\cdot(-p^0\nu [u]_-^+)
=-\nu\snabla\cdot(p^0 [u]_-^+)- p^0\snabla\nu\cdot[u]_-^+ ,
\end{eqnarray}
 the interface condition is (the terms $p^0\snabla\nu\cdot[u]_-^+$ cancel)
$$-\nu\snabla\cdot(p^0 [u]_-^+)+p^0\nu\cdot\left[\snabla u\right]_-^+ -[T^{\PK 1}]_-^+\cdot \nu=0,$$ that is,
\begin{equation}
[T^{\PK 1}]_-^+\cdot \nu+\nu\snabla\cdot(p^0 [u]_-^+)-p^0\nu\cdot\left[\snabla u\right]_-^+ =0.
\end{equation}
Finally, as $[u]_-^+\cdot\nu=0$, \eqref{eq:Jswap} allows us to pull out the jump operator: 
\begin{equation}
\left[ T^{\PK 1}\cdot \nu+\nu\snabla\cdot(p^0 u)-p^0\nu\cdot\snabla u\right]_-^+ =0.
\end{equation} 
Thus the dynamical traction NIBC on $\Sig^\fxs$ finally read
\begin{equation}\label{eq:tractionIBC}
[\tau^{\PK 1}]_-^+=0\qquad\textrm{on}\quad\Sig^\fxs
\end{equation}
for the modified surface traction vector
\begin{equation}
\tau^{\PK 1}:=T^{\PK 1}\cdot\nu+\nu\snabla \cdot(p^0u)-p^0\nu\cdot(\snabla  u).
\end{equation}
This is a special case of \eqref{eq:linCauchyIBC} and coincides with \cite[eq.\ (3.80)]{DaTr:98}.

We discuss the fluid-solid normality condition \eqref{eq:normality}, $f \cdot \nu =(\nu\cdot f\cdot\nu)\nu$ on $\Sig^\fxs$, which by Corollary \ref{cor:Frechet} is
required to obtain the NIBC \eqref{eq:Ham_IBC} or \eqref{eq:Ham_IBCS} by application of  Lemma \ref{lem:divthmComposite} to the weak EL. For a volume action of the form $\func=\int_V F\dvol$ we have  $f=\d_{D y} F$.  Consequently normality corresponds to the dynamical interface condition \eqref{eq:surfpressure} for prestress  
\begin{equation}\label{eq:norT0}
T^0\cdot \nu=(\nu\cdot T^0\cdot \nu)\nu\quad\text{on}\quad\Sig^\fxs
\end{equation} 
(as well as to \eqref{eq:BCnorhyd}, $T^1\cdot \nu=(\nu\cdot T^1\cdot \nu)\nu$  on $\Sig^\fxs$ in the hydrostatic case).
In the presence of a surface action, that is, if the action integral is of the form $\int_V F\dvol+\int_S F_S\dsurf$ with $S=\Sig^\fxs$, we have $f \cdot \nu =[\d_{D y} F]_-^+\cdot \nu+\sdiv(\d_{\wtil{D}y}F_{S})-\d_{y}F_{S}$. For the composite fluid-solid earth model only the second-order action possesses a surface term. Thus normality corresponds to the following dynamical interface condition for the modified surface traction vector:
%
%
\begin{equation}\label{eq:nortau}
   \tau^{\PK1} = (\tau^{\PK1} \cdot \nu) \nu
                       \quad\text{on}\quad\Sig^\fxs .
\end{equation}
Together, \eqref{eq:surfpressure} and \eqref{eq:nortau} are the linearization of the normality condition  $T^s\cdot \nu=(\nu\cdot T^s\cdot \nu)\nu$ \eqref{eq:Cauchynofriction}, on slipping fluid-solid interfaces.

The interior boundary conditions for equilibrium quantities and perturbations are summarized in the following table
(using index notation for $1\leq i\leq 3$  and summation convention):
\begin{center}
\begin{tabular}{|l|l|l|}
\hline
boundaries &  conditions for      & conditions for \\
                   &  equilibrium fields & perturbations\\
\hline \hline
all boundaries $S\sbs\earth$ & & \\
& $[\Phi^0]_-^+=0$ & $[\Phi^1]_-^+=0$ \\
& $[\partial_j\Phi^0]_{-}^{+}\nu_j=0$ & $[\partial_j\Phi^1+4\pi G\rho^0u_j] _-^+\nu_j=0$\\
& & \\ \hline
solid-solid boundaries $\Sig^{\sxs}$ & & \\
& $\qquad - \qquad$ & $ [u_i]_-^{+}=0$ \\
& $[T_{ij}^0]_-^{+}\nu_j=0$ & $[T_{ij}^{\PK 1}]_-^{+}\nu_j=0$ \\
& & \\ \hline
fluid-solid boundaries $\Sig^\fxs$  & & \\
& $\qquad - \qquad$ & $[u_j]_-^{+}\nu_j=0$ \\
& $[T_{ij}^0]_-^+\nu_j=0$ & $[\tau_i^{\PK 1}]_-^+=0 $ \\
& $T_{ij}^0\nu_j=(T_{kj}^0\nu_k\nu_j)\:\nu_i$
& $\tau_i^{\PK 1}=\tau_j^{\PK 1}\nu_j\nu_i $ \\
& & \\ \hline
free surface $\d\earth$ & & \\
& $T_{ij}^0\nu_j=0$ & $T_{ij}^{\PK 1}\nu_j=0 $ \\
& & \\ \hline
\end{tabular}
\end{center}

\begin{remark}
Not all these conditions follow from the variational principle. For
example, continuity of $\Phi^0$ and $\nabla\Phi^0\cdot\nu$ is clear
since $\Phi^0\in Y^\infty$ by Poisson's equation for $\rho^0\in
L^\infty$. The condition $T^0\cdot\nu=-p^0\nu$ with $p^0=-\nu\cdot
T^0\cdot\nu$ is valid by definition of a fluid, see
\eqref{eq:surfpressure}.
\end{remark}

Under the hydrostatic assumption, the natural interface conditions are
given by $[p^0]_-^+=0$ and $\left[T^1\right]_-^+\cdot\nu=0$ on any
interior boundary $S\sbs B$. The natural boundary conditions read
$p^0=0$ and $T^1\cdot\nu=0$ on $\d\earth$ and normality reduces to
$T^1\cdot \nu=(\nu\cdot T^1\cdot \nu)\nu$ on $\Sig^\fxs$.

Solving the weak EL already implies that the NBC and NIBC
hold. Indeed, as in the equilibrium case, the regularity conditions
required for the formulation of the dynamical NBC and NIBC also follow
from the solvability and regularity of solutions of the weak EL
\eqref{eq:wELint2}. This requires the additional assumption of bounded
and positive $\rho^0$ and $\La^{T^0}$, that are piecewise $\Lip$, see
\cite{dHHP:15}.

\section{Earthquake ruptures}
\label{sec:ruptures}

\newcommand{\dd}{\,\mathrm{d}}
\def\mathbi#1{\textbf{\em #1}}
\newcommand{\superscript}[1]{\ensuremath{^{\textrm{#1}}}}
\newcommand{\etal}{et al.}

\newcommand{\mm}{{-}}
\newcommand{\pp}{{+}}
\newcommand{\bmm}{{(-)}}
\newcommand{\bpp}{{(+)}}
\newcommand{\bpm}{{(\pm)}}
\newcommand{\TT}{{\mathrm{T}}}
\newcommand{\ii}{{(i)}}
\newcommand{\jj}{{(j)}}
\newcommand{\iijj}{{(ij)}}
\newcommand{\jjii}{{(ji)}}
\newcommand{\IM}{{\mathrm{IM}}}
\newcommand{\EX}{{\mathrm{EX}}}
\newcommand{\btb}{{(t)}}
\newcommand{\bDtb}{{(t+\Dt)}}
\newcommand{\ini}{\mathrm{ini}}

\newcommand{\njmp}[1]{{\,[\![#1\,]\!]}}
\newcommand{\njmpN}[1]{{\,[\![#1\,]\!]}_\NN}
\newcommand{\njmpT}[1]{{\,[\![#1\,]\!]}_\TT}
\newcommand{\jmp}[1]{{[#1\,]}}
\newcommand{\jmpT}[1]{{[#1\,]}_\TT}
\newcommand{\DDt}[1]{\frac{d#1}{dt}}
\newcommand{\ddt}[1]{\frac{\partial #1}{\partial t}}
\newcommand{\Norm}[1]{\left\Vert #1 \right\Vert}
\newcommand{\Amp}[1]{\left| #1 \right|}
\newcommand{\Dt}{{\Delta t}}

\newcommand{\Domn}{{\Omega}}
\newcommand{\Domne}{\Omega^{\mathrm{e}}}
\newcommand{\Surf}{{\Sigma}}
\newcommand{\Surft}{\Sigma_t}
\newcommand{\cSurf}{\Surf_\mathrm{c}}
\newcommand{\fSurf}{\Surf_\mathrm{f}}
\newcommand{\Surfe}{\mathcal{T}^{\mathrm{e}}}
\newcommand{\cSurfe}{{\Surfe_\mathrm{c}}}
\newcommand{\fSurfe}{{\Surfe_\mathrm{f}}}
\newcommand{\cfSurfe}{{\Surfe_\mathrm{c,f}}}
\newcommand{\Surfc}{\mathfrak{T}^\mathrm{c}}
\newcommand{\Surff}{\mathfrak{T}^\mathrm{f}}
\newcommand{\DomnIM}{{\Omega}_\mathrm{IM}}
\newcommand{\DomnEX}{{\Omega}_\mathrm{EX}}

\newcommand{\tsC}{\mathbi{C}}
\newcommand{\tsE}{\mathbi{E}}
\newcommand{\tsH}{\mathbi{H}}
\newcommand{\tsS}{{(\tsC \tsE\,)}}
\newcommand{\tsSsup}[1]{{(\tsC \tsE^{#1}\,)}}
\newcommand{\tsT}{{(\tsC \tsH\,)}}

\newcommand{\aSt}{{\tau_f}}
\newcommand{\vex}{\mathbi{x}}
\newcommand{\veu}{\mathbi{u}}
\newcommand{\ven}{\nu}
\newcommand{\venu}{\boldsymbol\nu}
\newcommand{\vev}{\mathbi{v}}
\newcommand{\vevn}{\vev_\NN}
\newcommand{\vevt}{\vev_\TT}
\newcommand{\dvevt}{{\mathbi{V}_\TT}}
\newcommand{\dvt}{{V_\TT}}
\newcommand{\tildvt}{{\widetilde{V}_\TT}}
\newcommand{\avn}{{\tfrac12(\vevn^\ii+\vevn^\jj)}}
\newcommand{\veS}{{\boldsymbol\tau}}
\newcommand{\veSn}{\veS_\NN}
\newcommand{\veSt}{\veS_\TT}
\newcommand{\aveSt}{{\boldsymbol\tau_f}}
\newcommand{\vew}{\mathbi{w}}

\newcommand{\Energy}{{\mathcal{E}}}
\newcommand{\WaveOp}{{\mathcal{L}}}
\newcommand{\FluxOp}{{\mathcal{F}}}
\newcommand{\FricOp}{{\mathcal{G}}}
\newcommand{\BndcOp}{{\mathcal{B}}}
\newcommand{\TintOp}{{\mathcal{R}}}

It is natural to add shear faulting or shear ruptures, incorporating a
friction law, to the spatially weak formulation of the EL
(cf.~\eqref{eq:wELint1}-\eqref{eq:wELint2}), namely, through a fault
boundary integral term that has certain similarities with the
fluid-solid boundary integral term. We discuss the details in this
section. Thus we extend the variational description of
elastic-gravitational deformations to one with earthquakes.

\subsection{The equation of motion as a first-order system in the
            spatially weak formulation}
\label{ssec:firstorder} 

We combine displacement $u$, velocity $\dot{u}$, and the derivative of
the displacement $\nabla u$ into a new set of independent variables,
that is,
\begin{equation}
   z = \begin{pmatrix}z_0\\ z_1\\ z_2 \end{pmatrix} = 
       \begin{pmatrix}u\\ \dot u\\ \nabla u \end{pmatrix}
       \colon\ \earth \times I
                \to \RR^{3} \times \RR^{3} \times \RR^{3\times 3} ,
\end{equation}
and briefly discuss how to turn each of the three variants of linear
equation of motion, classical elasticity and prestressed elasticity
with $T^{\PK 1}$ or $T^1$, into a first-order system.

The equation of motion for classical linearized elasticity on $\earth
\sbs \RR^3$ without prestress is the \textbf{elastic wave equation}
\begin{equation}\label{eq:eqmclassical}
   \rho^0 \ddot u - \nabla \cdot T = 0 ,
\end{equation}
where $T = c : \nabla u$ with $c\colon B\to \RR^{3\times 3\times
  3\times 3}$ the classical stiffness tensor (we may identify all
stress tensors $T = T^1 = T^s = T^{s1} = T^{\PK} = T^{\PK 1}$ in the
linearized setting, since $T^0 = 0$). Equation \eqref{eq:eqmclassical}
is equivalent to the first-order system
\begin{equation}\label{eq:eqmclassical1}
   \dot z = \begin{pmatrix} z_1\\
         \frac{1}{\rho^0} \nabla \cdot (c:z_2)\\
         \nabla z_1 \end{pmatrix} .
\end{equation}
In case of a self-gravitating, uniformly rotating, fluid-solid
composite earth model, the strong formulation of the equation of
motion with $T^{\PK 1} = \Lambda^{T^0} : \nabla u$ \eqref{eq:eqm},
\[
   \rho^0 (\ddot u + 2 \Om \times \dot u
        + \nabla\nabla (\Phi^0 + \Psi^s) \cdot u
            + \nabla S(u)) - \nabla \cdot T^{\PK 1} = 0 ,
\]
can be written as
\begin{equation}
   \dot z = \begin{pmatrix} z_1\\ 
   \frac{1}{\rho^0} \nabla \cdot (\La^{T^0} : z_2)
        - \nabla\nabla (\Phi^0 + \Psi^s) \cdot z_0
            - \nabla S(z_0) - 2 \Om \times z_1\\
   \nabla z_1 \end{pmatrix} .
\end{equation}
Here, we have omitted the external conservative field of force
contribution ($F^s = 0$). The dependence of $\Phi^1$ on $z_0 = u$ is
given in terms of the solution operator, $S$, of Poisson's equation
\eqref{eq:Phi1_Poisson}, $\triangle\Phi^1 = -4\pi G \nabla \cdot
(\rho^0 u)$: $u \mapsto \Phi^1 = S(u)$. Alternatively, we may start
from the equation of motion \eqref{eq:eqmT1} expressed in terms of the
incremental Lagrangian stress $T^1 = \Upsilon^{T^0} : \nabla u$
(cf.~\eqref{eq:TPK1T1})
\[
   \rho^0 (\ddot u + 2 \Om \times \dot u
        + \nabla S(u))
   - (\nabla \cdot (\rho^0 u)) \nabla(\Phi^0 + \Psi^s)
        + \nabla \cdot (\nabla T^0 \cdot u) - \nabla \cdot T^1 = 0 ,
\]
and obtain it in first-order form
\begin{equation}
   \dot z = \begin{pmatrix} z_1\\ 
   \frac{1}{\rho^0} \left(\nabla \cdot (\Upsilon^{T^0} : z_2
        - \nabla T^0 \cdot z_0)
        + (\nabla \cdot (\rho^0 z_0))
                                  \nabla(\Phi^0 + \Psi^s)\right)
        - \nabla S(z_0) - 2 \Om \times z_1\\
   \nabla z_1 \end{pmatrix} .
\end{equation} 
In both cases, the first and third component equations are exactly as
in the classical case \eqref{eq:eqmclassical1}. The second component
equation involving $\dot z_1$ contains a similar second-order term of
the form $\frac{1}{\rho^0}\nabla\cdot(c:z_2)$ with $c$ replaced by
$\La^{T^0}$ or $\Upsilon^{T^0}$ and, in addition, lower-order terms
due to self-gravitation, prestress, and Coriolis acceleration that
only depend on $z_0, z_1$.

The weak formulation of the first-order system can, in principle, be
obtained by taking the $L^2$-inner product of all components with
spatial test functions, $h_0, h_1 \in H^1(\earth)^3$ and $h_2 \in
H^1(\earth)^{3\times 3}$. The first component equation, ${\dot z}_0 =
z_1$, immediately gives
\begin{equation} \label{eq:weaksystem1}
   \frac{d}{d t} \int_{\earth} z_0 \cdot h_0 \:\dvol
                 = \int_{\earth} z_1 \cdot h_0 \:\dvol .
\end{equation}
In the case of classical linearized elasticity, we multiply and
contract the third component equation with $c : h_2$, while assuming
that the components of the elasticity tensor $c$ are piecewise $\Lip$
with respect to the interior boundary $\Sig \cup \d\earth$; we obtain
\begin{multline} \label{eq:weaksystem3}
   \frac{d}{d t} \int_{\earth} z_2 : (c : h_2) \:\dvol
   = \int_{\earth} \nabla z_1 : (c : h_2) \:\dvol
\\
   =-\int_{\earth} z_1 \cdot
            \left(\nabla \cdot (c : h_2)\right) \:\dvol
   - \int_{\Sig \cup \d\earth} [z_1 \cdot (c : h_2)]_-^+ \cdot \nu
            \:\dsurf
   = -\int_{\earth} z_1 \cdot
            \left(\nabla \cdot (c : h_2)\right) \:\dvol
\end{multline}
in view of the usual NBC, even at a fluid-solid boundary\footnote{This
  holds true since integration by parts in (\ref{eq:weaksystem3}) also
  holds on slipping surfaces. This follows by applying
  (\ref{eq:IBC_calc}) (from Lemma~\ref{lem:divthmComposite}) with $h =
  z_1 = \dot u \in H^1_\Sigma$ and $f = c : h_2$ piecewise in
  $H_{\div}$, assuming that $f$ satisfies the normality condition
  (\ref{eq:normality}), $f \cdot \nu = (\nu \cdot f \cdot \nu)
  \nu$.}. These weak equations also hold in the more general earth
model $\earthfsc$ while $\Sigma \cup \d\earth$ coincides with
$\Sig^{\sxs} \cup \Sig^{\fxs} \cup \Sig^{\fxf} \cup \d\earth$ as
before: \eqref{eq:weaksystem1} remains unchanged and in
\eqref{eq:weaksystem3} one just has to replace $c$ by $\Lambda^{T^0}$
or $\Upsilon^{T^0}$ in case of the $T^{\PK 1}$ or the $T^1$
formulation, respectively.

In the case of classical linearized elasticity, the second component
equation yields,
\begin{multline} \label{eq:weaksystem2}
   \frac{d}{d t} \int_{\earth} \rho^0 z_1 \cdot h_1 \:\dvol
   = \int_{\earth} \left(\nabla \cdot (c:z_2)\right) \cdot h_1
                     \:\dvol
\\
   =-\int_{\earth} (c : z_2) : \nabla h_1 \:\dvol
    - \int_{\Sig \cup \d\earth} [h_1 \cdot (c : z_2)]_-^+ \cdot \nu
                     \:\dsurf
    = -\int_{\earth} (c : z_2) : \nabla h_1 \:\dvol
\end{multline}
in view of the usual NBC in the purely elastic case. In the more
general earth model $\earthfsc$, one has to replace $c$ in the volume
integral in \eqref{eq:weaksystem2} by $\Lambda^{T^0}$ or
$\Upsilon^{T^0}$ in case of the $T^{\PK 1}$ or the $T^1$ formulation,
respectively, while adding the corresponding lower-order terms. In the
$T^{\PK 1}$ case, we get
\begin{multline} \label{eq:weaksystem4}
   \frac{d}{d t}
   \int_{\earthfsci} \rho^0 z_1 \cdot h_1 \:\dvol
   = -\int_{\earthfsci} \rho^0
       \left(\nabla\nabla (\Phi^0 + \Psi^s) \cdot z_0
          + \nabla S(z_0) + 2 \Om \times z_1\right) \cdot h_1 \:\dvol
\\
     - \int_{\earthfsci} (\Lambda^{T^0} : z_2) : \nabla h_1 \:\dvol
\\
     - \int_{\Sig^{\fxs}} p^0 \Big(
             -[z_0 \cdot (\snabla \nu) \cdot h_1]^+_-
            + [z_0]^+_- \cdot \snabla (\nu \cdot h_1)
            + [h_1]^+_- \cdot \snabla (\nu \cdot z_0) \Big) \dsurf
\end{multline}
(cf.~(\ref{eq:wELint2})). Using \cite[Lemma 3.1]{dHHP:15}, we find that
\begin{equation}
   \int_{\Sig^{\fxs}} p^0 \Big(
             -[z_0 \cdot (\snabla \nu) \cdot h_1]_-^+
            + [z_0]^+_- \cdot \snabla (\nu \cdot h_1)
            + [h_1]^+_- \cdot \snabla (\nu \cdot z_0) \Big) \dsurf
   = \int_{\Sig^{\fxs}} [(T^{\PK 1} \cdot \nu) \cdot h_1]_-^+ \:\dsurf ,
\end{equation}
while $\Lambda^{T^0} : z_2$ is identified with $T^{\PK 1}$ in the
second integral on the right-hand side.

We let $\fSurf$ be a potential fault plane that yields to slip. (We
let $\fSurf$ be part of $\Sigma$ in a $\Lip$ composite domain.) The
active portion of the rupture $\Surft \subset \fSurf$ dynamically
evolves as a function of time $t$. Incorporating a basic rupture
yields adding a term,
\[
   - \int_{\Surft} [\tau \cdot h_1]_-^+ \:\dsurf ,
\]
to (\ref{eq:weaksystem4}), where $\tau$ is the incremental
traction. That is,
\begin{multline} \label{eq:weaksystem4f}
   \frac{d}{d t}
   \int_{\earthfsci} \rho^0 z_1 \cdot h_1 \:\dvol
   = -\int_{\earthfsci} \rho^0
       \left(\nabla\nabla (\Phi^0 + \Psi^s) \cdot z_0
          + \nabla S(z_0) + 2 \Om \times z_1\right) \cdot h_1 \:\dvol
\\
     - \int_{\earthfsci} (\Lambda^{T^0} : z_2) : \nabla h_1 \:\dvol
\\
     - \int_{\Sig^{\fxs}} p^0 \Big(
             -[z_0 \cdot (\snabla \nu) \cdot h_1]^+_-
            + [z_0]^+_- \cdot \snabla (\nu \cdot h_1)
            + [h_1]^+_- \cdot \snabla (\nu \cdot z_0) \Big) \dsurf
     - \int_{\Surft} [\tau \cdot h_1]_-^+ \:\dsurf .
\end{multline}
We note that $\tau$ needs to belong to $H^{-1/2}(\Surft)^3$, and that
$[h_1]_-^+ \cdot \nu = 0$ on $\Surft$. The formulation in terms of
$T^1$ is analogous. Incorporating slip in (\ref{eq:weaksystem3})
yields
\begin{equation} \label{eq:weaksystem3-slip}
   \frac{d}{d t} \int_{\earth} z_2 : (c : h_2) \:\dvol
   =-\int_{\earth} z_1 \cdot
            \left(\nabla \cdot (c : h_2)\right) \:\dvol
   - \int_{\Surft} [z_1 \cdot (c : h_2)]_-^+ \cdot \nu
            \:\dsurf ,
\end{equation}
where
\begin{equation}
   [z_1 \cdot (c : h_2)]_-^+ \cdot \nu
         = \{ z_1 \} \cdot [(c : h_2) \cdot \nu]_-^+
               + [z_1]_-^+ \cdot \{ (c : h_2) \cdot \nu \} ,\quad
   \{ z_1 \} = \tfrac{1}{2} (z_1^+ + z_1^-) .
\end{equation}
(Here, $h_2$ is not subjected to the normality conditions on the
fault surface.)

\begin{remark}[{{\bf Velocity-strain formulation}}]
Due to the symmetries of the classical elasticity tensor ($c_{ijkl} =
c_{jikl} = c_{ijlk} = c_{klij}$) we have $T = c : \nabla u = c :
\varepsilon$ with the symmetric linearized strain tensor $\varepsilon
= \frac{1}{2} \left(\nabla u + (\nabla u)^T\right)$
\eqref{eq:linstrain}. Thus, in terms of
\begin{equation} 
  \tilde{z} = \begin{pmatrix} z_0\\ z_1\\ \tilde{z}_2 \end{pmatrix}
  = \begin{pmatrix} u\\ \dot u\\ \varepsilon \end{pmatrix}
    \colon \earthfs\times I \to \RR^{3} \times \RR^{3}
                   \times \RR^{3\times 3}_{\text{sym}} ,
\end{equation}
the elastic wave equation \eqref{eq:eqmclassical} may be written as a
first order system in velocity-strain formulation:
\begin{equation}
   \dot{\tilde{z}} = \begin{pmatrix} z_1\\
           \frac{1}{\rho^0}\nabla\cdot(c:\tilde{z}_2) \\ 
   \frac{1}{2}\left(\nabla z_1 + (\nabla z_1)^T\right) \end{pmatrix} .
\end{equation}
The same is possible under the hydrostatic assumption, because $T^1 =
\Ga : \nabla u = \Ga : \varepsilon$, as $\Ga$ possesses the classical
symmetries. Hence, in the hydrostatic earth model, the equation of
motion \eqref{eq:eqmhyd},
\[
   \rho^0 (\ddot u + 2 \Om \times \dot u + \nabla S(u))
   - (\nabla \cdot (\rho^0 u)) \nabla (\Phi^0 + \Psi^s)
           + \nabla (\rho^0 \nabla (\Phi^0 + \Psi^s) \cdot u)
          - \nabla \cdot T^1 = 0 ,
\]
has the velocity-strain formulation
\begin{equation}
   \dot{\tilde{z}} = \begin{pmatrix} z_1\\ 
      \frac{1}{\rho^0} \nabla\cdot(\Ga:\tilde{z}_2)
   + \left(\nabla \cdot (\rho^0 z_0)\right) \nabla(\Phi^0 + \Psi^s)
   - \nabla\Phi^1(z_0) - 2 \Om \times z_1\\
   \frac{1}{2}\left(\nabla z_1 + (\nabla z_1)^T\right)
                     \end{pmatrix} .
\end{equation}
\end{remark}

\begin{remark}[{{\bf Energy considerations}}]
Following \cite[Appendix C]{dHHP:15}, we find that the energy
dissipation is given by
\begin{equation}
   \dot E = -\int_{\Surft} [(\tau + \nu \cdot T^0)
                    \cdot \dot u]_-^+ \:\dsurf .
\end{equation}
Unlike the corresponding integral over a fluid-solid boundary, this
integral does not vanish. The energy released by the rupture follows
from time integration over its duration.
\end{remark}

\noindent
In fact, in general, we need to account for a possible change in
rotation rate in response to the final, static deformation.

\subsection{Rupture dynamics: Dissipation and augmented system}
\label{sec:rupt}

We use the subscript $\,_\TT$, here, for tangential components, such
as the tangential particle velocity,
\[
   z_{1;\TT} = z_1 - (z_1 \cdot \ven) \, \ven
\]
($z_{1;\TT} = z_1^{\|}$; see (\ref{eq:tangpart})). We then introduce
the \textbf{slip rate},
\begin{equation}\label{eq:sliprate}
   V_{\TT} = [z_{1;\TT}]_-^+ = z_{1;\TT}^+ - z_{1;\TT}^- ,
\end{equation}
with its norm denoted by $|V_{\TT}|$. The support of the slip rate
vector coincides with the active portion of the fault, $\Surft$; we
assume that the slip vanishes on the boundary $\partial \Surft$. We
recall that
\[
   \tau = T^{\PK1} \cdot \nu ,\quad
                 T^{\PK1} = \Lambda^{T^0} \colon z_2 .
\]
The boundary conditions (under shear faulting or shear ruptures) on
$\Surft$ are
\begin{equation} \label{eq:fbdc}
   \ven \cdot (z_1^+ - z_1^-) = 0 ,
\end{equation}
hence, $V_{\TT} = [z_1]_-^+$, and
\begin{equation*}
   [T^{\PK 1} \cdot \ven
         - \snabla \cdot ((T^0 \cdot \ven) z_0)]_-^+ = 0
   \quad\text{on}\quad \Surft
\end{equation*}
(cf.~(\ref{eq:linCauchyIBC})). These are supplemented by
\begin{equation}
\label{eq:rupture condition}
   \sigma_\TT = \aSt ,
\end{equation}
where
\[
   \sigma = T^{\PK 1} \cdot \ven
                   - \snabla \cdot ((T^0 \cdot \ven) z_0)
            + T^0 \cdot \ven
   \quad\text{on}\quad \Surft
\]
and $\aSt$ signifies the \textbf{frictional force}. The
\textbf{compressive normal stress} applied to the fault plane is
positive and is given by
\[
   \sigma_N = -\ven \cdot \sigma
   \quad\text{on}\quad \Surft .
\]
The changes in $T^0$ during rupture are typically assumed to be
negligible. The frictional force satisfies the relation \cite{Day2005}
\begin{equation}
\label{eq:friction direction}
   |\tau_f| \, V_\TT - \tau_f \, |V_\TT| = 0 ;
\end{equation}
that is, the frictional force is parallel to the slip rate vector.

The dependency of friction on slip rate, compressive normal stress and
a ``state'' variable has been studied extensively \cite{Dieterich1979,
  Ruina1981, Ruina1983, Rice1983, Rice1983a}. Laboratory experiments
have been conducted on various types of rocks or fault gouge layers
over a wide range of slip rates and compressive normal stresses. Many
studies have been carried out examining the effect of variable normal
stress on rock friction \cite{Linker1992, Prakash1998, Richardson1999,
  Bureau2000}. The stability of a rate- and state-dependent friction
law can be analyzed through a linearized perturbation of steady state
sliding. We briefly mention the extensive work of Rice \textit{et
  al}. \cite{Rice2001} and indicate how to integrate it in our
variational framework.

\subsubsection{General assumptions}

We summarize several features of the frictional force that have been
identified from experimental observations:
\begin{enumerate}[label=($\mathfrak{a}.\arabic*$)]
\item the instantaneous frictional force is positively related to the
  amplitude of normal stress;\label{asm:A1}
\item the instantaneous frictional force is positively related to the
  magnitude of slip rate;\label{asm:A2}
\item the long-term variation of frictional force is accumulatively
  affected by the history of slip rate and normal
  stress;\label{asm:A3}
\item a steady-state frictional force can be monotonically approached
  with any given constant slip rate and normal stress.\label{asm:A4}
\end{enumerate}

A universal representation capturing the features above can be found
in \cite[pp.1869-1870]{Rice2001}. The frictional force follows a
general nonlinear constitutive relation
\begin{equation}\label{eq:friction law}
   |\aSt| = F(\sigma_N,|\dvt|,\psi) .
\end{equation}
The state variable $\psi$ is introduced to measure the average contact
maturity, which evolves with time following the nonlinear ordinary
differential,
\begin{equation}\label{eq:state ODE}
   \dot\psi = -G(\sigma_N,\dot\sigma_N,|\dvt|,\psi) .
\end{equation}
From \ref{asm:A1}-\ref{asm:A3} it follows that
\[
   \frac{\partial F}{\partial \sigma_N} > 0 ,
\quad
   \frac{\partial F}{\partial |\dvt|} > 0
\quad\text{and}\quad
   \frac{\partial F}{\partial \psi} > 0 .
\]
These dependencies capture not only the long time effects of $|\dvt|$,
but also the memory dependency of variation in $\sigma_N$. Following
\ref{asm:A4}, we define the steady-state value, $\psi_\std$, of $\psi$
which satisfies
\begin{equation}\label{eq:steady stat}
   G(\sigma_N,0,|\dvt|,\psi_\std(\sigma_N,|\dvt|)) = 0 ,
\end{equation}
under the constraint that $|\dot\dvt| = 0$ and $\dot\sigma_N = 0$. The
corresponding frictional force is then given by
\begin{equation}\label{eq:steady force}
   |\aSt|_\std(\sigma_N,|\dvt|)
          = F(\sigma_N,|\dvt|,\psi_\std(\sigma_N,|\dvt|)) .
\end{equation}

\medskip\medskip

\noindent
The system of equations is then given by (\ref{eq:weaksystem1}),
(\ref{eq:weaksystem3-slip}) and (\ref{eq:weaksystem4f}) augmented with
(\ref{eq:sliprate}), (\ref{eq:fbdc}), (\ref{eq:friction direction}),
(\ref{eq:friction law}) and (\ref{eq:state ODE}). One can view this as
a linear hyperbolic system with nonlinear feedback to (interior)
Neumann boundaries. We need to impose initial conditions at $t = t_0 =
0$ for
\[
   z_0 = u ,\ z_1 = \dot u,\ z_2 = \nabla u\ \text{and}\ \psi
\]
which imply initial values for
\[
   V_\TT ,\ \sigma_N\ \text{and}\ \aSt .
\]
In fact, the initial values for $z_0$, $z_1$ and, hence, $V_\TT$ will
be zero, while the initial value for $\psi$ on $\Surft$ will be
positive and must be derived from the nucleation. Nucleation requires
that the shear stress locally exceeds a fault's static strength. It is
expected that earthquakes nucleate at heterogeneities
\cite{Schmitt2015}. The initial value for $z_2$ is also zero. We note
that the initial values for the total traction at $\Surft$ -- that is,
$\sigma_N$ and $\aSt$ -- are implicitly equal to $\nu \cdot T^0$.

\begin{remark}
The slip velocity ahead of the rupture front (that is, on what is
generally regarded as the ``locked'' part of the fault) is in reality
nonzero in contradiction with our assumptions. However, for typical
choices of initial state variable and shear stress, the resulting slip
velocity is very small and several orders of magnitude below even the
plate rate.

The direct effect of rate-and-state friction provides an increase of
frictional strength at the rupture front to a value that can be
thought of as ``static'' friction. The evolution effect then provides
a ``slip-weakening'' type reduction of fault strength. The resulting
plot of strength versus slip has nonzero fracture energy. So it is not
necessary to introduce any additional crack tip physics or procedures
to account for fracture energy; one just selects rate-and-state
parameters to give the desired fracture energy.
\end{remark}

\subsubsection{The quasi-static assumption}

Commonly the Amontons-Coulomb law is assumed to hold, when the
frictional force $\aSt$ and compressive normal stress $\sigma_N$ are
proportional to one another \cite[eq. (4a)]{Ruina1983}):
\begin{equation}\label{eq:Coulomb law}
   |\aSt| = \sigma_N \, f(|\dvt|,\psi) ;
\end{equation}
$f$ is referred to as the \textbf{coefficient of friction}. In
conjunction with this assumption, one can linearize (\ref{eq:state
  ODE}) in the time derivative of the compressive normal stress
\cite[p. 1870]{Rice2001}
\begin{equation}\label{eq:state ODE linear}
   \dot{\psi} = -G_1(\sigma_N,|\dvt|,\psi)
                     - \dot{\sigma}_N G_2(\sigma_N,|\dvt|,\psi) .
\end{equation}
The quasi-static assumption yields a sufficiently slow slip rate and
an almost constant compressive normal stress. A steady state must
satisfy
\[
   G_1(\sigma_N,|\dvt|,\psi_\std(\sigma_N,|\dvt|)) = 0 .
\]
The following expression has been proposed for the function $G_2$
\cite[eq. (8)]{Perfettini2001}:
\[
   G_2(\sigma_N,|\dvt|,\psi) = \gamma(|\dvt|,\psi) / \sigma_N ,
\]
in which $\gamma$ is chosen to be a constant \cite{Linker1992}, or to
be equal to the coefficient of friction, that is, $\gamma(|\dvt|,\psi)
= f(|\dvt|,\psi)$ \cite[eq. (9)]{Perfettini2001}, though the
limitation of this form is assessed in several studies
\cite{Richardson1999, Bureau2000, Bureau2006}.

Under the quasi-static assumption, the slip rate of sliding motion is
sufficiently slow that the inertia of block mass can be neglected. By
fixing the value of $\sigma_N$ (omitting the effect of normal stress
variations), there are further empirical results that lead to the
following assumptions

\begin{enumerate}[label=($\mathfrak{b}_\arabic*$)]
\item there is a characteristic length for a steady-sliding rupture
  evolving into the next steady state after a sudden change of slip
  rate, regardless of the value of slip rate;\label{asm:B1}
\item the instantaneous rate-dependent frictional force is
  approximately proportional to the logarithm of slip
  rate;\label{asm:B2}
\item the steady-state frictional force is approximately proportional
  to the logarithm of slip rate.\label{asm:B3}
\end{enumerate}

Assumption \ref{asm:B1} implies that the slip rate stays constant,
equal to $\tildvt$ say, after a sudden jump. A linearization of
equation (\ref{eq:state ODE linear}) as a perturbation of the steady
state with $\dot{\sigma}_N = 0$ yields \cite[eq. (7)]{Ruina1983}
\begin{equation}
    \DDt \psi = -\left.
        \frac{\partial G_1}{\partial \psi}\right|_{\psi_\std}
                (\psi - \psi_\std) ,
\end{equation}
which has the solution, at $t = L/|\dvt|$ \cite[eq. (8)]{Ruina1983}
\begin{equation}
   \psi(L/|\dvt|) = \psi_\std(\sigma_N,|\dvt|) +
                      (\psi(0) - \psi_\std(\sigma_N,|\dvt|))
   \exp\left( -\frac{L}{|\dvt|} \left.
       \frac{\partial G_1}{\partial \psi}\right|_{\psi_\std} \right) ,
\end{equation}
where $\psi(t) = \psi(\sigma_N,|\dvt|,t)$, $|\dvt| = \tildvt$, and $L$
signifies the slip distance. The characteristic length is defined as
\[
   L_c = \left[ \frac{1}{|\dvt|} \left.
       \frac{\partial G_1}{\partial \psi}\right|_{\psi_\std}
          \right]^{-1} .
\]
After slipping for a distance of $L_c$ under constant normal stress
and slip rate, the state variable evolves towards the steady state by
a ratio $1/e$. Assumption \ref{asm:B1} implies that $L_c$ must be
independent of $|\dvt|$.

With assumption \ref{asm:B2}, the friction law (\ref{eq:Coulomb law})
involves \cite[p. 1873]{Rice2001}
\begin{equation}\label{eq:friction log}
   f(|\dvt|,\psi) = f_0 + 
     a \, \operatorname{log}\left(\frac{|\dvt|}{V_0}\right)
               + \psi ,
\end{equation}
where $f_0$ and $V_0$ are reference values for the coefficient of
friction and slip rate. Assumption \ref{asm:B3} implies that, with $f$
as in (\ref{eq:friction log}), the steady state should take the form
\cite[p. 13,457]{Perfettini2001}
\begin{equation}\label{eq:friction steady log}
   \psi_\std(\sigma_N,|\dvt|) = \psi_\std(|\dvt|)
          = -b \, \operatorname{log}\left(\frac{|\dvt|}{V_0}\right) ,
\end{equation}
so that \cite[eq. (7)]{Perfettini2001}
\begin{equation}\label{eq:friction coeff std}
   f(|\dvt|,\psi_\std(|\dvt|)) = f_0 + 
       (a - b) \, \operatorname{log}\left(\frac{|\dvt|}{V_0}\right) .
\end{equation}
(We note that $f_0$ represents the coefficient of friction at steady
state at slip rate $V_0$.) The sign of $a-b$ characterizes whether the
steady-state dependency is slip-strengthening or slip-weakening. In
the above, the parameters $a, b$ and $L_c$ are independent of
$\sigma_N, |\dvt|$.

The logarithm in (\ref{eq:friction log})-(\ref{eq:friction coeff std})
causes a problem when the slip rate approaches zero. A common remedy
involves appending a small creep rate $V^{\mathrm{creep}}$ to the slip
caused by seismicity, that is, replacing $|\dvt|$ by
\[
   |\dvt|^\mathrm{reg} = |\dvt| + V^\mathrm{creep}
\]
\cite[p. 4]{Woodhouse2015}, \cite[p. 6]{Schmitt2015}\footnote{
  Typically $V^{\mathrm{creep}}$ is taken to be uniformly
  distributed.}. A further regularization of the family of fricton
laws, extending them to negative slip rate, is achieved by replacing
$f(|\dvt|,\psi)$ in (\ref{eq:Coulomb law}) by \cite[p. 1875]{Rice2001}
\begin{equation}
   f^\mathrm{reg}(|\dvt|,\psi) = a \, \operatorname{arcsinh}
    \Big(\tfrac12\exp\big(f(|\dvt|^\mathrm{reg},\psi)/a\big)\Big) ,
\end{equation}
which is based on the asymptotic behavior at large $|\dvt|$, and
(\ref{eq:friction log}) by
\begin{equation}\label{eq:friction asinh}
   f^\mathrm{reg}(|\dvt|,\psi) = a \, \operatorname{arcsinh}
   \left(\frac{|\dvt|^\mathrm{reg}}{2V_0}
            \exp\left(\frac{f_0+\psi}{a}\right)\right) .
\end{equation}
Indeed, now $f^\mathrm{reg}$ vanishes at $|\dvt|^\mathrm{reg} = 0$ as
it should. In this case, the creep velocity takes the role of
maintaining an initially quasi-static status before the rupture
propagates, in which the frictional force balances the traction on the
fault \cite[p. 5]{Ampuero2008} everywhere except at the nucleation
patch where the initial shear stress (that is, tangential component of
the traction) exceeds the initial strength given by the friction law,
and causes an abrupt increase in slip velocity
\cite[p. 5]{Kaneko2008a}.

\newpage

\appendix
\section{Identities for surface operations}\label{app:surf}

We recall that $\nu\colon S \to\RR^n$ denotes the unit normal vector
of a sufficiently regular surface $S$ in $\RR^n$, directed from its
$-$ to its $+$-side, see \eqref{eq:normal}. The {jump} across $S$ of
any vector(space)-valued function $f$ defined in a neighborhood of the
surface is given by $[f]_-^+:=f^+-f^-$, see \eqref{eq:jump}.

The jump operator commutes with any linear operator $P$ which is continuous across $S$, that is, $P^+=P^-$, since
$[Pf]_-^+=(Pf)^+-(Pf)^-=P^+f^+-P^-f^-=Pf^+-Pf^-=P(f^+-f^-)=P[f]^+_-$:
\begin{equation}
[P]_-^+=0\qquad\Longrightarrow\qquad [Pf]_-^+=P[f]_-^+.
\end{equation}
Moreover, as $[fg]_-^+=(fg)^+-(fg)^-=f^+(g^+-g^-)+(f^+-f^-)g^-$ the jump of the product of $f$ and $g$ satisfies the following {\bf Leibniz rule}:
\begin{equation}\label{eq:Leibnizjump}
[fg]_-^+=f^+[g]_-^++[f]_-^+g^-.
\end{equation}
Similarly, one deduces that
\begin{equation}\label{eq:Leibnizjump2}
   [fg]_-^+ = \{f\} [g]_-^+ + [f]_-^+ \{g\} ,
\end{equation}
where $\{f\} := \frac{1}{2} (f^+ + f^-)$ denotes the mean value of $f$
on $S$. If $f$ takes values in $\RR^{\ldots\times n}$ then
\begin{equation}\label{eq:tangpart}
f^\parallel:=f\cdot ( 1_{n\times n}-\nu\nu)=f-(f\cdot\nu)\nu
\end{equation} 
denotes its {\bf tangential part}, which is defined such that 
\begin{equation}
f^\parallel\cdot\nu=0.
\end{equation}
Consequently, if the normal part $(f\cdot\nu)\nu$ of $f$ is continuous across $S$, the jump involves the tangential part only. Hence,
\begin{equation}\label{eq:Jswap}
[f]_-^+\cdot\nu=0\qquad\Longrightarrow\qquad [f]_-^+=[f^\parallel]_-^+=([f]_-^+)^\parallel
\qquad\text{and}\qquad [\snabla f]_-^+=\snabla [f]_-^+.
\end{equation} 
Here $\snabla$ is the {\bf surface gradient}, given by the derivative in directions tangential to $S$, that is, 
\begin{equation}\label{eq:surfdiff}
\snabla f:=(\nabla f)^\parallel=(\nabla f)\cdot ( 1_{n\times n}-\nu\nu)=\nabla f-(\nabla f\cdot\nu)\nu.
\end{equation}
By definition
\begin{equation}
\snabla f\cdot\nu=0.
\end{equation}  
If $a,b$ are $\RR^n$-valued fields we have
\begin{equation}
\snabla f\cdot a=\snabla f\cdot (a^\parallel+(a\cdot\nu)\nu)=\snabla f\cdot a^\parallel
\end{equation}
and, as the {\bf Weingarten map} $\snabla\nu$ is a symmetric bilinear form projecting on the tangent space of $S$, 
\begin{equation}
a\cdot\snabla \nu\cdot b=a\cdot\snabla \nu\cdot b^\parallel=b^\parallel\cdot\snabla \nu\cdot  a
=b^\parallel\cdot\snabla \nu\cdot  a^\parallel=b\cdot\snabla \nu\cdot a.
\end{equation}
From 
$$
0=\snabla 0=\snabla(a^\parallel\cdot\nu)=\nu\cdot\snabla a^\parallel+a^\parallel\cdot\snabla\nu
$$  
we then also get the identity
\begin{equation}\label{eq:Wswap} 
\nu\cdot\snabla a^\parallel=-a^\parallel\cdot\snabla\nu
\quad(\:=-a\cdot\snabla\nu=-\snabla\nu\cdot a=-\snabla\nu\cdot a^\parallel\:).
\end{equation}
The {\bf surface divergence} is the trace of the surface gradient when contracting the last two indices:
\begin{equation}\label{eq:surfdiv}
\sdiv f:=\snabla\cdot f=\tr(\snabla f)=\nabla\cdot f-(\nabla f\cdot\nu)\cdot\nu.
\end{equation}
If $f$ is $\RR^n$-valued, the last term is $\nu\cdot\nabla f\cdot\nu$.  
Finally,
\begin{equation}
\div f=\nabla\cdot f=\snabla\cdot f^\parallel+(\nabla f\cdot\nu)\cdot\nu+(f\cdot\nu)(\snabla\cdot\nu),
\end{equation}
where $\snabla\cdot\nu$ is the sum of the $n-1$ {\bf principal curvatures} of $S$.

\section{Classical differential and variational calculus on Banach spaces}
\label{app:var}

Variational calculus in field theories is based on nonlinear functionals on function spaces given in the form of an integral as in (\ref{eq:Ham_action}) and crucial aspects of stationarity can be expressed in terms of derivatives of the functional on these infinite dimensional spaces of fields. Although the latter are often equipped with a natural inner product, in intermediate steps we also consider function spaces with norms that do not stem from an inner product,   e.g., (local) $L^p$-norms in Lemma \ref{lem:poisson1}, (iii),  and Lemma \ref{lem:poisson1} describing regularity properties of gravitational potential, or $L^\infty$ and Lipschitz norms in the description of admissible motions in Definition \ref{def:Am} and  the basic configuration spaces in Definition \ref{def:nonlin_reg}. Furthermore, the setting of classical differential calculus is a concept of normed spaces, independent of an inner product, and therefore we briefly review basic concepts of differential calculus on Banach spaces following \cite{BlBr:92}. Let $E_1$, $E_2$ be Banach spaces, $\emptyset\neq M\sbs E_1$ be open, 
and let $L(E_1,E_2)$ denote the set of continuous linear operators $E_1\to E_2$.

\begin{definition}[{{\bf Fr\'echet and G\^ateaux derivative}}]\label{def:FrechetGateaux}
Let $f\colon M\to E_2$.
\begin{enumerate}[label=(\roman*)]
\item The Fr\'echet derivative (strong derivative) of $f$ at $y\in M$ is an operator $Df(y)\in L(E_1,E_2)$ such that 
$$Df(y)(h)=f(y+h)-f(y)+o(y,h)$$
for all $h\in E_1$ with $y+h\in M$ and with $o(y,\:.\:)\colon E_1\to E_2$ such that $o(y,0)=0$ and 
$$\lim_{\norm{h}{E_1}\to 0}\frac{\norm{o(y,h)}{E_2}}{\norm{h}{E_1}}=0.$$
If $Df(y)$ exists, $f$ is called (Fr\'echet) differentiable at $y$ and if this is true for all $y\in M$ then $f$ is differentiable on $M$. 
If $Df\colon M\to L(E_1,E_2)$ is continuous, $f$ is called continuously differentiable on $M$.

\item The G\^ateaux derivative (weak derivative, directional derivative) of $f$ at $y\in M$ in the direction $h\in E_1$ is defined by 
$\de f(y,h):=\Phi'(0)$ with $\Phi\colon(-\veps_0,\veps_0)\to E_2$,
$\Phi(\veps)=f(y+\veps h)$ for $\veps_0>0$, that is,
$$\de f(y,h)=\frac{d}{d\veps}(f(y+\veps h))|_{\veps=0}=\lim_{\veps\to 0}
\frac{f(y+\veps h)-f(y)}{\veps}.$$
If $\de f(y,.)\in L(E_1,E_2)$, then $f$ is called G\^ateaux differentiable on $M$.
\end{enumerate}
\end{definition}

The Fr\'echet derivative is stronger than the G\^ateaux derivative in the sense that if $Df$ exists then also $\de f$ exists and $\de f=Df$. 
Conversely, if the G\^ateaux derivative enjoys additional regularity then it is actually a Fr\'echet derivative: If $\de f(y,.)$ exists for 
all $y$ in a neighborhood $U$ of $y_0\in M$ and if $y\mapsto\de f(y,.)$ is continuous $U\to L(E_1,E_2)$, then $D f(y_0)$ exists and 
$Df(y_0)=\de f(y_0,.)$ \cite[Lemma 2.3.2, p.\ 47]{BlBr:92}. 

\begin{example}[{{\bf Basic examples of Fr\'echet derivatives}}]\label{rem:frechet_linquad}
$\:$
\begin{enumerate}[label=(\roman*)]
\item Every linear operator $l\colon E_1\to E_2$ is Fr\'echet (and hence G\^ateaux) differentiable with $Dl(y)=l$ (indeed $l(y+h)=l(y)+l(h)$ 
for all $y$ and $h\in E_1$).

\item On a Hilbert space $H$ with inner product $\lara{.|.}$ and associated norm $\|.\|$, every quadratic form $Q\colon H\to\CC$, 
$Q(y)=\lara{Ay|y}$ with a bounded linear operator $A\colon H\to H$ is Fr\'echet differentiable on $H$ with $DQ(y)(h)=\lara{Ay|h}+\lara{Ah|y}$, 
since $Q(y+h)=\lara{A(y+h)|y+h}=\lara{Ay|y}+\lara{Ay|h}+\lara{Ah|y}+\lara{Ah|h}$, $|\lara{Ah|h}|/\norm{h}{}\leq\norm{A}{\op}\norm{h}{}\to 0$ as 
$\norm{h}{}\to 0$, and $h\mapsto\lara{Ay|h}+\lara{Ah|y}$ is bounded and linear.
    
\item For a functional $\func\colon M\to \CC$ we obtain $D\func(y_0)\in E_1':=L(E_1,\CC)$ (the normed dual space of $E_1$) for every $y_0\in M$.

\item The Fr\'echet derivative $Df$ of a vector-valued function $f\colon M\to\RR^m$ on $M\sbs\RR^n$ coincides with the linear operator 
corresponding to the Jacobi matrix of $f$, that is $Df(y_0)(h)=Df(y_0)\cdot h$ with $Df=(\d_j f_i)_{i,j=1}^{m,n}$ (identifying linear 
operators with matrices).
\end{enumerate}
\end{example}

\begin{definition}[{{\bf Regular and stationary points}}]\label{def:stationary}
Let $f\colon M\to E_2$ be differentiable at $y_0\in M$. If $Df(y_0)\colon E_1\to E_2$ is surjective, $y_0$ is called a regular point of $f$. 
Otherwise, $y_0$ is called a stationary (or critical) point of $f$. In particular, a point $y_0\in M\sbs E_1$ is a stationary point of a functional 
$\func\colon M\to\RR$ if $D\func(y_0)=0$ in $E_1'$, that is
\begin{equation}\label{eq:stationary}
D\func(y_0)(h)=0\quad\textrm{for all}\quad h\in E_1.
\end{equation}
\end{definition}

The Lagrange multiplier theorem gives necessary conditions for stationary points of a functional $\func\colon M\to\RR$ upon restriction to a subset of $M$ specified by a constraint condition (cf.\  \cite[Theorem 4.3.3, p.\ 74]{BlBr:92}). Let $g\colon M\to E_2$ and $y_0\in M$. Then $y_0$ is called a \textbf{stationary point} of $\func$ \textbf{subject to the constraint} $g=0$ on $M$, if $y_0$ is a stationary point for the 
restricted functional $\func|_{g^{-1}(\{0\})}$. Note that continuity of $g$ implies that $g^{-1}(\{0\})$ is a closed set.

\begin{theorem}[{{\bf Lagrange multiplier theorem}}]\label{rem:LagMult} 
Let $\func\colon M\to\RR$ and $g\colon M\to E_2$ be Fr\'echet differentiable and $y_0\in M$ be a stationary point of $\func$ subject 
to the constraint $g(y)=0$ ($y\in M$). If $y_0$ is a regular point of $g$ and $\ker(Dg(y_0))$ has a topological complement\footnote{that is, there exists a closed subspace $V \subseteq E_1$ such that $E_1 = \ker(Dg(y_0)) \oplus V$; note that $\ker(Dg(y_0))$ is closed, since $Dg(y_0) \in L(E_1,E_2)$ by the requirement of Fr\'echet differentiability.} in $E_1$, 
then there exists a continuous linear functional $\la\colon E_2\to\RR$, called Lagrange multiplier, such that $D\func(y_0)=\la\circ Dg(y_0)$, 
that is, $y_0$ is a stationary point of the functional $\func^\la:M\to\RR$, $\func^\la:=\func-\la\circ g$.
\end{theorem}

\begin{remark}
The topological complement condition (on the closed subspace) $K := \ker(Dg(y_0))$ is automatically satisfied if $E_1$ is a Hilbert space, since the orthogonal complement $K^\bot$  is closed and we always have $E_1=K\oplus K^\bot$.
\end{remark}

In the classical calculus of variations (\cite{GiHi:96i}) one
considers functionals $\func$ given by an integral in the
form\footnote{In favour of coherence with the notation $\dvol(x)$ for
  the standard Lebesgue measure used in the main sections -- to
  emphasize ``volume integral'' versus surface integrals -- we accept
  here an unfortunate clash of notation with the basic domain being
  denoted by $V$.}
\begin{equation}\label{eq:func}
\func(y):=\int_V F(x,y(x),Dy(x))\:\dvol(x)
\qquad \textrm{for}\quad y\in\cC^1(\ovl{V})^m,
\end{equation}
where $V\subseteq\RR^n$ is nonempty, open, and bounded and the Lagrangian $F\in\cC^1(U)$ is defined on the open set $U\sbs\RR^n\times\RR^m\times\RR^{m\times n}$. The space $\cC^1(\ovl{V})^m$
equipped with the norm $\norm{y}{\cC^1(\ovl{V})^m}:=\max_{x\in \ovl V}(\norm{y(x)}{\RR^m}+\norm{Dy(x)}{\RR^{m\times n}})$
is a Banach space. The set of all {\bf admissible} points
\begin{equation}
W:=\Big{\{}y\in\cC^1(\ovl{V})^m:\:\{(x,y(x),Dy(x))\:|\:x\in\ovl{V}\}\sbs U\Big{\}},
\end{equation}
that is, functions $y$ whose 1-graph is contained in $U$, is open in $\cC^1(\ovl{V})^m$.
Thus $\func:W\to\RR$ is a functional defined on an open subset of a Banach space.

\begin{definition}[{{\bf First variation}}]
The first variation of $\func$ at $y_0\in W$ in the direction $h\in\cC^1(\ovl{V})^m$ is defined as the
G\^ateaux derivative $\de \func(y_0,h)=\frac{d}{d\veps}(\func(y_0+\veps h))|_{\veps=0}$.
\end{definition}

The first variation of $\func$, given by (\ref{eq:func}), exists in the sense of Definition \ref{def:FrechetGateaux}.
Indeed we have $\de\func(y, h)=\Phi'(0)$ with the continuously differentiable function $\Phi:(-\veps_0,\veps_0)\to E_2$,
$\Phi(\veps):=\func(y+\veps h)$, defined for small $\veps$ (cf.\ \cite[p.\ 12]{GiHi:96i}).
Continuity and linearity of the functional $ h\mapsto\de\func(y, h)$  on $\cC^1(\ovl{V})^m$ will be seen directly from formula (\ref{eq:firstvar_calc}) below.
We introduce the short-hand notation 
$D^\al F_\veps:=(D^\al F)\circ\eta_\veps:\ovl V\to\RR$
for $\al\in\NN_0^n$, $\veps\in\RR$,
and the mapping $\eta_\veps:=(\Id_{\ovl V},y_0+\veps h,Dy_0+\veps Dh)\colon\ovl V\to\ovl V\times\RR^m\times\RR^{m\times n}$.
In particular, for $x\in\ovl V$,
\begin{equation}\label{eq:evalLag}
F_0(x)=F(x,y_0(x),Dy_0(x)) \quad \text{and}\quad
(D^\al F_0)(x)
=( D^\al F)(x,y_0(x),Dy_0(x)).
\end{equation}
The first variation of (\ref{eq:func}) at $y_0\in W$ in the direction $h\in\cC^1(\ovl{V})^m$ then reads
\begin{eqnarray}
\de\func(y_0,h)&=&\frac{d}{d\veps}\Big(\func(y_0+\veps h)\Big)\Big|_{\veps=0}
=\frac{d}{d\veps}\int_VF_\veps(x)\:\dvol(x)\:\Big|_{\veps=0}\nn\\
&=&\int_V\Big(\sum_{i=1}^m(\d_{y_i}F_0)h_i
+\sum_{i=1}^m\sum_{j=1}^n(\d_{\d_jy_i}F_0)\d_jh_i\Big)(x)\:\dvol(x).
\label{eq:firstvar_calc}
\end{eqnarray}
Indeed, interchanging differentiation and integration, that is,
$\frac{d}{d\veps}\int_VF_\veps(x)\:\dvol(x)=\int_V\d_\veps F_\veps(x)\:\dvol(x)$, is justified, since we have that $F_\veps\in L^1(V)$ with derivative $\d_\veps F_\veps$ existing almost everywhere in $V$ and for all $\veps$ near zero, and satisfying
$|\d_\veps F_\veps|\leq g$ with $g\in L^1(V)$ independent of $\veps$ (see \cite[Satz 5.7, p.\ 147]{Elstrodt:07}). Then by the chain rule,
$(\d_\veps F_\veps)(x)=\d_\veps(F\circ\eta_\veps)(x)
=(DF)(\eta_\veps(x))\cdot(\d_\veps\eta_\veps)(x)
=\sum_{i=1}^m(\d_{y_i}F)(\eta_\veps(x))h_i(x)+\sum_{i=1}^m\sum_{j=1}^n(\d_{\d_jy_i}F)(\eta_\veps(x))(\d_jh_i)(x)$
for $x\in\ovl V$. Finally, setting $\veps=0$ yields the form of the integrand above.

In accordance with Definition \ref{def:stationary}, a point $y_0\in W$ is called a \textbf{stationary point of $\func$}, if 
\begin{equation}\label{eq:firstvar_stat}
\de\func(y_0, h)=0\quad\textrm{for all}\: h\in\cC^1(\ovl{V})^m.
\end{equation}
In the context of classical calculus of variations one often considers variations with fixed boundary values of $y$.
In this case it suffices to require $\de\func(y_0, h)=0$ for all $ h\in\cC^1_0(\ovl{V})^m$. By density of $\cC^\infty_c({V})^m$ in
$\cC^1_0(\ovl{V})^m$ this is equivalent to $\de\func(y_0, h)=0$ for all $ h\in\cC^\infty_c({V})^m$.

The {\bf EL (Euler-Lagrange equations)} are necessary conditions for
stationary points which are twice continuously differentiable. If the
(open and bounded) domain $V$ has some additional regularity, then
also {\bf NBC (natural boundary conditions)} can be deduced. Both EL
and NBC follow from the formula (\ref{eq:firstvar_calc}) using the
divergence theorem (Lemma \ref{lem:divthm}) and applying the

\begin{theorem}[{\bf Fundamental lemma of calculus of variations}]\label{lem:fundlemma}

(i) Let $V\sbs\RR^n$ be open and $f\in L^1_\loc(V)$. Then $\int_Vf(x)\eta(x)\:\dvol(x)=0$
for all $\eta\in\cC_c^\infty(V)$ is equivalent to $f(x)=0$ for almost all $x\in V$.
In particular, if $f\in\cC^0(V)$ we may conclude that $f(x)=0$ for all $x\in V$.

(ii) Let $S\sbs\d V\sbs\RR^n$ be a subset of a $\Lip$-boundary and $f\in H^{-\frac{1}{2}}(S)$.
Then $\int_S f(x)\eta(x)\:\dsurf(x)=0$ for all $\eta\in\cC^1(\ovl{V})$ is equivalent to $f=0$
in  $H^{-\frac{1}{2}}(S)$. In particular, if $f\in L^2(S)$  we may conclude that $f=0$ almost everywhere on $S$. If $f\in \cC^0(S)$, we may conclude that $f(x)=0$ for all $x\in S$.

\end{theorem}

\begin{proof} Part (i) is a standard result in distribution theory (\cite[Theorem 1.2.5]{Hoermander:V1}).
For (ii) note that since $\cC^1(\ovl{V})\sbs H^1(V)$ is dense, $\int_S f\eta\:\dsurf=0$ for all $\eta\in\cC^1(\ovl{V})$,
implies $\int_S f\:T\eta\:\dsurf=0$ for all $\eta\in H^1(V)$, with $T:H^1(V)\to H^{\frac{1}{2}}(\d V)$ the trace operator.
Since $T$ is continuous and surjective (\cite{Wloka:82}, Thm.\ 8.8), it follows $f=0$ in $H^{-\frac{1}{2}}(S)$.
If $f\in L^2(S)$ we conclude $f=0$ almost everywhere on $S$. The converse direction is clear.
\end{proof}

Let $\func$ be as in (\ref{eq:func}), with $F$, $U$, $V$, $W$ defined as above.
In addition, let $V$ be a $\Lip$-domain with exterior normal $\nu\colon\d V\to\RR^n$ (defined almost everywhere),
let $F\in\cC^2(U)$, and let $y_0\in W \cap\, \cC^2(V)^m$ be a stationary point of $\func$, that is $y_0$ satisfies
(\ref{eq:firstvar_stat}).  
Then $y_0$ satisfies the EL
\begin{equation}\label{eq:func_EL}
\sum_{j=1}^n\d_j(\d_{\d_jy_i}F_0)-\d_{y_i}F_0=0\quad\textrm{in}\:\: V
\end{equation}
and the NBC 
\begin{equation}\label{eq:func_NBC}
\sum_{j=1}^n(\d_{\d_jy_i}F_0)\:\nu_j=0\quad\textrm{on}\:\: \d V,
\end{equation}
for $1\leq i\leq m$. These equations are deduced as follows: The divergence theorem (Lemma \ref{lem:divthm}) applied to the first variation of $\func$
given by (\ref{eq:firstvar_calc}) yields the {\bf weak EL}
\begin{equation}\label{eq:firstvar_scalc}
\de\func(y_0, h)=\int_V\Big(\sum_{i=1}^m\d_{y_i}F_0
-\sum_{i=1}^m\sum_{j=1}^n\d_j(\d_{\d_jy_i}F_0)\Big) h_i\:\dvol
+\int_{\d V}\sum_{i=1}^m\sum_{j=1}^n(\d_{\d_jy_i}F_0) h_i\:\nu_j\:\dsurf.
\end{equation}
The fundamental lemma of variational calculus applied
to $V$, $\de\func(y_0, h)=0$ for all
$ h\in\cC^\infty_c(\ovl{V})^m\sbs\cC^1(\ovl{V})^m$ implies the EL (\ref{eq:func_EL}) in $V$ (the surface integral
vanishes identically since $ h|_{\d V}=0$). Furthermore, requiring stationarity for all $ h\in\cC^1(\ovl{V})^m$ yields the EL as well as the
condition $\int_{\d V}\sum_{i=1}^m\sum_{j=1}^n(\d_{\d_jy_i}F_0) h_i\:\nu_j\:\dsurf=0$ for all $ h\in\cC^1(\ovl{V})^m$.
By the fundamental lemma applied to $\d V$ we deduce the NBC (\ref{eq:func_NBC}). 

The {\bf NIBC (natural interior boundary conditions)} are deduced
analogously based on the divergence theorem for $\Lip$-composite
domains and for surfaces, Lemma \ref{lem:divthmComposite} and Lemma
\ref{lem:surfdivthm} respectively.

\bibliographystyle{siam}
\bibliography{brazda_ref}

\end{document}